\documentclass[aps,prx,reprint]{revtex4-2}
\usepackage{blindtext}
\usepackage{amsmath}
\usepackage{amsfonts}
\usepackage{amsthm}
\usepackage{physics}
\usepackage{graphicx}
\usepackage{changes}
\usepackage{makecell}
\usepackage{mathtools}
\usepackage{bm}

\usepackage[table]{xcolor}

\usepackage{booktabs}
\usepackage{multirow} 
\usepackage{diagbox}

\usepackage{xcolor}
\usepackage[
    colorlinks=true,
    hypertexnames=false
    ]{hyperref}

\hypersetup{
  colorlinks   = true,
  urlcolor = magenta!90!black,
  linkcolor = blue!60!black,
  citecolor=black!60
}

\usepackage{url}
\usepackage[all]{hypcap}

\usepackage[capitalise,compress]{cleveref}

\usepackage{enumitem}   
\usepackage{tikz} 
\usetikzlibrary{arrows, shapes.gates.logic.US, calc}
\usetikzlibrary{shapes}
\usetikzlibrary{plotmarks}
\usetikzlibrary{quantikz}
\usetikzlibrary{patterns}

\newcommand{\imag}{\mathrm{i}}
\DeclareMathOperator{\Var}{Var}

\DeclareMathOperator{\poly}{poly}

\newcommand{\E}{\mathbb{E}}
\definecolor{hqred}{HTML}{EE1C25}

\usepackage{bbm}

\newtheorem{theorem}{Theorem}

\newtheorem{problem}{Problem}

\newtheorem{algorithm}{Algorithm}

\crefname{protocol}{Protocol}{Protocols}
\Crefname{protocol}{Protocol}{Protocols}

\newlength{\protowidth}

\usepackage{algpseudocode}

\usepackage{graphicx}

\newcommand{\apjl}{The Astrophysical Journal Letters}
\newcommand{\mnras}{Monthly Notices of the Royal Astronomical Society}

\let\originalleft\left
\let\originalright\right
\renewcommand{\left}{\mathopen{}\mathclose\bgroup\originalleft}
\renewcommand{\right}{\aftergroup\egroup\originalright}

\usepackage{xcolor}
\definecolor{someonesgreen}{RGB}{205, 222, 194}
\definecolor{hqred}{HTML}{EE1C25}

\begin{document}

\title{Measuring gravitational lensing time delays with quantum information processing}
\author{Zhenning~Liu,$^{1,2}$ William~DeRocco,$^{3,4}$ Shiming~Gu,$^{5}$ Emil~T.~Khabiboulline,$^{1,6}$ Soonwon~Choi,$^7$ Andrew~M.~Childs,$^{1,2,8}$ Anson~Hook,$^3$ Alexey~V.~Gorshkov,$^{1,6}$ and Daniel~Gottesman$^{1,2}$\\
\textit{
\footnotesize $^{1}$ Joint Center for Quantum Information and Computer Science, NIST/University of Maryland, College Park, MD 20742, USA\\
\footnotesize $^{2}$ Department of Computer Science, University of Maryland, College Park, MD 20742, USA\\
\footnotesize $^{3}$ Maryland Center for Fundamental Physics, University of Maryland,
College Park, MD 20742, USA
\footnotesize $^{4}$ Department of Physics \& Astronomy, The Johns Hopkins University, Baltimore, MD 21218, USA\\
\footnotesize $^{5}$ Department of Physics \& Astronomy, University of British
Columbia, Vancouver, BC V6T 1Z1, Canada\\
\footnotesize $^{6}$ Joint Quantum Institute, NIST/University of Maryland College Park, MD 20742, USA\\
\footnotesize $^{7}$ Department of Physics, Massachusetts Institute of Technology, Cambridge, MA 02139, USA\\
\footnotesize $^{8}$ Institute for Advanced Computer Studies, University of Maryland, College Park, MD 20742, USA\\
}}

\begin{abstract}
    The gravitational fields of astrophysical bodies bend the light around them, creating multiple paths along which light from a distant source can arrive at Earth. Measuring the difference in photon arrival time along these different paths provides a means of determining the mass of the lensing system, which is otherwise difficult to constrain. This is particularly challenging in the case of microlensing, where the images produced by lensing cannot be individually resolved; existing proposals for detecting time delays in microlensed systems are significantly constrained due to the need for large photon flux and the loss of signal coherence when the angular diameter of the light source becomes too large.
    
    In this work, we propose a novel approach to measuring astrophysical time delays. Our method uses exponentially fewer photons than previous schemes, enabling observations that would otherwise be impossible. Our approach, which combines a quantum-inspired algorithm and quantum information processing technologies, saturates a provable lower bound on the number of photons required to find the time delay. Our scheme has multiple applications: we explore its use both in calibrating optical interferometric telescopes and in making direct mass measurements of ongoing microlensing events. To demonstrate the latter, we present a fiducial example of microlensed stellar flares sources in the Galactic Bulge. Though the number of photons produced by such events is small, we show that our photon-efficient scheme opens the possibility of directly measuring microlensing time delays using existing and near-future ground-based telescopes.
\end{abstract}

\maketitle

\tableofcontents

\section{Introduction}

\subsection{Motivation and background}
 The rapidly developing field of quantum information technology has various promising applications. In particular, it has been widely accepted that certain computational problems can be solved more efficiently by quantum computers \cite{shor1999polynomial,grover1996fast}, and certain signals can be measured more precisely using quantum sensing \cite{giovannetti2006quantum}. 
Some quantum sensing protocols provide improvement by reducing the number of photons needed, for instance, but quantum technology can also provide other advantages for sensing.
An example along these lines is the \emph{quantum telescope} array proposed by \citet*{Gottesman2021TelescopesRepeaters}, a long-baseline optical interferometer scheme for astronomical observation. The proposed telescope provides ultra-high angular resolution (micro-arcsec) using error-corrected quantum communication techniques as the key building block. \textcite{Gottesman2021TelescopesRepeaters} stimulated various follow-up research in applying quantum information processing techniques to develop novel strategies for optical imaging and astronomical observation \cite{Khabiboulline2019OpticalInterferometry,Khabiboulline2019TelescopeArrays,liu2024low,brown2023interferometric}.

In this paper, we focus on measuring the time delay between optical signals. This is of great significance in astronomy due to applications to observing \emph{gravitational lensing} \cite{Dyson1920LightDeflection,Einstein1936Lensing} events. 
Gravitational lensing occurs when light from a distant source is bent by the gravitational field of an intervening massive object, leading to the formation of multiple images. As a result of this deflection, different light paths associated with the lensed images correspond to different geometric lengths and geodesics, thereby introducing relative time delays in their arrival times \cite{Refsdal1964Lensing}. These time delays provide a direct and powerful means to measure the mass profile of the lens system, including masses of rogue planets \cite{MiretRoig2023,Koshimoto2023,Mroz2017}, isolated black holes \cite{Sahu2022,Kaczmarek2025}, and even the spatial distribution of dark matter, an unknown component of the mass budget of the universe whose contribution dominates over that of visible matter \cite{Planck2018Results}.

While measurements of lensing time delays are of great scientific value, it is highly challenging to obtain them in practice. Successful measurements have only been done for strong gravitational lensing systems \cite{biggs1999time,H0LiCOW2017Intro,H0LiCOW2019Review}, where ``strong" means that different images of the source are spatially distinguishable. In strong lensing, the time delay is measured by exploiting variability of the source. For instance, a transient astrophysical event, such as a supernova explosion or a quasar flare, appears in each lensed image at different times due to their light path difference. However, in other lensing systems, such as instances of \emph{microlensing}, the different images usually cannot be individually resolved \footnote{Different images in a microlensing system are usually not resolved with seeing-limited imaging; however, long-baseline interferometry (e.g., VLTI/GRAVITY) has now resolved microlensed images in select cases (see e.g. Ref.~\cite{mroz2025observations}). Whereas long-baseline interferometers can resolve the micro-images and infer lens parameters from visibilities/closure phases, our method infers the time delay $\Delta t$ from single-photon spectra without resolving the images}. The time delay of a microlensing event is typically much shorter than the source’s variability timescale, unlike in strong lensing, hence transient signals propagating along two different paths overlap in the recombined light curve, hiding the small time delay between them.
Furthermore, the source in a microlensing system is often selected to be absent of well-defined, observable transient events and to remain stable over long baselines. 
Together, these factors make it infeasible to use techniques for strong lensing to extract time delays in microlensing systems.

Fortunately, there is a different theoretical framework that is designed particularly for measuring microlensing time delays. This framework has multiple observational proposals and one actual implementation. Throughout this paper, we assume there are only two images in the lensing system (this assumption is justified in \cref{sec:lens_intro}) and denote their time delay by $\Delta t$. References~\cite{deguchi1986diffraction,peterson1991gravitational} first realized that a fixed gravitational lensing time delay leads to an oscillatory modulation in the spectrum, with adjacent peaks separated by $1/\Delta t$. This observation laid the theoretical foundation for various follow-up works. Later, Refs.~\cite{gould1992femtolensing,ulmer1995femtolensing,katz2018femtolensing} predicted the existence of lensing events induced by extremely lightweight lensing objects and claimed that their (very short) time delays may create observable frequency-domain oscillations when the photon source is a gamma-ray burst. Similarly, Refs.~\cite{eichler2017nanolensed,jow2020wave,wucknitz2021cosmology} discuss lensing delay measurement in the radio wavelength, where they use fast radio bursts (FRBs) as their photon source. There is even one reported experimental attempt using the FRB-based proposal \cite{kader2022high} which enables constraining the abundance of primordial black holes. 
 
However, none of the above proposals provide successful measurement outcomes for any microlensing system, for multiple reasons. First, the femtolensing observations using gamma-ray bursts and most proposals beyond the radio wavelength suffer from the severe \emph{finite-source effect}. If the source is large, photons from different regions of the source have different $\Delta t$ values. When the time delay uncertainty $\delta_{\Delta t}$ is greater than one period of the carrier frequency, the gravitationally lensed light signal no longer contains any information about the time delay. Such an effect is studied in previous works \cite{eichler2017nanolensed,leung2025wave,sugiyama2020wave} and will also be shown information theoretically in our work. Second, the number of photons required in the aforementioned proposals is generally large, while microlensing systems have, on average, lower signal-to-noise ratio (SNR) than strong lensing systems because their dynamic nature limits opportunities to do stable long exposure. To achieve a high SNR, the integration times are sometimes required longer than the lensing event itself and forces researchers to use extremely luminous sources in their designated wavelengths, such as FRBs in radio wavelengths. However, the number of FRB events per day is (empirically) limited to a small value, and the number of gravitationally microlensed FRB events is even lower.
Finally, when $\Delta t$ is relatively large, say $\sim 1 \, \mathrm{ms}$, the distance between peaks of the spectrum is only $1/\Delta t \sim 1 \, \mathrm{kHz}$. Observing such a pattern requires prohibitively high frequency resolution for carrier frequencies higher than those of radio waves. 

As suggested above, to find an eligible microlensing event and conduct a successful measurement of its time delay, the observation proposal must have (roughly and qualitatively) two ingredients. First, we need a \emph{feasible} and \emph{sample-efficient} delay measurement approach to allow for a longer list of observable sources and enable measurements in a short time window (ingredient (i)). Second, we need a class of very tiny sources to avoid the finite-source effect (ingredient (ii)). While these two ingredients are independent, instantiating either is a great challenge, which makes microlensing time delay measurement an exceptionally difficult problem. In this paper, we  
address ingredient (i) using quantum mechanics and quantum information theory as theoretical tools and single-photon quantum devices as potential experimental platforms. Our novel delay-finding approach extends the list of observable lensing events and enables us to address ingredient (ii).

Time delays also appear in optical interferometric imaging systems (long-baseline telescope arrays). In particular, the spatial separation introduces extra distance that the incoming photon must travel to reach a neighboring site. To perform a successful interference, the ensuing time delays must be matched with precision at the level of wave packet 
duration to allow interference to happen. As such, there is an initial calibration stage where the time delays are tuned. Light is gathered from a small 
bright source in the sky (the \emph{guide star}), close to the object of interest. In practice, artificial sources such as satellites or laser guidestars are used. For this application, the finite-source problem is no longer an issue. Sample-efficient measurements are still vital, to allow rapid calibration, or alternatively, using dimmer sources. Therefore, in this paper, we also provide an efficient solution to the problem of learning time delays in telescope arrays, using the same approach as for microlensing.

\subsection{Our contribution}

In this work, we develop a novel technique that provides ingredient (i). To obtain ingredient (ii), we discuss a class of lensing events whose time delays could be measurable via our technique.

Our delay-finding approach relies on a key observation that every photon emitted in a spherical wave takes both paths created by gravitational lensing to reach the Earth in quantum \emph{superposition}. With this, inspired by the advancement of quantum information science and building upon the intuition of frequency-domain interference from Refs.~\cite{deguchi1986diffraction,peterson1991gravitational}, we propose a concrete sample-efficient delay measurement approach, \cref{alg:alg1}, that provably uses as few photons as possible. In particular, letting $T$ be the upper limit of $\Delta t$, $t_c$ be the coherence time of the photon without lensing (we consider $t_c$ to be an inherent parameter of the photons for now; we will explain how $t_c$ is defined in a realistic observation scenario later in this paper), and assuming $\Delta t\gg t_c$, our method consumes only $O(\log (T/t_c))$ photons to measure $\Delta t$ with precision $t_c$, while traditional proposals require $O(T/t_c)$ photons. Note that we work in the photon-starved regime, hence we expect to receive at most one photon per wave packet, and the wave function of each photon is the superposition of two wave packets separated by $\Delta t$ due to the microlensing effect. We also provide a rigorous proof that $\Omega(\log (T/t_c))$ is the information-theoretic lower bound, hence our method is optimal. One proof is based on modeling the gravitational lensing system as a communication channel and computing its channel capacity. We also exploit a surprising connection between the delay-finding problem and a well-studied problem in quantum computing, the dihedral hidden subgroup problem. We show that the dihedral hidden subgroup problem can be reduced to our problem, giving an alternative optimality proof in terms of both sample complexity and computational complexity.

The key ingredient for the exponential improvement in our scheme is that our algorithm uses quantum information processing technologies (including \emph{single-photon} spectrometers, and, depending on the specific implementation of our scheme, quantum memory and digital quantum computation) to perform single-photon frequency-basis measurements. This allows us to \emph{sample} from a certain distribution determined by the value of $\Delta t$. By feeding these samples into a data-processing procedure inspired by the sample-efficient algorithm for the dihedral hidden subgroup problem \cite{EH00}, we can estimate $\Delta t$ in the style of maximum-likelihood estimation.

Implementing our approach involves measuring the frequency of every photon, which requires a broadband high-resolution spectrometer with single-photon sensitivity. The difficulty of the implementation strongly depends on $T$, the upper limit of $\Delta t$ (which is determined by the lensing object) of our interest, because the required frequency resolution is $\sim 1/\Delta t$. Single-photon spectrometers based on dual-combs feature up to $\sim 100\,\mathrm{MHz}$ resolution \cite{coddington2016dual,picque2020photon,xu2024near,peng2025single} with $\sim 10\,\mathrm{GHz}$ bandwidth. State-of-the-art spectrometers based on a time lens \cite{mazelanik2020temporal,joshi2022picosecond,mazelanik2020temporal} even achieve $20\,\mathrm{kHz}$ resolution, but their bandwidth is limited to MHz-level. These results in principle enable measurements of $10\,\mathrm{ns}$ (dual-comb) or up-to-$0.1\,\mathrm{ms}$ (time lens) time delay, corresponding to lensing objects as heavy as brown dwarfs (dual-comb) or primordial black holes of multiple solar masses (time lens), as is explained later in \cref{sec:lens_intro}. However, one bottleneck for such an observation is that the photon sources we consider are thermal sources emitting broadband signals. The tiny bandwidth of existing single-photon spectrometers may require using prohibitively many of them in parallel. 
Therefore, long-$\Delta t$ measurements may only be achievable through next-generation single-photon spectrometers; 
however, for shorter $\Delta t$, which are also of significant interest to astronomy, resolving the frequency of single photons is much less challenging and can potentially be realized by combining existing technologies.

We also propose another version of the delay-finding approach (\cref{alg:alg2}) which relies on storing and processing the photon wave function in the \emph{time} domain and uses a different data processing procedure. We propose to first perform non-demolition frequency measurement on the received photon to localize it to a frequency range with width denoted by $1/t'_c$ (where $t'_c$ is the \emph{effective coherence time} satisfying $t'_c \geq t_c$). Then, using a quantum information \emph{discarding} process, one can store the photon in a digital quantum computer in an \emph{undersampling} manner in the time domain. More specifically, we propose to use a quantum memory that can only distinguish $O(T/t'_c)$ modes, which is far fewer than $\Theta(T\omega_0)$ modes for sampling at the Nyquist rate (where $\omega_0$ is the carrier frequency of the photon). With this approach, one can employ the quantum Fourier transform (QFT) to produce an aliased frequency as the output, which is fed into \cref{alg:alg2} to find $\Delta t$. The connection between delay finding and the dihedral hidden subgroup problem is established through this time-domain version. We show that this version can potentially be implemented by a linear optics system. Going further, digital quantum computing in principle enables compressed storage of the photonic modes with binary encoding of arrival time, which gives an exponential reduction in the resources required.

Our photon-efficient method enables the estimation of microlensing time delays in the optical and infrared (IR) bands. The measurable delay range spans from $10^{-10}\,\mathrm{s}$ to $10^{-3}\,\mathrm{s}$ depending on the capabilities of the frequency-resolving device. This, in principle, supports observation of many more sources than in radio/gamma wavelengths, and the photons can be lensed by the majority of interesting microlensing systems. However, optical/IR waves oscillate extremely fast, giving a rather stringent requirement on the variance of $\Delta t$ between different photons due to the finite-source effect. This means it is much more difficult to have ingredient (ii). Nevertheless, as part of our solution, we give a concrete use case of our measurement scheme in the optical/IR band that satisfies ingredient (ii), and present a comprehensive analysis of its scientific value and feasibility. Specifically, 
we consider
\emph{flares} of \emph{M-class red dwarfs} (M dwarfs). M dwarfs are relatively tiny and cold stars, and a flare is an event in which a small region of the dwarf becomes almost as bright as the whole dwarf in certain passbands. Our analysis shows that, for a significant fraction of flares in M dwarfs, the size of the light-emitting area may be small enough such that the uncertainty in $\Delta t$, denoted by $\delta_{\Delta t}$, is less than one period of the carrier frequency $\sim 10^{-15}$s. Moreover, our scheme not only enables the study of the lensing object, but also yields constraints on the actual spatial size of the flares in M dwarfs, which is currently poorly understood. Indeed, directly resolving flare kernels on even nearby M dwarfs likely requires $\geq 10\,\mathrm{km}$ optical baselines, which may only be achievable with quantum-assisted optical interferometry \cite{Gottesman2021TelescopesRepeaters}. We also perform a comprehensive analysis of the number of photons we can receive in realistic settings to observe microlensed flares on M dwarfs, taking into account the duration, size, and temperature of the flare, as well as astronomical dust extinction and telescope collecting area. Our result shows that near-term ground-based optical telescopes can achieve sufficiently high signal-to-noise ratio in such observations. Moreover, we improve our algorithms such that photons from temporally and spatially separated flares can be analyzed collectively to infer the average lensing time delay, allowing for potential implementation using existing optical telescopes.

To support the feasibility of our observation scheme, we also analyze the robustness of our approaches against several other potential issues. We prove that our algorithm still works when signal photons (with fixed time delay) are mixed with noise photons (without fixed time delay) in an indistinguishable manner, although more signal photons are needed than in the noiseless case.
We also prove that, although the majority of photons may be lost during transmission due to the interstellar medium, we can guarantee with high probability that the superposition of two paths is preserved in the 
received signal photons provided the dust ``particle" size is much lower than the telescope size.

Finally, we discuss the application of our methods to a different task: the calibration of time delays in telescope arrays. Light traveling from a source along different paths picks up a relative delay when observed at different sites. Learning these time delays is important to enable interference of photons arriving at different telescopes. In order to learn the time delays, we map 
the distributed problem to a lensing-like scenario, where a single detector observes the photons. In particular, we show how to use entanglement to transfer the information across the array to a single site, where we can apply our algorithm. The same compression and storage of photonic information in memory as used in prior work on telescope arrays~\cite{Khabiboulline2019OpticalInterferometry,Khabiboulline2019TelescopeArrays} is applicable here, such that our proposal is compatible with that scheme. The benefits over classical techniques are a replacement of long delay lines with memories that keep track of timing, allowing for longer baselines and thus larger resolution, and improved sample efficiency, as provided by our algorithm.

The remainder of the paper is structured as follows. We provide preliminary information in \cref{sec:prelim}: we give a technical introduction to gravitational lensing and the significance of measuring its time delay (\cref{{sec:lens_intro}}), explain the setup of the delay-finding problem with mathematical and physical rigor (\cref{sec:setup}), and review a traditional time delay measurement approach using a large number of photons (\cref{sec:review_classical}). In \cref{sec:freq_domain_interference}, we describe our sample-efficient delay-finding algorithm: we introduce the frequency-domain interference framework via a classical electromagnetism derivation (\cref{sec:classical_pic}), reproduce the same derivation in the quantum picture for photonic wave functions (\cref{sec:quantum_pic}), use this picture to propose our delay-finding algorithm and analyze its sample complexity (\cref{sec:ouralgo}), explain the consequence of varied $\Delta t$ values (the finite-source effect) and analyze the how our algorithm performs under this effect (\cref{sec:finite_source}), analyze the effect of noise photons (\cref{sec:noisysignal}) and the lensing magnification (\cref{sec:unequal}), and finally discuss the realistic scenario of broadband input photons which induces a short coherence time (\cref{sec:broadband}). In \cref{sec:lower_bound}, we prove the information-theoretic lower bound for the sample complexity, matching the actual sample complexity of our algorithm. In \cref{sec:quantum_solution}, we propose the undersampling algorithm using the quantum Fourier transform on a digital quantum computer, and discuss its connection to the dihedral hidden subgroup problem (\cref{sec:dhsp_discussion}), giving another proof of the optimal sample complexity as well as computational complexity. In \cref{sec:experiment}, we discuss possible experimental realizations of our algorithms. In \cref{sec:dwarfflare}, we present the astronomical observation plan: we carry out a case study for the example setup for microlensed M dwarf flares with an analysis of its feasibility and scientific value (\cref{sec:flare_setup} and a narrow-band version in \cref{sec:narrow}); we then introduce a modified version of our algorithm to combine photons from different flares and present numerical simulation results (\cref{sec:combineflares}).
In \cref{sec:robustness}, we perform a thorough analysis for the robustness of our approach to noises due to the medium between the source and the telescope: we investigate the effect of dust extinction (\cref{sec:dustExtinction} and a detailed proof in \cref{appendix:dust}), astronomical scintillation (\cref{sec:refractiveindex}), and atmospheric fluctuation (\cref{sec:atmos}). In \cref{sec:forarrays}, we discuss the application of our delay-finding protocol to calibrating quantum telescope arrays. Finally, we summarize our work and discuss open problems in \cref{sec:discussion}.

\section{Preliminaries}
\label{sec:prelim}
In this section, we provide background information for the rest of the paper. We first give a brief introduction to gravitational lensing and derive the corresponding time delay $\Delta t$ and its variation due to finite source size (\cref{sec:lens_intro}). Next, we describe both classical and quantum descriptions of the delay-finding problem (\cref{sec:setup}). Finally, we review one straightforward approach to measure the time delay based on Mach-Zehnder interferometry, which consumes $O(T/t_c)$ photons (\cref{sec:review_classical}).

Throughout this paper, we adopt a few non-SI units that are standard in astronomy. Specifically, we use $1\,\mathrm{pc} = 3.0857\times 10^{16}\,\mathrm{m}$ and $1\,\mathrm{erg} = 10^{-7}\,\mathrm{J}$.

\subsection{Gravitational lensing time delays}
\label{sec:lens_intro}

When light emitted by a distant source passes near a large distribution of mass on its way to Earth, the path of the light is altered in an effect known as gravitational lensing.
For sufficiently large masses, the distortion is large enough to form multiple images of the same source on the sky. However, at lower masses, these images cannot be individually resolved---they overlap with each other and with the true position of the source, causing the source to appear brighter. This transient magnification of a source is known as \textit{gravitational microlensing} \cite{Paczynski1986}. Such microlensing events provide one of the few ways to detect non-luminous astrophysical bodies.

For the purposes of this paper, we will restrict ourselves to a fiducial example, namely a star in the Galactic Bulge of the Milky Way as the source and a dark, isolated object such as a black hole as the lens. This example is of particular interest, as the Galactic Bulge has been the target of decades of ground-based microlensing searches \cite{Udalski2015,Nunota2025,Park2018}, and such surveys have yielded significant discoveries of dark objects such as isolated black holes \cite{Lam2022,Sahu2022,Kaczmarek2025} and free-floating planets \cite{Sumi2011,Mroz2017,Mroz2018,Mroz2019,Mroz2020a,Mroz2020b,Kim2021,Mroz2023,Mroz2024,Koshimoto2023,Gould2022,Gould2023}. In the next few years, future missions such as NASA's Nancy Grace Roman Space Telescope (henceforth, \textit{Roman}) \cite{Akeson2019,Penny2019} and the Chinese National Space Agency's Earth 2.0 satellite \cite{Ge2022} will conduct the first ever dedicated space-based microlensing surveys. They are expected to discover orders of magnitude more dark astrophysical bodies than current ground-based surveys \cite{Penny2019,Johnson2020}.

However, as powerful a tool as microlensing is for discovering dark astrophysical objects, it suffers from inherent degeneracies that make measuring the underlying properties of the lens, such as its mass, quite challenging \cite{Lee2017}. This can be seen from the fact that the primary observable associated with microlensing is a quantity called the \textit{Einstein crossing time} that corresponds to the approximate duration of the lensing event. The Einstein crossing time is defined as the time for the source to cross the angular Einstein radius $\theta_E$, which is the region of the sky surrounding the lens in which the source is magnified by the lens. The Einstein crossing time $t_E$ therefore depends on the distance to the lens ($D_L$), the distance to the source ($D_S$), the mass of the lens ($M$), and the relative proper motion of the lens and source $\mu_\text{rel}$:
\begin{equation}
\label{eq:tE}
    t_E = \frac{\theta_E}{\mu_{\text{rel}}}
\end{equation}
with
\begin{equation}
    \theta_E = \sqrt{\frac{4 G M (1 - D_{L}/D_{S})}{ D_{L} \, c^2}}.
\end{equation}
It is clear from these equations that, if only $t_E$ is measured, there is an inherent degeneracy between the lens mass $M$, lens distance $D_L$, and relative transverse velocity ($v_T = \mu_\text{rel} D_L$). Breaking this degeneracy is challenging, particularly for isolated, dark objects such as black holes, neutron stars, and free-floating planets.

\begin{figure*}
    \centering
    \includegraphics[width=\linewidth]{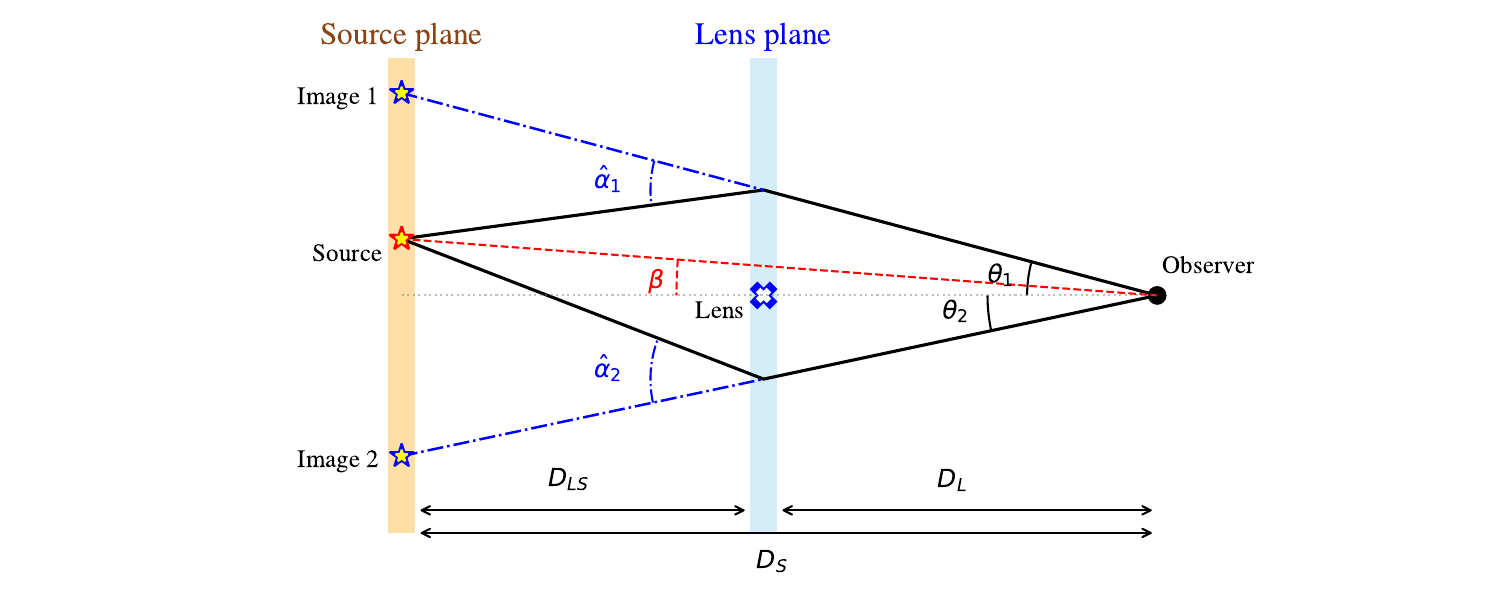}
    \caption{The simplified diagram of a single-lens system showing the light deflection under the gravity of the lens. In the diagram, the black solid line shows a small-angle-approximated light path from the source to the observer, $\theta_1$ and $\theta_2$ are the position angle of the lensed image, $\beta$ is the position angle of the source, and $\hat \alpha_1$ and $\hat \alpha_2$ are the deflection angles of the light paths.}
    \label{fig:lensingdiagram}
\end{figure*}

One means of breaking this degeneracy is through time delays. Since the paths taken by light from the source differ for the two  
images, there is a relative path-length difference between them. Working in a coordinate system in which the lens lies along the axis, we define $\vec{\beta}$ as the angular position of the source on the sky and $\vec{\theta}$ as the angular position of the corresponding image on the sky (see \cref{fig:lensingdiagram}).
Given this, we can compute the propagation time difference between paths as 
\begin{equation}
\label{eq:timingsurface}
    t = \frac{D_L D_S}{D_{LS}c} \tau ~~~\text{with}~~~ \tau \coloneqq \frac{1}{2}(\vec{\theta}-\vec{\beta})^2-\psi(\vec{\theta}),
\end{equation}
where---on sub-cosmological scales---$D_{LS} = D_S - D_L$ is the source-lens distance,
and $\psi(\vec{\theta})$ is the lensing potential, which reduces to $\psi(\vec{\theta}) = \theta_E^2 \ln|\vec{\theta}|$ for a point lens lying along the axis. A single derivative of this quantity yields the standard point-source lensing equation
\begin{equation}
    \vec{\theta} - \vec{\beta} = \frac{\theta_E^2}{|\vec{\theta}|},
\end{equation}
which is rotationally symmetric due to the geometry of the problem. Solving this equation for the major (+) and minor ($-$) image positions as a function of source position yields
\begin{equation}
    \theta_{\pm} = \frac{1}{2}(\beta \pm \sqrt{\beta^2 + 4 \theta_E^2}).
\end{equation}
Plugging this solution back into Eq.~(\ref{eq:timingsurface}) and defining the impact parameter $u \coloneqq \frac{\beta}{\theta_E}$ yields
\begin{equation}
    \Delta t = t_+ - t_- = \frac{4 G M}{c^3} f(u),
\end{equation}
where we have defined
\begin{equation}
\label{eq:fu}
    f(u) \coloneqq \left[\frac{1}{2} u \sqrt{u^2 +4} + \ln\left(\frac{\sqrt{u^2 +4}+u}{\sqrt{u^2+4} -u}\right)\right].
\end{equation}
The impact parameter can be independently measured from the magnification curve, as the magnification in the point-source point-lens regime is \cite{Refsdal1964Lensing}
\begin{equation}
\label{eq:Au}
    A(u) = \frac{u^2 +2}{u\sqrt{u^2 +4}},
\end{equation}
where $u$ varies with time as the source location approaches its minimum impact parameter $u_0$ and then moves away from the lens axis. We see that if $u_0 = 0$, the magnification is infinite, as expected \cite{Paczynski1986}. Given the magnification curve and the time delay, it is clear that one can solve for $M$, the lens mass, directly, breaking the inherent degeneracy present in the light curve alone. 

The behavior of Eqs.~(\ref{eq:fu}) and (\ref{eq:Au}), as plotted in Fig.~\ref{fig:fu}, constrains the regime in which microlensing is useful. For a microlensing event to be detectable, we normally require $u \approx 1$ (though with future space-based missions, this may be relaxed). Additionally, note that Eq.~(\ref{eq:Au}) is the total magnification, i.e., the sum of the major and minor image contributions given by $A_\pm (u) = \frac{1}{2}A(u)\pm \frac{1}{2}$.
The analysis is simpler when
there is comparable flux from the two images, i.e., $u \approx 1$.
For both of these reasons, we mainly consider examples in which $u\approx1$, hence $\Delta t \approx \frac{8 GM}{c^3}$ and $A \approx 1.34$ (a $34\,\%$ brightening).

\begin{figure}
    \centering
    \includegraphics[width=\linewidth]{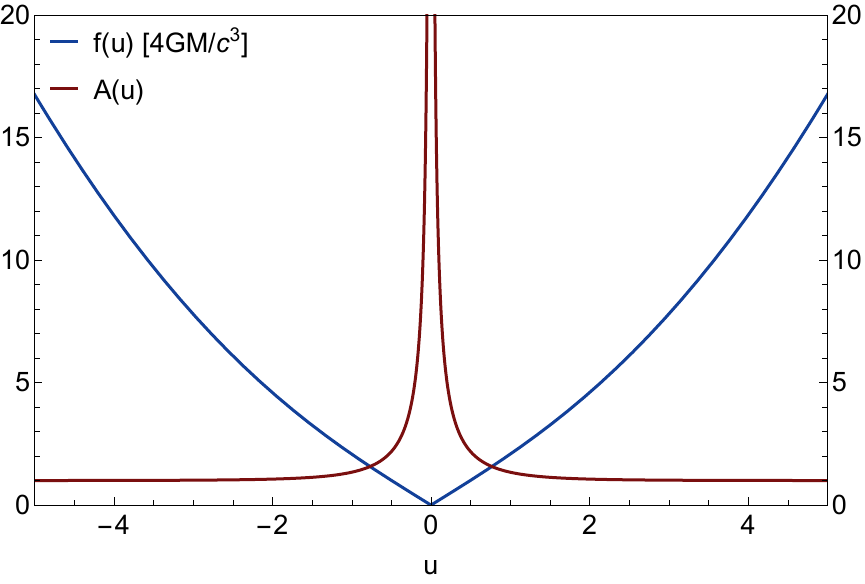}
    \caption{Plot of $A(u)$ and $f(u)$, the $u$-dependent parts of the magnification and time delay for a point-source point-lens microlensing configuration.}
    \label{fig:fu}
\end{figure}

When considering the fiducial use case of our delay-finding scheme (which uses photon sources near the Galactic Bulge), there is a simplified expression for the time delay:
\begin{equation}
\label{eq:deltat_mass}
    \Delta t = (2\times 10^{-5}\,\text{s}) \left(\frac{M}{M_\odot}\right) \times f(u).
\end{equation}
Hence, for Earth-mass lenses such as small rogue planets ($M \sim M_\oplus$), this is on the order of $0.1 \, \mathrm{ns}$, while for black holes such as those found in the LIGO band \cite{Nitz2021BHMerger,LIGO2023BHCatalogue} ($M \sim 30 M_\odot$), this is on the order of $1 \, \mathrm{ms}$, evaluated at $u\approx1$.

As mentioned in the previous subsection, the strategy of frequency-domain interference has been discussed primarily in the radio band, partially due to the finite-source effect:
if photons emitted from different regions of the source have respective path lengths for the same image that differ by more than a fraction of $\lambda$, averaging over the surface of the star eliminates the interferometric effect. Although we postpone the rigorous analysis to \cref{sec:finite_source}, we can derive a corresponding bound on the size of a source as a function of wavelength. To do this, we expand $\tau$ in small $\vec{\delta\theta}$ first to find
\begin{equation}
    \delta \tau = \frac{\partial\tau}{\partial\theta_i}\delta\theta_i+\frac{1}{2} \delta\theta_j \frac{\partial^2 \tau}{\partial \theta_j \partial \theta_k}\delta\theta_k + \cdots.
\end{equation}
By Fermat's principle, the first term in the expansion vanishes, since $\frac{\partial\tau}{\partial\theta_i} = 0$. As a result, the leading-order contribution is quadratic. The derivative expression in this term is simply the magnification matrix
\begin{equation}
    A_{jk} \coloneqq \frac{\partial^2 \tau}{\partial \theta_j \partial \theta_k} = \frac{\partial\beta_k}{\partial\theta_j},
\end{equation}
where the magnification defined in Eq.~\eqref{eq:Au} is simply $|\mathbf{A}^{-1}|$.
To convert to source-plane coordinates, we take $\vec{\delta \theta} \rightarrow \frac{\partial \theta}{\partial \beta} \vec{\delta \beta}$, but this is just a transformation by the inverse magnification matrix $\mathbf{A}^{-1}$. Hence the leading-order term just becomes
\begin{equation}
    \delta t = \frac{D_L D_S}{2 D_{LS} c}\vec{\delta\beta}^T \mathbf{A}^{-1} \vec{\delta\beta}.
\end{equation}
In order to find the finite-source path difference, we can take $\delta\beta \approx 2 R_S / D_S$, where $R_S$ is the physical radius of the source star and $D_S$ is the distance to the source star. (Note that this argument also applies to proper motion stability. In this case, we can take $\delta\beta \approx v_T t_{\text{exp}}$. However, since $R_S > v_T t_\text{exp}$ for all realistic values of $t_\text{exp}$, finite-source effects always place the stronger constraint on which targets are viable for this setup.)

Taking $D_L \approx D_S/2$ and with  $|\mathbf{A}^{-1}| = A$, we can coarsely approximate this constraint as
\begin{equation}
\label{eq:finitesource}
    \delta \lambda \gtrsim 2\frac{ R_S^2}{D_S}  A \approx (5\,\text{mm}) \left(\frac{R_S}{R_\odot}\right)^2\left(\frac{D_S}{8 \, \text{kpc}}\right)^{-1} \left(\frac{A}{1.34}\right).
\end{equation}
Light satisfying the above constraint largely lies in the radio band, making it difficult to detect such an effect in the optical. The radius of the source, however, remains a free parameter. If, instead, the source is only the radius of Earth, we have
\begin{equation}
\label{eq:finitesource2}
    \delta \lambda \gtrsim  (400\,\text{nm}) \left(\frac{R_S}{R_\oplus}\right)^2\left(\frac{D_S}{8 \, \text{kpc}}\right)^{-1} \left(\frac{A}{1.34}\right),
\end{equation}
which falls squarely in the optical range. The only isolated sources at these radii are white dwarfs, which would be a compelling target were it not for the fact that they are too dim at $8 \, \mathrm{kpc}$ to be detected in microlensing surveys.

Interestingly, however, the algorithm presented later in this work can still produce a mass measurement even if the source \textit{region} is only a small fraction of the source surface. All that is required is that this localized region is solitary and produces a larger flux in the bandpass of interest than the background flux from the rest of the source. For this reason, the microlensing of \textit{stellar flares} of M dwarfs, energetic emissions from localized regions on the surface of source stars, may provide a means of performing this mass measurement, even if the source star radius exceeds the wavelength limit of Eq.~(\ref{eq:finitesource}). 

Finally, observe that Eqs.~\eqref{eq:finitesource} and \eqref{eq:finitesource2} suggest that greater magnification will increase the finite-source effect, hence the $A=1.34$ scenario (corresponding to $u=1$) achieves a reasonable balance between brightening and finite-source problem. However, we notice that if a lensing event has more than 1.34 amplification, one can always choose to wait a while (may be from hours to days, depending on the total duration of the event) such that the magnification decreases to 1.34 due to the relative motion. Therefore, the $A\approx 1.34$ requirement will not significantly reduce the number of observable events. Furthermore, as is discussed later in \cref{sec:unequal}, small magnification will slightly increase the number of photons we need to perform a successful time-delay measurement. This implies that there may exist an optimal choice of $A$ that balances the finite-source effect, the photon number requirement, and the observable event rate. We leave this as an open problem.

\subsection{Problem setup for delay finding}
\label{sec:setup}
\paragraph{Classical setup.} We start by defining the problem from the most general classical description of the optical system. We consider the electromagnetic field at a specific position (say the position of our telescope) with a specific polarization in a time window from $0$ to $T$ (which is set by our measurement protocol). Note that this is a reasonable setting because $T$ is the upper limit of $\Delta t$ given by our prior knowledge, and the time window must be longer than $\Delta t$ to allow for detecting the time-delay phenomena. We let $E_0(t)$ be the electric field emitted by the source without microlensing and let $E(t)$ be the field with microlensing. With time delay $\Delta t$ and magnification $A$, we claim that $E(t)$ can be written as 
\begin{equation}
\label{eq:unequal_electricfield}
    E(t) = \sqrt{A_+} E_0(t) + \sqrt{A_-} E_0(t-\Delta t)
\end{equation}
where $A_\pm = \frac{1}{2}A\pm \frac{1}{2}$ is the magnification of each path, as introduced in \cref{sec:lens_intro}.
The reason why the actual electric field can be considered simply as the sum of $E_0$ and its delayed version is that the two images (corresponding to $E_0(t)$ and $E_0(t-\Delta t)$, respectively) are indistinguishable in microlensing. One can also verify that, when the coherence time $t_c$ of the unlensed electric field $E_0(t)$ is much smaller than $\Delta t$, $E_0(t)$ and $E_0(t-\Delta t)$ are incoherent, hence the average intensity of light $I$ satisfies
\begin{equation}
\begin{aligned}
    I &= \langle |E(t)|^2 \rangle_t \\
    &= A_+ \langle |E_0(t)|^2 \rangle_t + A_- \langle |E_0(t)|^2 \rangle_t \\
    &\quad + 2\sqrt{A_+ A_-} \langle E_0(t) E_0(t-\Delta t)\rangle_t\\
    &= A I_0,
\end{aligned}
\end{equation}
where $\langle\cdot\rangle_t$ denotes an average over time. This result agrees with the condition that light intensity is amplified by a factor of $A$ due to microlensing.
Since, in this paper, we focus on lensing events with $A=1.34$, 
corresponding to $\sqrt{A_+} = 1.08$ and $\sqrt{A_-} = 0.41$, the unequal magnification of the two paths will likely only cause small-constant-factor-level noise in the estimation of $\Delta t$ (indeed, we prove this in \cref{sec:unequal}). Therefore, we use a simpler model with $A_+ = A_- = 1$ in most of our analyses and postpone a rigorous discussion of the $A_\pm$ factors to \cref{sec:unequal}.

The delay-finding problem is simply evaluating $\Delta t$ from the light field $E(t)$. To give a better sense of a realistic form of $E(t)$, we can assume that $E_0(t)$ is a Gaussian wave packet, i.e.,
\begin{equation}
\label{eq:clasical_field}
    E_0 (t) = \mathcal{E} \alpha(t-t_0) e^{-\imag \omega_0 t},
\end{equation}
where $\omega_0$ is the carrier frequency, $t_0$ is centroid of the wave packet, $\mathcal{E}$ is the strength of the electric field, and $\alpha$ is the normalized Gaussian wave packet defined as
\begin{equation}
\label{eq:gausspacket}
    \alpha(\tau) \coloneqq \frac{1}{\sqrt[4]{\pi t_c^2}} e^{-\frac{\tau^2}{2 t_c^2}},
\end{equation}
where $t_c$ is the width of the wave packet, or the coherence time of $E_0(t)$.

\paragraph{Quantum setup.} Recall that a major difficulty in microlensing delay finding is the photon-starved condition, and one of our key objectives is to find a photon-efficient solution. Therefore, we must consider the problem setup from a quantum mechanical perspective. To do so, we analyze the wave function of an incident photon, which is of similar form as the classical electric field. We can interpret $\mathcal{E}^2$ as the photon rate and $\alpha(t-t_0) e^{-\imag \omega_0 t}$ as the wave function of an unlensed photon (note that $\alpha(t-t_0)$ is a normalized Gaussian wave packet), i.e.,
\begin{equation}
    \ket{\phi_0(t_0)} =  \int_{-\infty}^{\infty} \alpha(t-t_0) e^{-i\omega_0 t} \ket{t} dt,
\end{equation}
where $\ket{t}$ represents the state that the photon is received at time $t$.

Now, let us consider the state of a lensed photon. Since the two images in microlensing are indistinguishable for the observer, one can imagine that the directions of the two emission paths are also indistinguishable from the perspective of the photon emitter. Also, this implies that the angle between the two paths is sufficiently small that we can consider the photon as an excitation of a spherical wave, which is a superposition of all possible directions.
Thus we can describe
every received photon by a superposition of two paths. This allows us to write down the state of a photon when a microlensing event occurs:
\begin{equation}
\label{eq:purestate}
\begin{aligned}
    &\quad \ket{\phi(t_0,\Delta t)}\\
    &= 
   \frac{1}{\sqrt{2}} \int_{-\infty}^{\infty} \alpha(t-t_0) e^{-i\omega_0 t} \left(\ket{t} + \ket{t + \Delta t} \right)dt \\
   &= 
   \frac{1}{\sqrt{2}} \int_{-\infty}^{\infty}\left[ \alpha(t-t_0) e^{-i\omega_0 t}  \right.\\
   &\quad \quad \quad \quad + \left. \alpha(t-t_0-\Delta t) e^{-i\omega_0 (t-\Delta t)} \right] \ket{t} dt.
\end{aligned}
\end{equation}
Note that this state is normalized correctly only if $\Delta t\gg t_c$, which is the scenario of interest.

Note that, in reality, $t_0$ is as a uniformly random quantity because one can never predict at what time a photon arrives. Therefore, the most rigorous way of expressing the state is a density operator $\rho(\Delta t)$, defined as a (classical) uniform mixture over all $\ket{\phi(t_0,\Delta t)}\bra{\phi(t_0,\Delta t)}$ with fixed $\Delta t$ and $t_0\in [0,T_\mathrm{w}]$ (with $T_\mathrm{w}\gg \Delta t$), i.e.,
\begin{equation} \label{eq:mixedState}
    \rho(\Delta t) = \frac{1}{T_\mathrm{w}} \int_{0}^{T_\mathrm{w}} \ket{\phi(t_0,\Delta t)}\bra{\phi(t_0,\Delta t)} d t_0.
\end{equation}
Now, we have all the theoretical ingredients to formulate the problem as follows.

\begin{problem}[Delay finding]
\label{prob:prob1}
    Learn $\Delta t$ with error up to $t_c$ from as few copies of $\rho(\Delta t)$ as possible.
\end{problem}
Note that each copy of $\rho(t)$ corresponds to one incident photon. We therefore use the phrase ``sample complexity" to represent the number of photons needed. We present our solution to this problem in \cref{sec:freq_domain_interference}.

\subsection{Review of a sample-inefficient approach}
\label{sec:review_classical}
Measuring the time delay between two paths is not unique to the topic of gravitational lensing. However, since there is not a stringent restriction on the number of photons in most scenarios, existing delay-finding approaches cannot be simply applied to our problem due to their sample complexity. This even includes previous works based on the same intuition as ours (frequency-domain interference). In this subsection, we review a straightforward method to measure $\Delta t$ with $t_c$ precision using $O(T/t_c)$ photons, while our algorithm in \cref{sec:freq_domain_interference} needs only $O(\log(T/t_c))$ photons.

From Eq.~\eqref{eq:purestate}, we realize that the wave function of a photon is a superposition of two wave packets separated by $\Delta t \gg t_c$. Therefore, if one can move one of the packets in the time domain by a time $\tau$ with $|\Delta t -\tau|\leq t_c$, then the two wave packets would overlap with each other and create interference.

To observe the above phenomena, one can use a standard Mach–Zehnder interferometer: the input light is split by a beam splitter into two paths, where a tunable delay line (of length $\tau$) is place in one of them, then two paths are recombined on the second beam splitter followed by single-photon detectors at two output ports. We can see that, with constant probability, the state at one port is
\begin{equation}
\begin{aligned}
    &\frac{1}{\sqrt{2}} \int_{-\infty}^{\infty}\left( \alpha(t-t_0 - \tau) e^{-i\omega_0 (t-\tau)} \right.\\ 
    &\quad \quad \quad \quad \left. + \alpha(t-t_0-\Delta t) e^{-i\omega_0 (t-\Delta t)} \right)\ket{t} dt
\end{aligned}
\end{equation}
while the state at the other port is
\begin{equation}
\begin{aligned}
    &\frac{1}{\sqrt{2}} \int_{-\infty}^{\infty}\left( -\alpha(t-t_0 - \tau) e^{-i\omega_0 (t-\tau)} \right.\\ 
    &\quad \quad \quad\quad \left. +\alpha(t-t_0-\Delta t) e^{-i\omega_0 (t-\Delta t)} \right)\ket{t} dt
\end{aligned}
\end{equation}
Now, if $\tau > t_c$, the two packets are not overlapped, hence the probability of receiving the photon at each port is the same, i.e., $\Pr[\mathrm{port}\,1] \approx \Pr[\mathrm{port}\,2] \approx 1/2$. However, if $\tau \leq t_c$, the probabilities will be approximately $\frac{1}{2}(1+\cos(\omega_0(\Delta t-\tau))$ and $\frac{1}{2}(1-\cos(\omega_0(\Delta t-\tau))$, respectively.

Using this result, one can try to scan over many possible $\tau$ values and check whether the photon distribution of the two ports is biased or not for each $\tau$. The $\tau$ with significant bias must satisfy $|\tau - \Delta t|\leq t_c$. However, since the search space is of size $O(T/t_c)$, 
the sample complexity of this approach is also $O(T/t_c)$. In the photon-starved regime, this method does not work well. To reduce the photon number requirement, we consider measuring the photons in a different basis, as explained in the next section.

\section{Frequency-domain interference}
\label{sec:freq_domain_interference}
In this section, we propose our quantum-inspired algorithm for sample-efficient delay finding. We first introduce the main theoretical intuition in \cref{sec:classical_pic} and \cref{sec:quantum_pic}, using the fact that a fixed delay between two signals in the time domain corresponds to a modulation in the frequency domain. Next, we describe our algorithm based on frequency-basis measurements in \cref{sec:ouralgo}. Then, in the context of our Fourier-basis analysis, we explain in \cref{sec:finite_source} how the variance in $\Delta t$ caused by the finite-source effect may destroy the signal from an information-theoretic perspective. Next, in \cref{sec:noisysignal} and \cref{sec:unequal}, we discuss the performance of our algorithm in the presence of noise and unequal magnification, respectively. Finally, we discuss the realistic \emph{broadband} scenario where photons have different carrier frequencies in \cref{sec:broadband}.

\subsection{Fully classical picture}
\label{sec:classical_pic}
Recall from \cref{sec:setup} that the classical description of the lensed electromagnetic field (with the equal-magnification assumption) is
\begin{equation}
    E(t) = E_0(t) + E_0(t-\Delta t).
\end{equation}
We denote the Fourier transform 
of $E_0(t)$ by
\begin{equation}
    \tilde{E}_0(\omega) \coloneqq \frac{1}{\sqrt{2\pi}} \int_{-\infty}^\infty E_0(t) e^{i\omega t} dt.
\end{equation}
The \emph{power spectrum} of $E_0(t)$ is then $\left|\tilde{E}_0(\omega)\right|^2$. Next, we compute the Fourier transform of $E(t)$,
\begin{equation}
\begin{aligned}
    \tilde{E}(\omega) &= \frac{1}{\sqrt{2\pi}} \left[ \int_{-\infty}^\infty E_0(t) e^{i\omega t} dt +\right.\\
    &\quad \left.\int_{-\infty}^\infty E_0(t-\Delta t) e^{i\omega (t-\Delta t)} e^{i\omega \Delta t} dt \right] \\
    &= \tilde{E}_0(\omega) (1+e^{i\omega \Delta t}),
\end{aligned}
\end{equation}
and its power spectrum
\begin{equation}
    \left|\tilde{E}(\omega)\right|^2 = 2 \left|\tilde{E}_0(\omega)\right|^2 (1+\cos(\omega \Delta t)).
\end{equation}
Now we can see that, if the original power spectrum $|\tilde{E}_0(\omega)|^2$ is known, then $\Delta t$ can be seen as interference fringes in the power spectrum provided the spectrum can be observed with resolution ${1}/{\Delta t}$.

\begin{figure*}
    \centering
    \includegraphics[width=\textwidth]{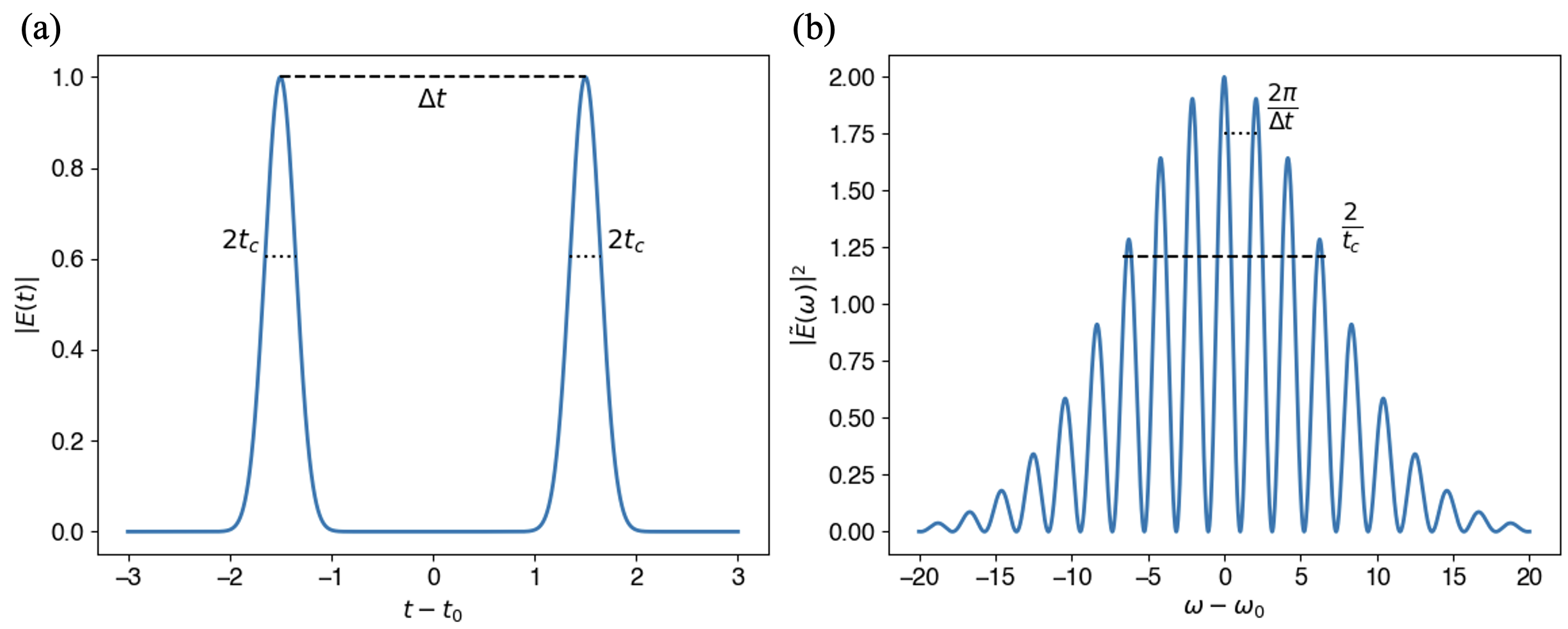}
    \caption{(a) Two Gaussian wave packets separated by $\Delta t$ in the time domain  
    correspond to (b) one Gaussian packet in the frequency domain with $(1+\cos(\omega \Delta t ))$ modulation.  
    In the simple example of this figure, we let $\Delta t=3$ and $t_c = 0.15$, hence two adjacent peaks in  
    (b) are separated by $2\pi/\Delta t \approx 2.1$ and the Gaussian envelop has width $2/t_c \approx 13$. All values in this example are unitless.}
    \label{fig:fig_modulation}
\end{figure*}

More formally, we can also find $\Delta t$ by Fourier transforming the power spectrum:
\begin{equation}
\begin{aligned}
    F(\tau) 
    &\coloneqq \frac{2}{\sqrt{2\pi}} \int_{-\infty}^\infty \left|\tilde{E}_0(\omega) \right|^2 e^{i\omega \tau} \cdot \\
    &\quad \left( 1+ \frac{e^{i\omega \Delta t} + e^{-i\omega \Delta t}}{2}\right) d\omega \\
    &=  2F_0(\tau) + F_0(\Delta t + \tau) + F_0(\Delta t - \tau)
\end{aligned}
\end{equation}
where $F_0(\tau) \coloneqq \frac{1}{\sqrt{2\pi}} \int_{-\infty}^\infty \left|\tilde{E}_0(\omega)\right|^2e^{i\omega \tau} d\omega$ is the Fourier transform of the spectrum of the unlensed signal.

We consider again the wave packet as a realistic model of the $E_0$ field, which can give a concrete example of the above behavior. Recall from Eq.~\eqref{eq:clasical_field} that $E_0 (t) = \mathcal{E} \alpha(t-t_0) e^{-\imag \omega_0 t}$. Using the definition of a normalized wave packet $\alpha(t)$ in Eq.~\eqref{eq:gausspacket}, we can compute its Fourier transform, which is another Gaussian wave packet:
\begin{equation}
    \tilde{\alpha}(\omega) = \frac{1}{\sqrt{2\pi}} \int_{-\infty}^\infty \alpha(\tau) e^{-\imag \omega \tau} d\omega =  \sqrt[4]{\frac{t_c^2}{\pi}} e^{-\frac{t_c^2 \omega^2}{2}}.
\end{equation}
Since the form of the wave packet can be controlled by placing a filter in front of the light receiver, it is reasonable to assume a Gaussian shape as above.
Now, one can compute the Fourier transform of $E_0$,
\begin{equation}
\begin{aligned}
    \tilde{E}_0(\omega) &= \frac{\mathcal{E}}{\sqrt{2\pi}} \int_{-\infty}^\infty \alpha(t-t_0) e^{-\imag \omega_0 t} e^{\imag \omega t} dt \\
    &= \mathcal{E} \tilde{\alpha}(\omega - \omega_0) e^{\imag (\omega-\omega_0)t_0}.
\end{aligned}
\end{equation}
See \cref{fig:fig_modulation}(a) for the plot of an example microlensed electric field $E(t)$ with $E_0(t)$ being a Gaussian wave packet, and \cref{fig:fig_modulation}(b) for the plot of its power spectrum, $|\tilde{E}(\omega)|^2 = 2|\tilde{E}_0(\omega)|^2(1+\cos(\omega \Delta t) $.

Next, we compute the Fourier transform of $|\tilde{E}_0(\omega)|^2$,
\begin{equation}
\begin{aligned}
    F_0(\tau) &= \frac{\mathcal{E}^2}{\sqrt{2\pi}} \sqrt{\frac{t_c^2}{\pi}} \int_{-\infty}^\infty   e^{-t_c^2 (\omega-\omega_0)^2} e^{\imag \omega \tau} d\omega \\
    &= \frac{\mathcal{E}^2}{\sqrt{2\pi}}
    e^{-\frac{\tau^2}{4 t_c^2}} e^{\imag \omega_0 \tau}.
\end{aligned}
\end{equation}
We can now compute $F(\tau)$ as a Fourier transform of $ \left|\tilde{E}_0(\omega)\right|^2 (1+\cos(\omega \Delta t))$:
\begin{equation}
\label{eq:Ftau_certain}
\begin{aligned}
    F(\tau) &= \frac{\mathcal{E}}{\sqrt{2\pi}} \left[2e^{-\frac{\tau^2 }{4 t_c^2}}e^{i\omega_0 \tau} + e^{-\frac{(\Delta t + \tau)^2}{4 t_c^2}} e^{i\omega_0 (\Delta t + \tau)} \right.\\
    &\quad \quad  \quad  \quad  \left. + e^{-\frac{(\Delta t - \tau)^2}{4 t_c^2}} e^{i\omega_0 (\Delta t - \tau)} \right].
\end{aligned}
\end{equation} 
We observe that $|F(\tau)|^2$ has peaks at $\tau \rightarrow 0$ and $\tau\rightarrow \Delta t$. Therefore, the time delay can be directly read out from the Fourier transform of the power spectrum as a nonzero peak. This gives a straightforward approach to measuring the time delay by measuring the whole power spectrum of the optical signal, as in Refs.~\cite{deguchi1986diffraction,peterson1991gravitational,gould1992femtolensing,ulmer1995femtolensing,katz2018femtolensing,eichler2017nanolensed,jow2020wave,wucknitz2021cosmology}.

However, measuring the full power spectrum of a extremely weak light signal may consume prohibitively many photons. Instead, we show that if each photon can give one \emph{sample} from the power spectrum $|\tilde{E}|^2$, then the time delay $\Delta t$ can be measured using exponentially fewer photons. To show this, we derive the frequency-domain interference in the photonic picture in the next subsection.

Before moving on, we emphasize that the modulation $\cos(\omega \Delta t)$ in the power spectrum has an extremely short period, $1/\Delta t$. (Note that this ``period'' has units of $\mathrm{time}^{-1}$ because the oscillation is in the frequency domain.) For typical microlenses, $\Delta t$ can be as high as $10^{-3} \, \mathrm{s}$, corresponding to a $1 \, \mathrm{kHz}$ period. Therefore, according to the Nyquist–Shannon sampling theorem, one must be able to obtain frequency information with up to kHz-level error to find $\Delta t$, regardless of how the data are processed.
Furthermore, due to the uncertainty principle, measuring the frequency with precision $O(1/\Delta t)$ implies that a measurement of the time has error $\Omega(\Delta t)$. This is necessary for a lensed photon with $\Delta t$ time delay, and implies that the device must not be able to tell the difference between two time points in the same $\Delta t$ interval (see e.g.\ Ref.\ \cite{vittorini2014entanglement} for a similar effect with the roles of time and frequency domains reversed). In other words, a device with $1/\Delta t$ frequency resolution must somehow ``store'' a photon for time at least $\Delta t$.

\subsection{Photonic (quantum) picture}
\label{sec:quantum_pic}
In fact, the same derivation as in the classical picture can be reproduced in the photonic picture.
Recall from Eq.~\eqref{eq:purestate} that the pure state of a photon can be written as
\begin{equation}
\begin{aligned}
    &\quad \ket{\phi(t_0,\Delta t)} \\
    &=\frac{1}{\sqrt{2}} \int_{-\infty}^{\infty}\left( \alpha(t-t_0) e^{-i\omega_0 t} \right.\\
   &\quad \quad \quad \quad \quad \left.+ \alpha(t-t_0-\Delta t) e^{-i\omega_0 (t-\Delta t)} \right)\ket{t} dt.
\end{aligned}
\end{equation}
We can express the state in the frequency domain by performing the Fourier transform:
\begin{equation}
\begin{aligned}
    &\quad \ket{\phi(t_0,\Delta t)}\\
    &= \frac{1}{\sqrt{2}} \int_{-\infty}^\infty \tilde{\alpha}(\omega - \omega_0) e^{\imag (\omega-\omega_0)t_0} (1+e^{\imag \omega \Delta t}) \ket{\omega} d\omega,
\end{aligned}
\end{equation}
which has the same form as in the classical picture. Now, if we measure the state in the frequency basis, the probability density of obtaining $\omega$ is
\begin{equation}
\label{eq:thechannel}
p(\omega | \Delta t) = \left|\tilde{\alpha} (\omega - \omega_0) \right|^2 (1+\cos(\omega \Delta t)),
\end{equation}
which is independent of $t_0$. The Fourier transform of $p(\omega|\Delta t)$, denoted by $F_{\Delta t}(\tau)$, 
has the same form as $F(\tau)$:
\begin{equation}
\label{eq:Ftau_certain_q}
\begin{aligned}
    &\quad F_{\Delta t}(\tau) \\
    &=\frac{1}{\sqrt{2\pi}} \int_{-\infty}^\infty e^{\imag \omega \tau} p(\omega|\Delta t) d\omega\\
    &=\frac{1}{2\sqrt{2\pi}} \left[2e^{-\frac{\tau^2 }{4 t_c^2}}e^{\imag \omega_0 \tau} + e^{-\frac{(\Delta t + \tau)^2}{4 t_c^2}} e^{\imag \omega_0 (\Delta t + \tau)} \right.\\
    &\quad\quad\quad \quad\quad \left. +e^{-\frac{(\Delta t - \tau)^2}{4 t_c^2}} e^{i\omega_0 (\Delta t - \tau)} \right] \\
    &\approx \frac{e^{-\frac{(\Delta t - \tau)^2}{4 t_c^2}} e^{i\omega_0 (\Delta t - \tau)}}{2\sqrt{2\pi}}.
\end{aligned}
\end{equation}

In addition, we can show that the density operator $\rho(\Delta t)$ is diagonal in the Fourier basis for large time window $T_\mathrm{w}$. (In fact, it must be diagonal for large $T$ in the Fourier basis because $\rho(\Delta t)$ is time-translation invariant.) To see this, we simply express the density operator in the frequency domain:
\begin{equation}
\begin{aligned}
     \rho(\Delta t) &= \int_0^{T_\mathrm{w}} \ket{\phi(t_0,\Delta t)} \bra{\phi(t_0,\Delta t)} p(t_0) dt_0 \\
    &= \frac{1}{2} \int_{-\infty}^\infty d\omega_1 d\omega_2 \cdot \tilde{\alpha}(\omega_1 - \omega_0) \tilde{\alpha}(\omega_2 -\omega_0) \cdot \\
    &\quad\quad  (1+e^{i\omega_1\Delta t})(1+e^{-i\omega_2 \Delta t}) \ket{\omega_1}\bra{\omega_2} \cdot \\
    &\quad\quad \int_0^{T_\mathrm{w}} p(t_0) e^{\imag (\omega_1-\omega_2)t_0} dt_0.
\end{aligned}
\end{equation}
Now, we can see that, for the diagonal terms ($\omega_1 = \omega_2=\omega$), the coefficient is simply $p(\omega|\Delta t)$; for off-diagonal terms ($\omega_1 \neq \omega_2$),
\begin{equation}
\begin{aligned}
    \int_0^{T_\mathrm{w}} p(t_0) e^{i(\omega_1 - \omega_2) t_0} dt_0 &= \frac{e^{i(\omega_1-\omega_2) T_\mathrm{w}}-1}{iT_\mathrm{w}(\omega_1 -\omega_2)}  \\
    &=O\left(\frac{1}{T_\mathrm{w}}\right) \rightarrow 0.
\end{aligned}
\end{equation}
Therefore, when $T_\mathrm{w}$ is sufficiently large,
\begin{equation}
\label{eq:rhoDeltaT}
\begin{aligned}
&\quad \rho(\Delta t)  \approx  \rho_\mathrm{diag}(\Delta t) \\
& = \int_{-\infty}^\infty  \left|\tilde{\alpha} (\omega - \omega_0) \right|^2 (1+\cos(\omega \Delta t)) \ket{\omega}\bra{\omega} d\omega \\
&=  \sqrt{\frac{t_c^2}{\pi}} \int_{-\infty}^\infty   e^{-t_c^2 (\omega-\omega_0)^2} (1+\cos(\omega \Delta t)) \ket{\omega}\bra{\omega} d\omega  
\end{aligned}
\end{equation}
is a diagonal density operator. Moreover, since the same operation  (the Fourier transform) diagonalizes $\rho(\Delta t)$ for arbitrary $\Delta t$ without knowing its value, we can treat the gravitational lensing system as a \emph{classical communication channel} where Alice (the gravitational lens) sends information about $\Delta t$ to Bob (observers on the Earth) through a continuous-variable channel $p(\omega | \Delta t) = \left|\tilde{\alpha} (\omega - \omega_0) \right|^2 (1+\cos(\omega \Delta t))$.

As a remark, we note that one can easily compute the Fisher information with respect to $\Delta t$ in the above distribution of $\omega$ and apply the Cramér-Rao bound to derive the asymptotic scaling of the optimal number of samples needed to achieve a certain precision of $\Delta t$ estimation. However, this is irrelevant to the problem in this paper because we are only interested in a rough estimate with up-to-$t_c$ precision, which is not in the regime addressed by the Fisher information.

\subsection{The sample-efficient algorithm}
\label{sec:ouralgo}

Recall that the value of $\Delta t$ can be directly read from a peak in the Fourier transform of the power spectrum, as shown in Eqs.~(\ref{eq:Ftau_certain}, \ref{eq:Ftau_certain_q}). Our sample-efficient algorithm reconstructs the peak from a limited number of frequency-domain measurement outcomes. We are inspired by maximum likelihood estimation algorithm and studies of the dihedral hidden subgroup problem 
\cite{EH00}
to propose the \emph{score function}
\begin{equation}
    f(\tau,\nu_{1},\dots,\nu_{n}) = \sum_j \cos(\nu_{j} \tau),
\end{equation}
where $\tau$ denotes a candidate for the unknown $\Delta t$ and $\nu_1, \nu_2,\dots,\nu_n$ are the $n$ samples obtained by measuring the photons in the frequency domain. It is not hard to prove that if the $\tau$s are sufficiently dense in $[0,T]$ and $n$ is sufficiently large (as quantified below), then the $\tau$ maximizing $f(\tau,\nu_1,\dots,\nu_n)$ is the closest to $\Delta t$ among all candidates. In particular, the expectation value of $\cos(\nu_j \tau)$ for any $j$ corresponds directly to the Fourier transform of the power spectrum:
\begin{equation}
\begin{aligned}
    \E[\cos(\nu_j \tau)] &= \int_{-\infty}^\infty p(\nu_j|\Delta t) \cos(\nu_j \tau) d\nu_j \\
    &= \int_{-\infty}^\infty p(\nu_j|\Delta t) \Re[e^{i\nu_j \tau}] d\nu_j \\
    &= \sqrt{2\pi}\Re[F_{\Delta t}(\tau)].
\end{aligned}
\end{equation}
According to Eq.~\eqref{eq:Ftau_certain_q}, considering the case where $\tau \gg t_c$, we conclude that
\begin{equation}
\begin{aligned}
    &\quad \E[f(\tau,\nu_{1},\dots,\nu_{n})]\\ &\approx \begin{cases}
         \frac{1}{2} n \cos(\omega_0 (\Delta t-\tau)),&|\tau-\Delta t|< t_c  \\
         0,&|\tau-\Delta t| \geq t_c, 
    \end{cases}
\end{aligned}
\end{equation}
which gives a $\Theta(n)$ gap between correct and incorrect candidates. We note that the number of $\tau$s need only be $O(T/t_c)$ to find a $\tau$ approximating $\Delta t$ with up to $t_c$ precision. In practice, to avoid the ``unlucky" cases where $\cos(\omega_0 (\Delta t-\tau)) \approx 0$, we also check $\tau + \frac{2\pi k} {10 \omega_0}$ for $k\in\{0,1,\dots, 9\}$, so the number of candidates is $10 T/t_c$.

More formally, our algorithm is as follows.

\begin{algorithm}[Sample-efficient delay finding.]
\label{alg:alg1}
    Step (i): measure $n$ incident photons in the frequency basis to obtain $\nu_1,\nu_2,\dots,\nu_n$. 
    Step (ii): evaluate $f(\tau,\nu_1,\dots,\nu_n)$ for all $ 10 T/t_c$ candidate $\tau$s. 
    Step (iii): accept any $\tau$ with $f(\tau,\nu_1,\dots,\nu_n) \geq n/4$ as an estimate of the gravitational lensing time delay.
\end{algorithm}

Finally, we establish the logarithmic sample complexity ($n=O(\log(T/t_c))$) of this method. For a ``bad candidate" $\tau$ with $|\tau - \Delta t|\geq t_c$, we let $Y_{\tau}$ denote a ``bad event" that $f(\tau, \nu_1,\dots,\nu_n) \geq n /4$. Hoeffding's inequality gives a bound for the probability of $Y_\tau$:
\begin{equation}
\begin{aligned}
    &\quad \Pr[f(\tau, \nu_1,\dots,\nu_n) \geq \frac{n}{4}] \\
    &= \Pr[f(\tau, \nu_1,\dots,\nu_n) - \E[f(\tau, \nu_1,\dots,\nu_n)] \geq \frac{n}{4}] \\
    &\leq e^{-\frac{n}{32}}.
\end{aligned}
\end{equation}

To ensure that, with high probability (say $95\,\%$), no bad event happens, the union bound gives 
\begin{equation}
    \Pr[\bigcup_{\tau} Y_\tau] \leq \sum_{\tau} \Pr[Y_\tau] \leq \frac{10T}{t_c}  e^{-\frac{n}{32}} \leq 0.05.
\end{equation}
This implies that 
$n\geq 32 \left[\ln(T/t_c) + \ln(0.005) \right]= O(\log (T/t_c))$
photons are sufficient to find $\Delta t$ with $t_c$ precision and $95\,\%$ confidence.

\subsection{The finite-source effect}
\label{sec:finite_source}
In previous analyses, we assume all incoming photons share exactly the same lensing time delay $\Delta t$, which only holds when the photon source is pointlike. However, in reality, almost all photon sources are \emph{extended}, including stars, planets, quasars, etc. Different regions of an extended source have different paths to the observer with different lensing time delays. If all these regions have the same emission power spectra and the angular resolution of the telescope is smaller than the angular distance between different regions, photons with different $\Delta t$ values will be mixed in an indistinguishable manner. This suggests the possibility that the strategy of measuring $\Delta t$ may fail, and even worse, from an information-theoretic point of view, it may be fundamentally impossible to obtain any information about the delay. Indeed, as we show in this subsection, this \emph{finite-source effect} turns out to be a major challenge in time-delay measurements based on frequency-domain interference because the information about $\Delta t$ is exponentially suppressed in the optical signal.

By taking into account the distribution of $\Delta t$, denoted by $p_\mathrm{FS}(\Delta t)$, we can derive the marginal distribution of $\omega$:
\begin{equation}
\begin{aligned}
    p(\omega) &= \int_{-\infty}^\infty p_\mathrm{FS}(\Delta t) p(\omega | \Delta t) d\Delta t \\
    &= \left|\tilde{\alpha} (\omega - \omega_0) \right|^2 \int_{-\infty}^\infty p_\mathrm{FS}(\Delta t) (1+\cos(\omega \Delta t)) d\Delta t.
\end{aligned}
\end{equation}
Since $p(\omega)$ has a Gaussian envelope centered at $\omega_0$, one can see that, if the uncertainty in $\Delta t$ is greater than $ 1/ \omega_0$, then the integral of $\cos(\omega \Delta t)$ will be washed out.
We can plug in some realistic settings into $p_\mathrm{FS}(\Delta t)$: assuming $p_\mathrm{FS}(\Delta t)$ is a Gaussian centered at $\Delta t_0$ with standard deviation $\delta_{\Delta t}$, we find
\begin{equation}
\label{eq:pomega_finite_source}
\begin{aligned}
    p(\omega) &= \frac{\left|\tilde{\alpha} (\omega - \omega_0) \right|^2}{\sqrt{2\pi} \delta_{\Delta t}} \cdot \\
    &\quad \int_{-\infty}^\infty e^{-\frac{(\Delta t - \Delta t_0)^2}{2\delta_{\Delta t}^2}} (1+\cos(\omega\Delta t)) d\Delta t
    \\
    &=\left|\tilde{\alpha} (\omega - \omega_0) \right|^2 \left[ 1 + e^{-\frac{\omega^2 \delta_{\Delta t}^2} {2} } \cos(\omega \Delta t_0)
    \right]
\end{aligned}
\end{equation} 
or, equivalently,
\begin{equation}
\begin{aligned}
     \rho(\Delta t_0,\delta_{\Delta t}) & = \int_{-\infty}^\infty d\omega \left|\tilde{\alpha} (\omega - \omega_0) \right|^2 \cdot \\
     &\quad \left[ 1 + e^{-\frac{\omega^2 \delta_{\Delta t}^2} {2} } \cos(\omega \Delta t_0)
    \right] \ket{\omega} \bra{\omega} .
\end{aligned}
\end{equation}
Now, compared with $p(\omega|\Delta t)$ in Eq.~\eqref{eq:thechannel}, the $\cos(\omega \Delta t_0)$ oscillation, which carries the information about $\Delta t_0$, is exponentially suppressed in the marginal distribution when $\delta_{\Delta t} \gtrsim 1/\omega_0$. Note that $1/\omega_0$ can be extremely tiny---as an example, $1/\omega_0 \sim 10^{-15} \, \mathrm{s}$ for visible light.

To quantify how robust our score-function-based algorithm is, we also compute the expectation value of $\cos(\nu_j \tau)$ when $\delta_{\Delta t}$ is taken into account, which is essentially the Fourier transform of the $p(\omega)$ function in Eq.~\eqref{eq:pomega_finite_source}:
\begin{equation}
    \begin{aligned}
         F_{\Delta t_0,\delta_{\Delta t}} (\tau) &= \frac{1}{\sqrt{2\pi}} \int_{-\infty}^\infty p(\omega) e^{i\omega \tau} d\omega \\
        &= \frac{t_c}{\sqrt{2}\pi} \int_{-\infty}^\infty  e^{-t_c^2 (\omega-\omega_0)^2} \cdot \\ 
        &\quad \left[ 1 +  e^{-\frac{\omega^2 \delta_{\Delta t}^2} {2} } \cos(\omega \Delta t_0)\right]  e^{i\omega \tau} d\omega\\
       &= \frac{e^{i\omega_0 \tau}  e^{-\frac{\tau^2}{4 t_c^2}  } }{\sqrt{2\pi}} +\frac{t_c}{2\sqrt{2}\pi} \cdot \\
       &\quad \int_{-\infty}^\infty \left[ e^{-(\omega - \omega_0)^2 t_c^2 + i\omega (\tau+\Delta t_0) -\frac{\omega^2 \delta_{\Delta t}^2} {2} } \right. \\
       & \quad +\left. e^{-(\omega - \omega_0)^2 t_c^2 + i\omega (\tau-\Delta t_0) -\frac{\omega^2 \delta_{\Delta t}^2} {2} } \right] d\omega.
    \end{aligned}
\end{equation}
Similar to Eqs.~(\ref{eq:Ftau_certain}, \ref{eq:Ftau_certain_q}), only the third term in $F_{\Delta t_0,\delta_{\Delta t}}$ matters when $\tau\gg t_c$. Evaluating the integral in the third term yields
\begin{equation}
\begin{aligned}
    &\quad F_{\Delta t_0,\delta_{\Delta t}}(\tau) \\
    &\approx \frac{t_c}{2\sqrt{2} \pi} \sqrt{\frac{\pi}{t_c^2 + \delta_{\Delta t}^2 / 2} }  \\ 
    &\quad \exp\left[- \omega_0^2 t_c^2 + \frac{\left( 2\omega_0  t_c^2+ i(\tau \pm \Delta t_0) \right)^2 }{4 (t_c^2+\delta_{\Delta t}^2/2)}\right]\\
    &= \frac{1}{2\sqrt{2\pi}}\sqrt{\frac{1}{ 1  + \delta_{\Delta t^2} / (2t_c^2) } } \exp\left[ \frac{\imag \omega_0 (\tau-\Delta t_0)}{1+\delta_{\Delta t}^2 / (2t_c^2)} \right] \\
    &\quad \exp\left[\frac{-\omega_0^2 \delta_{\Delta t}^2 / 2 -(\tau-\Delta t_0)^2 / (4t_c^2) }{1+\delta_{\Delta t}^2 / (2t_c^2)}\right].
\end{aligned}
\end{equation}    
In the limit  $\delta_{\Delta t}\ll t_c$, we have
\begin{equation}
    F_{\Delta t_0,\delta_{\Delta t}}(\tau) \approx \exp(-\frac{\omega_0^2 \delta_{\Delta t}^2} {2}) \cdot F_{\Delta t_0}(\tau),
\end{equation}
an exponentially suppressed version of $F_{\Delta t_0}(\tau)$. Hence the expectation value of the score function, $\sqrt{2\pi}F_{\Delta t_0, \delta_{\Delta t}}$, is also exponentially suppressed with $\delta_{\Delta t}$.

\subsection{Noisy-signal performance}
\label{sec:noisysignal}
In this subsection, we analyze the performance of \cref{alg:alg1} in a realistic scenario in astronomical observations where we receive not only signal photons with information about $\Delta t$, but also noise photons. More specifically, suppose there is a signal-to-noise ratio (or signal-to-background ratio) $Q$ (with $0\leq Q\leq 1$) such that, among $n$ incident photons, only $nQ$ photons are samples of the $\rho(\Delta t)$ state of our interest. This setting is relevant to stellar flare observation where $n_\mathrm{sig}\coloneqq nQ$ ``good" photons are signal photons from the flare region and in state $\rho(\Delta t,\delta_{\Delta t})$ with $\delta_{\Delta t}\ll 2\pi / \omega_0$; while $n_\mathrm{bg} \coloneqq n(1-Q)$ ``bad" photons are background photons from the M dwarf host and suffer from a severe finite-source effect with $\delta_{\Delta t} \gg 2\pi / \omega_0$. See \cref{sec:flare_setup} for a detailed calculation for $n_\mathrm{sig}$ and $n_\mathrm{bg}$.

We can now analyze the expectation value of the score function when the photons are from the above flare scenario. Due to the finite-source effect, all ``bad" photons give almost zero contribution to the expectation value, and the separation between $|\tau -\Delta t|<t_c$ and $|\tau - \Delta t|\geq t_c$ is created only by the $nQ$ ``good" photons. Therefore,
\begin{equation}
\begin{aligned}
    &\quad \E[f(\tau,\nu_{1},\dots,\nu_{n})] \\
    &\approx \begin{cases}
         \frac{1}{2} n Q\cos(\omega_0 (\Delta t-\tau)),&|\tau-\Delta t|< t_c  \\
         0,&|\tau-\Delta t| \geq t_c. 
    \end{cases}
\end{aligned}
\end{equation}
We can now analyze the number of photons needed to achieve the same precision and confidence as in the noiseless scenario $Q=1$. Assuming $Q$ is known, we can set the threshold to be $nQ/4$ rather than $n/4$. Now, Hoeffding's inequality implies that
\begin{equation}
    \Pr[f(\tau, \nu_1,\dots,\nu_n) \geq \frac{nQ}{4}]   \leq e^{-\frac{n Q^2}{32}}.
\end{equation}
Therefore, we need 
\begin{equation}
\begin{aligned}
    n &\geq \frac{32}{Q^2} \left[\ln(T/t_c) + \ln(0.005) \right] / Q^2 \\
    &= \Theta\left(\frac{\log (T/t_c)}{Q^2}\right)
\end{aligned}
\end{equation} photons in total, among which $nQ = \Theta(\log (T/t_c)/Q)$ are ``good" photons. In conclusion, if the fraction of noise photons among all received photons is $1-Q$, then we need $1/Q$ times as many signal photons as in the noiseless case.

\subsection{Unequal magnification}
\label{sec:unequal}
In this subsection, we discuss how the fact that the two paths have different amplification affects the performance of our sample-efficient delay-finding scheme. In the classical picture, recall from Eq.~\eqref{eq:unequal_electricfield} that the superposition of two electric waves has $\sqrt{A_\pm}$ as coefficients. In the quantum picture, the increased light intensity leads to a higher number of photons received, and the wave function of each photon should be normalized. 
Therefore, the pure state of a given lensed photon is generalized from Eq.\ (\ref{eq:purestate}) to
\begin{equation}
\begin{aligned}
    \ket{\phi(t_0,\Delta t)} &= \frac{1}{\sqrt{A}}
   \int_{-\infty}^{\infty}\left[ \sqrt{A_+} \alpha(t-t_0) e^{-i\omega_0 t} + \right.\\
   &\quad \left. \sqrt{A_-} \alpha(t-t_0-\Delta t) e^{-i\omega_0 (t-\Delta t)} \right]\ket{t} dt.
\end{aligned}
\end{equation}
Following the same derivation as in \cref{sec:quantum_pic}, one can show that the classical communication channel now becomes
\begin{equation}
\label{eq:channel_with_magnification}
    p_A(\omega|\Delta t) = |\tilde{\alpha}(\omega-\omega_0)|^2 (1+ \gamma_A \cos(\omega \Delta t)), \end{equation}
where $\gamma_A = \sqrt{A^2-1}/A$, and the expectation value of the score function becomes
\begin{equation}
\label{eq:unequalmagscore}
\begin{aligned}
    &\quad \E[f(\tau,\nu_{1},\dots,\nu_{n})] \\
    &\approx \begin{cases}
         \frac{1}{2} n \gamma_A \cos(\omega_0 (\Delta t-\tau)),&|\tau-\Delta t|< t_c  \\
         0,&|\tau-\Delta t| \geq t_c,
    \end{cases}
\end{aligned}
\end{equation}
which is equivalent to having noisy photons with signal-to-noise ratio $\gamma_A$. For the case we focus on, $A=1.34$, we have $\gamma_A = 0.666$, meaning that the number of required photons only increases by a factor of $(1/0.666)^2 = 2.25$ due to the unequal magnification. Additionally, the effects of finite magnification $A$ and noisy photons with rate $Q$ can be combined such that the expectation value of the score function becomes $\frac{1}{2} n Q \gamma_A  \cos(\omega_0 (\Delta t-\tau))$ for $|\tau - \Delta t|<t_c$.

\subsection{Broadband input and coherence time}
\label{sec:broadband}
In the previous analysis, we assume that all photons are within the frequency range set by the bandwidth of the single-photon spectrometer. In other words, we assume all photons are Gaussian wave packets centered at $\omega_0$ with width $\sim 1/t_c$. However, as later discussed in \cref{sec:dwarfflare}, in a realistic setting, a typical photon source emits broadband light of bandwidth $\sim 10^{14}\,\mathrm{Hz}$, and we only have access to several hundred photons that come from this bandwidth. In this case, the carrier frequency $\omega_0$ can be drastically different for different photons, and we must consider the potential effect of these broadband input photons.

Fortunately, Eq.~\eqref{eq:rhoDeltaT} indicates that the $1+\cos(\omega \Delta t)$ modulation in the wave function of the incoming photon is independent of the carrier frequency $\omega_0$. Therefore, the density matrix of the broadband input photons can be written as the integration over $\rho(\Delta t)$ of all possible carrier frequencies, i.e.,
\begin{equation}
    \rho_\mathrm{B}(\Delta t) \approx \mathcal{A} \int_{\omega_\mathrm{L}}^{\omega_\mathrm{R}} p_\mathrm{B}(\omega)(1+\cos(\omega \Delta t)) \ket{\omega}\bra{\omega} d\omega
\end{equation}
where $\mathcal{A}$ is the normalization factor, $\omega_\mathrm{L},\omega_\mathrm{R}$ are the lower and upper limits of the passband, and $p_\mathrm{B}(\omega)$ is the envelope of the frequency distribution, which is usually the black-body spectrum and can be approximately considered as a constant. With the constant $p_\mathrm{B}$ assumption, the classical communication channel can be written as
\begin{equation}
    p(\omega|\Delta t) \approx \frac{1+\cos(\omega \Delta t)}{\omega_\mathrm{R}-\omega_\mathrm{L}}.
\end{equation}
Now, we observe that inferring $\Delta t$ using photons with carrier frequency sampled from a wide range is equivalent to using photons with a broadband wave packet profile. This implies that the \emph{coherence time} in our broadband time-delay estimation task is $\sim 1/(\omega_\mathrm{R}-\omega_\mathrm{L}) \sim 10^{-15}\,\mathrm{s}$, independent of the bandwidth of every single-photon spectrometer. Indeed, one can compute the expectation value of the score function with $p_\mathrm{B}$:
\begin{equation}
\begin{aligned}
    \E&[\cos(\nu_j \tau)] =\int_{\omega_\mathrm{L}}^{\omega_\mathrm{R}} \cos(\nu_j \tau)  p_\mathrm{B}(\omega|\Delta t)  d\nu_j\\
    &=\frac{\left[\sin(\omega(\tau-\Delta t)) \right]_{\omega_\mathrm{L}}^{\omega_\mathrm{R}}}{2(\omega_\mathrm{R} - \omega_\mathrm{L}) (\tau-\Delta t)} + O\left( ((\omega_\mathrm{R} - \omega_\mathrm{L}) \tau)^{-1}  \right),
\end{aligned}
\end{equation}
whose absolute value is close to $0$ when $(\omega_\mathrm{R} - \omega_\mathrm{L}) |\tau - \Delta t| \gg 1$ and is close to $1/2$ when $(\omega_\mathrm{R} - \omega_\mathrm{L}) |\tau - \Delta t| \leq  1$. Therefore, the precision of $\Delta t$ estimation is set by the coherence time $t_c \sim 1/(\omega_\mathrm{R}-\omega_\mathrm{L}) \sim 10^{-15}\,\mathrm{s}$. Henceforth in this paper, when considering the broadband input scenario, $t_c$ is the coherence time determined by the signal bandwidth. Since the former is not tunable for a given passband, the number of required photons in the broadband case (which scales with $\log(T/t_c)$) is solely determined by $T$, the upper limit of $\Delta t$.

\section{Sample complexity lower bound}
\label{sec:lower_bound}
In this section, we prove that $\Omega(\log (T/t_c))$ photons are needed to estimate $\Delta t$ with $t_c$ precision. In other words, no strategy can outperform our \cref{alg:alg1} in terms of sample complexity. We present one rigorous proof based on channel capacity in this section. Note that we will later discuss the connection between the discretized version of the delay-finding problem and the dihedral hidden subgroup problem in \cref{sec:dhsp_discussion}, which gives another proof of the lower bound.

If we consider the gravitational lens as a quantum communication channel in which $\Delta t$ is encoded as $\rho(\Delta t)$, then the \emph{Holevo capacity} quantifies the number of bits encoded in a single copy of the state. Let $p_\mathrm{prior}(\Delta t)$ be the (prior) probability that the lensing time delay is $\Delta t$; then the Holevo capacity is
\begin{equation}
\begin{aligned}
    \chi &= S\left(\int_{-\infty}^\infty p_\mathrm{prior}(\Delta t) \rho(\Delta t) d\Delta t\right) \\
    &\quad  -\int_{-\infty}^\infty p_\mathrm{prior}(\Delta t) S(\rho(\Delta t))d\Delta t \\
    &=: S_\mathrm{left}  - S_\mathrm{right},
\end{aligned}
\end{equation}
where $S(\rho) = -\Tr[\rho \ln\rho]$ is the von Neumann entropy. However, according to Eq.~\eqref{eq:rhoDeltaT}, for all $\Delta t$, the mixed state $\rho(\Delta t)$ can be diagonalized by the Fourier transform if the time window for the photon to arrive is infinitely large. Note that the actual $\rho(\Delta t)$ with finite time window length $T_\mathrm{w}$ can be obtained by truncating $\rho_\mathrm{diag}(\Delta t)$ at $t\in[0,T_\mathrm{w}]$. This implies that the information stored in $\rho(\Delta t)$ is upper bounded by that in the perfectly diagonalized $\rho_\mathrm{diag}(\Delta t)$. In other words, the sample complexity proved using $\rho_\mathrm{diag}(\Delta t)$ is a \emph{lower bound} on the actual sample complexity. Since we aim to prove a lower bound in this section, we simply assume $\rho(\Delta t) = \rho_\mathrm{diag}(\Delta t)$ here, which means that the communication channel is essentially classical, i.e.,
\begin{equation}
\begin{aligned}
    \quad \rho(\Delta t) &= \sqrt{\frac{t_c^2}{\pi}} \int_{-\infty}^\infty   e^{-t_c^2 (\omega-\omega_0)^2} \cdot\\
    &\quad  (1+\cos(\omega \Delta t)) \ket{\omega}\bra{\omega} d\omega
\end{aligned}
\end{equation}
and
\begin{equation}
\begin{aligned}
    p(\omega|\Delta t) &= \sqrt{\frac{t_c^2}{\pi}}   e^{-t_c^2 (\omega-\omega_0)^2} (1+\cos(\omega \Delta t)).
\end{aligned}
\end{equation}

For classical channels, Holevo capacity reduces to classical mutual information, where the von Neumann entropy ($S(\rho)$ for density operator $\rho$) is replaced by Shannon entropy ($S(p)$ for the pdf $p$ corresponding to $\rho$). Note that, since $\Delta t$ and $\nu$ are continuous variables, the Shannon entropy should be replaced by the differential entropy, i.e., for pdf $p(x)$,
\begin{equation}
    S(p(x)) = -\int_{-\infty}^\infty p(x) \ln(p(x))dx.
\end{equation}
Now, to compute the left-hand side of the Holevo capacity / mutual information, we first evaluate the $\Delta t$-averaged density operator in the Fourier basis. Since our prior knowledge about $\Delta t$ is that it may be any value between $0$ and $T$, we can simply set $p_\mathrm{prior}(\Delta t) = 1/T$. Therefore,
\begin{equation}
\begin{aligned}
    &\quad \int_{-\infty}^\infty p_\mathrm{prior}(\Delta t) \rho(\Delta t) d\Delta t\\
    &= \frac{\sqrt{t_c^2}}{T\sqrt{\pi}} \int_{-\infty}^\infty e^{- t_c^2 (\omega_0-\nu)^2} d\nu \cdot  \\
    &\quad \int_0^T (1+\cos(\nu\Delta t)) \ket{\nu}\bra{\nu} d\Delta t\\
    &= \int_{-\infty}^\infty \frac{t_c}{\sqrt{\pi}} e^{-t_c^2 (\omega_0-\nu)^2} \ket{\nu}\bra{\nu} d\nu.
\end{aligned}
\end{equation}
Hence, the corresponding pdf is $p(\nu) = \frac{t_c}{\sqrt{\pi}} e^{- t_c^2 (\omega_0-\nu)^2}$. To avoid putting dimensional quantities into the logarithmic function, we also use its alternative form $p(T\nu) = e^{- t_c^2 (\omega_0-\nu)^2} \cdot t_c / (\sqrt{\pi} T)$, where the continuous variable is changed to $T\nu$. Now,
    \begin{equation}
\begin{aligned}
    S_\mathrm{left} &\coloneqq S\left(\int_{-\infty}^\infty p_\mathrm{prior}(\Delta t) \rho(\Delta t) d\Delta t\right) \\
    &= -\int_{-\infty}^\infty \frac{t_c}{\sqrt{\pi} T} e^{-t_c^2 (\omega_0-\nu)^2} \cdot \\
    &\quad \ln(\frac{t_c}{\sqrt{\pi} T} e^{- t_c^2 (\omega_0-\nu)^2}) d(T\nu).
\end{aligned}
\end{equation}
Next, we compute the second term in the Holevo capacity. Similarly, we also use $p(T\nu|\Delta t) = e^{-4\pi^2 t_c^2(\omega_0-\nu)^2} (1+\cos(\nu \Delta t)) \cdot t_c/ (\sqrt{\pi}T) $:
    \begin{equation}
\begin{aligned}
    S_\mathrm{right} &\coloneqq \int_{-\infty}^\infty p_\mathrm{prior}(\Delta t) S(\rho(\Delta t))d\Delta t \\
    &= \int_{-\infty}^\infty \int_{-\infty}^\infty -\frac{ t_c}{\sqrt{\pi}T^2} e^{- t_c^2(\omega_0-\nu)^2}\cdot \\ 
    &\quad \ln(\frac{ t_c}{\sqrt{\pi}T} e^{- t_c^2(\omega_0-\nu)^2} (1+\cos(\nu \Delta t))) \cdot \\
    &\quad (1+\cos(\nu \Delta t)) \cdot d\Delta t \cdot d(T\nu).
\end{aligned}
\end{equation}
Using $\ln(AB) = \ln A + \ln B$ and the fact that $\int_{0}^T\cos(\nu \Delta t) d\Delta t$ is much less than $T$ when $T\gg 2\pi/\nu$, the above integration can be simplified as
\begin{equation}
\begin{aligned}
     S_\mathrm{right}
    &\approx S_\mathrm{left} - \frac{1}{T} \int_0^T d\Delta t \int_{-\infty}^\infty d\nu  \frac{t_c}{\sqrt{\pi}} e^{- t_c^2 (\omega_0-\nu)^2} \cdot \\
    &\quad \quad \quad \quad \quad \left[ \ln(1+\cos(\nu \Delta t)) \right.\\
    &\quad \quad \quad \quad \quad \quad \left.+ \cos(\nu \Delta t) \ln(1+\cos(\nu \Delta t)) \right]  \\
    &=\,S_\mathrm{left} - \frac{2}{T} \int_{-\infty}^\infty d\nu \frac{t_c}{\sqrt{\pi}}  e^{- t_c^2 (\omega_0-\nu)^2} \cdot  \\
    &\quad \int_0^T d\Delta t \, \left[ \ln(\left|\cos(\frac{\nu \Delta t}{2})\right|) \right. \\
    &\quad\quad  \quad \left.+ \cos(\nu \Delta t) \ln(\left|\cos(\frac{\nu \Delta t}{2})\right|) + \frac{\ln 2}{2} \right] .
\end{aligned}
\end{equation}
We evaluate the $\Delta t$-integrations by first considering the integral over a period:
\begin{equation}
\begin{aligned}
    \int_{0}^{2\pi} \ln(\left|\cos x\right|) dx &= 4\int_{0}^{\pi/2} \ln(\cos x) dx \\
    &= -2\pi \ln 2,
\end{aligned}
\end{equation}
and
\begin{equation}
\begin{aligned}
    &\quad \int_{0}^{2\pi} \cos(2x) \ln(\left|\cos(x)\right|)) dx  \\
    &= 8\int_{0}^{\pi/2} \cos(2x) \ln(\cos x) dx.
\end{aligned}
\end{equation}
To evaluate the second integral above, we use the Fourier series of $\ln(\cos x)$:
\begin{equation}
    \ln (\cos x) = \sum_{k=1}^\infty (-1)^{k+1} \frac{\cos(2kx)}{k} + \ln 2.
\end{equation}
Using
\begin{equation}
\begin{aligned}
    \quad \int_{0}^{\pi/2} \cos(2x) \cos(2kx) dx
    = \frac{\pi \delta_{k,1}}{4},
\end{aligned}
\end{equation}
we find that
\begin{equation}
    \int_{0}^{2\pi} \cos(2x) \ln(\left|\cos(x)\right|) dx = 2\pi.
\end{equation}
We can now compute the original integrals by a change of variables $\nu \Delta t/2 \mapsto x$:
\begin{equation}
\begin{aligned}
    &\quad \int_0^T  \ln(\left|\cos(\frac{\nu \Delta t}{2})\right|) d\Delta t \\
    &= \int_{0}^{\nu T/2} \ln(\left|\cos x\right|) \frac{2}{\nu} dx\\
    &\approx \frac{2}{\nu} \cdot \frac{\nu T / 2}{2\pi} \cdot (-2\pi \ln 2) = -T\ln 2
\end{aligned} 
\end{equation}
and
\begin{equation}
\begin{aligned}
    &\quad \int_0^T \cos(\nu \Delta t) \ln(\left|\cos(\frac{\nu \Delta t}{2})\right|) \\
    &= \int_{0}^{\nu T / 2}\cos(2x) \ln(\left|\cos(x)\right|) \frac{2}{\nu} dx\\
    &\approx  \frac{2}{\nu} \cdot \frac{\nu T /2}{2\pi} \cdot 2\pi = T.
\end{aligned} 
\end{equation}
Therefore,
\begin{equation}
\begin{aligned}
     S_\mathrm{right} & = S_\mathrm{left} - \frac{2}{T} \int_{-\infty}^\infty \frac{t_c}{\sqrt{\pi}} e^{-t_c^2 (\omega-\nu)^2} \cdot \\
     & \quad\quad\quad\quad\quad \quad\quad \left( T - \frac{T\ln 2}{2} \right) d\nu \\
    &= S_\mathrm{left} - (2-\ln 2),
\end{aligned}
\end{equation}
and the Holevo capacity/mutual information is
\begin{equation}
    \chi = S_\mathrm{left} - S_\mathrm{right} = 2-\ln 2,
\end{equation}
a constant value. This implies that one photon in state $\rho(\Delta t)$ encodes up to a constant number of bits of $\Delta t$. Therefore, if one wishes to determine $\Delta t$ with precision $t_c$ in the range of $[0,T]$, the total number of bits needed is $\log(T/t_c)$, hence the optimal sample complexity is $\Omega(\log (T/t_c))$. 

Finally, as a remark, we emphasize the importance that we work in the photon-starved regime such that we receive photons one by one. Consider the scenario where photons appear in pairs such that the mode whose shape is given in Eq.\ (\ref{eq:purestate}) is occupied not by one photon but by two photons. 
In this case, if we simply measure each photon in the time basis, then with $1/2$ probability the two outcomes will be separated by $\Delta t$ with error $t_c$, hence the sample complexity is only $O(1)$, rather than $\Omega(\log(T/t_c))$. However, this scenario is irrelevant to the case of our interest for two reasons. First, as explained later in \cref{sec:dwarfflare}, in a fiducial use case, we only expect to obtain several hundreds of photons in the wide spectrum from optical to near-IR bands in a 1-minute time interval using a state-of-the-art ground-based telescope. Therefore, the number of photons per mode is extremely close to 0, and the probability of obtaining a pair of photons in the same state is even lower. Second, to enable the above constant sample complexity measurement, one must be able to recognize which pair of photons corresponds to two photons emitted into the same mode, as opposed to two photons independently emitted within $\Delta t$. We are unlikely to be able to recognize this unless 
we have a variable source with timescale of variability much shorter than $\Delta t$, which is a property generally associated with strong lensing rather than microlensing. Indeed, when the variability is much faster than $\Delta t$, one can use the classical approach (comparing two arrival times of the same explosion that happened at the source) to measure the time delay, and this approach takes only a constant number of photons in principle, but does not apply to the most general microlensing scenario without a variable source.

\section{Quantum undersampling}
\label{sec:quantum_solution}
Recall that our sample-efficient \cref{alg:alg1} takes as input the frequency-basis measurement outcome of each incident photon with precision $\sim 1/\Delta t$. This requires using a high-resolution spectrometer 
with single-photon sensitivity. Brown-dwarf-mass lenses have $\Delta t \sim 1\,\mathrm{ns}$ corresponding to GHz-level resolution, which is potentially feasible with existing devices such as dual-comb spectrometers, as discussed in \cref{sec:review_sps}. However, for $\Delta t\sim 1\,\mathrm{ms}$, the kHz-level resolution in the optical domain is extremely demanding partially due to the typically limited bandwidth for high-resolution devices (see \cref{sec:review_sps} for a detailed discussion). Direct measurement of single-photon frequency requires the spectrometer to distinguish $\sim \frac{\omega_\mathrm{max}}{1/T} = T\omega_\mathrm{max}$ modes, where $\omega_\mathrm{max}$ is the upper limit of carrier frequency allowed. Given how challenging it is to directly realize such a spectrometer, we would like to explore in this section other, indirect, ways of realizing it.

However, suppose all wave packets are of $1/t'_c$ width (with $t'_c \gg t_c$) in the frequency domain, it suffices to only distinguish the $\sim \frac{1/t'_c}{1/T} = T/t'_c$ frequency modes within the wave packet if one can localize the photon's frequency to a range of $1/t'_c$ width. This localization process can be thought of as a non-demolition measurement of photon frequency with $1/t'_c$ resolution. One possible way to realize such a measurement is via a two-step process, which first splits photons into ports with $1/t'_c$ frequency range (using e.g.\ a diffraction grating) and then uses non-demolition photon detectors to determine which port the photon is in. Once the photon is localized to a frequency range with $1/t'_c$ resolution, it can be sent to a single-photon spectrometer with $1/t'_c$ bandwidth and $1/T$ resolution. One spectrometer suffices provided that a coherent frequency converter is used or the spectrometer has tunable frequency range.
Alternatively, one can use 
$\sim \frac{\omega_\mathrm{max}}{1/t'_c} = \omega_\mathrm{max} t'_c$ spectrometers.  
Combining the photon splitting device with spectrometry still distinguishes $\sim T\omega_\mathrm{max}$ modes, although the implementation is potentially less challenging than direct frequency measurement with a single broadband high-resolution single-photon spectrometer.

In this section, we present another way---that doesn't use standard high-resolution spectrometers---of implementing the delay-finding procedure under the condition that we have already performed non-demolition measurement of the photon's frequency with $\sim 1/t'_c$ resolution. 
In particular, we propose to then use quantum information processing techniques to store the discretized and \emph{undersampled} wave function of a photon in the \emph{time} domain in a quantum memory (\cref{sec:discretize}) and perform the quantum Fourier transform on it (\cref{sec:qft}).

In addition, recall that our quantum-inspired data processing algorithm uses $O(T/t_c)$ classical computation, scaling exponentially with the number of photons. Although the overall cost of our delay estimation procedure is not sensitive to classical data processing, it is still valuable to understand whether the $O(T/t_c)$ algorithm is optimal. Interestingly, the time-domain undersampling approach described in this section allows us to formulate the discretized version of the delay-finding problem (\cref{prob:prob2}), which has a surprising connection to the famous dihedral hidden subgroup problem. Indeed, as presented in \cref{sec:dhsp_discussion}, we prove not only the optimal sample complexity, but also the computational hardness of the delay finding problem. The proof is based on reduction from the dihedral hidden subgroup problem, hence it is highly unlikely any data processing procedure can outperform ours.

\subsection{Discretization by undersampling}
\label{sec:discretize}
We call the approach described in this section \emph{quantum undersampling} because we store the photon's wave function in the time domain in a quantum computer, but we let our quantum memory distinguish only $O(T/t'_c)$ temporal modes, many fewer than $O(T\omega_0)$. This means that the $\tilde{\alpha}(\omega - \omega_0)$ envelope in the frequency domain cannot be faithfully recorded in our quantum memory, and will instead be \emph{aliased} into lower frequencies. To outline the idea, we first store a discretized version of the photonic state $\ket{\phi(t_0,\Delta t)}$ in the time domain using qubits, then perform the QFT to map the state to the (aliased) frequency domain, and run a slightly modified version of \cref{alg:alg1}.

Recall that the state of a lensed photon $\ket{\phi(t_0,\Delta t)}$ is a superposition over real numbers $t$. We divide the time domain into $n_s$ equal bins with length $\tau_s = T/n_s$. To obtain the discretized state, we simply discard bits of each $\ket{t}$ that are less significant than the information indicating which bin it belongs to. This leaves the system in a mixture over the following discretized states with different $\tau_0$ values:
\begin{equation}
\label{eq:state_discretized}
\begin{aligned}
    &\quad \ket{\phi_\mathrm{d}(\tau_0, t_0,\Delta t)}\\
    &\propto e^{-i\omega_0 (\tau_0 -t_0)} \sum_{j=0}^{n_s-1} e^{-\imag \omega_0 \tau_s j} \left[ \alpha(\tau_0 + \tau_s j-t_0) \right. \\
    &\quad \quad \quad \quad \quad \left. + \alpha(\tau_0 + \tau_s j-t_0 - \Delta t) e^{i\omega_0 \Delta t} \right] \ket{j}\\
\end{aligned}
\end{equation}
where 
$\tau_0\in[0,\tau_s]$ is the discarded information.

In order to record both width-$t'_c$ wave packets in the discretized state, we need $\tau_s \ll t'_c$. Since we also wish to reduce resource requirements by setting $\tau_s \gg \pi/\omega$, a reasonable choice is $\tau_s = t'_c / 10$. In addition, since the only operation we need to perform to obtain $\ket{\phi_\mathrm{d}(\tau_0,t_0,\Delta t)}$ from the actual photonic state is discarding partial information, the number of discretized states we can get is the same as the number of photons we can receive, in principle.

Taking into account the fact that both $\tau_0$ and $t_0$ are uniformly random, the actual state is a density operator defined by the pdf $p(\tau_0,t_0)$ and states $\ket{\phi_\mathrm{d}(\tau_0,t_0,\Delta t)}$:
\begin{equation}
\begin{aligned}
    &\rho_\mathrm{d}(\Delta t) = \int_{-\infty}^\infty \int_{-\infty}^\infty p(\tau_0,t_0) \cdot \\ 
    &\quad \quad \quad \ket{\phi_\mathrm{d}(\tau_0, t_0,\Delta t)} \bra{\phi_\mathrm{d}(\tau_0, t_0,\Delta t)} \, dt_0 \, d\tau_0.
\end{aligned}
\end{equation}
Our goal is now to learn $\Delta t$ from copies of $\rho_\mathrm{d}(\Delta t)$. With this, we can formulate another problem as follows.
\begin{problem}[Discretized delay finding]
\label{prob:prob2}
    Learn $\Delta t$ with error up to $t_c$ from as few copies of $\rho_\mathrm{d}(\Delta t)$ as possible.
\end{problem}
In fact, $\rho_\mathrm{d}(\Delta t)$ can be produced from $\rho(\Delta t)$ by discarding unnecessary information, which means that if one can solve \cref{prob:prob2}, then one can also solve \cref{prob:prob1}. In other words, \cref{prob:prob2} is at least as hard as \cref{prob:prob1}.

\subsection{Quantum Fourier transform}
\label{sec:qft}
With the photonic state stored in the digital quantum computer, we would like to read it in the frequency basis, just as in \cref{alg:alg1}. To do so with qubits, we need to perform quantum Fourier transform to the $\phi_\mathrm{d}(\tau_0,t_0,\Delta t)$ state. 
Note that if we use an array of $n_s$ real values to store $\phi_\mathrm{d}(\tau_0, t_0, \Delta t)$, then the QFT of the state corresponds to the discrete Fourier transform (DFT) of the array, hence we can employ results from classical signal processing to analyze the output state.
We first notice that the carrier-wave oscillation $\omega_0/(2\pi)$ is fast compared to the sampling rate $n_s/T$, so
\begin{equation}
\begin{aligned}
    &\quad \exp(-\imag \omega_0  \tau_s j ) \\
    &= \exp(-2\pi  \imag j \frac{ \omega_0/(2\pi) }{n_s/T})\\
    &= \exp(-2\pi  \imag j \frac{\left(\omega_0/(2\pi)\right) \bmod (n_s/T) }{n_s/T})\\
    &= \exp(-2\pi \imag j \frac{f_\mathrm{alias} T}{n_s}),
\end{aligned}
\end{equation}
where the aliased frequency is $f_\mathrm{alias}\coloneqq\left(\omega_0 /(2\pi)\right) \bmod (n_s/T)$.

To simplify the presentation, we let $\alpha[j]$ denote the list of $\alpha(\tau_s j)$, and use $\tilde{\alpha}[k]$ to denote the DFT of $\alpha[j]$. Similarly, $\alpha[j+(\tau_0 -t_0)/\tau_s] = \alpha(\tau_0 + \tau_s j - t_0 )$ and $\alpha[j+(\tau_0 -t_0 - \Delta t)/\tau_s] = \alpha(\tau_0 + \tau_s j - t_0-\Delta t)$. Now, we can evaluate the QFT by evaluating the DFT of $e^{-2\pi \imag j \frac{f_\mathrm{alias} T}{n_s}} \alpha[j+(\tau_0 -t_0)/\tau_s] $ and of $e^{-2\pi \imag j \frac{f_\mathrm{alias} T}{n_s}} \alpha[j+(\tau_0 -t_0 - \Delta t)/\tau_s] $.

We use the time-shift property of the DFT to derive that
\begin{equation}
    \mathrm{DFT}\left(\alpha[j+(\tau_0 -t_0)/\tau_s]\right) = \hat{\alpha}[k] e^{2\pi \imag (\tau_0 -t_0) k / T}
\end{equation}
and
\begin{equation}
\begin{aligned}
    &\quad \mathrm{DFT}\left(\alpha[j+(\tau_0 -t_0-\Delta t)/\tau_s]\right) \\
    &=\hat{\alpha}[k] e^{2\pi \imag (\tau_0 -t_0-\Delta t) k / T}.
\end{aligned}
\end{equation}
Next, for any sequence $\beta[j]$, the frequency-shift property of the DFT implies that
\begin{equation}
    \mathrm{DFT}\left( e^{-2\pi \imag j \frac{f_\mathrm{alias} T}{n_s}} \beta[j] \right) = \hat{\beta}[k + f_\mathrm{alias} T].
\end{equation}
Let $\beta[j]$ be $\alpha[j+(\tau_0 -t_0)/\tau_s]$ and $\alpha[j+(\tau_0 -t_0 -\Delta t)/\tau_s]$, respectively. Then
\begin{equation}
\begin{aligned}
   &\quad \mathrm{DFT}( e^{-2\pi \imag j \frac{f_\mathrm{alias} T}{n_s}}  \alpha[j+(\tau_0 -t_0)/\tau_s] )\\
    &= \hat{\alpha}[k+f_\mathrm{alias} T] e^{2\pi \imag (\tau_0 -t_0) (k+f_\mathrm{alias} T) / T}
\end{aligned}
\end{equation}
and
\begin{equation}
\begin{aligned}
    &\quad \mathrm{DFT}( e^{-2\pi \imag j \frac{f_\mathrm{alias} T}{n_s}}  \alpha[j+(\tau_0 -t_0 - \Delta t)/\tau_s] ) \\
    &= \hat{\alpha}[k+f_\mathrm{alias} T] e^{2\pi \imag (\tau_0 -t_0 -\Delta t) (k+f_\mathrm{alias} T) / T}.
\end{aligned}
\end{equation}
Now, we can write down the DFT of the array representing the state:
    \begin{equation}
\begin{aligned}
    &\quad \mathrm{DFT}\left(e^{-2\pi \imag j \frac{f_\mathrm{alias} T}{n_s}} \alpha[j+(\tau_0 -t_0)/\tau_s] \right.\\
    &\quad \left.+ e^{\imag \omega \Delta t} e^{-2\pi \imag j \frac{f_\mathrm{alias} T}{n_s}} \alpha[j+(\tau_0 -t_0 - \Delta t)/\tau_s]\right) \\
    &= \hat{\alpha}[k+f_\mathrm{alias} T] e^{2\pi \imag (\tau_0 -t_0 ) (k+f_\mathrm{alias} T) / T} \\
    &\quad \left(1+e^{\imag\omega_0 \Delta t} e^{-2\pi\imag \Delta t (k+f_\mathrm{alias}T)/T}\right).
\end{aligned}
\end{equation}
In other words,
\begin{equation}
\begin{aligned}
&\mathrm{QFT} \ket{\phi_\mathrm{d}(\tau_0, t_0, \Delta t)} \propto e^{2\pi \imag (\tau_0 -t_0 ) (k+f_\mathrm{alias} T) / T} \cdot\\
&\sum_{k=0}^{n_s - 1}  \hat{\alpha}[k+f_\mathrm{alias} T]  \left(1+e^{\imag\Delta t(\omega_0 -2\pi f_\mathrm{alias} - 2\pi k/T}\right) \ket{k}.
\end{aligned}
\end{equation}
Note that here $\hat{\alpha}[k + f_\mathrm{alias}T]$ is still a Gaussian but centered around the alias frequency. The oscillatory feature is now $e^{-2\pi k \imag \Delta t/T }$ rather than $e^{-\imag \omega_0 \Delta t}$ in the continous case, and there is a constant phase factor $e^{\imag \Delta t (\omega_0 -2\pi f_\mathrm{alias})}$ due to the carrier frequency and its alias. Also, since both $\tau_0$ and $t_0$ contribute to the global phase only, QFT of the density operator (which takes into account the distribution of $t_0$ and $\tau_0$) should have the same distribution in the $k$-basis as any $\mathrm{QFT} \ket{\alpha_\mathrm{d}(\tau_0, t_0, \Delta t)}$, which is
\begin{equation}
\begin{aligned}
    p_\mathrm{d}&(k|\Delta t) \propto \left| \hat{\alpha}[k+f_\mathrm{alias} T] \right|^2 \cdot \\ 
    &\cdot \left( 1+\cos(\Delta t(\omega_0 -2\pi f_\mathrm{alias} - 2\pi k/T)) \right). 
\end{aligned}
\end{equation}
Finally, we realize that \cref{alg:alg1} needs to be adapted to the discretized and undersampled scenario. Note that, for the $j$th photon, we can not only measure the (integer) frequency $k_j$, but also obtain the carrier frequency of the photon's wave packet, because every photon detector has a filter and we know which detector has a click. Therefore, we use $\omega_j$ to represent the carrier frequency of the $j$th photon, rather than the same $\omega_0$. We can also compute $f_{\mathrm{alias},j}$ because it is a function of $\omega_j, n_s, T$.
Now, with measurement outcomes $k_1,k_2,\dots,k_n$, the score function in the discretized scenario is
\begin{equation}
\begin{aligned}
    &\quad f_{\mathrm{d}}(\tau,k_1,k_2,\dots,k_n, \omega_1,\omega_2,\dots,\omega_n) \\
    &= \sum_{j=1}^n \cos(\tau (\omega_j-2\pi f_{\mathrm{alias},j} - 2\pi k_j/T)).
\end{aligned}
\end{equation}
We can also formally write down the new algorithm, which solves \cref{prob:prob2}.
\begin{algorithm}[Sample-efficient delay finding by quantum undersampling.]
\label{alg:alg2}
    Step (i): measure the undersampled wave function of $n$ incident photons in the frequency basis to obtain $k_1,k_2,\dots,k_n$ and $\omega_1,\omega_2,\dots,\omega_n$. Step (ii): evaluate $f_{\mathrm{d}}(\tau,k_1,k_2,\dots,k_n, \omega_1,\omega_2,\dots,\omega_n)$ for all $ 10 T/t_c$ candidate $\tau$s. Step (iii): accept any $\tau$ with $f_{\mathrm{d}}(\tau,k_1,k_2,\dots,k_n, \omega_1,\omega_2,\dots,\omega_n) \geq n/4$ as an estimation of the gravitational lensing time delay.
\end{algorithm}
Since $\mathrm{QFT} \ket{\phi_\mathrm{d}(\tau_0, t_0, \Delta t)}$ is simply the discretized and undersampled version of the continuous-variable state, and $f_\mathrm{d}$ is simply the discretized and aliased version of the continuous-variable score function, the correctness and sample complexity results of \cref{alg:alg2} simply follow from the proof for \cref{alg:alg1} in \cref{sec:ouralgo}.

\subsection{Connection to the dihedral hidden subgroup problem}
\label{sec:dhsp_discussion}
The dihedral hidden subgroup problem (DHSP) is a well-studied, notoriously hard problem in 
quantum computing. In this subsection, we observe that the DHSP 
can be reduced to 
\cref{prob:prob2}. In other words, the lensing delay-finding problem is at least as hard as the DHSP. This implies a computational hardness result for our lensing delay-finding problem under a widely held cryptographic assumption.

It is well known that the DHSP reduces to a quantum state learning problem called the \emph{dihedral coset problem} (DCP).
(See Ref.~\cite{childs2010quantum} for a comprehensive introduction to the DHSP and the DCP.) 
In the DCP of size $N$, we are given multiple samples of a quantum \emph{coset state}
\begin{equation}
    \frac{1}{\sqrt{2}}\left(\ket{0,a}+\ket{1,a+l}\right),
\end{equation}
where $a\in \{0,1,\dots,N-1\}$ is uniformly random for each sample, and $l\in\{0,1,\dots,N-1\}$ is a fixed unknown parameter. The goal of the DCP is to find the value of $l$ using the given states.  
The sample complexity of the DCP is the number of states required to determine $l$ with bounded error, and the computational complexity is the amount of computation needed to process the states and learn $l$. Both complexities are typically analyzed in terms of their asymptotic scaling with $N$. An early result of Ettinger and H{\o}yer showed that the sample complexity of DCP is $O(\log N)$ \cite{EH00} (indeed, an analysis of the optimal recovery procedure shows the sample complexity is $\Theta(\log N)$ \cite{bacon2005optimal}). However, no known quantum algorithm for the DCP is efficient (i.e., runs in time $\poly(\log N)$), and indeed, there is no known efficient classical or quantum algorithm for the DHSP. Indeed, the belief that no such algorithm exists underlies the presumed security of lattice-based public-key cryptography \cite{regev2004quantum}.

We show that the delay-finding problem is at least as hard as the DCP.

\begin{theorem}
    There is an efficient quantum reduction from the DCP to \cref{prob:prob2}.
\end{theorem}

\begin{proof}
Observe that the state provided in \cref{prob:prob2} has a similar structure to the coset states in the DCP. In fact, if we let $t_0$ and $\Delta t$ be integers times $t'_c$, the sampling point number be $n_s = T/t'_c$ (rather than $10T/t'_c$), and the carrier frequency $\omega$ be sufficiently slow (even much lower than $1/T$), then the discretized state $\ket{\phi_\mathrm{d}}$ has the form
\begin{equation}
\begin{aligned}
    \ket{\phi_\mathrm{d}(\tau_0, t_0,\Delta t)} &\approx
   \frac{1}{\sqrt{2}}   \left(\ket{\frac{t_0}{t_c}} + \ket{\frac{t_0 + \Delta t}{ t_c}} \right) \\
   &= \frac{1}{\sqrt{2}} (\ket{a} + \ket{a+l}),
\end{aligned}
\end{equation}
where we map between quantities in the DCP and those in the delay-finding problem as follows: $\frac{t_0}{t'_c} \mapsto a$, $\frac{\Delta t}{t'_c}\mapsto l$, $\frac{T}{t'_c} = n_s\mapsto N$. Such states can be obtained by measuring the first qubit of the DCP states in the $X$ basis and post-selecting on $+1$ outcomes. This postselection succeeds with probability $1/2$, resulting in only a factor-of-$2$ overhead in the production of time-delay states from dihedral coset states. 
Thus an efficient algorithm for solving the delay-finding problem can be used to efficiently solve the DCP.
\end{proof}

Our approach to solving the delay-finding problem follows the same strategy as in Ettinger and H{\o}yer's procedure for solving the DHSP by producing coset states, measuring them in an appropriate basis, and inferring $l$ from the results \cite{EH00}. While it uses polynomially many samples, this procedure is computationally inefficient, using exponential (in $\log N$) processing.
Fortunately, the corresponding $\poly (T/t_c)$ classical data processing cost is acceptable for the delay-finding problem, since $T$ and $t_c$ are both constants for a certain observation, while the sample complexity is the real bottleneck in microlensing observation.

\section{Experimental realization}
\label{sec:experiment}
In this section, we outline potential experimental schemes to realize our sample-efficient delay-finding strategies. Note that our system is simply a single-photon spectrometer (i.e.\ a spectrometer with single-photon sensitivity) 
connected to an output port of a large optical telescope. Therefore, in this section, we focus on the implementation of high-resolution single-photon spectrometers. 

We first review the resolution and bandwidth of existing approaches to single-photon spectrometry using direct frequency readout without digitization (\cref{sec:review_sps}). These are the only experimental building blocks needed for \cref{alg:alg1}.
For the quantum undersampling algorithm (\cref{alg:alg2}), we present two schemes. Both schemes are digital and operate with a classical switch acting on the digitization timescale $t'_c/10$. First, we present a linear-optics implementation for the discrete (quantum) Fourier transform of the photonic state, albeit in practice some of the components are used in quantum optics (\cref{sec:linearoptics}). Second, we point out that the state of the light can be transferred to quantum memories, followed by quantum computation (\cref{sec:QCExperiment}). Although the experimental realization of high-quality quantum memories is challenging in the near term, the advantages of the approach in \cref{sec:QCExperiment} compared to that in \cref{sec:linearoptics} are longer possible storage times of the light (compared to delays achievable with delay lines in \cref{sec:linearoptics}) and hence larger allowed values of $T$, and an exponential reduction in the number of gates (compared to the number of beam splitters in \cref{sec:linearoptics})  and memory (compared to the number of delay lines in \cref{sec:linearoptics}) when binary encoding is available. The approach in \cref{sec:QCExperiment} is particularly suited for the application of our algorithm in telescope arrays (\cref{sec:forarrays}) in order to minimize entanglement consumption.

\subsection{Single-photon spectrometry}
\label{sec:review_sps}
Due to the great significance in quantum optics and quantum information processing, a variety of approaches have been developed to measure the frequency of a individual photons. Indeed, single-photon spectrometry is well within the quantum optics toolbox. In this subsection, we review the major achievements and state-of-the-art results in this area. Since longer time delays (corresponding to heavier lenses) require higher frequency resolution, and we only expect to receive a limited number of photons spread over a wide range of the spectrum, we focus on summarizing the spectral resolution and bandwidth of each approach.

As a popular scheme for single-photon spectrometry, frequency-to-time mapping can achieve $10\,\mathrm{GHz}$-level resolution with $10^{12}\,\mathrm{Hz}$-level bandwidth via chirped fiber Bragg gratings \cite{davis2017pulsed} or integrated thin-film lithium niobate phase modulators \cite{zhu2022spectral}. The \emph{time lens} \cite{mazelanik2020temporal,joshi2022picosecond} is also capable of transferring frequency information into temporal information. Remarkably, Ref.~\cite{mazelanik2020temporal} employs a spin-wave modulation method and a gradient echo memory to achieve $20\,\mathrm{kHz}$ resolution with $\mathrm{MHz}$-level bandwidth. Next, on-chip spectrometers based on superconducting nanowire single-photon
detectors  \cite{cheng2019broadband,xiao2022superconducting,kahl2017spectrally} support broadband input with $ 10^{14}\,\mathrm{Hz}$-level bandwidth with typical spectral resolution around $100\,\mathrm{GHz}$. Additionally, the dual-comb \cite{coddington2016dual,picque2020photon,zhong2025broadband} approach has been shown to be powerful in implementing single-photon spectrometry: Ref.~\cite{xu2024near} demonstrates $200\,\mathrm{MHz}$-level resolution with $50\,\mathrm{GHz}$ bandwidth; Ref.~\cite{peng2025single} provides $125\,\mathrm{MHz}$ resolution with $\sim 10\,\mathrm{GHz}$ bandwidth. There is also a frequency-to-space mapping scheme \cite{nagoro2025single} that achieves $120\,\mathrm{MHz}$ resolution with $15\,\mathrm{GHz}$ bandwidth using a single-photon
avalanche diode array. 

Let us analyze the feasibility of our delay finding scheme in the near term. Note that (we will elaborate on this in \cref{sec:dwarfflare}) that our observation plan requires measurements of $\Delta t$ in the range from $10^{-10}\,\mathrm{s}$ to $10^{-3}\,\mathrm{s}$, and this range is mainly set by the notorious finite-source effect, a fundamental problem that cannot be overcome by the development of technology. Therefore, we focus on time-delay measurement in this range.
We notice that a major obstacle is that high-resolution spectrometers are typically narrow-band devices, which limits their practicality for long-$\Delta t$ measurements. The good news is, if we only wish to measure short time delays at or below $10^{-8}\,\mathrm{s}$ level corresponding to lensing objects like brown dwarfs (or we have prior knowledge that promises the lens to be lightweight), then the state-of-the-art dual-comb single-photon spectrometers seem to have the required $10^8\,\mathrm{Hz}$-level resolution with a reasonable bandwidth. The $10\,\mathrm{GHz}$ device bandwidth requires using $\sim 10^4$ spectrometers together to cover the $\sim 10^{14}\,\mathrm{Hz}$ total bandwidth.
However, if we wish to measure the time delay corresponding to black holes of stellar mass (with $\Delta t \gtrsim 10^{-4}\,\mathrm{s}$), then the $\mathrm{kHz}$-level resolution can only be potentially achieved by the most precise spectrometer listed above (the spin-wave modulation method), and its extremely narrow bandwidth requires using prohibitively many spectrometers in parallel.  In conclusion, the difficulty of experimental realization strongly depends on the range of $\Delta t$ of interest: it is much more practical to work with short time delays, while a general-purpose experimental platform for all $\Delta t$ allowed by our observation scheme requires either massive investment or next-generation single-photon spectrometry with both high resolution and large bandwidth.

Finally, we briefly discuss another possible approach to realize high-resolution spectrometry based on ensemble-based quantum memories \cite{heshami16,shinbrough23}. First, if one uses quantum memories based on inhomogeneous broadening 
\cite{afzelius2009multimode,ma2021one,Mejia2025-rf}, one might be able to measure the frequency of the incoming light by projectively measuring which broadening class of atoms the light is stored in. Second, if the frequency of light is mapped to spatial frequency of the ensemble memory, one possible realization of frequency readout is putting a cavity around the memory and coupling the desired spatial frequency components to the cavity sequentially one by one. We leave the implementation of these ensemble-memory-based spectrometry approaches, along with improving existing ones, as potential future directions of research.

\subsection{Linear optics}
\label{sec:linearoptics}
We sketch a possible implementation with linear optics (Fig.~\ref{fig:optics}). Incoming light from the source is collected by a telescope. It is routed onto different paths by a classical switch: e.g., movable mirrors that reflect the light. The switch operates on the digitization timescale $\tau_s = t'_c/10$ set by the coherence time of the light, such that the different paths have $\Theta(1)$ amplitude of interference. Delay lines are introduced for each path such that they arrive at an interferometer at the same time. The delay can be realized by additional fiber that the light needs to traverse, or by an atomic cloud with a high index of refraction. The interference is done by a network of beamsplitters and phase shifters comprising a $P$-port, which can realize the QFT unitary of dimension $P$~\cite{Reck1994OpticsUnitary}. Given the digitization time $O(1/t'_c)$ and the total observation time $T$, the number of ports is $P=O(T/t'_c)$, and the number of beam splitters is $O((T/t'_c)^2)$. Implementation of the fast Fourier transform reduces the count to $O((T/t'_c)\log(T/t'_c))$~\cite{Barak2007OpticalFT,Hillerkuss2010OpticalFTProcessing}. Finally, the intensity of the light at the output of the $P$-port is measured. Since the number of incident photons is small compared to the number of ports, single-photon detectors are necessary and sufficient. Furthermore, the large number of elements and the associated precision, along with the limited delay time achievable in practice (about a microsecond~\cite{Guo2025HighlyEfficientBroadbandOptical}), can make implementation challenging.

\begin{figure}
    \centering
    \includegraphics[width=\linewidth]{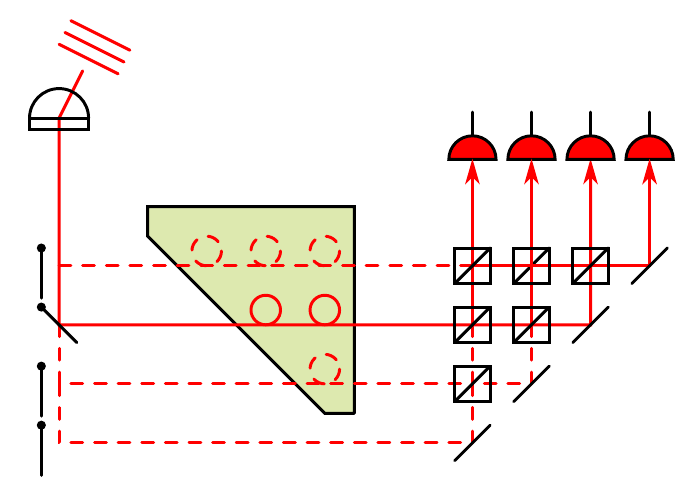}
    \caption{An implementation of the protocol with linear optics. Incoming light is routed by a classical switch onto $O(T/t_c)$ different paths. Time delays are introduced such that the paths jointly interfere at a $P$-port, where $P=O(T/t_c$), followed by measurement with photodetectors.}
    \label{fig:optics}
\end{figure}

\subsection{Quantum computing} \label{sec:QCExperiment}
As a potentially more efficient approach to implement the discretization process and the quantum Fourier transform of \cref{alg:alg2}, storage of light in digital quantum memory followed by digital quantum computation was previously considered in the context of quantum telescope arrays~\cite{Khabiboulline2019OpticalInterferometry}.
In particular, an incoming photon in a superposition of $O(T/t'_c)$ arrival times is coherently stored in $O(\log(T/t'_c))$ qubits. An incoming photon ($\gamma$) is mapped to quantum memory ($a$) in the following way:
\begin{equation}
\begin{aligned}
    &\quad \frac{1}{\sqrt{2}}(\ket{t}+\ket{t+\Delta t})_\gamma \ket{\bar{0}}_a\\
    &\to \ket{0}_\gamma \frac{1}{\sqrt{2}}(\ket{\bar{t}}_a+\ket{\overline{t+\Delta t}}_a) \,,
\end{aligned}
\end{equation}
where $\bar{t}$ denotes the binary representation of $t$. In contrast to collapsing the temporal superposition and performing a Fourier transform on the spatial coherences, as in Ref.~\cite{Khabiboulline2019TelescopeArrays}, now we directly apply a quantum Fourier transform on the quantum memory. Readout of the qubits in the computational basis completes the quantum portion of \cref{alg:alg2}.

The state transfer from light to quantum memory may be realized with entangling gates between flying photons and qubits, mediated by cavities~\cite{Duan2004PhotonicGate}, and non-destructive photon measurements~\cite{Khabiboulline2019TelescopeArrays}. Despite having exponentially-reduced resource scaling with $T/t'_c$, the experimental realization of such a binary-encoded quantum memory remains challenging.

\section{Observation proposal: M dwarf flares}
\label{sec:dwarfflare}
In this section, we explore flare stars as potential sources for a fiducial experimental realization of the mass-measurement algorithm. The idea is to connect the photon measurement setup described in the previous section to a ground-based telescope that will collect these photons and pass them into the measurement device. This telescope would follow up on ongoing microlensing events detected by wide-field microlensing surveys, pivoting to focus specifically on the associated source and to perform a mass measurement of the lens. A given event duration can be estimated via the Einstein crossing timescale for the event, defined in Eq. \ref{eq:tE}: 
\begin{equation}
\label{eq:teval}
\begin{aligned}
    t_E &= \frac{\sqrt{4 G M D_L (1-\frac{D_L}{D_S})}}{v_T c} \approx 4\,\text{days}\,\left(\frac{M}{M_\text{Jup}}\right)^{1/2}
\end{aligned}
\end{equation}
where $G$ is Newton's gravitational constant, $c$ is the speed of light, $M$ is the lens mass, $D_L$ is the lens distance, $D_S$ is the source distance, and $v_T$ is the relative transverse velocity of the lens. We have evaluated at $D_L = 4$ kpc, $D_S$ = 8 kpc, and $v_T = 55$ km/s, which are typical values for Bulge-oriented microlensing surveys. $M_\text{Jup}$ is the mass of Jupiter. As a result, microlensing events in Bulge-oriented surveys can last between  hours and months, depending on the mass of the lens, providing sufficient time for a survey to detect an ongoing microlensing event and alert on it.

We describe the example setup in \cref{sec:flare_setup}. Then in \cref{sec:combineflares}, we present a modification to our algorithm such that photons from different flares in the same M dwarf can be combined to contribute to the same $\Delta t$ estimation.

\subsection{Example setup}
\label{sec:flare_setup}
We describe a particular experimental realization of this protocol, indicating the potential for compelling use cases. We emphasize that this is only one example, and there may be other scenarios in which our approach can improve microlensing observations. 

For our fiducial setup, we will focus on flare stars in the Galactic Bulge as sources. In particular, active M dwarfs may serve as a good target for this protocol. M dwarfs are the most populous stellar type in the Galaxy and exist in great abundance near the Galactic Bulge; therefore, M dwarfs are likely light sources for microlensing events. Their effective surface temperature declines with their radius, so larger M dwarfs ($R\approx 0.5\, R_\odot$) have temperatures of $\approx 3600$ K, while the lowest-mass M dwarfs ($R\approx 0.1\, R_\odot$) have temperatures of $\approx 2400$ K. These temperatures correspond to emission that peaks in the red/near-infrared part of the electromagnetic spectrum, leading to their more colloquial name of ``red dwarfs.''

Despite their small size, M dwarfs are one of the most active stellar types, producing flares that last on the order $1-10$ minutes and release energy in the range $10^{28}\,\mathrm{ergs} - 10^{34}\,\mathrm{ergs}$. Flare temperatures are difficult to constrain without multiwavelength spectra. As such, flare temperatures are usually assumed to be $T_\text{flare} = 9\,000$ K despite evidence that typical flare temperatures may actually be closer to $11\,000\,\mathrm{K}$ \cite{Jackman2023}, with some rare superflares even reaching peak temperatures of $\gtrsim 15\,000$ K \cite{Howard2020}.

For the sake of specificity, we will consider a fiducial M5V dwarf \footnote{The ‘5’ in M5V indicates a spectral subclass, specifying the star’s temperature within the M-type sequence (a lower number indicates a hotter star). The ‘V’ signifies luminosity class five, denoting a main-sequence star that generates energy through hydrogen fusion in its core.
} in what follows. We set the source parameters as $R_S = 0.2 \, R_\odot$, an effective temperature $T = 3060$ K, and a flare frequency distribution that falls off with $\tilde{\nu} = (3 \, \text{day}^{-1}) (\frac{E}{10^{30} \, \text{ergs}})^{-0.65}$, where $\tilde{\nu}$ is the cumulative frequency of flares occurring per unit time for flares with energy $\geq E$ \cite{Paudel2024}. This power law has a range of validity of $\approx 10^{29} \, \mathrm{erg} - 10^{32} \, \mathrm{erg}$.

The physical processes that give rise to flares on M dwarfs are not fully understood, but are believed to be related to magnetic reconnection events in active regions. Both the size and temperature of these regions are not well known, though energetics arguments and the spectra of observed flares suggest emission regions on the order of $2\times10^9 \, \mathrm{cm}$ (or $\approx 3 R_\oplus$) \cite{Paudel2024} for an assumed flare temperature of 9000 K. However, if the flare is actually at higher temperature, then a smaller region can still reproduce the observed energy. We see that,  as a result, smaller, hotter flares are more likely to satisfy the constraint in Eq.~(\ref{eq:finitesource2}), hence are better targets for our lensing scenario. However, applying our algorithm to larger, cooler flares that do not satisfy the finite-source constraint still yields interesting results; since such events would \textit{not} have detectable time delays, a non-detection for a given flare places constraints on the spatial size of the flare emission region.

\begin{table*}
    \centering
    \begin{tabular}{c|c|c}
        \textbf{Symbol} & \textbf{Value} & \textbf{Description}\\
        \hline 
        $D_S$ & 8 kpc &  Distance to source (M dwarf in Galactic Bulge)\\
        $R_\text{dwarf}$ & 0.2 $R_\odot$ & Radius of M dwarf\\
        $T_\text{dwarf} $ & 3060 K & Temperature of M dwarf\\
        $R_\text{flare}$ & 3000 km & Size of flare\\
        $T_\text{flare} $ & $10\,000\,\mathrm{K}$ & Temperature of flare\\
         $\tau_\text{flare}$ & 2 min& Flare duration\\
         $\tilde{\nu}$ & 0.67 day$^{-1}$ & Flare rate for flares with $E > 10^{31}$ erg\\
         $[\lambda_\text{min}, \lambda_\text{max}]$ & [365 nm, 510 nm] & Telescope passband
    \end{tabular}
    \caption{Fiducial model parameters of flaring M dwarf and detector.}
\label{tab:fiducialdwarf}
\end{table*}

We wish to explore the parameter space in which stellar flares may be a viable target for such an experiment. As described above, we adopt a fiducial M5V dwarf as our target, situating it in the Galactic Bulge at $D_S = 8 \, \mathrm{kpc}$. 
We assume a flare temperature $T_\text{flare}$ of $10\,000 \, \mathrm{K}$, a flare size $R_\text{flare}$ of 3000 km, and a flare duration $\tau_\text{flare}$ of 1 minute. This corresponds to a flare with total energy $\approx 10^{31} \, \mathrm{erg}$, which, given our fiducial flare frequency distribution, would occur at a rate of $\approx 0.67/$day.
To compute the number of signal/background photons incident on our telescope per flare, we integrate  over the flare/dwarf spectrum (modeling both as blackbodies at their respective temperatures) in the passband ranging from $\lambda_\text{max} = 510 \, \mathrm{nm}$, the peak wavelength of the flare blackbody spectrum, to $\lambda_\text{min} = 365 \, \mathrm{nm}$, the cutoff wavelength determined by Eq.~(\ref{eq:finitesource2}), with an additional factor $\varepsilon = 0.2$, a heuristically determined prefactor necessary for finite-source effects to be negligible.
We summarize fiducial model parameters in \cref{tab:fiducialdwarf}.

Additionally, we account for the effects of dust and atmospheric extinction in the following way. We calculate the dust extinction using the observationally-derived extinction coefficients towards the Galactic Bulge found by Ref.~\cite{Saha2019}, interpolating logarithmically between the values displayed in their Fig.~8. We set the color excess between $r$ and $z$-band filters, $E(r-z)$, to 0.5, as taken from the results of Ref.~\cite{Saha2019} in Baade's window. We find that the resulting $A_v$ is consistent with the findings of Ref.~\cite{Stanek1996}. With the interpolated $A_\lambda(f)$ from their Fig.~8 with $E(r-z) = 0.5$, we convert to optical depth as $\tau_\text{dust}(f) = A_\lambda(f) / 1.086$. Additionally, we include the effects of atmospheric extinction using the fiducial values provided for Mauna Kea in the Gemini Observer's Guide \cite{GeminiObservatory2024}, which result in a typical suppression of $\approx 20\,\%$ in our passband. The total optical depth is therefore $\tau(f) = \tau_\text{dust}(f) + \tau_\text{atmo}(f)$.

The resulting expressions for the number of signal (background) photons per fiducial flare per minute $n_\text{sig}$ ($n_\text{bg}$) are given by
\begin{widetext}
    \begin{equation}
\label{eq:nsig}
    n_\text{sig} = A_\text{telescope} \tau_\text{flare} \left(\frac{R_\text{flare}}{c\, D_S}\right)^2 \int_{f_\text{min}}^{f_\text{max}}  e^{-\tau(f)}  \frac{2\pi f^2}{\exp[hf/kT_\text{flare}] - 1} df = 0.44 \times \left(\frac{A_\text{telescope}}{1\,\text{m}^2}  \right),
\end{equation}
\begin{equation}
\label{eq:nbg}
    n_\text{bg} = A_\text{telescope} \tau_\text{flare} \left(\frac{R_\text{dwarf}}{c\, D_S}\right)^2 \int_{f_\text{min}}^{f_\text{max}}  e^{-\tau(f)}  \frac{2\pi f^2}{\exp[hf/kT_\text{dwarf}] - 1} df = 0.69 \times \left( \frac{A_\mathrm{telescope}} {1\,\text{m}^2} \right).
\end{equation}
\end{widetext}
Our technique performs best when telescope collecting area is large.  However, we find that often, due to the limited number of photons received per flare, a single flare is insufficient to measure the lens mass with $95\,\%$ confidence. As a result, we must combine measurements across multiple flares to make a confident measurement. In order to achieve this, we modify \cref{alg:alg1} and \cref{alg:alg2} to propose a new algorithm, which is elaborated in \cref{sec:combineflares}. Fortunately, the total number of photons needed to measure the mean value of $\Delta t$ to $t_c$ precision with $95\,\%$ confidence is still $O(\log (T/t_c))$ when the number of photons per flare is not too small.
We defer the technical details of the derivation to \cref{sec:combineflares}.

Note that numerical evaluation of $n_\mathrm{sig}$ yields up to several hundred signal photons for a typical $1$-minute flare for the largest ground-based telescope in the foreseeable future. Since this is the photon number corresponding to the entire passband from $\lambda_\mathrm{max}=510\,\mathrm{nm}$ to $\lambda_\mathrm{min}=365\,\mathrm{nm}$, it is a broadband signal with coherence time $t_c \sim 10^{-15}\,\mathrm{s}$. According to \cref{sec:broadband}, when $\lambda_\mathrm{min},\lambda_\mathrm{max}$ are fixed, the size of the search space, $T/t_c$, is solely determined by $T$, the upper bound of the time delay.

We present numerical simulation results for our delay-finding algorithm for a single flare in \cref{fig:conf_vs_n}. The vertical axis is the confidence level of the delay measurement achieved for one flare that produces $n_\text{sig}$ photons, and the horizontal axis is the number of signal photons received per flare.
We see that in our fiducial case, when $T=10^{-9}\,\mathrm{s}$ (corresponding to $T/t_c=10^6$), a single flare would need to yield $\approx 300$ signal photons collected in our telescope to achieve a $70\,\%$ confidence detection, and $\approx 500$ to achieve a $95\,\%$ confidence detection. For a longer possible time delay, $T=10^{-7}\,\mathrm{s}$ (corresponding to $T/t_c=10^8$), a single flare with $400$ signal photons can achieve $\geq 60\,\%$ confidence detection. Note that we have also included a zero-background limit ($Q=1$, blue), which while not applicable to the flare scenario, is more broadly applicable to possible other isolated sources. We see that if an isolated source satisfies the finite-source condition, it would require collecting only up to 200 total photons from that source to measure the time delay at $90\,\%$ confidence. This impressive sensitivity to even very faint sources may have broader applications beyond the flare scenario studied here.

\begin{figure}
    \centering
    \includegraphics[width=\linewidth]{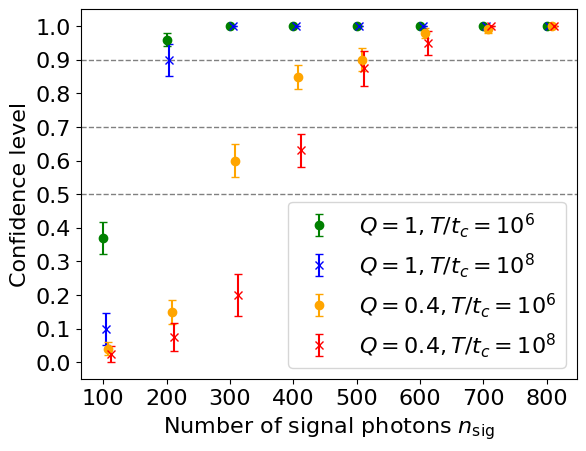}
    \caption{This figure shows, with two possible values of the signal-to-background ratio $Q$, how the confidence level of $\Delta t$ measurements increases with the number of signal photons. The confidence level for each $n_\mathrm{sig}$ is computed by numerical simulation of a scenario with $A=1.34$.
    Here $Q=1$ corresponds to the scenario with no background or noise photons, $Q=0.4$ corresponds to the fiducial example of M-dwarf flares considered in our observation proposal, $T/t_c = 10^6$ corresponds to $\Delta t\leq T=10^{-9}\,\mathrm{s}$, and $T/t_c = 10^8$ corresponds to $T=10^{-7}\,\mathrm{s}$.
    The $90\,\%$ confidence, $70\,\%$ confidence, and $50\,\%$ confidence are marked by horizontal dashed lines.}
    \label{fig:conf_vs_n}
\end{figure}

\begin{table*}[t]
\centering
\begin{tabular}{c|c||c|c|c||c|c|c||c|c|c||c|c|c}
\toprule
\hline
\multicolumn{2}{c||}{$T/t_c$} & \multicolumn{3}{c||}{$10^6$} & \multicolumn{3}{c||}{$10^7$} & \multicolumn{3}{c||}{$10^8$} & \multicolumn{3}{c}{$10^9$}\\
\multicolumn{2}{c||}{$M/M_\mathrm{Jup}$} & \multicolumn{3}{c||}{$\sim 0.1$} & \multicolumn{3}{c||}{$\sim 1$} & \multicolumn{3}{c||}{$\sim 10$} & \multicolumn{3}{c}{$\sim 100$}\\
\multicolumn{2}{c||}{$t_E$\,(days)} & \multicolumn{3}{c||}{$\sim 1.3$} & \multicolumn{3}{c||}{$\sim 4$} & \multicolumn{3}{c||}{$\sim 13$} & \multicolumn{3}{c}{$\sim 40$}\\
\multicolumn{2}{c||}{$m_\mathrm{typ}$} & \multicolumn{3}{c||}{$ \lesssim 1$} & \multicolumn{3}{c||}{$\lesssim 3$} & \multicolumn{3}{c||}{$\lesssim 9$} & \multicolumn{3}{c}{$\lesssim 27$}\\
\hline
\hline
\multicolumn{2}{c||}{$n_\mathrm{sig}$} & 150 & 200 &400 & 150 & 200 &400  & 150 & 200 &400  & 150 & 200 &400\\
\hline
\hline
\multirow{3}{*}{$m$}& $50\,\%$ confidence  & \cellcolor{red!15}6 & \cellcolor{red!15}4 &\cellcolor{blue!15}1 & \cellcolor{red!15}10 & \cellcolor{red!15}5 &\cellcolor{blue!15}1 & \cellcolor{red!15}11 & \cellcolor{blue!15}5  & \cellcolor{blue!15}1 & \cellcolor{blue!15}18 & \cellcolor{blue!15}6 &\cellcolor{blue!15}2\\
\cline{2-14}
&$70\,\%$ confidence  & \cellcolor{red!15}10 & \cellcolor{red!15}5 & \cellcolor{blue!15}1 & \cellcolor{red!15}13 & \cellcolor{red!15}6 & \cellcolor{blue!15}1 & \cellcolor{red!15}15 & \cellcolor{blue!15}6 & \cellcolor{blue!15}2 &\cellcolor{blue!15}22 & \cellcolor{blue!15}8 & \cellcolor{blue!15}2 \\
\cline{2-14}
&$90\,\%$ confidence   & \cellcolor{red!15}>10 & \cellcolor{red!15}7 & \cellcolor{blue!15}2 & \cellcolor{red!15}$>13$ & \cellcolor{red!15}8& \cellcolor{blue!15}2 & \cellcolor{red!15}$>15$ & \cellcolor{red!15}11 & \cellcolor{blue!15}2 &\cellcolor{red!15}$\geq 27$ & \cellcolor{blue!15}12 & \cellcolor{blue!15}2\\
\hline
\bottomrule
\end{tabular}
\caption{Summary of the number of flares needed to achieve certain confidence levels. We use the baseline setting discussed in our M-dwarf flare observation proposal: $Q=0.4, A=1.34, t_c = 10^{-15}\,\mathrm{s}$. For each $T/t_c$, we compute its corresponding lens mass $M$ and roughly estimate the event duration $t_E$ using Eq.~\eqref{eq:deltat_mass} and Eq.~\eqref{eq:teval}, respectively. With the assumption that a flare happens every $\sim 1.5$ days on average, we estimate the typical number of flares one can expect for each duration of event, denoted by $m_\mathrm{typ}$.
}
\label{tab:confidences}
\end{table*}

Simulation results for the multiple-flare combination are shown in \cref{tab:confidences}. This table shows how the number of flares needed to provide a detection changes with the number of photons per flare $n_\text{sig}$, confidence level, and the ratio of maximum time delay to coherence time, $T/t_c$.

If we consider an example case of using the two Keck telescopes atop Mauna Kea, which have a total collecting area of 152 m$^2$, the resulting signal and background yields are $\approx 132$ ($\approx 198$) and $\approx 210$ ($\approx 315$) per flare, respectively, assuming each flare has $2$-minute ($3$-minute) duration. As a result, we see from \cref{tab:confidences} that Keck's collecting area is insufficient to make a confident mass measurement with one flare alone. As indicated in \cref{tab:confidences}, for Keck, we would need to combine $\gtrsim 7$ flares with $3$-minute duration to measure $\Delta t$ with $70\,\%$ confidence when $T/t_c = 10^8$, corresponding to $M \approx 10\,M_\mathrm{Jup}$. Given our fiducial flare rate, this corresponds to a roughly 11-day observation, which is not unreasonable since the typical duration of a lensing event with $10\,M_\mathrm{Jup}$ lens mass is $\gtrsim 13\,\mathrm{days}$, according to Eq.~\eqref{eq:teval}.
Additionally, it is worth noting that while we have restricted ourselves to studying 2-minute or 3-minute flares that produce $10^{31}$ ergs, the observation of a single 10-minute flare ($n_\text{sig}$ = 660) at the same temperature would allow the mass to be estimated with very high confidence for various settings of $T/t_c$ and $Q$, as suggested by \cref{fig:conf_vs_n}.

The prospects are even better for next-generation extremely large telescopes like the currently under construction Extremely Large Telescope \cite{Padovani2023}. With a collecting area of 978 m$^2$, the per-flare photon yield for a 1-minute (2-minute) flare would be $\approx 426$ ($\approx 852$) for signal photons and $\approx 677$ ($\approx 1352$) for background photons. We see from Fig. \ref{fig:conf_vs_n} that in this scenario, photons \textit{from a single $1$-minute flare} are sufficient to measure $\Delta t$ with $> 70\,\%$ confidence for $T/t_c=10^6$, and a \emph{single 2-minute flare} would enable $\Delta t$-measurement for $T/t_c > 10^8$ with $>95\%$ confidence. This would enable the observation of lensing events with duration at least $\approx 1.5$ days, or by Eq.~\eqref{eq:teval}, lens masses above $\approx 0.2 M_\text{Jup}$. Note that much below this mass, the events would become so short that it would be challenging for microlensing surveys to identify that microlensing is occurring in time for us to reorient our telescope. As such, increasing the collecting area of the telescope does not appreciably improve upon the lower mass limit.

We see that with existing and near-future telescopes, our proposed method could potentially measure the mass of lenses with masses greater than roughly that of Jupiter. The mass range of isolated objects near and above $M_\text{Jup}$ is an exceptionally interesting range to explore, as it is currently poorly understood whether the dominant contribution in this mass range arises from free-floating planets or sub-stellar objects \cite{MiretRoig2023}.
Above $13 \,M_\text{Jup}$, the majority of nonluminous lenses would be dim brown dwarfs that are otherwise unobservable, providing a unique way of building a brown dwarf mass function at Galactic distances. Above $\approx 0.1 M_\odot$, direct mass measurements of observed stellar lenses would allow better calibration of mass-luminosity relations as well as the discovery and characterization of compact objects such as neutron stars and white dwarfs. 
Finally, at super-solar masses ($M\gtrsim M_\odot$), this technique would provide the opportunity to measure the masses of isolated black holes in the mass range probed by LIGO's observation of black hole mergers \cite{LIGO2023BHCatalogue}. Note that observations of this duration would also provide a complementary measurement that would help further break lensing degeneracies, even if, for these longer events, orbital parallax can be detected in the light curve. This technique would also enable the direct mass measurement of primordial black holes \cite{Carr1974,Carr2020} and other hypothesized macroscopic dark matter candidates if the dark matter is in fact composed of such objects. As such, if successful, the application of this technique as a follow-up strategy for microlensing surveys would provide the opportunity to do interesting science across many different sub-fields of astronomy. Additionally, as mentioned above, even a non-detection would provide new insight into the temperature and spatial scale of flare emission regions on M dwarfs.

It is worth noting that in the above analysis, we assumed all photons received from the stellar flare in the 1-minute or 2-minute window are in the state close to the $\rho(\Delta t)$ state in Eq.~(\ref{eq:mixedState}), i.e., a superposition of two light paths with stable relative phase and amplitude, even after experiencing dust extinction and passing through the interstellar medium. We justify this robustness assumption in \cref{sec:robustness}.

\subsection{Combination of multiple flares}
\label{sec:combineflares}

When the number of photons that can be received per flare is smaller, we consider using photons from multiple flares that happen in different areas of the star to make a joint analysis. We can safely assume that the durations of the flares do not overlap. Since the size of the M dwarf is generally much greater than the size limit set by the finite-source effect, the difference in $\Delta t$ between different flares is generally much greater than $2\pi / \omega_0$. Since the photons result from broadband emission, this difference is also much greater than the coherence time $t_c$. Note that if this difference is smaller than the coherence time, this will be the narrow-band scenario, in which combining multiple flares may have better performance. Although the narrow-band case is less realistic, we present its strategy in \cref{sec:narrow} in case this is of independent interest.

Let $n$ denote the expected number of photons received per flare and $m$ be the number of flares.  For simplicity, suppose that all flares yield exactly $n$ photons. We consider a realistic scenario that takes into account the signal-to-background ratio $Q$ and magnification $A$. Recall that the score function of $\tau$ for the frequencies $\nu_{i,j}$ (where $j\in\{1,2,\dots,n\}$) corresponding to the $i$th flare is $f(\tau,\nu_{i,1},\dots,\nu_{i,n}) =  \sum_{j=1}^n \cos( \nu_{i,j} \tau) $ with
\begin{equation}
\begin{aligned}
    &\quad \mathbb{E}[f(\tau,\nu_{i,1},\dots,\nu_{i,n})] \\
    &\approx \begin{cases}
         \frac{1}{2} n Q \gamma_A \cos(\omega_0(\Delta t_i-\tau)),&|\tau-\Delta t_i|< t_c  \\
         0,&|\tau-\Delta t_i| \geq t_c 
    \end{cases}
\end{aligned}
\end{equation}
where $\Delta t_i$ is the time delay for the $i$th flare.

Since the delays for different flares may vary significantly, we aim to find the smallest possible time window of $\tau$ that almost all $\Delta t_i$s fall into. An intuitive way to determine the size of the 
time window is to match the standard deviation of the $\Delta t_i$s, denoted by $\delta_{\Delta t,\mathrm{f}}$. This quantity is determined by the linear size of the host M dwarf, which is typically between $0.1\,R_\odot$ and $0.7\,R_\odot$ (note that in the fiducial case in \cref{tab:fiducialdwarf}, we use $0.2\,R_\odot$). This corresponds to approximately $\delta_{\Delta t,\mathrm{f}}\in [400t_c,20000t_c]$. If we take $\delta_{\Delta t,\mathrm{f}} = 1000t_c$ as a fiducial setting and assume we have prior knowledge that $T=10^{-8}\,\mathrm{s}$ (corresponding to $T/t_c = 10^6$), then our task is to find the correct window containing almost all $\Delta t_i$s among the $T/\delta_{\Delta t,\mathrm{f}} = 1000$ time windows.

To design an algorithm for the above task, we propose a score function for the combination of multiple flares. A straightforward idea is directly summing the absolute values of every flare's score function, and checking which time window has the highest mean score or maximum score. Note that such a sum of score functions is unlikely to create a significantly large value at one specific $\tau$ when the number of flares is small, because every peak in the score function only has width $t_c$ and the probability that two peaks overlap is tiny.

We generalize the above idea to give the following \emph{sum-$L_p$} score functions. Letting $p\in[0,\infty]$ be the order of the $L_p$ norm, and $j\in\{0,1,\dots,T/\delta_{\Delta t,\mathrm{f}}-1\}$ be the index of time window, we find that
\begin{equation}
    G_{L_p}(j) \coloneqq \sum_{i=0}^{m-1} \left[\sum_{k=0}^{\delta_{\Delta t,f}/t_c -1} \left|f_i(j \delta_{\Delta t,\mathrm{f}} + kt_c)\right|^p\right]^{1/p}
\end{equation}
is our score function for the $j$th window (corresponding to $\tau \in [j \delta_{\Delta t,\mathrm{f}}, (j+1) \delta_{\Delta t,\mathrm{f}}]$). The time window inferred from $m$ flares using the above score function is
\begin{equation}
    j_\mathrm{opt} = \arg\max_j G_{L_p}(j).
\end{equation}
Note that $p$ can be freely chosen to adjust our focus---whether we emphasize the average score of the window or its peak score. For $p=1$, this is the standard summation of score function values, while for $p=\infty$, we take the maximum among all score function values inside the window. To choose the best $L_p$ norm, we perform a numerical analysis for a small-scale example where $m=10$ flares are combined, each flare has $n_\mathrm{sig}=132$, with $A=1.34$, $Q=0.4$, $T/t_c = 10^4$, and $\delta_{\Delta t,\mathrm{f}}/t_c = 400$. We summarize the results in \cref{tab:norms}.
\begin{table}
\centering
\begin{tabular}{c|c}
$p$ & confidence level\\ 
\hline
$1$ & $52\pm7\,\%$ \\ 
$2$ &  $64\pm7\,\%$ \\
$\infty$ & $96\pm 3\,\%$
\end{tabular}
\caption{Results for numerical simulation of the performance of sum-$L_p$ score functions for $n_\mathrm{sig}=132$, $m=10$, $T/t_c = 10^4$, and $\delta_{\Delta t,\mathrm{f}}/t_c = 400$. The confidence level is the probability for the algorithm to find the correct time window containing the $\Delta t_i$s.}
\label{tab:norms}
\end{table}
Clearly, $G_{L_\infty}$ is the score function with best performance. Intuitively, if $\tau$ is close to $\Delta t_i$, $f_i(\tau)$ will be a peak with high probability, hence the window containing correct $\Delta t_i$s will more likely have a higher peak value.

\begin{figure}
    \centering
    \includegraphics[width=1.0\linewidth]{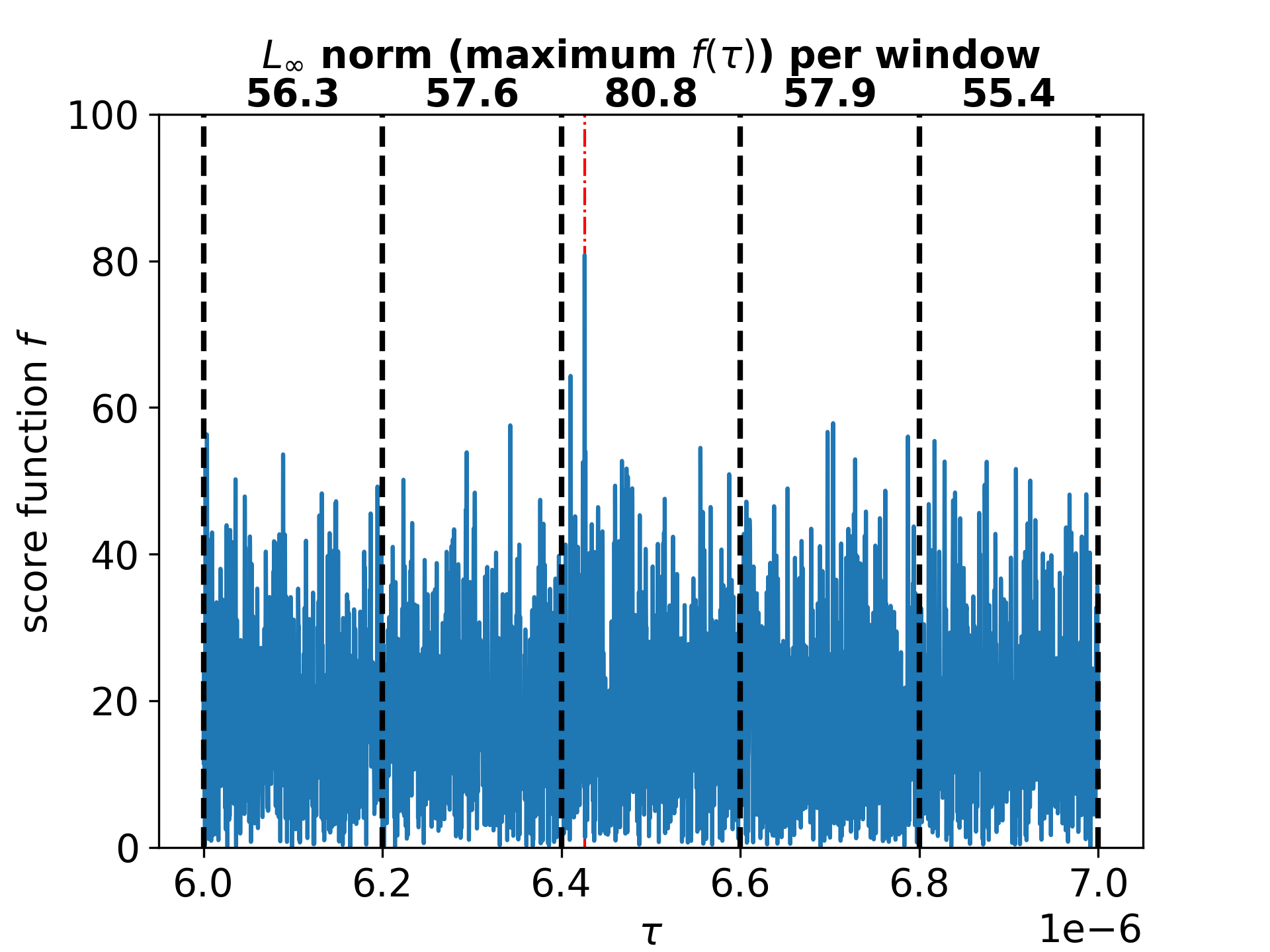}
    \caption{A schematic illustration of the \(L_\infty\)-norm–based flare combination algorithm. The plot shows the score function $f(\tau)$ evaluated over a range of $\tau$s for a single flare with $\Delta t \approx 6.426\times 10^{-6}\,\mathrm{s}$ (marked by the red dash-dot line). In this example, $\Delta t$ is promised to lie within $[6,7]\times 10^{-6}\,\mathrm{s}$ with standard deviation $\sim 10^{-7}\,\mathrm{s}$ caused by the finite-source effect. Therefore, we partition the $\tau$ range into five time windows of width $2\times 10^{-7}\,\mathrm{s}$, indicated by dashed vertical lines. For each window, the $L_\infty$-norm (i.e. the maximum) of all $f(\tau)$ values within each window is computed and displayed above the corresponding window. When combining multiple flares, the $L_\infty$ norms are summed across flares for each time window, and the window with the largest sum of $L_\infty$ norms is the estimator for $\Delta t$.}
    \label{fig:placeholder}
\end{figure}

Indeed, our numerical simulation (see \cref{tab:confidences}) shows that our flare combination strategy offers significant improvement over the single-flare case even if combining just a limited number of flares. These numerical results are based on more realistic settings where $T/t_c \in\{ 10^5,10^6,10^7, 5\times 10^7\}$ and $\delta_{\Delta t,\mathrm{f}}/t_c = 10^3$. Note that each $T/t_c$ corresponds to a typical duration of the lensing event, $t_E$, and subsequently a typical number of flares to expect, $m_\mathrm{typ}$. When the number of flares needed is smaller than $m_\mathrm{typ}$ for a certain confidence level, then with high probability one can get sufficiently many flares to obtain a $\Delta t$ estimation with that confidence. Observe that for a certain $n_\mathrm{sig}$, the required $m$ increases slowly while $T/t_c$ increases exponentially, which aligns with the result that the total number of photons ($n_\mathrm{sig}m/Q$) is proportional to $\log(T/t_c)$. Therefore, typically, our flare combination strategy works better with longer time delay. Indeed, Keck is capable of measuring $\Delta t$ with $98\,\%$ confidence when $T/t_c = 5\times 10^7$.

We emphasize that the numerical results presented in this subsection are based on a specific fiducial setting. In practice, there are more realistic factors affecting the algorithm's performance. For instance, the time delay will drift during the whole lensing event due to the transverse velocity. However, we claim that this effect, along with many other possible issues, can be taken into account in algorithm design. For the drift of time delay, one can use the magnification data and \cref{fig:fu} to model the change of $\Delta t$ and fit the frequency measurement outcomes to recover the entire time-delay curve. A more detailed exploration of possible variants of this algorithm is beyond the scope of this paper.

\section{Robustness of photonic states against the medium}
\label{sec:robustness}
As mentioned in the previous section, the interstellar medium may cause various effects on the observation, including photon loss due to dust extinction \cite{Saha2019} (corresponding to the imaginary part of the refractive index) and additional phases imprinted by gases (corresponding to the real part of the refractive index). In this section, we perform a thorough analysis of the robustness of our observation scheme (or, more fundamentally, the $\Delta t$-information in the optical signal) to both effects. We also briefly discuss the final section of the medium between the source and the telescope, which is our atmosphere, in \cref{sec:atmos}, and claim that it does not affect the stability of our measurement.

\subsection{Dust extinction}
\label{sec:dustExtinction}
The increased telescope size requirement due to the photon loss caused by dust extinction has been taken into account in the analysis of the feasibility of our example setup in Eqs.~(\ref{eq:nsig}, \ref{eq:nbg}) in \cref{sec:flare_setup}. However, since our delay-finding approach is based on the interference of two branches of the same particle, we additionally require each photon received from the source to be in one of the pure states described by Eq.~\eqref{eq:purestate} (or a state close to it) with a random $t_0$. This means that the state must be a superposition of two branches with comparable and stable weight with a stable relative phase between them. In this subsection, we discuss the potential effect of dust
on the weights in the superposition.

If dust extinction rate were tiny, the photons would barely interact with the dust, and the weights in the superposition would thus be essentially unchanged. However, as presented in \cref{sec:flare_setup}, a significant fraction of photons is indeed lost in the observation setup of our interest, hence we must take a closer look at the physical process of dust extinction. To determine the effect of dust extinction on the superposition, we assume that dust extinction in the two paths is uncorrelated. We therefore begin by studying the effect of dust extinction on one path. Intuitively, there are two models of dust extinction that can lead to the same photon loss rate (denoted by $p_\mathrm{loss}$) along one path:
\begin{enumerate}
    \item The \emph{random wall} model. At any moment, there is a $p_\mathrm{loss}$ chance that an opaque wall will appear randomly. The wall completely blocks the path from the source to the telescope.
    \item The \emph{beam splitter} model. There is always a beam splitter with reflection amplitude $\sqrt{p_\mathrm{loss}}$ in between the path from the source to the telescope.
\end{enumerate}
There is a drastic difference in the photonic state between the two models when considering microlensing, in which every photon takes \emph{two} paths in superposition and both paths are independently experiencing dust extinction described by the same model, either the random wall or the beam splitter model. For simplicity, we assume that both paths in microlensing have the same $p_\mathrm{loss}$ and the same amplitude when there is no photon loss. In the random wall model, the probability that the received photon is in the superposition over the two paths is
\begin{equation}
\begin{aligned}
    &\quad \Pr_\mathrm{wall}[\mathrm{superposition} | \mathrm{received}] \\
    &= \frac{\Pr_\mathrm{wall}[\mathrm{superposition} ,\mathrm{received}]}{\Pr_\mathrm{wall}[\mathrm{received}]} \\
    &= \frac{(1-p_\mathrm{loss})^2 }{p_\mathrm{loss} (1-p_\mathrm{loss}) + 
(1-p_\mathrm{loss})^2 } = 1-p_\mathrm{loss}.
\end{aligned}
\end{equation}
Note that $2p_\mathrm{loss} (1-p_\mathrm{loss}) $ is the probability that only one of the two paths is blocked by a wall, and half of it is the probability that one specific path is blocked and a photon is still received. This probability implies that, even among the (already limited number of) photons, only $1-p_\mathrm{loss}$ of them are signal photons (or ``good" photons), and the remaining $p_\mathrm{loss}$ of them are ``bad" photons. Combining with the analysis in \cref{sec:noisysignal}, we would need a prohibitively high number of photons to perform a successful measurement.

However, the situation is much more favorable in the beam splitter model. Here we use $\ket{1}$ and $\ket{2}$ to denote the state that the photon takes the first path and the second path, respectively. Now, the photonic state without dust extinction is $\frac{1}{\sqrt{2}}(\ket{1}+\ket{2})$. We use an additional qubit to describe the operation of beam splitters: $\ket{\mathrm{received}}$ means the photon passes the beam splitter (i.e., received by the telescope), and $\ket{\mathrm{blocked}}$ means it is reflected by the beam splitter (i.e., blocked by the dust). Hence, the final state is
\begin{equation}
\begin{aligned}
    \ket{\psi_\mathrm{BS}} &= \frac{1}{\sqrt{2}}\left( \sqrt{1-p_\mathrm{loss}}\ket{1,\mathrm{received}} + \sqrt{p_\mathrm{loss}}\ket{1,\mathrm{blocked}}\right.\\
    &\quad \quad \left.+ \sqrt{1-p_\mathrm{loss}}\ket{2,\mathrm{received}} + \sqrt{p_\mathrm{loss}}\ket{2,\mathrm{blocked}}   \right).
\end{aligned}
\end{equation}
We notice that receiving the photon is equivalent to \emph{postselecting} $\ket{\psi_\mathrm{BS}}$ on the state of the second qubit being $\ket{\mathrm{received}}$. The success probability of this postselection is $1-p_\mathrm{loss}$, and the state we obtain is always $\frac{1}{\sqrt{2}}(\ket{1,\mathrm{received}}+\ket{2,\mathrm{received}})$, i.e.,
$\Pr_\mathrm{BS}[\mathrm{superposition} | \mathrm{received}] = 1$, drastically different from the random wall model when $1-p_\mathrm{loss}$ is small. Therefore, it is crucial to determine which of the two models is a better approximation of the effect of dust extinction.

Fortunately, our close examination shows that the reality is much closer to the beam splitter model as long as the size of dust particles are sufficiently small compared to the radius of the telescope. For meter-level telescope, this is almost certainly true because Ref.~\cite{mathis1977size} suggests that $1\,\mathrm{\mu m}$ is a reasonable estimate of particle size. We provide a detailed analysis in \cref{appendix:dust}, in which we model the microscopic dust particle configuration as a binary tree coloring problem and derive an explicit relation between the level of decoherence and the size of dust particles.

\subsection{Variations in refractive index}
\label{sec:refractiveindex}
To maintain a stable time delay measurement, it is essential to ensure the light path remains stable along the light's propagation path during our targeted 1-minute window. Therefore, any medium along the light path may significantly impact the interference critical for the lensing to work, hence we wish to estimate the phase fluctuation alonog a path during our 1-minute observation window.

The phase difference accumulated by passing through the interstellar medium (ISM) has been well-studied in the context of pulsar dispersion and scintillation. The dispersion measure (DM) is the line of sight integral of the electron number density, $\int_0^{D_S} n_e(l) dl$, and can be derived from measured time delays of pulsar radio emissions. The associated time delay is given by
\begin{equation}
    \Delta t_\mathrm{DM} = \frac{\text{DM}}{k \nu^2}
\end{equation}
where $\Delta t_\mathrm{DM}$ is the accumulated time delay, $k$ is a numerical constant equal to $2.41 \times 10^{-4} \, \mathrm{cm}^{-3} \, \mathrm{pc} \, \mathrm{MHz}^{-2} \, \mathrm{s}^{-1}$.

The ISM is composed of several constituents and is often modeled as an approximately Komolgorov turbulent spectrum. As such, fluctuations in the ISM can be describe by the structure function of the dispersion measure, which for a Komologorov spectrum, has a simple power-law scaling with timescale $\tau$,
\begin{equation}
D_\mathrm{DM}(\tau) = \langle[\mathrm{DM}(t+\tau)-\mathrm{DM}(t)]^2\rangle \propto \tau^{5/3},
\end{equation}
and scales linearly with with $D_S$, the distance to the source \cite{Armstrong1995DMVariation,You2007DMVariation}.
The structure function also allows for the estimation of a root-mean-square variability of $\mathrm{DM}/dt$ for a given scale as $\sigma_\text{DM} \approx \sqrt{D_\text{DM}(\tau)}/\tau$ \cite{Donner_2020}.

The dispersion measure variation can be converted into a phase shift using the relation $D_\phi = (\frac{k\nu}{2\pi})^{-2} D_\text{DM}$. Hence we have that typical random phase shifts for a source at distance $D_S$ over timescales of $\tau$ are given by
\begin{equation}
\label{eq:sigphi}
    \sigma_\phi(\tau, D_S, \nu) \approx \frac{2\pi}{k\nu}\sqrt{D_{\text{DM}}(\tau_0, D_{S,0})} \left(\frac{\tau}{\tau_0}\right)^{5/6}\left(\frac{D_S}{D_{S,0}}\right)^{1/2}
\end{equation}

Though there are few distant pulsars with well-measured structure functions, local pulsars \cite{Donner_2020} indicate typical values of $D_\text{DM}(1000\,\text{day})\approx 1\times10^{-6}~\frac{\text{pc}^2}{\text{cm}^6}$ and $D_S \approx 1$ kpc. Using this, and evaluating Eq.~(\ref{eq:sigphi}) for typical values of interest for our proposal, we find
\begin{widetext}
    \begin{equation}
    \sigma_\phi(\tau, D_S, \nu) \approx 10^{-7}\;
    \left(\frac{D_\text{\text{DM}}(1000\,\text{day})}{1\times10^{-6}\,\frac{\text{pc}^2}{\text{cm}^6}}\right)^{1/2}
    \left(\frac{\tau}{1\,\text{min}}\right)^{5/6}\left(\frac{D_S}{8\,\text{kpc}}\right)^{1/2}\left(\frac{\nu}{750\,\text{THz}}\right)^{-1}
\end{equation}
\end{widetext}
over 1 minute, hence
is negligible for the timescales of interest to our proposal. Even taking the pulsar with the highest $D_\mathrm{DM}(\mathrm{1000\,day}) = 2.2\times10^{-4}\,\mathrm{pc}^2\,\mathrm{cm}^{-6}$ value from \cite{Donner_2020}, we still get only $\sigma_\phi = 9.9\times10^{-7}$ (average value for this dataset is $\sigma_\phi = 0.98\times10^{-7}$).

\subsection{Atmospheric fluctuations}
\label{sec:atmos}
We also discuss the noise in $\Delta t$ and the relative phase generated by atmospheric fluctuations and claim that ground-based telescopes are sufficient for our time-delay measurement. Note that the temporal variation in the refractive index of the atmosphere can be quantified by the atmospheric coherence time \cite{roddier1981v,kellerer2007atmospheric}, which is the timescale over which the wave path varies more than a significant fraction of the wavelength. This quantity is typically related to the wavelength and the wind speed. Note that, in our microlensing case, the two branches of the photon take almost the same path in the atmosphere because they are indistinguishable from Earth. Therefore, the refractive index of the atmosphere is the same for both branches, creating no noise in the measured $\Delta t$, provided that the change of refractive index during $\Delta t$ is sufficiently small.
A recent result \cite{osborn2018optical} gives an estimate of the atmospheric coherence time ($4.18\,\mathrm{ms}$) for the Very Large Telescope, much longer than the upper limit of our observable time delay ($1\,\mathrm{ms}$). Therefore, we claim that atmospheric fluctuations are irrelevant for the stability of our measurement, and it suffices to use ground-based telescopes, which is a much more economical choice---compared with space telescopes---for achieving larger telescope sizes.

\section{Time-delay calibration in telescope arrays}
\label{sec:forarrays}

A distributed version of our protocol is naturally suited for synchronizing time delays in telescope arrays (Fig.~\ref{fig:array}). In that scenario~\cite{Gottesman2021TelescopesRepeaters,Khabiboulline2019OpticalInterferometry,Khabiboulline2019TelescopeArrays}, there are $N$ sites that observe incoming light. We can consider a narrow-band point source, as is the case with an artificial guide star. Then the problems associated with finite-source size and undersampling are alleviated, and the receiver can operate over a small range of frequencies. Learning the delays between sites with small sample complexity enables faster calibration of the array before interferometry, or the ability to use fainter sources. 

The state of an incoming photon at the $N$ detectors is described as a mixture over (similar to Eqs.~(\ref{eq:purestate}) and (\ref{eq:mixedState}))
\begin{equation}
\begin{aligned}
    &\quad \ket{\phi(t_0,\Delta t)} \\
    &= \frac{1}{\sqrt{N}} \int_{-\infty}^{\infty} dt\, \alpha(t-t_0) e^{-i\omega_0 t} \cdot \\
    &\quad \left(\ket{t,0,\ldots,0}  + g_2 \ket{0,t + \Delta t_2,0,\ldots,0} \right. \\
    &\quad \, \left.  + \cdots +g_N\ket{0,\ldots,0,t+\Delta t_N} \right) \,,
\end{aligned}
\end{equation}
where $t_0$ is the centroid of the wave packet,  $\omega_0$ is the carrier frequency, $g_i \in \mathbb{C}_2$ ($\lvert g_i \rvert=1$ for a point source) is the spatial coherence between the $i$th detector and the first one, and where $\ket{t,0,\ldots,0}\coloneqq \ket{t}_1\ket{0}_2\ldots\ket{0}_N$ indicates a photon arriving at time $t$ at the first site. Our task is to learn $\{\Delta t_i\}_{i=2}^N$, such that they can be calibrated away, while $\{g_i\}_{i=2}^N$ are unknown. The next step would be to perform a QFT to learn the spatial intensity distribution of the source from $\{g_i\}_{i=2}^N$, since they are Fourier duals, as stipulated by the van Cittert-Zernike theorem~\cite{Zernike1938}. The two steps learn the angular distribution of the incoming light with increasing precision.

In contrast to the setup for gravitational lensing discussed earlier, the incoming light arrives at spatially separated sites. If we had lossless channels, we could bring the paths together and make a measurement in frequency space. However, in practice, channels are lossy with exponential degradation in distance, so we use teleportation \cite{Gottesman2021TelescopesRepeaters, Khabiboulline2019OpticalInterferometry} to overcome this limit, given that entanglement can be purified. Teleportation with minimal entanglement resources necessitates throwing out unnecessary information; hence, we introduce the qubit discretization. We comment on the possibility of a continuous-variable approach employing two-mode squeezed states containing the same number of ebits, but we do not explore this alternative here. 

\begin{figure}[t]
    \centering
    \includegraphics[width=\linewidth]{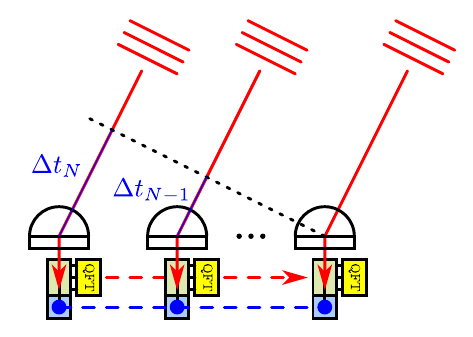}
    \caption{In an array, incoming light acquires relative time delays between $N$ different telescope sites. A distributed version of our protocol, where the state of the light is mapped to one memory using entanglement and then measured in the Fourier basis, learns all the delays efficiently.}
    \label{fig:array}
\end{figure}

Assume that we discretize the arrival times in bins of size $\sim t_c$ and store the light in qubits, as described in Section~\ref{sec:QCExperiment}, and $\ket{0}$ still denotes the absence of a photon. Since there are multiple possible arrival times, $O(\log(T/t_c))$ qubits of memory are necessary, which can be represented as a qudit of dimension $O(T/t_c)$. We map the information about the light from $N$ registers to one using $Z$-teleportation~\cite{Zhou2000GateConstruction} for qudits as follows. Consider $N = 2$ registers for now. First, we perform a generalized controlled-NOT ($\mathrm{CX}$) from register 2 to register 1. Then, we apply the QFT on register 2, followed by measurement of that register. Finally, we apply a measurement-dependent phase correction. The state transforms as
\begin{equation}
\begin{aligned}
    &\quad \quad \frac{1}{\sqrt{2}}(\ket{t,0}+g\ket{0,t+\Delta t}) \\
    &\xrightarrow{\mathrm{CX}_{21}} \frac{1}{\sqrt{2}}(\ket{t,0} + g \ket{t+\Delta t,t+\Delta t}) \\
    &\xrightarrow{\mathrm{QFT}_2} \frac{1}{\sqrt{2}} \left[ \ket{t}\frac{1}{\sqrt{T}} \sum_{j=0}^{T-1} \ket{j} +\right. \\
    & \quad \quad \quad + \left. g \ket{t+\Delta t} \frac{1}{\sqrt{T}} \sum_{j=0}^{T-1} e^{2\pi i j(t+\Delta t)/T} \ket{j}\right] \\
    &= \frac{1}{\sqrt{2T}} \sum_{j=0}^{N-1} \left[\ket{t}\ket{j} + ge^{2\pi i j (t+\Delta t)/T} \ket{t+\Delta t}\ket{j} \right] \\
    &\xrightarrow{\ket{j}\bra{j}_2} \frac{1}{\sqrt{2}}(\ket{t}+g e^{2\pi i j (t+\Delta t)/T} \ket{t+\Delta t}) \\
    &\xrightarrow{Z^{-j}} \frac{1}{\sqrt{2}}(\ket{t}+ g \ket{t+\Delta t}) \,.
\end{aligned}
\end{equation}
For multiple sites, we repeat the procedure: e.g., perform teleportation from site $j$ to site $1$, for every $j\in \{2,\ldots,N\}$. Using preshared entanglement, we can do the two-qudit $\mathrm{CX}$ gates nonlocally, or perform quantum teleportation to transfer all the registers to a single site, so that subsequent operations are local.

At this point, we have a similar setup to lensing.  Now there are multiple delays and spatial coherences. The state in memory is of the form
\begin{align}
\rho(\Delta t) &= \frac{1}{T} \sum_{t=1}^{T} \ket{\phi(t,\Delta t)}\bra{\phi(t,\Delta t)} \,,\\   
\ket{\phi(t,\Delta t)}&=\frac{1}{\sqrt{N}}(\ket{t}+ \sum_{i=2}^N g_i \ket{t+\Delta t_i}) \,,
\end{align}
where we have neglected the vacuum and multiple-photon components. Running \cref{alg:alg2} (undersampling, followed by measurement in the Fourier basis and maximum likelihood estimation) proceeds with the following modifications. First, the probability density (\cref{eq:thechannel}) of measurement outcomes acquires cross-terms $\cos(\omega (\Delta t_i-\Delta t_j))$ where $i\neq j$. These oscillations correspond to pairwise relative delays: there are $\binom{N}{2}$ possibilities, one for each pair of telescope sites. For example, $\Delta t_3 -\Delta t_2$ is the delay between site 3 and site 2. Classical Fourier analysis will identify these frequencies. Second, the constant spatial coherences $g_i$ introduce phase shifts in the sinusoids but do not change the frequency, given by the delay. For repeated values, the score function is multiplied by the number of repetitions. We can assign the pairwise relative delays to a graph, given the spatial configuration of the array. Alternatively, we can extract each $\Delta t_i$ relative to the first site. While this approach has the same entanglement consumption, the sample complexity to output all the delays one-by-one incurs $O(N)$ overhead.

Classical optical interferometers use physical, tunable delay lines. By storing light in quantum memory, we replace spatial delay with temporal delay, which can simplify requisite engineering especially for long baselines. Furthermore, interference accomplished with a Fourier transform achieves the optimal, small sample complexity. Consequently, the calibration can be done quickly even for faint sources.

\section{Summary and discussions}
\label{sec:discussion}
In this work, we investigated a time-delay measurement scheme for microlensing in the optical/IR wavelengths. We first developed novel delay-finding approaches for optical/IR signals based on single-photon quantum information processing technology and quantum-inspired data-processing algorithms. The first approach takes the measured frequencies of individual photons as input, thus requiring high-resolution broadband single-photon spectrometers.
The second approach takes carrier frequencies and aliased frequencies as input, and thus can be implemented by digital quantum computation with undersampling. Our approaches excel in the photon-starved regime because of our provably optimal sample complexity, as established by a channel capacity computation. For the second approach, we also prove a reduction from the dihedral hidden subgroup problem, which gives another proof of optimal sample complexity as well as evidence for the optimal (classical) computational cost.

Although our optical/IR lensing delay-finding approach extends the list of potential observation targets in principle, the more stringent size requirement (compared to the classical proposals for radio frequencies based on FRB) posed by the finite-source effect might also limit its applicability. Fortunately, the logarithmic sample complexity allows us to find use cases of our approaches that are not limited by the finite-source effect. In particular, we proposed a concrete scheme to observe microlensed stellar flares on M dwarfs. These rather short and relatively faint events may have a sufficiently small source size, and our proposal would test this hypothesis as a byproduct of its implementation. We perform comprehensive analysis of the number of photons one can obtain in fiducial cases using existing or near-term ground-based telescopes. To further support the feasibility of such an observation on our proposed platform, we also conducted robustness analysis regarding the coherence with dust extinction and astronomical scintillation, which may be of independent interest.

The main challenge left by our work in conducting the first successful microlensing time delay estimation is the implementation of high-resolution broadband single-photon frequency measurements. We briefly discussed several candidate experimental schemes suitable for proofs of principle or for measuring in a limited range of $\Delta t$, which is already of scientific significance. However, we believe that the ideal devices---ones that measure single-photon frequency with $1~\mathrm{kHz}$ precision in the broad optical/IR band, or ones that support the undersampling process and can store the discretized photonic state in binary encoding---have not yet been demonstrated. This is certainly an interesting future direction of research.

We emphasize that various factors may affect the performance of our proposed scheme, but they are not of the same level of difficulty. Some problems are information-theoretic ones caused by the universe and the laws of physics that erase the $\Delta t$ information from the photons, including the finite-source effect and decoherence due to the interstellar medium. Therefore, it is fundamentally impossible to overcome these difficulties, and we carry out a rather careful analysis in this paper to claim that our observation plan is not seriously affected by them. The other problems are technical difficulties, including the requirement of large light-collecting area and broadband single-photon spectrometry with high resolution. We claim that, although implementing the complete version of our proposal may require next-generation technology, it is not fundamentally impossible, and even the near-term-feasible incomplete version is also of scientific interest. 

We also emphasize that time-delay measurements in the photon-starved regime may find other applications beyond microlensing delay estimation. In fact, we presented one concrete example for long-baseline quantum telescope arrays \cite{Khabiboulline2019TelescopeArrays,Gottesman2021TelescopesRepeaters}. These high-resolution interferometers benefit from the capability to determine $\Delta t$ with $t_c$ precision from a very limited number of photons, since $\Delta t$ is rapidly changing, the search space is large, or the source is faint. Our approach provides a sample-optimal solution to this problem. Also, since a pointlike narrow-band artificial guide star can be used as the source in this case, the finite-source effect can be avoided, and the physical realization of the frequency-resolving device is much easier.

\section*{Acknowledgments}

We thank Tom Barclay, Ivan Burenkov, Ali Fahimniya, Shawn Geller, Elizabeth Goldschmidt, Chaitanya Karamchedu, Emanuel Knill, Alan Migdall, Krister Shalm, Shi Jie Samuel Tan, Ludovic van Waerbeke, Yuxin Wang, Zhi-Yuan Wei, and Gengzhi Yang for helpful discussions. We thank 
Sylvain Veilleux for useful input on estimating the variation in the interstellar medium. We thank Brittany McClinton, Jayadev Rajagopal, and Stephen Ridgway for discussions on timing issues and delays in telescope arrays. W.D. and Z.L. wish to thank Greg Kahanamoku-Meyer for introducing them to each other.
A.M.C., A.V.G., D.G., and Z.L.\ acknowledge support from the National Science Foundation (QLCI grant OMA-2120757). A.V.G.~was supported in part by ONR MURI, AFOSR MURI, ARL (W911NF-24-2-0107), DARPA SAVaNT ADVENT, NSF STAQ program,   and NQVL:QSTD:Pilot:FTL. A.V.G.~also acknowledges support from the U.S.~Department of Energy, Office of Science, National Quantum Information Science Research Centers, Quantum Systems Accelerator (QSA).
A.M.C.\ and A.V.G.\ acknowledge support from
the DoE ASCR Quantum Testbed Pathfinder program (awards No.~DE-SC0019040 and No.~DE-SC0024220) and from the U.S.~Department of Energy, Office of Science, Accelerated Research in Quantum Computing, Fundamental Algorithmic Research toward Quantum Utility (FAR-Qu).
W.D. and A.H. were supported by NSF grant PHY-2210361 and the Maryland Center for Fundamental Physics. 
S.G.~acknowledges support from the Natural Sciences and Engineering Research Council of Canada, John I. Watters Research Fellowship (award No.~6701), and the Ruhr-Universit\"at Bochum Research Scholarship for Doctoral Researchers from International Partners. S.G.~also thanks the German Center for Cosmological Lensing at the Ruhr-Universit\"at Bochum for hosting him as a visiting scholar during the course of this work. The German Center for Cosmological Lensing is supported by the Max Planck Society and the Alexander von Humboldt Foundation in the framework of the Max Planck-Humboldt Research Award endowed by the German Federal Ministry of Education and Research.
E.T.K.~acknowledges support from the NRC Research Associateship Program at the National Institute of Standards and Technology (NIST), administered by the Fellowships Office of the National Academies of Sciences, Engineering, and Medicine.

\section*{Data Availability}
The simulation data underlying the results shown in \cref{fig:conf_vs_n} and \cref{tab:confidences} are publicly available on Zenodo at \doi{10.5281/zenodo.18408227}.

\appendix
\onecolumngrid
\section{Beam splitter model for dust extinction}
\label{appendix:dust}

The modeling of dust extinction depends on various properties of the interstellar medium. Exploring the exact composition of the interstellar medium is beyond the scope of this paper and is unlikely to succeed due to the limited amount of decisive research progress in this area. However, since we would be satisfied with solid evidence supporting either model, it suffices to study some reasonable approximations of \emph{Mie scattering}, the general theory of interaction between light and dust particles, and derive upper/lower bounds for the  fraction of ``bad" photons. Therefore, in this analysis, we consider all dust particles as opaque spherical objects with extinction cross-section $A_c$, where the exact value of $A_c$ is determined by the Mie theory using particles' (linear) size $a$ and the wavelength of light $\lambda$ \cite{mie1908beitrage,li2008optical}, which is typically upper bounded by a constant factor times the geometric cross section $\pi a^2$. When a photon falls within the extinction cross-section of the particle, it is either scattered in various directions or absorbed by the particle. Since the probability that the scattered photon still goes to our telescope is tiny, we assume the photon is lost whenever it is scattered. Under this assumption, we can derive a theory to unify both the random wall model and the beam splitter model as follows.

For simplicity, we assume dust particles are of the same size and material, and focus on 2-dimensional space in this derivation.
Despite being far from reality, we claim that results from our simplest model can be generalized to the actual scenario. In 2D, the cross section area $A_c$ has units of length, so we denote it with letter $r$. 
The value of $r$ is determined by various properties of the dust particle, including its geometrical size and refractive index, through the Mie scattering theory. Here, we derive our theory in full generality,  regardless of the value of $r$ (as long as $r<d$). To take into account the random fluctuations of the microscopic configuration of dust particles, we assume that, at any moment in time, particles are uniformly randomly sampled in any region of the space according to a Poisson distribution with mean number of particle $\rho_\mathrm{N} V$ where $V$ is the volume of the space, and $\rho_\mathrm{N}$ is the average number density. We assume that $\rho_\mathrm{N}$ is the same everywhere along both paths from the source to the telescope, although in reality dust is denser near the Galactic Bulge. However, we claim that more realistic scenarios can be reconstructed by modifying certain parameters in our model, which will be elaborated later in this section.

We consider the evolution of every photon's wave function during the transmission through the dust particles. Since a photon is emitted by an atom, we assume the wave is from a \emph{point} source. The emitted photonic wave function is in a dipole pattern which can be considered as a spherical wave when the telescope size is sufficiently small compared with the distance. Therefore, it is an equal superposition of rays going in all possible directions. However, we claim that only rays pointing towards the area of the telescope need to be taken into consideration when ignoring the possibility that rays pointing in other directions are diffracted to the telescope. Therefore, all relevant rays (or directions) of the photon's propagation in space can be illustrated in \cref{fig:dust_extinction} as an isosceles triangle, since we are considering a simplified model in 2D. (In 3D, the collection of relevant rays is a \emph{cone}.) The triangle has a tiny top angle $\approx d/R$, where $d$ denotes the telescope's linear size and $R$ denotes the distance from the source to the telescope. For typical choices of target in our M-dwarf observation proposal, $R\sim 10~\mathrm{kpc}$ and $d\sim 10\,\mathrm{m}$, hence the top angle is $\sim 10^{-15}\,\mathrm{arcsec}$.  In this analysis, as shown in \cref{fig:dust_extinction}, we \emph{discretize} the light rays (or the spatial support of the wave function) into multiple layers of trapezoids (except the first layer, which contains a single triangle) where each trapezoid has one base of length $r/2$ that is closer to the source and the other base of length $r$ that is closer to the telescope. The $k$th layer contains all trapezoids with the distance from the source to the base of length $r$ being $2^{k-1} Rr/d$. Observe that the number of trapezoids in the $k$th layer is $2^{k-1}$, which reflects the fact that, as the light rays go far away from the source, little differences in their angle become significant such that they need to be distinguished by different trapezoids. More specifically, when constructing the trapezoids, we split the $r$-base of each trapezoid into two $r/2$-bases of the trapezoids in the following layer. One can also see that the number of layers is $1+\log(d/r)$ because the sum of last layer's longer bases is the telescope size $d$, and the total number of trapezoids in that layer is $d/r$ (without loss of generality, we assume $d/r$ is an integer power of $2$). Also, all trapezoids in the same layer have the same area. In fact, letting $V_k$ represent the area of a trapezoid in the $k$th layer (we use letter $V$ because these trapezoids are in the simplified 2D model, while in 3D their area becomes volume), for $k\geq 2$, we have
\begin{equation}
    V_k =  \frac{3r^2 R }{d} 2^{k-4},
\end{equation}
and the expected number of particles in the trapezoid at any moment is $\rho_\mathrm{N} V_k$.

\begin{figure*}[t]
    \centering
    \includegraphics[width=\linewidth]{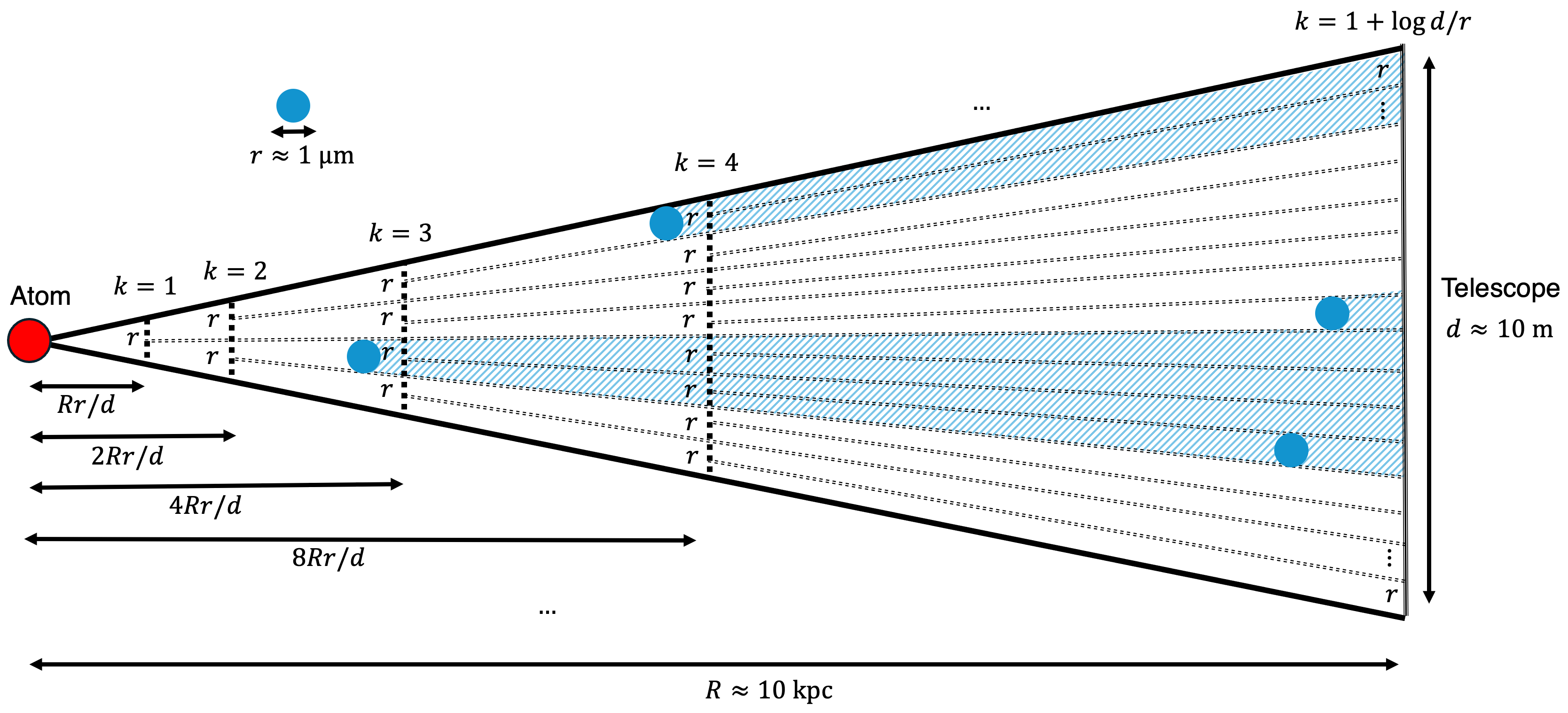}
    \caption{Our modeling of the dust extinction process.}
    \label{fig:dust_extinction}
\end{figure*}

The dust extinction process is modeled microscopically as a collection of discretized ``bad" events. We say a bad event happens to a trapezoid at any moment in time if at least one dust particle appears in it (recall that we assume particles are sampled at any moment according to a Poisson distribution with $\rho_\mathrm{N} V_k$ as the mean value). A good event happens if no dust particle is sampled in the trapezoid. When a bad event happens, the trapezoid is ``blocked" since all light rays through it fall in the scattering cross-section of the particle. Now, if a trapezoid in the $k$th layer is blocked, one can see that the two trapezoids following it in the $(k+1)$th layer should also be blocked because all of their light rays are already blocked in the preceding trapezoid, and then four trapezoids in the $(k+2)$th layer are blocked, etc. This modeling reflects the intuition that, if a dust particle is closer to the photon source, then more light rays are blocked, and vice versa. Therefore, a dust particle in the $k$th layer blocks $2^{1+\log(r/d) - k}$ trapezoids in the final layer. An example configuration of blocked trapezoids is shown in \cref{fig:dust_extinction}.

The connection between this model and the photonic state is that, since the photon takes all rays in superposition, if a fraction $\xi$ of all rays toward the telescope is blocked (equivalent to a fraction $\xi$ of trapezoids in the final layer being blocked) at a given moment, then the state of the photon at that moment is $\sqrt{\xi}\ket{\mathrm{blocked}}+\sqrt{1-\xi} \ket{\mathrm{received}}$. The physical process of dust absorbing the photon or telescope receiving the photon is just a measurement of the quantum state. Therefore, $\xi$ satisfies $\langle \xi \rangle = p_\mathrm{loss},$ where the average is taken over all possible configurations of dust particles blocking the rays.

One can compute the probability, denoted by $q_k$, that a good event (no dust particle is sampled) happens to a trapezoid in the $k$th layer (for $k\geq 2)$ using properties of the Poisson distribution:
\begin{equation}
    q_k = \Pr[X=0|X\sim \mathrm{Pois}(V_k \rho_\mathrm{N})] = \exp(-V_k \rho_\mathrm{N}) = \exp(-\frac{3r^2 R \rho_\mathrm{N}}{d} 2^{k-4} ).
\end{equation}
For $k=1$, the area of the little triangle is $V_1 = Rr^2/(2d)$, hence
\begin{equation}
    q_1 = \Pr[X=0|X\sim \mathrm{Pois}(V_1 \rho_\mathrm{N})] = \exp(-V_1 \rho_\mathrm{N}) = \exp(-\frac{r^2 R \rho_\mathrm{N}}{2d}  ).
\end{equation}
Now, observe that, if a trapezoid is not blocked, no bad event should happen to any trapezoid connecting the source to it. Hence the probability for any trapezoid in the \emph{last} layer to not be blocked is simply the product of all $q_k$s from $k=1$ to $k=1+\log(d/r)$, because every trapezoid in the final layer has exactly one preceding trapezoid in every layer before it. We can now construct the connection between the microscopic configuration of dust particles and the macroscopic observable, the dust extinction rate $p_\mathrm{loss}$, as follows. $p_\mathrm{loss}$ is simply the expected value of the  fraction of blocked trapezoids in the final layer, i.e.,
\begin{equation}
\label{eq:ploss}
    1 - p_\mathrm{loss} = \frac{\E[\#\mathrm{unblocked}~\mathrm{trapezoids}]}{d/r} = \frac{\prod_{k=1}^{1+\log(d/r)} q_k d/r}{d/r} = \prod_{k=1}^{1+\log(d/r)} q_k.
\end{equation}
We can further compute the product of $q_k$s: 
\begin{equation}
\label{eq:productofqks}
    \prod_{k=1}^{1+\log(d/r)} q_k = \exp(-\frac{r^2 R \rho_\mathrm{N}}{2d}  )  \exp(-\frac{3r^2 R \rho_\mathrm{N}}{d} \sum_{k=2}^{1+\log(d/r)} 2^{k-4} ) = \exp(\frac{ r^2 R \rho_\mathrm{N}}{4d}   -\frac{3r R \rho_\mathrm{N}}{4}  ),
\end{equation}
which is approximately $e^{-\frac{3}{4}rR\rho_\mathrm{N}}$ for small particles and large telescopes (a reasonable assumption since particle sizes are up to micrometer level while telescope sizes are several meters). In this limit, the photon loss rate does not depend on telescope size $d$, as expected.

Now, we observe that two dust extinction models mentioned above, the random wall model and the beam splitter model, are unified in this trapezoid model as follows. To achieve the random wall model, the number of unblocked trapezoids must be either zero or $d/r$ with probability $p_\mathrm{loss}$ and $1-p_\mathrm{loss}$, respectively. The beam splitter model, on the contrary, corresponds to the case where there are always $p_\mathrm{loss} d/r$ blocked trapezoids and $(1-p_\mathrm{loss}) d/r$ unblocked ones. Now, one can see that the \emph{variance} of the number of unblocked trapezoids is the quantity that differentiates between a wide spectrum of models including the random wall model and the beam splitter model. Indeed, the random wall model has the largest possible variance with a given $p_\mathrm{loss}$, while the beam splitter model has zero variance.

We can also establish the effect of the variance more quantitatively by considering it as a source of decoherence of the density operator. We again consider a highly simplified model that captures the coherence between the two paths: letting $\ket{1}$ represent the state of the photon in the first path and $\ket{2}$ in the second path and assuming a constant phase difference, the equal superposition is $\ket{\psi_\mathrm{simple}}\coloneqq\frac{1}{\sqrt{2}}(\ket{1}+\ket{2})$ with density matrix
\begin{equation}
    \rho_\mathrm{simple} = \begin{pmatrix}
1/2 & 1/2 \\
1/2 & 1/2
\end{pmatrix}.
\end{equation}
This density matrix corresponds to the situation that the number of blocked trapezoids is the same for both paths, hence they have the same amplitude in the superposition.

Next, we introduce the variation of the amplitudes: we let $\eta\in\mathbb{R}$ be the amplitude of $\ket{1}$ and assume that $\eta^2$ is a random variable following the Gaussian distribution centered at $1/2$ with variance $\sigma^2$. Then the density matrix over the distribution is
\begin{equation}
    \rho_\mathrm{var} = \int_{-1}^1 \begin{pmatrix}
\eta^2 & \eta\sqrt{1-\eta^2}  \\
\eta\sqrt{1-\eta^2} & 1-\eta^2
\end{pmatrix} p_\mathrm{var}(\eta)d\eta
= \begin{pmatrix}
1/2 & g(\eta,\sigma)  \\
g(\eta,\sigma) & 1/2
\end{pmatrix},
\end{equation}
where $p_\mathrm{var}(\eta)$ is the pdf of $\eta$ and
\begin{equation}
    g(\eta,\sigma) = \int_{-1}^1 \eta\sqrt{1-\eta^2} \cdot p_\mathrm{var}(\eta)d\eta = \frac{1}{2} - \sigma^2 + O(\sigma^4).
\end{equation}
Now, one can see that $\rho_\mathrm{var}$ becomes a maximally mixed state when the variance of the amplitude squared $\sigma^2$ is close to $1/2$. Note that this is also the maximally possible variance for a density matrix, since we need to ensure that both $\eta^2$ and $1-\eta^2$ are within $[0,1]$, so this corresponds to the random-wall model. On the other hand, if $\sigma=0$, no decoherence exists, and we have the beam splitter model.

Therefore, to decide which model is closer to reality, we compute the variance of unblocked trapezoids. As mentioned in the above observations, every trapezoid in the final layer is connected to the photon source through a chain of trapezoids with one trapezoid in each layer. This suggests that there is a \emph{tree} structure in the trapezoids: we let the triangle in the first layer and all trapezoids be nodes in the graph, and let an edge connect two nodes if they are in two adjacent layers and share a base. One can see that there is a unique path from the root node to any other node, hence it is a tree. Moreover, it is a \emph{binary} tree with $1+\log(d/r)$ layers, and the \emph{leaf} nodes (nodes in the deepest layer) correspond to the trapezoids in the final layer of the trapezoid model. Observe that the dust particle sampling process is the same as a \emph{coloring} process where each node in layer $k$ is colored as green with probability $q_k$ or red with probability $1-q_k$, individually and independently. The event that one trapezoid (with index $i$, for instance) in the final layer is not blocked corresponds to the event that the coloring of the binary tree has an \emph{all-green} path from the root node to the $i$th leaf node. Now, the number of such all-green paths from the top to the bottom is the number of unblocked trapezoids in the last layer, and the ratio of the average number of such events to the total number of events is also determined by Eqs.~(\ref{eq:ploss}, \ref{eq:productofqks}).

The variance of the number of unblocked trapezoids in the final layer is the same as the variance of the the number of all-green paths in the binary tree. Let $Y_i$ denote the random variable such that $Y_i = 1$ means there exists an all-green path to leaf node $i$, and $Y_i=0$ otherwise. Then the variance we need is 
\begin{equation}
    \Var\left[\sum_{i=1}^{d/r} Y_i\right] = \sum_{i=1}^{d/r} \Var[Y_i] + 2\sum_{1\leq i < j \leq d/r} \mathrm{Cov}[Y_i, Y_j].
\end{equation}
Note that $\Pr[Y_i = 1] = 1-p_\mathrm{loss}$ and for all $i\in\{1,2,\dots,d/r\}$,
\begin{equation}
    \Var[Y_i] = \E[Y_i^2] - \E[Y_i]^2 = \E[Y_i] (1-\E[Y_i]) =p_\mathrm{loss} (1-p_\mathrm{loss}) < 1-p_\mathrm{loss}
\end{equation}
Therefore, the main difficulty is to compute the covariances. Note that $Y_i$ and $Y_j$ are not independent random variables because the all-green paths from the root node to leaf node $i$ and leaf node $j$ must share some common nodes before their lowest common ancestor. Suppose the LCA of $i$ and $j$ is in the $k_{i,j}$th layer, then
\begin{equation}
\begin{aligned}
    \mathrm{Cov}[Y_i, Y_j] &= \Pr[Y_i = 1, Y_j = 1] - \Pr[Y_i = 1] \Pr[Y_j = 1]\\
&= \prod_{k=1}^{k_{i,j}} q_k \prod_{k=k_{i,j}+1}^{1+\log(d/r)} q_k^2 - (1-p_\mathrm{loss})^2\\
&= (1-p_\mathrm{loss})^2 \left(\frac{1}{\prod_{k=1}^{k_{i,j}} q_k} -  1\right),
\end{aligned}
\end{equation}
where
\begin{equation}
    \prod_{k=1}^{k_{i,j}} q_k = \exp(-\frac{r^2 R \rho_\mathrm{N}}{2d}  )  \exp(-\frac{3r^2 R \rho_\mathrm{N}}{d} \sum_{k=2}^{k_{i,j}} 2^{k-4} ) =  \exp(\frac{r^2 R \rho_\mathrm{N}}{4 d} - \frac{3r^2 R \rho_\mathrm{N}}{4 d} 2^{k_{i,j}-1} ).
\end{equation}
Now, for any node in the $k$th layer of the binary tree, if and only if leaf nodes $i,j$ are in its left and right subtrees, respectively, it is the lowest common ancestor of $i$ and $j$. This means that the number of $(i,j)$ pairs (with $i<j$) with lowest common ancestor being a specific node in the $k$th layer is $\left[2^{1+\log(d/r) - (k+1)}\right]^2 = 2^{2\log(d/r) - 2k}$.
Now, since the number of nodes in the $k$th layer is $2^{k-1}$, the total number of $(i,j)$ pairs with lowest common ancestor in the $k$th layer is $2^{2\log(d/r) - k - 1}$. This allows us to compute the sum of covariances:
\begin{equation}
    \sum_{1\leq i < j \leq d/r} \mathrm{Cov}[Y_i, Y_j] = (1-p_\mathrm{loss})^2 \sum_{k=1}^{\log(d/r)} 2^{2\log(d/r)-k-1} \left[\exp(-\frac{r^2 R \rho_\mathrm{N}}{4 d} + \frac{3r^2 R \rho_\mathrm{N}}{4 d} 2^{k-1} ) -1 \right].
\end{equation}
Since the summand takes its maximum value when $k = \log(d/r)$, we can derive an upper bound:
\begin{equation}
\begin{aligned}
     \sum_{1\leq i < j \leq d/r} \mathrm{Cov}[Y_i, Y_j] &\leq \frac{1}{2} (1-p_\mathrm{loss})^2 \log(d/r) \frac{d}{r} \exp(-\frac{r^2 R \rho_\mathrm{N}}{4 d} + \frac{3r R \rho_\mathrm{N}}{8} ) \leq  \log(d/r) \frac{d}{2r} (1-p_\mathrm{loss}).
\end{aligned}
\end{equation}

Finally, we can derive a (rather loose, but sufficiently good) upper bound for the variance of $\sum Y_j$:
\begin{equation}
    \Var\left(\sum_{i=1}^{d/r} Y_i\right) \leq (1+\log(d/r)) \frac{d}{r} (1-p_\mathrm{loss}). 
\end{equation}
We can now compute the variance of the fraction of unblocked trapezoids (i.e., variance of $\sum_j Y_j / (d/r)$): 
\begin{equation}
    \Var\left(\frac{\sum_j Y_j}{d/r}\right) = \frac{\Var\left(\sum_j Y_j\right)}{d^2/r^2} = \frac{(1+\log(d/r)) (1-p_\mathrm{loss})}{d/r}.
\end{equation}
Observe that the variance in the  fraction of unblocked trapezoids is approximately inversely proportional to $d/r$, the ratio between the telescope size and the dust particle cross-section. Therefore, the variance is large only if the dust ``particles" are so large that their size $r$ is comparable to the telescope aperture $d$, and the ``particles" are now huge rocks blocking almost all rays in the triangle. 
Fortunately, $r\lesssim 1\,\mathrm{\mu m}$ \cite{mathis1977size}, hence, for $d\approx 10\,\mathrm{m}$, this ratio is $10^7$. Moreover, in our 3D universe, the number of 3D trapezoids (frustums) in the final layer is $d^2/r^2$ rather than $d/r$, suggesting that the above variance should be replaced by
\begin{equation}
    \Var\left(\frac{\sum_j Y_{j,\mathrm{3D}}}{d^2/r^2}\right) = \frac{(1+2\log(d/r))(1-p_\mathrm{loss})}{d^2/r^2}.
\end{equation} 
This implies that, for a typical rate of obtaining a photon in the Baade window $1-p_\mathrm{loss} \geq 0.1$ given by Ref.~\cite{Saha2019} and \cref{sec:dwarfflare}, the standard deviation is $\sim 10^{-7}$, and the probability of any relative variation greater than $5\,\%$ is upper bounded by $10^{-9}$. Therefore, the fraction of every photon's wave function that arrives at the telescope is rather stable, and we conclude that the beam splitter model is much closer to the physical reality. This demonstrates the robustness of  $\Delta t$ signal against dust extinction.

\section{Narrowband flare combination}
\label{sec:narrow}
In this appendix, we discuss the flare combination approach when the coherence time $t_c$ is much longer than the difference in $\Delta t$ between different flares, i.e., when the photons are from a narrow-band source. In this case, peaks in the score function corresponding to different $\Delta t_i$s overlap with each other, so the data processing is much easier than in the broadband case. This may not occur in practice because the emissions of stellar flares are unlikely to be restricted within several $\mathrm{GHz}$. Nonetheless, this regime is potentially relevant to \cref{sec:forarrays}, when multiple guidestars are used, and artificial guidestars could be narrowband photon sources.

We introduce a new score function for the $i$th flare, denoted by $f_{\mathrm{c},i}$, and a score function for the combination of all $m$ flares, denoted by $G_\mathrm{\Delta t_0}$:
\begin{equation}
\begin{aligned}
    f_{\mathrm{c},i}(\tau) &\coloneqq \left| \sum_{j=1}^n \exp(\imag \nu_{i,j} \tau) \right|^2,\\ G_\mathrm{\Delta t_0}(\tau) &\coloneqq \sum_{i=1}^m f_{\mathrm{c},i}(\tau).
\end{aligned}
\end{equation}
Let us now compute the expectation value of $f_{\mathrm{c},i}$ for different values of $\tau$. First, observe that 
\begin{equation}
\begin{aligned}
    \E[\exp(\imag \nu_{i,j}\tau)] &= \int_{-\infty}^\infty e^{\imag \nu_{i,j} \tau} p_A(\nu_{i,j}|\Delta t) d\nu_{i,j} \\
    &= \sqrt{2\pi} \gamma_A F_{\Delta t}(\tau) \approx \frac{\gamma_A}{2} e^{-\frac{(\Delta t - \tau)^2}{4 t_c^2}} e^{i\omega_0 (\Delta t - \tau)}
\end{aligned}
\end{equation}
according to Eqs.~(\ref{eq:Ftau_certain_q}) and (\ref{eq:channel_with_magnification}). Therefore, the expectation value of each $\exp(\imag \nu_{i,j}\tau)$ is close to $0$ if $|\tau-\Delta t|\geq t_c$ and close to $\exp(\imag \omega_0 (\Delta t_i -\tau))$ otherwise. To simplify the presentation, we let $\exp(\imag \nu_{i,j}\tau)=:R_j+\imag I_j$, hence $f_{\mathrm{c},i} = (\sum_{j=1}^n R_j)^2+(\sum_{j=1}^n I_j)^2$. 
Now, for $|\tau-\Delta t_i|\geq t_c$, using $\E[R_j]=\E[I_j]=0$ and the independence between different $\nu_{i,j}$s, we obtain
\begin{equation}
\begin{aligned}
\E[f_{\mathrm{c},i}] &= \E\left[\left(\sum_{j=1}^n R_j\right)^2\right] + \E\left[\left(\sum_{j=1}^n I_j\right)^2\right] \\
&= \sum_{j=1}^n \left( \E[R_j^2] + \E[I_j^2] \right) = \sum_{j=1}^n \E[R_j^2+I_j^2]  d= n.
\end{aligned}
\end{equation}
Next, for $|\tau-\Delta t_i|<t_c$, using $\E[R_j] = \frac{\gamma_A}{2}\cos(\nu_{i,j}\tau)$, $\E[I_j]=\frac{\gamma_A}{2}\sin(\nu_{i,j}\tau)$, and the independence of different $\nu_{i,j}$s, we obtain
\begin{equation}
\begin{aligned}
    \E[f_{\mathrm{c},i}] &= \left|\E \left[\sum_{j=1}^n \exp(\imag \nu_{i,j}\tau) \right] \right|^2 + \Var\left(\sum_{j=1}^n \exp(\imag \nu_{i,j}\tau)\right)\\
    &= \frac{n^2 Q^2 \gamma_A^2}{4} + \sum_{j=1}^n \Var(R_j + \imag I_j)\\
    & = \frac{n^2 Q^2 \gamma_A^2}{4} + \sum_{j=1}^n \E[R_j^2+I_j^2]  - \sum_{j=1}^n \left[\E[R_j]^2 - \E[I_j]^2\right] \\
    &= \frac{n^2 Q^2 \gamma_A^2}{4} - \frac{nQ \gamma_A^2}{4} + n\\
    &= \frac{\gamma_A^2}{4}(n_\mathrm{sig}^2-n_\mathrm{sig}) + n
\end{aligned}
\end{equation}
where $n_\mathrm{sig}=nQ$ is the number of signal photons among the $n$ received photons. To summarize the above results,
\begin{equation}
\begin{aligned}
    &\quad \mathbb{E}[f_{\mathrm{c},i}(\tau)] \approx \begin{cases}
         n + \frac{\gamma_A^2}{4} (n_\mathrm{sig}^2 - n_\mathrm{sig}) ,&|\tau-\Delta t_i|< t_c  \\
         n,&|\tau-\Delta t_i| \geq t_c.
    \end{cases}
\end{aligned}
\end{equation}
Finally, using the assumption that $|\Delta t_i - \Delta t_k|\ll t_c$, one can conclude that, if $\Delta t_0$ is the mean value of all $\Delta t_i$s, then $|\tau-\Delta t_i|<t_c$ implies $|\tau-\Delta t_0|<t_c$ for all $i$. Therefore,
\begin{equation}
\begin{aligned}
    &\quad \E[G_\mathrm{\Delta t_0}(\tau)] \approx \begin{cases}
         mn + \frac{\gamma_A^2}{4} m\left(n_\mathrm{sig}^2 - n_\mathrm{sig} \right) ,&|\tau-\Delta t_0|< t_c  \\
         mn,&|\tau-\Delta t_0| \geq t_c.
    \end{cases}
\end{aligned}
\end{equation}

One can observe that for any $n_\mathrm{sig}\geq 2$, (i.e., if there are at least 2 signal photons per flare) there is a separation between good and bad estimates, and one can find $\Delta t_0$ by accumulating sufficiently many flares.
However, for $n_\mathrm{sig}=1$, we observe that $\E[G_\mathrm{\Delta t_0}]=m$ for any possible $\tau$. This implies that we need at least $2$ signal photons per flare to enable a successful $\Delta t_0$ measurement. A alternative way of seeing this is that, finding $\Delta t_0$ from $m$ flares with only $1$ signal photon is no different from finding $\Delta t_0$ from $m$ signal photons from a single source that is larger than the finite source effect limit, which is information-theoretically impossible.

Now that we have proved that combining multiple flares does give some information about $\Delta t_0$, we also wish to understand the sample complexity of this new strategy. In particular, it would be ideal if the total number of photons needed ($nm$) still scales linearly with $\log(T/t_c)$, just like in the single-flare scenario. To do so, we derive a \emph{sufficient} condition for the number of flares $m$ such that one can estimate $\Delta t_0$ with $t_c$ precision and $95\,\%$ confidence. The upper bound depends on the combination of parameters ($n_\mathrm{sig}, Q, A$, and $T/t_c$).

\begin{figure*}
    \centering
    \includegraphics[width=0.45\linewidth]{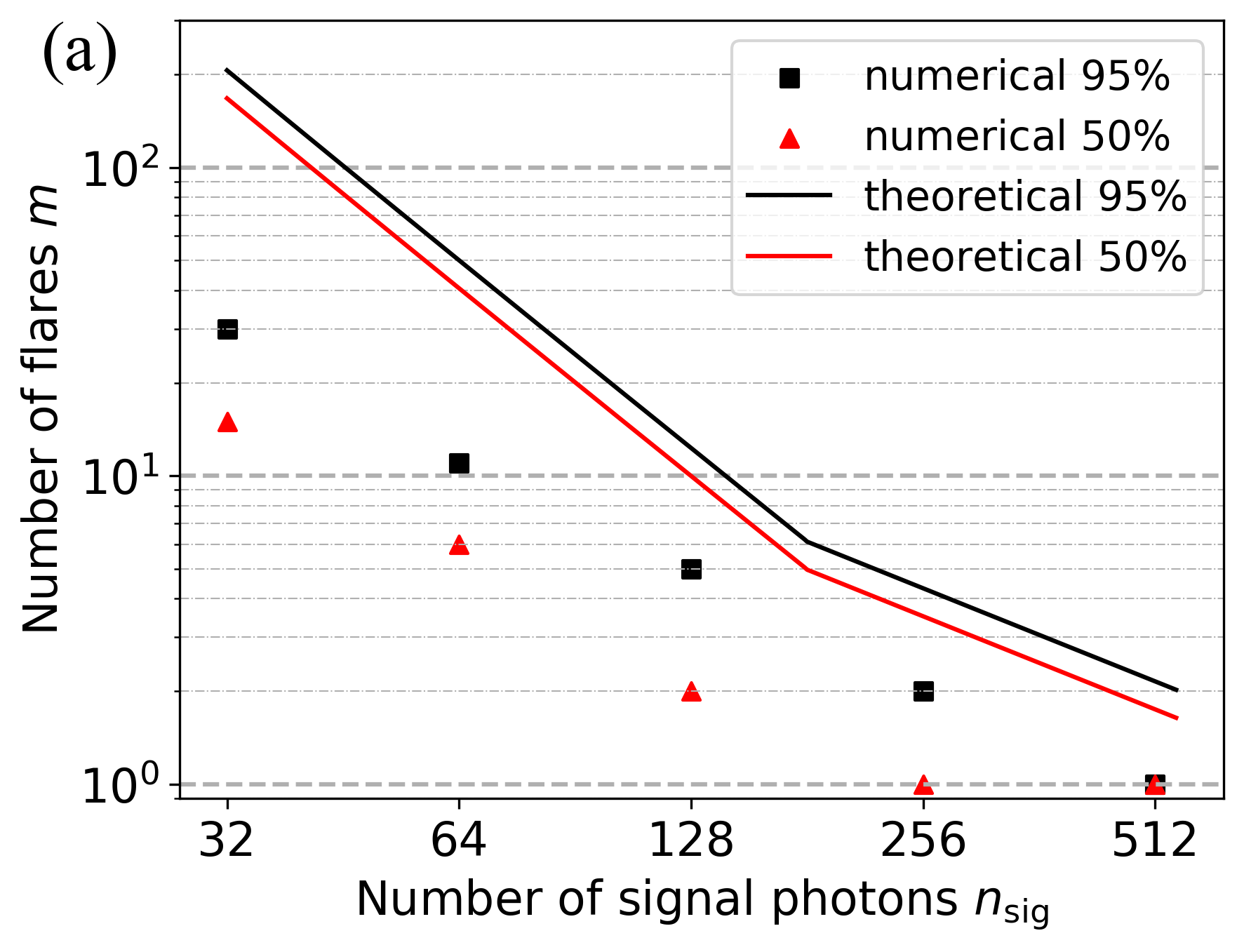}
    \includegraphics[width=0.45\linewidth]{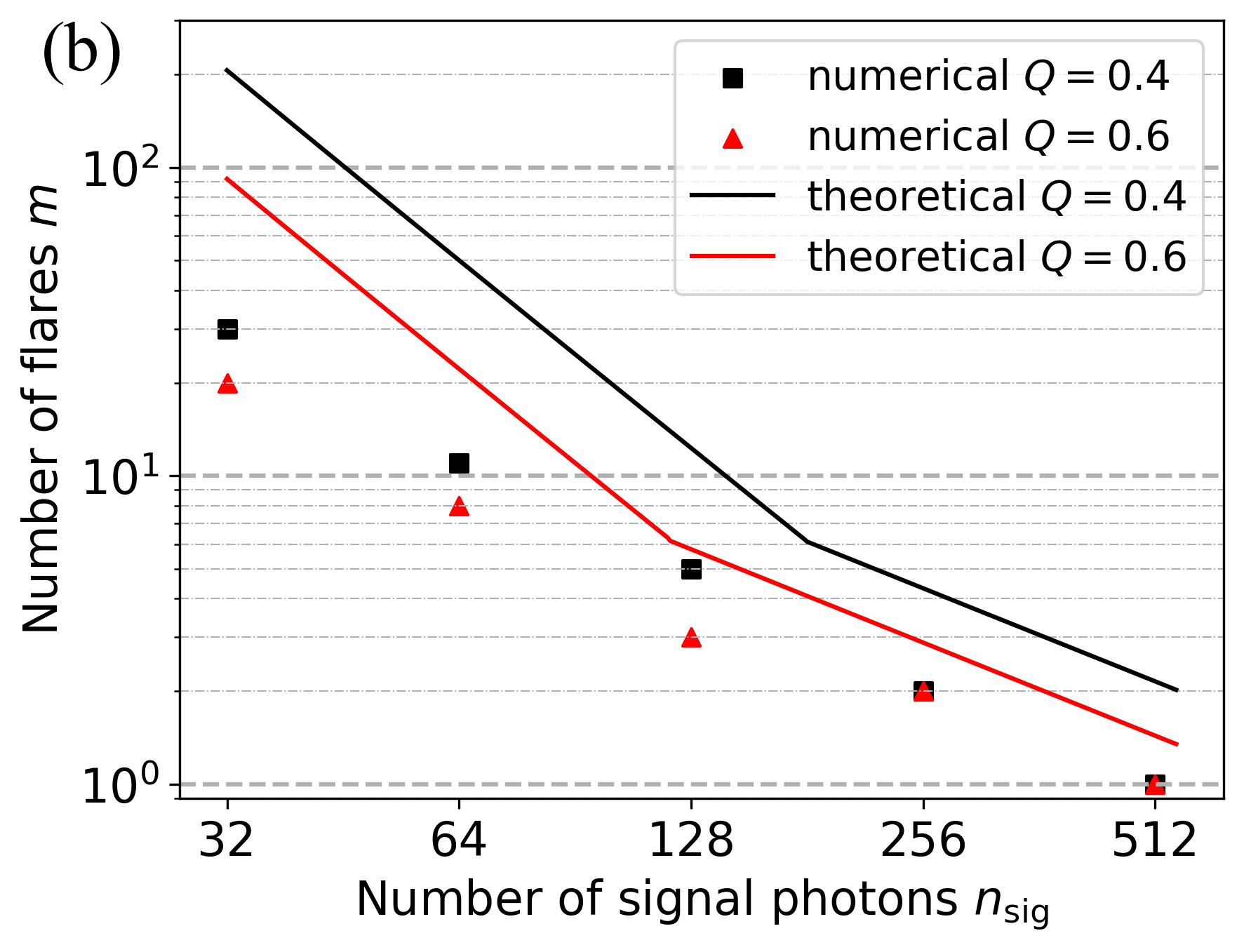}
    \includegraphics[width=0.45\linewidth]{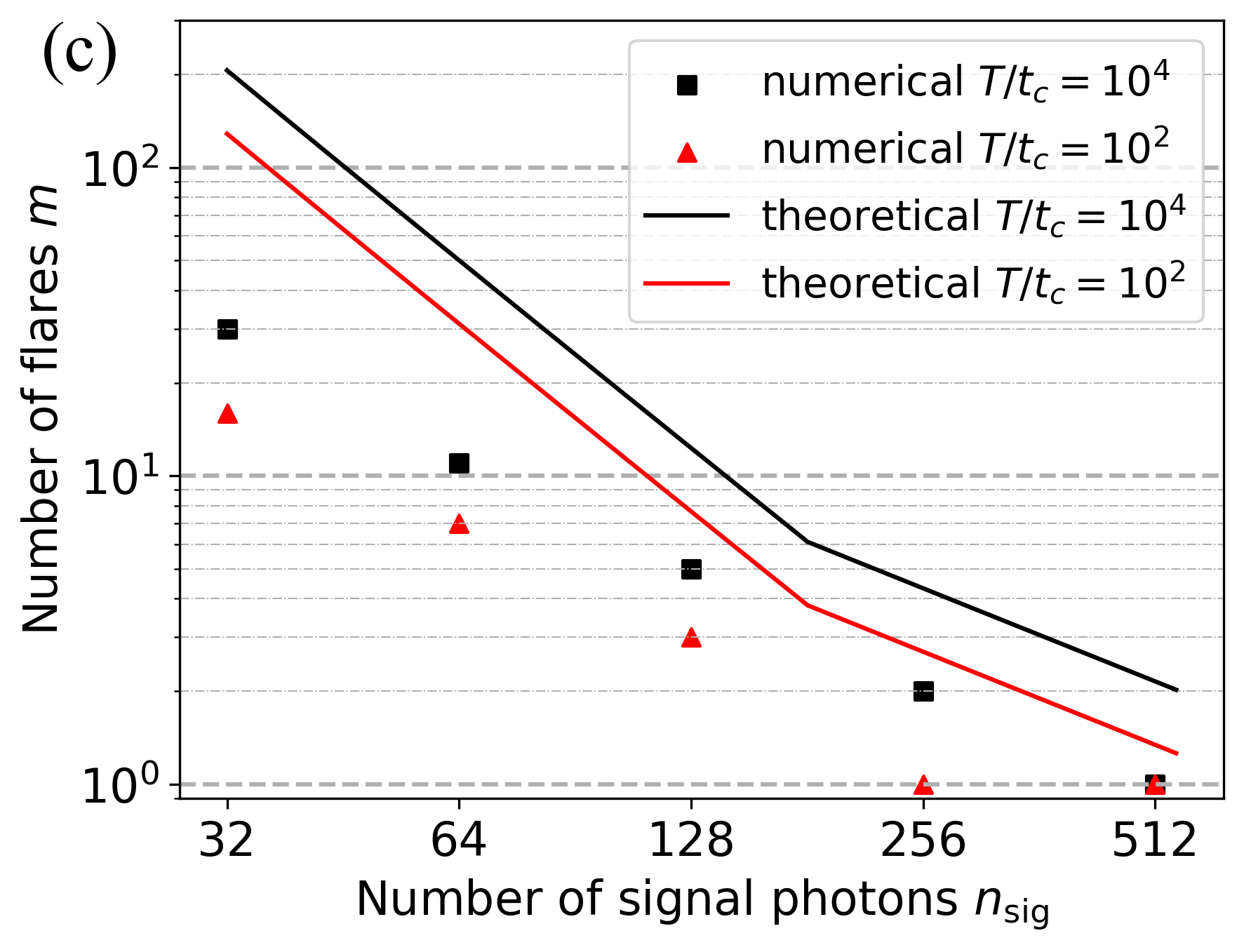}
    \caption{Numerical and analytical results for the flare number $m$ necessary to yield a detection as a function the number of signal photons $n_\mathrm{sig}$ per flare under various settings. All curves are the analytically-derived sufficient conditions (see \cref{eq:narrowband_bound}), and points represent the sufficient conditions obtained from numerical simulation, which are closer to optimality. The black curves and points are the baseline setting discussed in our M-dwarf flare observation proposal: $T/t_c = 10^4, Q=0.4$, and confidence level is $95\%$. The red curves and points show variations on these three parameters, with one parameter varied in each plot. All simulations are based on magnification $A=1.34$.
    In (a), we compare the flare number requirement for different confidence levels ($95\%$ and $50\%$). In (b), we alter the signal-to-noise ratio $Q$ from $0.4$ to $0.6$. In (c) we change $T/t_c$ from $10^4$ to $10^2$.}
    \label{fig:numericalplots}
\end{figure*}

First, we set the threshold value of the score function to be the mean value of the two cases: $G_\mathrm{th}\coloneqq mn+m\gamma_A^2(n_\mathrm{sig}^2-n_\mathrm{sig})/8$. To start with, we derive a tail bound for $f_{\mathrm{c},i}(\tau)$ when $|\tau-\Delta t_i|\geq t_c$ using $\E[\sqrt{f}]\leq \sqrt{\E[f]}$ and McDiarmid's inequality:
\begin{equation}
\begin{aligned}
    &\quad \Pr[f_{\mathrm{c},i}(\tau) - n > \epsilon]\\
    &= \Pr\left[\sqrt{f_{\mathrm{c},i}(\tau)} > \sqrt{\epsilon+n} \right]\\
    &\leq \Pr\left[\sqrt{f_{\mathrm{c},i}(\tau)} - \E[\sqrt{f_{\mathrm{c},i}}] > \sqrt{\epsilon+n} - \sqrt{n} \right] \\
    &\leq  \exp(-\frac{(\sqrt{\epsilon +n } - \sqrt{n})^2}{2n}) \\
    &= \exp(-\frac{\epsilon}{2n} -1 + \sqrt{1+\frac{\epsilon}{n}}).
\end{aligned}
\end{equation}
The above upper bound converges to $\exp(-\frac{\epsilon^2}{8n^2})$ when $\epsilon \ll n$ and becomes $\exp(-\frac{\epsilon}{2n})$ when $\epsilon \gg n$.
This implies that $f_{\mathrm{c},i}$ is approximately a \emph{sub-exponential} random variable with parameters $(4n^2, n)$. Therefore, as a sum of sub-exponential random variables, Bernstein's inequality \cite{chafai2012interactions} gives the following bounds for $G_\mathrm{\Delta t_0}$:
\begin{equation}
\begin{aligned}
    &\quad \Pr[G_\mathrm{\Delta t_0}(\tau) > G_\mathrm{th}] \\
    &= \Pr[G_\mathrm{\Delta t_0}(\tau) - mn > \frac{\gamma_A^2}{8}m(n_\mathrm{sig}^2 - n_\mathrm{sig})] \\
    &\leq \begin{cases}
         \exp(-\frac{mQ^2 \gamma_A^4 (n_\mathrm{sig}-1)^2}{512}),&0\leq n_\mathrm{sig}-1 < \frac{32}{Q \gamma_A^2} \\
         \exp(-\frac{mQ\gamma_A^2(n_\mathrm{sig}-1)}{16}),&n_\mathrm{sig}-1 > \frac{32}{Q \gamma_A^2}.
    \end{cases}
\end{aligned}
\end{equation}
We observe that there are two different scalings for different number of signal photons. We derive the expression of $m$ for each case. Recall that we need to ensure all $T/t_c - 1$ incorrect $\tau$s have score function value less than the threshold with probability at least $95\,\%$. Using the same strategy as in \cref{sec:ouralgo}, this probability can be established by the union bound: when $n_\mathrm{sig}$ is sufficiently large, the sufficient condition is
\begin{equation}
    \frac{T}{t_c} \exp(-\frac{mQ \gamma_A^2 (n_\mathrm{sig}-1)}{16})  \leq 0.05,
\end{equation}
and for $n_\mathrm{sig}$ close to $2$, the sufficient condition is 
\begin{equation}
    \frac{T}{t_c} \exp(-\frac{mQ^2 \gamma_A^4 (n_\mathrm{sig}-1)^2}{512})  \leq 0.05.
\end{equation}
These inequalities give the following sufficient condition for the number of flares to achieve $95\,\%$ confidence:
\begin{equation}
\label{eq:narrowband_bound}
    m\geq \begin{cases}
         \frac{512}{Q^2 \gamma_A^4 (n_\mathrm{sig}-1)^2} \ln(\frac{20 T}{t_c}) ,&0\leq n_\mathrm{sig}-1< \frac{32}{Q \gamma_A^2}\\
         \frac{16}{Q \gamma_A^2 (n_\mathrm{sig}-1)} \ln(\frac{20 T}{t_c})  ,&n_\mathrm{sig}-1 > \frac{32}{Q \gamma_A^2}.
    \end{cases}
\end{equation}

\begin{figure}
    \centering
    \includegraphics[width=0.5\linewidth]{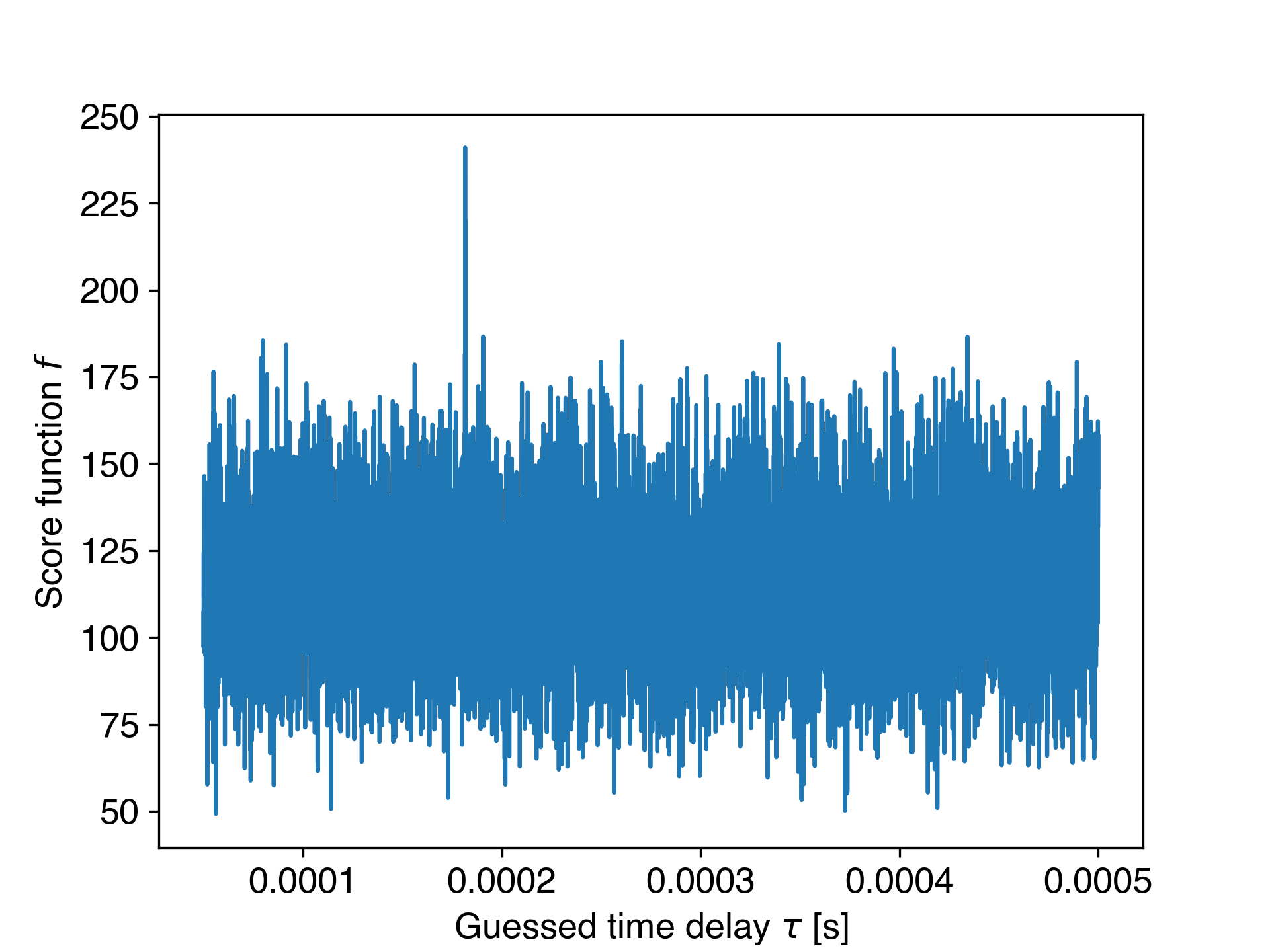}
    \caption{The score function $f(\tau,\nu_{1,1},\dots,\nu_{n,m})$ for different $\tau$ guesses in a numerical experiment for the multiple-flare combination algorithm. In this experiment, we simulate a realistic observation setting: we combine 5 flares where each flare has 66 signal photons and 103 noise photons. Every flare has a different $\Delta t_i$ value, and the difference between them is more than $2\pi/\omega_0$ but less than $t_c$. The average value of their $\Delta t_i$s is around $0.00018\,\mathrm{s}$, corresponding to the high peak of the score function we calculated, as shown in the plot.}
    \label{fig:numerical_multiflare}
\end{figure}

Therefore, when $n_\mathrm{sig}\gg 1$, it suffices to have $mn_\mathrm{sig} = O\left(\frac{\log(T/t_c)} {Q\gamma_A^2}\right)$, and the total number of good photons still scale linearly with $\log(T/t_c)$. However, when $n_\mathrm{sig}$ is tiny, the relation becomes more complicated: the number of flares no longer scales inversely with the number of photons per flare, and the total number of photons increases significantly as $n_\mathrm{sig}$ decreases. Nonetheless, due to the finite duration of microlensing events and the fact that flares happen only twice in three days on average, the number of flares we can expect is limited. This means that even if the inverse scaling still holds in the small-$n_\mathrm{sig}$ case, the number of flares needed would still be prohibitively high. Therefore, we are mainly interested in the case where $n_\mathrm{sig}\gg 0$.

Again, we emphasize that the above derivation only gives the sufficient condition for a successful time delay measurement, or an upper bound (in fact, a loose upper bound) of $m$ for any given $(n_\mathrm{sig}, T/t_c, Q, A)$ parameters. The actual lower bound for $m$ will likely have different constant factors, but will never be higher than the upper bounds given by the sufficient-condition analysis. Indeed, we perform numerical simulations to find the number of flares actually needed for various realistic parameters, and the results are significantly less than the theoretical upper bounds. We present the numerical results and their comparison with theoretical sufficient conditions in \cref{fig:numericalplots}. We also give one realistic example for the fiducial M-dwarf flare setting where $A=1.34$ (hence $\gamma_A = 0.666$) and $Q=0.4$. If we set $T/t_c = 10000$, use the Extremely Large Telescope (which gives $n_\mathrm{sig}=426$), and aim for $95\,\%$ confidence, then the sufficient condition given by our the analytic bound is $m\geq 3$, while numerical simulation shows that even $m=1$ works.

To give an intuitive demonstration of how our algorithm works, in \cref{fig:numerical_multiflare}, we present the score function values from the numerical simulation of one successful time-delay measurement. The results are from (the simulation of) the same fiducial observation of M dwarf flares but using Keck rather than the Extremely Large Telescope. It combines 5 flares with 66 signal photons and 103 noise photons per flare.

\bibliography{ref}

\begin{thebibliography}{105}%
\makeatletter
\providecommand \@ifxundefined [1]{%
 \@ifx{#1\undefined}
}%
\providecommand \@ifnum [1]{%
 \ifnum #1\expandafter \@firstoftwo
 \else \expandafter \@secondoftwo
 \fi
}%
\providecommand \@ifx [1]{%
 \ifx #1\expandafter \@firstoftwo
 \else \expandafter \@secondoftwo
 \fi
}%
\providecommand \natexlab [1]{#1}%
\providecommand \enquote  [1]{``#1''}%
\providecommand \bibnamefont  [1]{#1}%
\providecommand \bibfnamefont [1]{#1}%
\providecommand \citenamefont [1]{#1}%
\providecommand \href@noop [0]{\@secondoftwo}%
\providecommand \href [0]{\begingroup \@sanitize@url \@href}%
\providecommand \@href[1]{\@@startlink{#1}\@@href}%
\providecommand \@@href[1]{\endgroup#1\@@endlink}%
\providecommand \@sanitize@url [0]{\catcode `\\12\catcode `\$12\catcode `\&12\catcode `\#12\catcode `\^12\catcode `\_12\catcode `\%12\relax}%
\providecommand \@@startlink[1]{}%
\providecommand \@@endlink[0]{}%
\providecommand \url  [0]{\begingroup\@sanitize@url \@url }%
\providecommand \@url [1]{\endgroup\@href {#1}{\urlprefix }}%
\providecommand \urlprefix  [0]{URL }%
\providecommand \Eprint [0]{\href }%
\providecommand \doibase [0]{https://doi.org/}%
\providecommand \selectlanguage [0]{\@gobble}%
\providecommand \bibinfo  [0]{\@secondoftwo}%
\providecommand \bibfield  [0]{\@secondoftwo}%
\providecommand \translation [1]{[#1]}%
\providecommand \BibitemOpen [0]{}%
\providecommand \bibitemStop [0]{}%
\providecommand \bibitemNoStop [0]{.\EOS\space}%
\providecommand \EOS [0]{\spacefactor3000\relax}%
\providecommand \BibitemShut  [1]{\csname bibitem#1\endcsname}%
\let\auto@bib@innerbib\@empty
\bibitem [{\citenamefont {Shor}(1999)}]{shor1999polynomial}%
  \BibitemOpen
  \bibfield  {author} {\bibinfo {author} {\bibfnamefont {P.~W.}\ \bibnamefont {Shor}},\ }\href@noop {} {\bibfield  {journal} {\bibinfo  {journal} {SIAM Review}\ }\textbf {\bibinfo {volume} {41}},\ \bibinfo {pages} {303} (\bibinfo {year} {1999})}\BibitemShut {NoStop}%
\bibitem [{\citenamefont {Grover}(1996)}]{grover1996fast}%
  \BibitemOpen
  \bibfield  {author} {\bibinfo {author} {\bibfnamefont {L.~K.}\ \bibnamefont {Grover}},\ }in\ \href@noop {} {\emph {\bibinfo {booktitle} {Proceedings of the twenty-eighth annual ACM symposium on Theory of computing}}}\ (\bibinfo {year} {1996})\ pp.\ \bibinfo {pages} {212--219}\BibitemShut {NoStop}%
\bibitem [{\citenamefont {Giovannetti}\ \emph {et~al.}(2006)\citenamefont {Giovannetti}, \citenamefont {Lloyd},\ and\ \citenamefont {Maccone}}]{giovannetti2006quantum}%
  \BibitemOpen
  \bibfield  {author} {\bibinfo {author} {\bibfnamefont {V.}~\bibnamefont {Giovannetti}}, \bibinfo {author} {\bibfnamefont {S.}~\bibnamefont {Lloyd}},\ and\ \bibinfo {author} {\bibfnamefont {L.}~\bibnamefont {Maccone}},\ }\href@noop {} {\bibfield  {journal} {\bibinfo  {journal} {Physical Review Letters}\ }\textbf {\bibinfo {volume} {96}},\ \bibinfo {pages} {010401} (\bibinfo {year} {2006})}\BibitemShut {NoStop}%
\bibitem [{\citenamefont {Gottesman}\ \emph {et~al.}(2012)\citenamefont {Gottesman}, \citenamefont {Jennewein},\ and\ \citenamefont {Croke}}]{Gottesman2021TelescopesRepeaters}%
  \BibitemOpen
  \bibfield  {author} {\bibinfo {author} {\bibfnamefont {D.}~\bibnamefont {Gottesman}}, \bibinfo {author} {\bibfnamefont {T.}~\bibnamefont {Jennewein}},\ and\ \bibinfo {author} {\bibfnamefont {S.}~\bibnamefont {Croke}},\ }\href {https://doi.org/10.1103/PhysRevLett.109.070503} {\bibfield  {journal} {\bibinfo  {journal} {Physical Review Letters}\ }\textbf {\bibinfo {volume} {109}},\ \bibinfo {pages} {070503} (\bibinfo {year} {2012})}\BibitemShut {NoStop}%
\bibitem [{\citenamefont {Khabiboulline}\ \emph {et~al.}(2019{\natexlab{a}})\citenamefont {Khabiboulline}, \citenamefont {Borregaard}, \citenamefont {De~Greve},\ and\ \citenamefont {Lukin}}]{Khabiboulline2019OpticalInterferometry}%
  \BibitemOpen
  \bibfield  {author} {\bibinfo {author} {\bibfnamefont {E.~T.}\ \bibnamefont {Khabiboulline}}, \bibinfo {author} {\bibfnamefont {J.}~\bibnamefont {Borregaard}}, \bibinfo {author} {\bibfnamefont {K.}~\bibnamefont {De~Greve}},\ and\ \bibinfo {author} {\bibfnamefont {M.~D.}\ \bibnamefont {Lukin}},\ }\href {https://doi.org/10.1103/PhysRevLett.123.070504} {\bibfield  {journal} {\bibinfo  {journal} {Physical Review Letters}\ }\textbf {\bibinfo {volume} {123}},\ \bibinfo {pages} {070504} (\bibinfo {year} {2019}{\natexlab{a}})}\BibitemShut {NoStop}%
\bibitem [{\citenamefont {Khabiboulline}\ \emph {et~al.}(2019{\natexlab{b}})\citenamefont {Khabiboulline}, \citenamefont {Borregaard}, \citenamefont {De~Greve},\ and\ \citenamefont {Lukin}}]{Khabiboulline2019TelescopeArrays}%
  \BibitemOpen
  \bibfield  {author} {\bibinfo {author} {\bibfnamefont {E.~T.}\ \bibnamefont {Khabiboulline}}, \bibinfo {author} {\bibfnamefont {J.}~\bibnamefont {Borregaard}}, \bibinfo {author} {\bibfnamefont {K.}~\bibnamefont {De~Greve}},\ and\ \bibinfo {author} {\bibfnamefont {M.~D.}\ \bibnamefont {Lukin}},\ }\href {https://doi.org/10.1103/PhysRevA.100.022316} {\bibfield  {journal} {\bibinfo  {journal} {Physical Review A}\ }\textbf {\bibinfo {volume} {100}},\ \bibinfo {pages} {022316} (\bibinfo {year} {2019}{\natexlab{b}})}\BibitemShut {NoStop}%
\bibitem [{\citenamefont {Liu}\ \emph {et~al.}(2024)\citenamefont {Liu}, \citenamefont {Childs},\ and\ \citenamefont {Gottesman}}]{liu2024low}%
  \BibitemOpen
  \bibfield  {author} {\bibinfo {author} {\bibfnamefont {Z.}~\bibnamefont {Liu}}, \bibinfo {author} {\bibfnamefont {A.~M.}\ \bibnamefont {Childs}},\ and\ \bibinfo {author} {\bibfnamefont {D.}~\bibnamefont {Gottesman}},\ }\href@noop {} {\bibfield  {journal} {\bibinfo  {journal} {arXiv preprint arXiv:2411.04019}\ } (\bibinfo {year} {2024})}\BibitemShut {NoStop}%
\bibitem [{\citenamefont {Brown}\ \emph {et~al.}(2023)\citenamefont {Brown}, \citenamefont {Allgaier}, \citenamefont {Thiel}, \citenamefont {Monnier}, \citenamefont {Raymer},\ and\ \citenamefont {Smith}}]{brown2023interferometric}%
  \BibitemOpen
  \bibfield  {author} {\bibinfo {author} {\bibfnamefont {M.~R.}\ \bibnamefont {Brown}}, \bibinfo {author} {\bibfnamefont {M.}~\bibnamefont {Allgaier}}, \bibinfo {author} {\bibfnamefont {V.}~\bibnamefont {Thiel}}, \bibinfo {author} {\bibfnamefont {J.~D.}\ \bibnamefont {Monnier}}, \bibinfo {author} {\bibfnamefont {M.~G.}\ \bibnamefont {Raymer}},\ and\ \bibinfo {author} {\bibfnamefont {B.~J.}\ \bibnamefont {Smith}},\ }\href@noop {} {\bibfield  {journal} {\bibinfo  {journal} {Physical Review Letters}\ }\textbf {\bibinfo {volume} {131}},\ \bibinfo {pages} {210801} (\bibinfo {year} {2023})}\BibitemShut {NoStop}%
\bibitem [{\citenamefont {{Dyson}}\ \emph {et~al.}(1920)\citenamefont {{Dyson}}, \citenamefont {{Eddington}},\ and\ \citenamefont {{Davidson}}}]{Dyson1920LightDeflection}%
  \BibitemOpen
  \bibfield  {author} {\bibinfo {author} {\bibfnamefont {F.~W.}\ \bibnamefont {{Dyson}}}, \bibinfo {author} {\bibfnamefont {A.~S.}\ \bibnamefont {{Eddington}}},\ and\ \bibinfo {author} {\bibfnamefont {C.}~\bibnamefont {{Davidson}}},\ }\href {https://doi.org/10.1098/rsta.1920.0009} {\bibfield  {journal} {\bibinfo  {journal} {Philosophical Transactions of the Royal Society of London Series A}\ }\textbf {\bibinfo {volume} {220}},\ \bibinfo {pages} {291} (\bibinfo {year} {1920})}\BibitemShut {NoStop}%
\bibitem [{\citenamefont {{Einstein}}(1936)}]{Einstein1936Lensing}%
  \BibitemOpen
  \bibfield  {author} {\bibinfo {author} {\bibfnamefont {A.}~\bibnamefont {{Einstein}}},\ }\href {https://doi.org/10.1126/science.84.2188.506} {\bibfield  {journal} {\bibinfo  {journal} {Science}\ }\textbf {\bibinfo {volume} {84}},\ \bibinfo {pages} {506} (\bibinfo {year} {1936})}\BibitemShut {NoStop}%
\bibitem [{\citenamefont {{Refsdal}}(1964)}]{Refsdal1964Lensing}%
  \BibitemOpen
  \bibfield  {author} {\bibinfo {author} {\bibfnamefont {S.}~\bibnamefont {{Refsdal}}},\ }\href {https://doi.org/10.1093/mnras/128.4.307} {\bibfield  {journal} {\bibinfo  {journal} {Monthly Notices of the Royal Astronomical Society}\ }\textbf {\bibinfo {volume} {128}},\ \bibinfo {pages} {307} (\bibinfo {year} {1964})}\BibitemShut {NoStop}%
\bibitem [{\citenamefont {Miret-Roig}(2023)}]{MiretRoig2023}%
  \BibitemOpen
  \bibfield  {author} {\bibinfo {author} {\bibfnamefont {N.}~\bibnamefont {Miret-Roig}},\ }\bibfield  {journal} {\bibinfo  {journal} {Astrophysics and Space Science}\ }\textbf {\bibinfo {volume} {368}},\ \href {https://doi.org/10.1007/s10509-023-04175-5} {10.1007/s10509-023-04175-5} (\bibinfo {year} {2023})\BibitemShut {NoStop}%
\bibitem [{\citenamefont {{Koshimoto}}\ \emph {et~al.}(2023)\citenamefont {{Koshimoto}}, \citenamefont {{Sumi}}, \citenamefont {{Bennett}}, \citenamefont {{Bozza}}, \citenamefont {{Mr{\'o}z}}, \citenamefont {{Udalski}}, \citenamefont {{Rattenbury}}, \citenamefont {{Abe}}, \citenamefont {{Barry}}, \citenamefont {{Bhattacharya}}, \citenamefont {{Bond}}, \citenamefont {{Fujii}}, \citenamefont {{Fukui}}, \citenamefont {{Hamada}}, \citenamefont {{Hirao}}, \citenamefont {{Silva}}, \citenamefont {{Itow}}, \citenamefont {{Kirikawa}}, \citenamefont {{Kondo}}, \citenamefont {{Matsubara}}, \citenamefont {{Miyazaki}}, \citenamefont {{Muraki}}, \citenamefont {{Olmschenk}}, \citenamefont {{Ranc}}, \citenamefont {{Satoh}}, \citenamefont {{Suzuki}}, \citenamefont {{Tomoyoshi}}, \citenamefont {{Tristram}}, \citenamefont {{Vandorou}}, \citenamefont {{Yama}},\ and\ \citenamefont {{Yamashita}}}]{Koshimoto2023}%
  \BibitemOpen
  \bibfield  {author} {\bibinfo {author} {\bibfnamefont {N.}~\bibnamefont {{Koshimoto}}}, \bibinfo {author} {\bibfnamefont {T.}~\bibnamefont {{Sumi}}}, \bibinfo {author} {\bibfnamefont {D.~P.}\ \bibnamefont {{Bennett}}}, \bibinfo {author} {\bibfnamefont {V.}~\bibnamefont {{Bozza}}}, \bibinfo {author} {\bibfnamefont {P.}~\bibnamefont {{Mr{\'o}z}}}, \bibinfo {author} {\bibfnamefont {A.}~\bibnamefont {{Udalski}}}, \bibinfo {author} {\bibfnamefont {N.~J.}\ \bibnamefont {{Rattenbury}}}, \bibinfo {author} {\bibfnamefont {F.}~\bibnamefont {{Abe}}}, \bibinfo {author} {\bibfnamefont {R.}~\bibnamefont {{Barry}}}, \bibinfo {author} {\bibfnamefont {A.}~\bibnamefont {{Bhattacharya}}}, \bibinfo {author} {\bibfnamefont {I.~A.}\ \bibnamefont {{Bond}}}, \bibinfo {author} {\bibfnamefont {H.}~\bibnamefont {{Fujii}}}, \bibinfo {author} {\bibfnamefont {A.}~\bibnamefont {{Fukui}}}, \bibinfo {author} {\bibfnamefont {R.}~\bibnamefont {{Hamada}}}, \bibinfo {author} {\bibfnamefont {Y.}~\bibnamefont {{Hirao}}}, \bibinfo {author}
  {\bibfnamefont {S.~I.}\ \bibnamefont {{Silva}}}, \bibinfo {author} {\bibfnamefont {Y.}~\bibnamefont {{Itow}}}, \bibinfo {author} {\bibfnamefont {R.}~\bibnamefont {{Kirikawa}}}, \bibinfo {author} {\bibfnamefont {I.}~\bibnamefont {{Kondo}}}, \bibinfo {author} {\bibfnamefont {Y.}~\bibnamefont {{Matsubara}}}, \bibinfo {author} {\bibfnamefont {S.}~\bibnamefont {{Miyazaki}}}, \bibinfo {author} {\bibfnamefont {Y.}~\bibnamefont {{Muraki}}}, \bibinfo {author} {\bibfnamefont {G.}~\bibnamefont {{Olmschenk}}}, \bibinfo {author} {\bibfnamefont {C.}~\bibnamefont {{Ranc}}}, \bibinfo {author} {\bibfnamefont {Y.}~\bibnamefont {{Satoh}}}, \bibinfo {author} {\bibfnamefont {D.}~\bibnamefont {{Suzuki}}}, \bibinfo {author} {\bibfnamefont {M.}~\bibnamefont {{Tomoyoshi}}}, \bibinfo {author} {\bibfnamefont {P.~J.}\ \bibnamefont {{Tristram}}}, \bibinfo {author} {\bibfnamefont {A.}~\bibnamefont {{Vandorou}}}, \bibinfo {author} {\bibfnamefont {H.}~\bibnamefont {{Yama}}},\ and\ \bibinfo {author} {\bibfnamefont {K.}~\bibnamefont
  {{Yamashita}}},\ }\href {https://doi.org/10.3847/1538-3881/ace689} {\bibfield  {journal} {\bibinfo  {journal} {The Astronomical Journal}\ }\textbf {\bibinfo {volume} {166}},\ \bibinfo {eid} {107} (\bibinfo {year} {2023})},\ \Eprint {https://arxiv.org/abs/2303.08279} {arXiv:2303.08279 [astro-ph.EP]} \BibitemShut {NoStop}%
\bibitem [{\citenamefont {Mr{\'{o}}z}\ \emph {et~al.}(2017)\citenamefont {Mr{\'{o}}z}, \citenamefont {Udalski}, \citenamefont {Skowron}, \citenamefont {Poleski}, \citenamefont {Koz{\l}owski}, \citenamefont {Szyma{\'{n}}ski}, \citenamefont {Soszy{\'{n}}ski}, \citenamefont {Wyrzykowski}, \citenamefont {Pietrukowicz}, \citenamefont {Ulaczyk}, \citenamefont {Skowron},\ and\ \citenamefont {Pawlak}}]{Mroz2017}%
  \BibitemOpen
  \bibfield  {author} {\bibinfo {author} {\bibfnamefont {P.}~\bibnamefont {Mr{\'{o}}z}}, \bibinfo {author} {\bibfnamefont {A.}~\bibnamefont {Udalski}}, \bibinfo {author} {\bibfnamefont {J.}~\bibnamefont {Skowron}}, \bibinfo {author} {\bibfnamefont {R.}~\bibnamefont {Poleski}}, \bibinfo {author} {\bibfnamefont {S.}~\bibnamefont {Koz{\l}owski}}, \bibinfo {author} {\bibfnamefont {M.~K.}\ \bibnamefont {Szyma{\'{n}}ski}}, \bibinfo {author} {\bibfnamefont {I.}~\bibnamefont {Soszy{\'{n}}ski}}, \bibinfo {author} {\bibfnamefont {{\L}.}~\bibnamefont {Wyrzykowski}}, \bibinfo {author} {\bibfnamefont {P.}~\bibnamefont {Pietrukowicz}}, \bibinfo {author} {\bibfnamefont {K.}~\bibnamefont {Ulaczyk}}, \bibinfo {author} {\bibfnamefont {D.}~\bibnamefont {Skowron}},\ and\ \bibinfo {author} {\bibfnamefont {M.}~\bibnamefont {Pawlak}},\ }\href {https://doi.org/10.1038/nature23276} {\bibfield  {journal} {\bibinfo  {journal} {Nature}\ }\textbf {\bibinfo {volume} {548}},\ \bibinfo {pages} {183} (\bibinfo {year} {2017})}\BibitemShut
  {NoStop}%
\bibitem [{\citenamefont {{Sahu}}\ \emph {et~al.}(2022)\citenamefont {{Sahu}}, \citenamefont {{Anderson}}, \citenamefont {{Casertano}}, \citenamefont {{Bond}}, \citenamefont {{Udalski}}, \citenamefont {{Dominik}}, \citenamefont {{Calamida}}, \citenamefont {{Bellini}}, \citenamefont {{Brown}}, \citenamefont {{Rejkuba}}, \citenamefont {{Bajaj}}, \citenamefont {{Kains}}, \citenamefont {{Ferguson}}, \citenamefont {{Fryer}}, \citenamefont {{Yock}}, \citenamefont {{Mr{\'o}z}}, \citenamefont {{Koz{\l}owski}}, \citenamefont {{Pietrukowicz}}, \citenamefont {{Poleski}}, \citenamefont {{Skowron}}, \citenamefont {{Soszy{\'n}ski}}, \citenamefont {{Szyma{\'n}ski}}, \citenamefont {{Ulaczyk}}, \citenamefont {{Wyrzykowski}}, \citenamefont {{Barry}}, \citenamefont {{Bennett}}, \citenamefont {{Bond}}, \citenamefont {{Hirao}}, \citenamefont {{Silva}}, \citenamefont {{Kondo}}, \citenamefont {{Koshimoto}}, \citenamefont {{Ranc}}, \citenamefont {{Rattenbury}}, \citenamefont {{Sumi}}, \citenamefont {{Suzuki}}, \citenamefont
  {{Tristram}}, \citenamefont {{Vandorou}}, \citenamefont {{Beaulieu}}, \citenamefont {{Marquette}}, \citenamefont {{Cole}}, \citenamefont {{Fouqu{\'e}}}, \citenamefont {{Hill}}, \citenamefont {{Dieters}}, \citenamefont {{Coutures}}, \citenamefont {{Dominis-Prester}}, \citenamefont {{Bennett}}, \citenamefont {{Bachelet}}, \citenamefont {{Menzies}}, \citenamefont {{Albrow}}, \citenamefont {{Pollard}}, \citenamefont {{Gould}}, \citenamefont {{Yee}}, \citenamefont {{Allen}}, \citenamefont {{Almeida}}, \citenamefont {{Christie}}, \citenamefont {{Drummond}}, \citenamefont {{Gal-Yam}}, \citenamefont {{Gorbikov}}, \citenamefont {{Jablonski}}, \citenamefont {{Lee}}, \citenamefont {{Maoz}}, \citenamefont {{Manulis}}, \citenamefont {{McCormick}}, \citenamefont {{Natusch}}, \citenamefont {{Pogge}}, \citenamefont {{Shvartzvald}}, \citenamefont {{J{\o}rgensen}}, \citenamefont {{Alsubai}}, \citenamefont {{Andersen}}, \citenamefont {{Bozza}}, \citenamefont {{Novati}}, \citenamefont {{Burgdorf}}, \citenamefont {{Hinse}},
  \citenamefont {{Hundertmark}}, \citenamefont {{Husser}}, \citenamefont {{Kerins}}, \citenamefont {{Longa-Pe{\~n}a}}, \citenamefont {{Mancini}}, \citenamefont {{Penny}}, \citenamefont {{Rahvar}}, \citenamefont {{Ricci}}, \citenamefont {{Sajadian}}, \citenamefont {{Skottfelt}}, \citenamefont {{Snodgrass}}, \citenamefont {{Southworth}}, \citenamefont {{Tregloan-Reed}}, \citenamefont {{Wambsganss}}, \citenamefont {{Wertz}}, \citenamefont {{Tsapras}}, \citenamefont {{Street}}, \citenamefont {{Bramich}}, \citenamefont {{Horne}}, \citenamefont {{Steele}},\ and\ \citenamefont {{RoboNet Collaboration}}}]{Sahu2022}%
  \BibitemOpen
  \bibfield  {author} {\bibinfo {author} {\bibfnamefont {K.~C.}\ \bibnamefont {{Sahu}}}, \bibinfo {author} {\bibfnamefont {J.}~\bibnamefont {{Anderson}}}, \bibinfo {author} {\bibfnamefont {S.}~\bibnamefont {{Casertano}}}, \bibinfo {author} {\bibfnamefont {H.~E.}\ \bibnamefont {{Bond}}}, \bibinfo {author} {\bibfnamefont {A.}~\bibnamefont {{Udalski}}}, \bibinfo {author} {\bibfnamefont {M.}~\bibnamefont {{Dominik}}}, \bibinfo {author} {\bibfnamefont {A.}~\bibnamefont {{Calamida}}}, \bibinfo {author} {\bibfnamefont {A.}~\bibnamefont {{Bellini}}}, \bibinfo {author} {\bibfnamefont {T.~M.}\ \bibnamefont {{Brown}}}, \bibinfo {author} {\bibfnamefont {M.}~\bibnamefont {{Rejkuba}}}, \bibinfo {author} {\bibfnamefont {V.}~\bibnamefont {{Bajaj}}}, \bibinfo {author} {\bibfnamefont {N.}~\bibnamefont {{Kains}}}, \bibinfo {author} {\bibfnamefont {H.~C.}\ \bibnamefont {{Ferguson}}}, \bibinfo {author} {\bibfnamefont {C.~L.}\ \bibnamefont {{Fryer}}}, \bibinfo {author} {\bibfnamefont {P.}~\bibnamefont {{Yock}}}, \bibinfo {author}
  {\bibfnamefont {P.}~\bibnamefont {{Mr{\'o}z}}}, \bibinfo {author} {\bibfnamefont {S.}~\bibnamefont {{Koz{\l}owski}}}, \bibinfo {author} {\bibfnamefont {P.}~\bibnamefont {{Pietrukowicz}}}, \bibinfo {author} {\bibfnamefont {R.}~\bibnamefont {{Poleski}}}, \bibinfo {author} {\bibfnamefont {J.}~\bibnamefont {{Skowron}}}, \bibinfo {author} {\bibfnamefont {I.}~\bibnamefont {{Soszy{\'n}ski}}}, \bibinfo {author} {\bibfnamefont {M.~K.}\ \bibnamefont {{Szyma{\'n}ski}}}, \bibinfo {author} {\bibfnamefont {K.}~\bibnamefont {{Ulaczyk}}}, \bibinfo {author} {\bibfnamefont {{\L}.}~\bibnamefont {{Wyrzykowski}}}, \bibinfo {author} {\bibfnamefont {R.~K.}\ \bibnamefont {{Barry}}}, \bibinfo {author} {\bibfnamefont {D.~P.}\ \bibnamefont {{Bennett}}}, \bibinfo {author} {\bibfnamefont {I.~A.}\ \bibnamefont {{Bond}}}, \bibinfo {author} {\bibfnamefont {Y.}~\bibnamefont {{Hirao}}}, \bibinfo {author} {\bibfnamefont {S.~I.}\ \bibnamefont {{Silva}}}, \bibinfo {author} {\bibfnamefont {I.}~\bibnamefont {{Kondo}}}, \bibinfo {author}
  {\bibfnamefont {N.}~\bibnamefont {{Koshimoto}}}, \bibinfo {author} {\bibfnamefont {C.}~\bibnamefont {{Ranc}}}, \bibinfo {author} {\bibfnamefont {N.~J.}\ \bibnamefont {{Rattenbury}}}, \bibinfo {author} {\bibfnamefont {T.}~\bibnamefont {{Sumi}}}, \bibinfo {author} {\bibfnamefont {D.}~\bibnamefont {{Suzuki}}}, \bibinfo {author} {\bibfnamefont {P.~J.}\ \bibnamefont {{Tristram}}}, \bibinfo {author} {\bibfnamefont {A.}~\bibnamefont {{Vandorou}}}, \bibinfo {author} {\bibfnamefont {J.-P.}\ \bibnamefont {{Beaulieu}}}, \bibinfo {author} {\bibfnamefont {J.-B.}\ \bibnamefont {{Marquette}}}, \bibinfo {author} {\bibfnamefont {A.}~\bibnamefont {{Cole}}}, \bibinfo {author} {\bibfnamefont {P.}~\bibnamefont {{Fouqu{\'e}}}}, \bibinfo {author} {\bibfnamefont {K.}~\bibnamefont {{Hill}}}, \bibinfo {author} {\bibfnamefont {S.}~\bibnamefont {{Dieters}}}, \bibinfo {author} {\bibfnamefont {C.}~\bibnamefont {{Coutures}}}, \bibinfo {author} {\bibfnamefont {D.}~\bibnamefont {{Dominis-Prester}}}, \bibinfo {author} {\bibfnamefont
  {C.}~\bibnamefont {{Bennett}}}, \bibinfo {author} {\bibfnamefont {E.}~\bibnamefont {{Bachelet}}}, \bibinfo {author} {\bibfnamefont {J.}~\bibnamefont {{Menzies}}}, \bibinfo {author} {\bibfnamefont {M.}~\bibnamefont {{Albrow}}}, \bibinfo {author} {\bibfnamefont {K.}~\bibnamefont {{Pollard}}}, \bibinfo {author} {\bibfnamefont {A.}~\bibnamefont {{Gould}}}, \bibinfo {author} {\bibfnamefont {J.~C.}\ \bibnamefont {{Yee}}}, \bibinfo {author} {\bibfnamefont {W.}~\bibnamefont {{Allen}}}, \bibinfo {author} {\bibfnamefont {L.~A.}\ \bibnamefont {{Almeida}}}, \bibinfo {author} {\bibfnamefont {G.}~\bibnamefont {{Christie}}}, \bibinfo {author} {\bibfnamefont {J.}~\bibnamefont {{Drummond}}}, \bibinfo {author} {\bibfnamefont {A.}~\bibnamefont {{Gal-Yam}}}, \bibinfo {author} {\bibfnamefont {E.}~\bibnamefont {{Gorbikov}}}, \bibinfo {author} {\bibfnamefont {F.}~\bibnamefont {{Jablonski}}}, \bibinfo {author} {\bibfnamefont {C.-U.}\ \bibnamefont {{Lee}}}, \bibinfo {author} {\bibfnamefont {D.}~\bibnamefont {{Maoz}}}, \bibinfo
  {author} {\bibfnamefont {I.}~\bibnamefont {{Manulis}}}, \bibinfo {author} {\bibfnamefont {J.}~\bibnamefont {{McCormick}}}, \bibinfo {author} {\bibfnamefont {T.}~\bibnamefont {{Natusch}}}, \bibinfo {author} {\bibfnamefont {R.~W.}\ \bibnamefont {{Pogge}}}, \bibinfo {author} {\bibfnamefont {Y.}~\bibnamefont {{Shvartzvald}}}, \bibinfo {author} {\bibfnamefont {U.~G.}\ \bibnamefont {{J{\o}rgensen}}}, \bibinfo {author} {\bibfnamefont {K.~A.}\ \bibnamefont {{Alsubai}}}, \bibinfo {author} {\bibfnamefont {M.~I.}\ \bibnamefont {{Andersen}}}, \bibinfo {author} {\bibfnamefont {V.}~\bibnamefont {{Bozza}}}, \bibinfo {author} {\bibfnamefont {S.~C.}\ \bibnamefont {{Novati}}}, \bibinfo {author} {\bibfnamefont {M.}~\bibnamefont {{Burgdorf}}}, \bibinfo {author} {\bibfnamefont {T.~C.}\ \bibnamefont {{Hinse}}}, \bibinfo {author} {\bibfnamefont {M.}~\bibnamefont {{Hundertmark}}}, \bibinfo {author} {\bibfnamefont {T.-O.}\ \bibnamefont {{Husser}}}, \bibinfo {author} {\bibfnamefont {E.}~\bibnamefont {{Kerins}}}, \bibinfo {author}
  {\bibfnamefont {P.}~\bibnamefont {{Longa-Pe{\~n}a}}}, \bibinfo {author} {\bibfnamefont {L.}~\bibnamefont {{Mancini}}}, \bibinfo {author} {\bibfnamefont {M.}~\bibnamefont {{Penny}}}, \bibinfo {author} {\bibfnamefont {S.}~\bibnamefont {{Rahvar}}}, \bibinfo {author} {\bibfnamefont {D.}~\bibnamefont {{Ricci}}}, \bibinfo {author} {\bibfnamefont {S.}~\bibnamefont {{Sajadian}}}, \bibinfo {author} {\bibfnamefont {J.}~\bibnamefont {{Skottfelt}}}, \bibinfo {author} {\bibfnamefont {C.}~\bibnamefont {{Snodgrass}}}, \bibinfo {author} {\bibfnamefont {J.}~\bibnamefont {{Southworth}}}, \bibinfo {author} {\bibfnamefont {J.}~\bibnamefont {{Tregloan-Reed}}}, \bibinfo {author} {\bibfnamefont {J.}~\bibnamefont {{Wambsganss}}}, \bibinfo {author} {\bibfnamefont {O.}~\bibnamefont {{Wertz}}}, \bibinfo {author} {\bibfnamefont {Y.}~\bibnamefont {{Tsapras}}}, \bibinfo {author} {\bibfnamefont {R.~A.}\ \bibnamefont {{Street}}}, \bibinfo {author} {\bibfnamefont {D.~M.}\ \bibnamefont {{Bramich}}}, \bibinfo {author} {\bibfnamefont
  {K.}~\bibnamefont {{Horne}}}, \bibinfo {author} {\bibfnamefont {I.~A.}\ \bibnamefont {{Steele}}},\ and\ \bibinfo {author} {\bibnamefont {{RoboNet Collaboration}}},\ }\href {https://doi.org/10.3847/1538-4357/ac739e} {\bibfield  {journal} {\bibinfo  {journal} {The Astrophysical Journal}\ }\textbf {\bibinfo {volume} {933}},\ \bibinfo {eid} {83} (\bibinfo {year} {2022})},\ \Eprint {https://arxiv.org/abs/2201.13296} {arXiv:2201.13296 [astro-ph.SR]} \BibitemShut {NoStop}%
\bibitem [{\citenamefont {{Kaczmarek}}\ \emph {et~al.}(2025)\citenamefont {{Kaczmarek}}, \citenamefont {{McGill}}, \citenamefont {{Perkins}}, \citenamefont {{Dawson}}, \citenamefont {{Huston}}, \citenamefont {{Ho}}, \citenamefont {{Abrams}},\ and\ \citenamefont {{Lu}}}]{Kaczmarek2025}%
  \BibitemOpen
  \bibfield  {author} {\bibinfo {author} {\bibfnamefont {Z.}~\bibnamefont {{Kaczmarek}}}, \bibinfo {author} {\bibfnamefont {P.}~\bibnamefont {{McGill}}}, \bibinfo {author} {\bibfnamefont {S.~E.}\ \bibnamefont {{Perkins}}}, \bibinfo {author} {\bibfnamefont {W.~A.}\ \bibnamefont {{Dawson}}}, \bibinfo {author} {\bibfnamefont {M.}~\bibnamefont {{Huston}}}, \bibinfo {author} {\bibfnamefont {M.-F.}\ \bibnamefont {{Ho}}}, \bibinfo {author} {\bibfnamefont {N.~S.}\ \bibnamefont {{Abrams}}},\ and\ \bibinfo {author} {\bibfnamefont {J.~R.}\ \bibnamefont {{Lu}}},\ }\href {https://doi.org/10.3847/1538-4357/adb1d7} {\bibfield  {journal} {\bibinfo  {journal} {The Astrophysical Journal}\ }\textbf {\bibinfo {volume} {981}},\ \bibinfo {eid} {183} (\bibinfo {year} {2025})},\ \Eprint {https://arxiv.org/abs/2410.14098} {arXiv:2410.14098 [astro-ph.SR]} \BibitemShut {NoStop}%
\bibitem [{\citenamefont {{Planck Collaboration}}\ \emph {et~al.}(2020)\citenamefont {{Planck Collaboration}}, \citenamefont {{Aghanim}}, \citenamefont {{Akrami}}, \citenamefont {{Ashdown}}, \citenamefont {{Aumont}}, \citenamefont {{Baccigalupi}}, \citenamefont {{Ballardini}}, \citenamefont {{Banday}}, \citenamefont {{Barreiro}}, \citenamefont {{Bartolo}}, \citenamefont {{Basak}}, \citenamefont {{Battye}}, \citenamefont {{Benabed}}, \citenamefont {{Bernard}}, \citenamefont {{Bersanelli}}, \citenamefont {{Bielewicz}}, \citenamefont {{Bock}}, \citenamefont {{Bond}}, \citenamefont {{Borrill}}, \citenamefont {{Bouchet}}, \citenamefont {{Boulanger}}, \citenamefont {{Bucher}}, \citenamefont {{Burigana}}, \citenamefont {{Butler}}, \citenamefont {{Calabrese}}, \citenamefont {{Cardoso}}, \citenamefont {{Carron}}, \citenamefont {{Challinor}}, \citenamefont {{Chiang}}, \citenamefont {{Chluba}}, \citenamefont {{Colombo}}, \citenamefont {{Combet}}, \citenamefont {{Contreras}}, \citenamefont {{Crill}}, \citenamefont
  {{Cuttaia}}, \citenamefont {{de Bernardis}}, \citenamefont {{de Zotti}}, \citenamefont {{Delabrouille}}, \citenamefont {{Delouis}}, \citenamefont {{Di Valentino}}, \citenamefont {{Diego}}, \citenamefont {{Dor{\'e}}}, \citenamefont {{Douspis}}, \citenamefont {{Ducout}}, \citenamefont {{Dupac}}, \citenamefont {{Dusini}}, \citenamefont {{Efstathiou}}, \citenamefont {{Elsner}}, \citenamefont {{En{\ss}lin}}, \citenamefont {{Eriksen}}, \citenamefont {{Fantaye}}, \citenamefont {{Farhang}}, \citenamefont {{Fergusson}}, \citenamefont {{Fernandez-Cobos}}, \citenamefont {{Finelli}}, \citenamefont {{Forastieri}}, \citenamefont {{Frailis}}, \citenamefont {{Fraisse}}, \citenamefont {{Franceschi}}, \citenamefont {{Frolov}}, \citenamefont {{Galeotta}}, \citenamefont {{Galli}}, \citenamefont {{Ganga}}, \citenamefont {{G{\'e}nova-Santos}}, \citenamefont {{Gerbino}}, \citenamefont {{Ghosh}}, \citenamefont {{Gonz{\'a}lez-Nuevo}}, \citenamefont {{G{\'o}rski}}, \citenamefont {{Gratton}}, \citenamefont {{Gruppuso}}, \citenamefont
  {{Gudmundsson}}, \citenamefont {{Hamann}}, \citenamefont {{Handley}}, \citenamefont {{Hansen}}, \citenamefont {{Herranz}}, \citenamefont {{Hildebrandt}}, \citenamefont {{Hivon}}, \citenamefont {{Huang}}, \citenamefont {{Jaffe}}, \citenamefont {{Jones}}, \citenamefont {{Karakci}}, \citenamefont {{Keih{\"a}nen}}, \citenamefont {{Keskitalo}}, \citenamefont {{Kiiveri}}, \citenamefont {{Kim}}, \citenamefont {{Kisner}}, \citenamefont {{Knox}}, \citenamefont {{Krachmalnicoff}}, \citenamefont {{Kunz}}, \citenamefont {{Kurki-Suonio}}, \citenamefont {{Lagache}}, \citenamefont {{Lamarre}}, \citenamefont {{Lasenby}}, \citenamefont {{Lattanzi}}, \citenamefont {{Lawrence}}, \citenamefont {{Le Jeune}}, \citenamefont {{Lemos}}, \citenamefont {{Lesgourgues}}, \citenamefont {{Levrier}}, \citenamefont {{Lewis}}, \citenamefont {{Liguori}}, \citenamefont {{Lilje}}, \citenamefont {{Lilley}}, \citenamefont {{Lindholm}}, \citenamefont {{L{\'o}pez-Caniego}}, \citenamefont {{Lubin}}, \citenamefont {{Ma}}, \citenamefont
  {{Mac{\'\i}as-P{\'e}rez}}, \citenamefont {{Maggio}}, \citenamefont {{Maino}}, \citenamefont {{Mandolesi}}, \citenamefont {{Mangilli}}, \citenamefont {{Marcos-Caballero}}, \citenamefont {{Maris}}, \citenamefont {{Martin}}, \citenamefont {{Martinelli}}, \citenamefont {{Mart{\'\i}nez-Gonz{\'a}lez}}, \citenamefont {{Matarrese}}, \citenamefont {{Mauri}}, \citenamefont {{McEwen}}, \citenamefont {{Meinhold}}, \citenamefont {{Melchiorri}}, \citenamefont {{Mennella}}, \citenamefont {{Migliaccio}}, \citenamefont {{Millea}}, \citenamefont {{Mitra}}, \citenamefont {{Miville-Desch{\^e}nes}}, \citenamefont {{Molinari}}, \citenamefont {{Montier}}, \citenamefont {{Morgante}}, \citenamefont {{Moss}}, \citenamefont {{Natoli}}, \citenamefont {{N{\o}rgaard-Nielsen}}, \citenamefont {{Pagano}}, \citenamefont {{Paoletti}}, \citenamefont {{Partridge}}, \citenamefont {{Patanchon}}, \citenamefont {{Peiris}}, \citenamefont {{Perrotta}}, \citenamefont {{Pettorino}}, \citenamefont {{Piacentini}}, \citenamefont {{Polastri}},
  \citenamefont {{Polenta}}, \citenamefont {{Puget}}, \citenamefont {{Rachen}}, \citenamefont {{Reinecke}}, \citenamefont {{Remazeilles}}, \citenamefont {{Renzi}}, \citenamefont {{Rocha}}, \citenamefont {{Rosset}}, \citenamefont {{Roudier}}, \citenamefont {{Rubi{\~n}o-Mart{\'\i}n}}, \citenamefont {{Ruiz-Granados}}, \citenamefont {{Salvati}}, \citenamefont {{Sandri}}, \citenamefont {{Savelainen}}, \citenamefont {{Scott}}, \citenamefont {{Shellard}}, \citenamefont {{Sirignano}}, \citenamefont {{Sirri}}, \citenamefont {{Spencer}}, \citenamefont {{Sunyaev}}, \citenamefont {{Suur-Uski}}, \citenamefont {{Tauber}}, \citenamefont {{Tavagnacco}}, \citenamefont {{Tenti}}, \citenamefont {{Toffolatti}}, \citenamefont {{Tomasi}}, \citenamefont {{Trombetti}}, \citenamefont {{Valenziano}}, \citenamefont {{Valiviita}}, \citenamefont {{Van Tent}}, \citenamefont {{Vibert}}, \citenamefont {{Vielva}}, \citenamefont {{Villa}}, \citenamefont {{Vittorio}}, \citenamefont {{Wandelt}}, \citenamefont {{Wehus}}, \citenamefont {{White}},
  \citenamefont {{White}}, \citenamefont {{Zacchei}},\ and\ \citenamefont {{Zonca}}}]{Planck2018Results}%
  \BibitemOpen
  \bibfield  {author} {\bibinfo {author} {\bibnamefont {{Planck Collaboration}}}, \bibinfo {author} {\bibfnamefont {N.}~\bibnamefont {{Aghanim}}}, \bibinfo {author} {\bibfnamefont {Y.}~\bibnamefont {{Akrami}}}, \bibinfo {author} {\bibfnamefont {M.}~\bibnamefont {{Ashdown}}}, \bibinfo {author} {\bibfnamefont {J.}~\bibnamefont {{Aumont}}}, \bibinfo {author} {\bibfnamefont {C.}~\bibnamefont {{Baccigalupi}}}, \bibinfo {author} {\bibfnamefont {M.}~\bibnamefont {{Ballardini}}}, \bibinfo {author} {\bibfnamefont {A.~J.}\ \bibnamefont {{Banday}}}, \bibinfo {author} {\bibfnamefont {R.~B.}\ \bibnamefont {{Barreiro}}}, \bibinfo {author} {\bibfnamefont {N.}~\bibnamefont {{Bartolo}}}, \bibinfo {author} {\bibfnamefont {S.}~\bibnamefont {{Basak}}}, \bibinfo {author} {\bibfnamefont {R.}~\bibnamefont {{Battye}}}, \bibinfo {author} {\bibfnamefont {K.}~\bibnamefont {{Benabed}}}, \bibinfo {author} {\bibfnamefont {J.~P.}\ \bibnamefont {{Bernard}}}, \bibinfo {author} {\bibfnamefont {M.}~\bibnamefont {{Bersanelli}}}, \bibinfo
  {author} {\bibfnamefont {P.}~\bibnamefont {{Bielewicz}}}, \bibinfo {author} {\bibfnamefont {J.~J.}\ \bibnamefont {{Bock}}}, \bibinfo {author} {\bibfnamefont {J.~R.}\ \bibnamefont {{Bond}}}, \bibinfo {author} {\bibfnamefont {J.}~\bibnamefont {{Borrill}}}, \bibinfo {author} {\bibfnamefont {F.~R.}\ \bibnamefont {{Bouchet}}}, \bibinfo {author} {\bibfnamefont {F.}~\bibnamefont {{Boulanger}}}, \bibinfo {author} {\bibfnamefont {M.}~\bibnamefont {{Bucher}}}, \bibinfo {author} {\bibfnamefont {C.}~\bibnamefont {{Burigana}}}, \bibinfo {author} {\bibfnamefont {R.~C.}\ \bibnamefont {{Butler}}}, \bibinfo {author} {\bibfnamefont {E.}~\bibnamefont {{Calabrese}}}, \bibinfo {author} {\bibfnamefont {J.~F.}\ \bibnamefont {{Cardoso}}}, \bibinfo {author} {\bibfnamefont {J.}~\bibnamefont {{Carron}}}, \bibinfo {author} {\bibfnamefont {A.}~\bibnamefont {{Challinor}}}, \bibinfo {author} {\bibfnamefont {H.~C.}\ \bibnamefont {{Chiang}}}, \bibinfo {author} {\bibfnamefont {J.}~\bibnamefont {{Chluba}}}, \bibinfo {author} {\bibfnamefont
  {L.~P.~L.}\ \bibnamefont {{Colombo}}}, \bibinfo {author} {\bibfnamefont {C.}~\bibnamefont {{Combet}}}, \bibinfo {author} {\bibfnamefont {D.}~\bibnamefont {{Contreras}}}, \bibinfo {author} {\bibfnamefont {B.~P.}\ \bibnamefont {{Crill}}}, \bibinfo {author} {\bibfnamefont {F.}~\bibnamefont {{Cuttaia}}}, \bibinfo {author} {\bibfnamefont {P.}~\bibnamefont {{de Bernardis}}}, \bibinfo {author} {\bibfnamefont {G.}~\bibnamefont {{de Zotti}}}, \bibinfo {author} {\bibfnamefont {J.}~\bibnamefont {{Delabrouille}}}, \bibinfo {author} {\bibfnamefont {J.~M.}\ \bibnamefont {{Delouis}}}, \bibinfo {author} {\bibfnamefont {E.}~\bibnamefont {{Di Valentino}}}, \bibinfo {author} {\bibfnamefont {J.~M.}\ \bibnamefont {{Diego}}}, \bibinfo {author} {\bibfnamefont {O.}~\bibnamefont {{Dor{\'e}}}}, \bibinfo {author} {\bibfnamefont {M.}~\bibnamefont {{Douspis}}}, \bibinfo {author} {\bibfnamefont {A.}~\bibnamefont {{Ducout}}}, \bibinfo {author} {\bibfnamefont {X.}~\bibnamefont {{Dupac}}}, \bibinfo {author} {\bibfnamefont {S.}~\bibnamefont
  {{Dusini}}}, \bibinfo {author} {\bibfnamefont {G.}~\bibnamefont {{Efstathiou}}}, \bibinfo {author} {\bibfnamefont {F.}~\bibnamefont {{Elsner}}}, \bibinfo {author} {\bibfnamefont {T.~A.}\ \bibnamefont {{En{\ss}lin}}}, \bibinfo {author} {\bibfnamefont {H.~K.}\ \bibnamefont {{Eriksen}}}, \bibinfo {author} {\bibfnamefont {Y.}~\bibnamefont {{Fantaye}}}, \bibinfo {author} {\bibfnamefont {M.}~\bibnamefont {{Farhang}}}, \bibinfo {author} {\bibfnamefont {J.}~\bibnamefont {{Fergusson}}}, \bibinfo {author} {\bibfnamefont {R.}~\bibnamefont {{Fernandez-Cobos}}}, \bibinfo {author} {\bibfnamefont {F.}~\bibnamefont {{Finelli}}}, \bibinfo {author} {\bibfnamefont {F.}~\bibnamefont {{Forastieri}}}, \bibinfo {author} {\bibfnamefont {M.}~\bibnamefont {{Frailis}}}, \bibinfo {author} {\bibfnamefont {A.~A.}\ \bibnamefont {{Fraisse}}}, \bibinfo {author} {\bibfnamefont {E.}~\bibnamefont {{Franceschi}}}, \bibinfo {author} {\bibfnamefont {A.}~\bibnamefont {{Frolov}}}, \bibinfo {author} {\bibfnamefont {S.}~\bibnamefont {{Galeotta}}},
  \bibinfo {author} {\bibfnamefont {S.}~\bibnamefont {{Galli}}}, \bibinfo {author} {\bibfnamefont {K.}~\bibnamefont {{Ganga}}}, \bibinfo {author} {\bibfnamefont {R.~T.}\ \bibnamefont {{G{\'e}nova-Santos}}}, \bibinfo {author} {\bibfnamefont {M.}~\bibnamefont {{Gerbino}}}, \bibinfo {author} {\bibfnamefont {T.}~\bibnamefont {{Ghosh}}}, \bibinfo {author} {\bibfnamefont {J.}~\bibnamefont {{Gonz{\'a}lez-Nuevo}}}, \bibinfo {author} {\bibfnamefont {K.~M.}\ \bibnamefont {{G{\'o}rski}}}, \bibinfo {author} {\bibfnamefont {S.}~\bibnamefont {{Gratton}}}, \bibinfo {author} {\bibfnamefont {A.}~\bibnamefont {{Gruppuso}}}, \bibinfo {author} {\bibfnamefont {J.~E.}\ \bibnamefont {{Gudmundsson}}}, \bibinfo {author} {\bibfnamefont {J.}~\bibnamefont {{Hamann}}}, \bibinfo {author} {\bibfnamefont {W.}~\bibnamefont {{Handley}}}, \bibinfo {author} {\bibfnamefont {F.~K.}\ \bibnamefont {{Hansen}}}, \bibinfo {author} {\bibfnamefont {D.}~\bibnamefont {{Herranz}}}, \bibinfo {author} {\bibfnamefont {S.~R.}\ \bibnamefont {{Hildebrandt}}},
  \bibinfo {author} {\bibfnamefont {E.}~\bibnamefont {{Hivon}}}, \bibinfo {author} {\bibfnamefont {Z.}~\bibnamefont {{Huang}}}, \bibinfo {author} {\bibfnamefont {A.~H.}\ \bibnamefont {{Jaffe}}}, \bibinfo {author} {\bibfnamefont {W.~C.}\ \bibnamefont {{Jones}}}, \bibinfo {author} {\bibfnamefont {A.}~\bibnamefont {{Karakci}}}, \bibinfo {author} {\bibfnamefont {E.}~\bibnamefont {{Keih{\"a}nen}}}, \bibinfo {author} {\bibfnamefont {R.}~\bibnamefont {{Keskitalo}}}, \bibinfo {author} {\bibfnamefont {K.}~\bibnamefont {{Kiiveri}}}, \bibinfo {author} {\bibfnamefont {J.}~\bibnamefont {{Kim}}}, \bibinfo {author} {\bibfnamefont {T.~S.}\ \bibnamefont {{Kisner}}}, \bibinfo {author} {\bibfnamefont {L.}~\bibnamefont {{Knox}}}, \bibinfo {author} {\bibfnamefont {N.}~\bibnamefont {{Krachmalnicoff}}}, \bibinfo {author} {\bibfnamefont {M.}~\bibnamefont {{Kunz}}}, \bibinfo {author} {\bibfnamefont {H.}~\bibnamefont {{Kurki-Suonio}}}, \bibinfo {author} {\bibfnamefont {G.}~\bibnamefont {{Lagache}}}, \bibinfo {author} {\bibfnamefont
  {J.~M.}\ \bibnamefont {{Lamarre}}}, \bibinfo {author} {\bibfnamefont {A.}~\bibnamefont {{Lasenby}}}, \bibinfo {author} {\bibfnamefont {M.}~\bibnamefont {{Lattanzi}}}, \bibinfo {author} {\bibfnamefont {C.~R.}\ \bibnamefont {{Lawrence}}}, \bibinfo {author} {\bibfnamefont {M.}~\bibnamefont {{Le Jeune}}}, \bibinfo {author} {\bibfnamefont {P.}~\bibnamefont {{Lemos}}}, \bibinfo {author} {\bibfnamefont {J.}~\bibnamefont {{Lesgourgues}}}, \bibinfo {author} {\bibfnamefont {F.}~\bibnamefont {{Levrier}}}, \bibinfo {author} {\bibfnamefont {A.}~\bibnamefont {{Lewis}}}, \bibinfo {author} {\bibfnamefont {M.}~\bibnamefont {{Liguori}}}, \bibinfo {author} {\bibfnamefont {P.~B.}\ \bibnamefont {{Lilje}}}, \bibinfo {author} {\bibfnamefont {M.}~\bibnamefont {{Lilley}}}, \bibinfo {author} {\bibfnamefont {V.}~\bibnamefont {{Lindholm}}}, \bibinfo {author} {\bibfnamefont {M.}~\bibnamefont {{L{\'o}pez-Caniego}}}, \bibinfo {author} {\bibfnamefont {P.~M.}\ \bibnamefont {{Lubin}}}, \bibinfo {author} {\bibfnamefont {Y.~Z.}\ \bibnamefont
  {{Ma}}}, \bibinfo {author} {\bibfnamefont {J.~F.}\ \bibnamefont {{Mac{\'\i}as-P{\'e}rez}}}, \bibinfo {author} {\bibfnamefont {G.}~\bibnamefont {{Maggio}}}, \bibinfo {author} {\bibfnamefont {D.}~\bibnamefont {{Maino}}}, \bibinfo {author} {\bibfnamefont {N.}~\bibnamefont {{Mandolesi}}}, \bibinfo {author} {\bibfnamefont {A.}~\bibnamefont {{Mangilli}}}, \bibinfo {author} {\bibfnamefont {A.}~\bibnamefont {{Marcos-Caballero}}}, \bibinfo {author} {\bibfnamefont {M.}~\bibnamefont {{Maris}}}, \bibinfo {author} {\bibfnamefont {P.~G.}\ \bibnamefont {{Martin}}}, \bibinfo {author} {\bibfnamefont {M.}~\bibnamefont {{Martinelli}}}, \bibinfo {author} {\bibfnamefont {E.}~\bibnamefont {{Mart{\'\i}nez-Gonz{\'a}lez}}}, \bibinfo {author} {\bibfnamefont {S.}~\bibnamefont {{Matarrese}}}, \bibinfo {author} {\bibfnamefont {N.}~\bibnamefont {{Mauri}}}, \bibinfo {author} {\bibfnamefont {J.~D.}\ \bibnamefont {{McEwen}}}, \bibinfo {author} {\bibfnamefont {P.~R.}\ \bibnamefont {{Meinhold}}}, \bibinfo {author} {\bibfnamefont
  {A.}~\bibnamefont {{Melchiorri}}}, \bibinfo {author} {\bibfnamefont {A.}~\bibnamefont {{Mennella}}}, \bibinfo {author} {\bibfnamefont {M.}~\bibnamefont {{Migliaccio}}}, \bibinfo {author} {\bibfnamefont {M.}~\bibnamefont {{Millea}}}, \bibinfo {author} {\bibfnamefont {S.}~\bibnamefont {{Mitra}}}, \bibinfo {author} {\bibfnamefont {M.~A.}\ \bibnamefont {{Miville-Desch{\^e}nes}}}, \bibinfo {author} {\bibfnamefont {D.}~\bibnamefont {{Molinari}}}, \bibinfo {author} {\bibfnamefont {L.}~\bibnamefont {{Montier}}}, \bibinfo {author} {\bibfnamefont {G.}~\bibnamefont {{Morgante}}}, \bibinfo {author} {\bibfnamefont {A.}~\bibnamefont {{Moss}}}, \bibinfo {author} {\bibfnamefont {P.}~\bibnamefont {{Natoli}}}, \bibinfo {author} {\bibfnamefont {H.~U.}\ \bibnamefont {{N{\o}rgaard-Nielsen}}}, \bibinfo {author} {\bibfnamefont {L.}~\bibnamefont {{Pagano}}}, \bibinfo {author} {\bibfnamefont {D.}~\bibnamefont {{Paoletti}}}, \bibinfo {author} {\bibfnamefont {B.}~\bibnamefont {{Partridge}}}, \bibinfo {author} {\bibfnamefont
  {G.}~\bibnamefont {{Patanchon}}}, \bibinfo {author} {\bibfnamefont {H.~V.}\ \bibnamefont {{Peiris}}}, \bibinfo {author} {\bibfnamefont {F.}~\bibnamefont {{Perrotta}}}, \bibinfo {author} {\bibfnamefont {V.}~\bibnamefont {{Pettorino}}}, \bibinfo {author} {\bibfnamefont {F.}~\bibnamefont {{Piacentini}}}, \bibinfo {author} {\bibfnamefont {L.}~\bibnamefont {{Polastri}}}, \bibinfo {author} {\bibfnamefont {G.}~\bibnamefont {{Polenta}}}, \bibinfo {author} {\bibfnamefont {J.~L.}\ \bibnamefont {{Puget}}}, \bibinfo {author} {\bibfnamefont {J.~P.}\ \bibnamefont {{Rachen}}}, \bibinfo {author} {\bibfnamefont {M.}~\bibnamefont {{Reinecke}}}, \bibinfo {author} {\bibfnamefont {M.}~\bibnamefont {{Remazeilles}}}, \bibinfo {author} {\bibfnamefont {A.}~\bibnamefont {{Renzi}}}, \bibinfo {author} {\bibfnamefont {G.}~\bibnamefont {{Rocha}}}, \bibinfo {author} {\bibfnamefont {C.}~\bibnamefont {{Rosset}}}, \bibinfo {author} {\bibfnamefont {G.}~\bibnamefont {{Roudier}}}, \bibinfo {author} {\bibfnamefont {J.~A.}\ \bibnamefont
  {{Rubi{\~n}o-Mart{\'\i}n}}}, \bibinfo {author} {\bibfnamefont {B.}~\bibnamefont {{Ruiz-Granados}}}, \bibinfo {author} {\bibfnamefont {L.}~\bibnamefont {{Salvati}}}, \bibinfo {author} {\bibfnamefont {M.}~\bibnamefont {{Sandri}}}, \bibinfo {author} {\bibfnamefont {M.}~\bibnamefont {{Savelainen}}}, \bibinfo {author} {\bibfnamefont {D.}~\bibnamefont {{Scott}}}, \bibinfo {author} {\bibfnamefont {E.~P.~S.}\ \bibnamefont {{Shellard}}}, \bibinfo {author} {\bibfnamefont {C.}~\bibnamefont {{Sirignano}}}, \bibinfo {author} {\bibfnamefont {G.}~\bibnamefont {{Sirri}}}, \bibinfo {author} {\bibfnamefont {L.~D.}\ \bibnamefont {{Spencer}}}, \bibinfo {author} {\bibfnamefont {R.}~\bibnamefont {{Sunyaev}}}, \bibinfo {author} {\bibfnamefont {A.~S.}\ \bibnamefont {{Suur-Uski}}}, \bibinfo {author} {\bibfnamefont {J.~A.}\ \bibnamefont {{Tauber}}}, \bibinfo {author} {\bibfnamefont {D.}~\bibnamefont {{Tavagnacco}}}, \bibinfo {author} {\bibfnamefont {M.}~\bibnamefont {{Tenti}}}, \bibinfo {author} {\bibfnamefont {L.}~\bibnamefont
  {{Toffolatti}}}, \bibinfo {author} {\bibfnamefont {M.}~\bibnamefont {{Tomasi}}}, \bibinfo {author} {\bibfnamefont {T.}~\bibnamefont {{Trombetti}}}, \bibinfo {author} {\bibfnamefont {L.}~\bibnamefont {{Valenziano}}}, \bibinfo {author} {\bibfnamefont {J.}~\bibnamefont {{Valiviita}}}, \bibinfo {author} {\bibfnamefont {B.}~\bibnamefont {{Van Tent}}}, \bibinfo {author} {\bibfnamefont {L.}~\bibnamefont {{Vibert}}}, \bibinfo {author} {\bibfnamefont {P.}~\bibnamefont {{Vielva}}}, \bibinfo {author} {\bibfnamefont {F.}~\bibnamefont {{Villa}}}, \bibinfo {author} {\bibfnamefont {N.}~\bibnamefont {{Vittorio}}}, \bibinfo {author} {\bibfnamefont {B.~D.}\ \bibnamefont {{Wandelt}}}, \bibinfo {author} {\bibfnamefont {I.~K.}\ \bibnamefont {{Wehus}}}, \bibinfo {author} {\bibfnamefont {M.}~\bibnamefont {{White}}}, \bibinfo {author} {\bibfnamefont {S.~D.~M.}\ \bibnamefont {{White}}}, \bibinfo {author} {\bibfnamefont {A.}~\bibnamefont {{Zacchei}}},\ and\ \bibinfo {author} {\bibfnamefont {A.}~\bibnamefont {{Zonca}}},\ }\href
  {https://doi.org/10.1051/0004-6361/201833910} {\bibfield  {journal} {\bibinfo  {journal} {Astronomy \& Astrophysics}\ }\textbf {\bibinfo {volume} {641}},\ \bibinfo {eid} {A6} (\bibinfo {year} {2020})},\ \Eprint {https://arxiv.org/abs/1807.06209} {arXiv:1807.06209 [astro-ph.CO]} \BibitemShut {NoStop}%
\bibitem [{\citenamefont {Biggs}\ \emph {et~al.}(1999)\citenamefont {Biggs}, \citenamefont {Browne}, \citenamefont {Helbig}, \citenamefont {Koopmans}, \citenamefont {Wilkinson},\ and\ \citenamefont {Perley}}]{biggs1999time}%
  \BibitemOpen
  \bibfield  {author} {\bibinfo {author} {\bibfnamefont {A.}~\bibnamefont {Biggs}}, \bibinfo {author} {\bibfnamefont {I.}~\bibnamefont {Browne}}, \bibinfo {author} {\bibfnamefont {P.}~\bibnamefont {Helbig}}, \bibinfo {author} {\bibfnamefont {L.}~\bibnamefont {Koopmans}}, \bibinfo {author} {\bibfnamefont {P.}~\bibnamefont {Wilkinson}},\ and\ \bibinfo {author} {\bibfnamefont {R.}~\bibnamefont {Perley}},\ }\href@noop {} {\bibfield  {journal} {\bibinfo  {journal} {Monthly Notices of the Royal Astronomical Society}\ }\textbf {\bibinfo {volume} {304}},\ \bibinfo {pages} {349} (\bibinfo {year} {1999})}\BibitemShut {NoStop}%
\bibitem [{\citenamefont {{Suyu}}\ \emph {et~al.}(2017)\citenamefont {{Suyu}}, \citenamefont {{Bonvin}}, \citenamefont {{Courbin}}, \citenamefont {{Fassnacht}}, \citenamefont {{Rusu}}, \citenamefont {{Sluse}}, \citenamefont {{Treu}}, \citenamefont {{Wong}}, \citenamefont {{Auger}}, \citenamefont {{Ding}}, \citenamefont {{Hilbert}}, \citenamefont {{Marshall}}, \citenamefont {{Rumbaugh}}, \citenamefont {{Sonnenfeld}}, \citenamefont {{Tewes}}, \citenamefont {{Tihhonova}}, \citenamefont {{Agnello}}, \citenamefont {{Blandford}}, \citenamefont {{Chen}}, \citenamefont {{Collett}}, \citenamefont {{Koopmans}}, \citenamefont {{Liao}}, \citenamefont {{Meylan}},\ and\ \citenamefont {{Spiniello}}}]{H0LiCOW2017Intro}%
  \BibitemOpen
  \bibfield  {author} {\bibinfo {author} {\bibfnamefont {S.~H.}\ \bibnamefont {{Suyu}}}, \bibinfo {author} {\bibfnamefont {V.}~\bibnamefont {{Bonvin}}}, \bibinfo {author} {\bibfnamefont {F.}~\bibnamefont {{Courbin}}}, \bibinfo {author} {\bibfnamefont {C.~D.}\ \bibnamefont {{Fassnacht}}}, \bibinfo {author} {\bibfnamefont {C.~E.}\ \bibnamefont {{Rusu}}}, \bibinfo {author} {\bibfnamefont {D.}~\bibnamefont {{Sluse}}}, \bibinfo {author} {\bibfnamefont {T.}~\bibnamefont {{Treu}}}, \bibinfo {author} {\bibfnamefont {K.~C.}\ \bibnamefont {{Wong}}}, \bibinfo {author} {\bibfnamefont {M.~W.}\ \bibnamefont {{Auger}}}, \bibinfo {author} {\bibfnamefont {X.}~\bibnamefont {{Ding}}}, \bibinfo {author} {\bibfnamefont {S.}~\bibnamefont {{Hilbert}}}, \bibinfo {author} {\bibfnamefont {P.~J.}\ \bibnamefont {{Marshall}}}, \bibinfo {author} {\bibfnamefont {N.}~\bibnamefont {{Rumbaugh}}}, \bibinfo {author} {\bibfnamefont {A.}~\bibnamefont {{Sonnenfeld}}}, \bibinfo {author} {\bibfnamefont {M.}~\bibnamefont {{Tewes}}}, \bibinfo {author}
  {\bibfnamefont {O.}~\bibnamefont {{Tihhonova}}}, \bibinfo {author} {\bibfnamefont {A.}~\bibnamefont {{Agnello}}}, \bibinfo {author} {\bibfnamefont {R.~D.}\ \bibnamefont {{Blandford}}}, \bibinfo {author} {\bibfnamefont {G.~C.~F.}\ \bibnamefont {{Chen}}}, \bibinfo {author} {\bibfnamefont {T.}~\bibnamefont {{Collett}}}, \bibinfo {author} {\bibfnamefont {L.~V.~E.}\ \bibnamefont {{Koopmans}}}, \bibinfo {author} {\bibfnamefont {K.}~\bibnamefont {{Liao}}}, \bibinfo {author} {\bibfnamefont {G.}~\bibnamefont {{Meylan}}},\ and\ \bibinfo {author} {\bibfnamefont {C.}~\bibnamefont {{Spiniello}}},\ }\href {https://doi.org/10.1093/mnras/stx483} {\bibfield  {journal} {\bibinfo  {journal} {Monthly Notices of the Royal Astronomical Society}\ }\textbf {\bibinfo {volume} {468}},\ \bibinfo {pages} {2590} (\bibinfo {year} {2017})},\ \Eprint {https://arxiv.org/abs/1607.00017} {arXiv:1607.00017 [astro-ph.CO]} \BibitemShut {NoStop}%
\bibitem [{\citenamefont {{Chen}}\ \emph {et~al.}(2019)\citenamefont {{Chen}}, \citenamefont {{Fassnacht}}, \citenamefont {{Suyu}}, \citenamefont {{Rusu}}, \citenamefont {{Chan}}, \citenamefont {{Wong}}, \citenamefont {{Auger}}, \citenamefont {{Hilbert}}, \citenamefont {{Bonvin}}, \citenamefont {{Birrer}}, \citenamefont {{Millon}}, \citenamefont {{Koopmans}}, \citenamefont {{Lagattuta}}, \citenamefont {{McKean}}, \citenamefont {{Vegetti}}, \citenamefont {{Courbin}}, \citenamefont {{Ding}}, \citenamefont {{Halkola}}, \citenamefont {{Jee}}, \citenamefont {{Shajib}}, \citenamefont {{Sluse}}, \citenamefont {{Sonnenfeld}},\ and\ \citenamefont {{Treu}}}]{H0LiCOW2019Review}%
  \BibitemOpen
  \bibfield  {author} {\bibinfo {author} {\bibfnamefont {G.~C.~F.}\ \bibnamefont {{Chen}}}, \bibinfo {author} {\bibfnamefont {C.~D.}\ \bibnamefont {{Fassnacht}}}, \bibinfo {author} {\bibfnamefont {S.~H.}\ \bibnamefont {{Suyu}}}, \bibinfo {author} {\bibfnamefont {C.~E.}\ \bibnamefont {{Rusu}}}, \bibinfo {author} {\bibfnamefont {J.~H.~H.}\ \bibnamefont {{Chan}}}, \bibinfo {author} {\bibfnamefont {K.~C.}\ \bibnamefont {{Wong}}}, \bibinfo {author} {\bibfnamefont {M.~W.}\ \bibnamefont {{Auger}}}, \bibinfo {author} {\bibfnamefont {S.}~\bibnamefont {{Hilbert}}}, \bibinfo {author} {\bibfnamefont {V.}~\bibnamefont {{Bonvin}}}, \bibinfo {author} {\bibfnamefont {S.}~\bibnamefont {{Birrer}}}, \bibinfo {author} {\bibfnamefont {M.}~\bibnamefont {{Millon}}}, \bibinfo {author} {\bibfnamefont {L.~V.~E.}\ \bibnamefont {{Koopmans}}}, \bibinfo {author} {\bibfnamefont {D.~J.}\ \bibnamefont {{Lagattuta}}}, \bibinfo {author} {\bibfnamefont {J.~P.}\ \bibnamefont {{McKean}}}, \bibinfo {author} {\bibfnamefont {S.}~\bibnamefont
  {{Vegetti}}}, \bibinfo {author} {\bibfnamefont {F.}~\bibnamefont {{Courbin}}}, \bibinfo {author} {\bibfnamefont {X.}~\bibnamefont {{Ding}}}, \bibinfo {author} {\bibfnamefont {A.}~\bibnamefont {{Halkola}}}, \bibinfo {author} {\bibfnamefont {I.}~\bibnamefont {{Jee}}}, \bibinfo {author} {\bibfnamefont {A.~J.}\ \bibnamefont {{Shajib}}}, \bibinfo {author} {\bibfnamefont {D.}~\bibnamefont {{Sluse}}}, \bibinfo {author} {\bibfnamefont {A.}~\bibnamefont {{Sonnenfeld}}},\ and\ \bibinfo {author} {\bibfnamefont {T.}~\bibnamefont {{Treu}}},\ }\href {https://doi.org/10.1093/mnras/stz2547} {\bibfield  {journal} {\bibinfo  {journal} {Monthly Notices of the Royal Astronomical Society}\ }\textbf {\bibinfo {volume} {490}},\ \bibinfo {pages} {1743} (\bibinfo {year} {2019})},\ \Eprint {https://arxiv.org/abs/1907.02533} {arXiv:1907.02533 [astro-ph.CO]} \BibitemShut {NoStop}%
\bibitem [{Note1()}]{Note1}%
  \BibitemOpen
  \bibinfo {note} {Different images in a microlensing system are usually not resolved with seeing-limited imaging; however, long-baseline interferometry (e.g., VLTI/GRAVITY) has now resolved microlensed images in select cases (see e.g. Ref.~\cite {mroz2025observations}). Whereas long-baseline interferometers can resolve the micro-images and infer lens parameters from visibilities/closure phases, our method infers the time delay $\Delta t$ from single-photon spectra without resolving the images}\BibitemShut {NoStop}%
\bibitem [{\citenamefont {Deguchi}\ and\ \citenamefont {Watson}(1986)}]{deguchi1986diffraction}%
  \BibitemOpen
  \bibfield  {author} {\bibinfo {author} {\bibfnamefont {S.}~\bibnamefont {Deguchi}}\ and\ \bibinfo {author} {\bibfnamefont {W.~D.}\ \bibnamefont {Watson}},\ }\href@noop {} {\bibfield  {journal} {\bibinfo  {journal} {Astrophysical Journal, Part 1 (ISSN 0004-637X), vol. 307, Aug. 1, 1986, p. 30-37.}\ }\textbf {\bibinfo {volume} {307}},\ \bibinfo {pages} {30} (\bibinfo {year} {1986})}\BibitemShut {NoStop}%
\bibitem [{\citenamefont {Peterson}\ and\ \citenamefont {Falk}(1991)}]{peterson1991gravitational}%
  \BibitemOpen
  \bibfield  {author} {\bibinfo {author} {\bibfnamefont {J.}~\bibnamefont {Peterson}}\ and\ \bibinfo {author} {\bibfnamefont {T.}~\bibnamefont {Falk}},\ }\href@noop {} {\bibfield  {journal} {\bibinfo  {journal} {Astrophysical Journal, Part 2-Letters (ISSN 0004-637X), vol. 374, June 10, 1991, p. L5-L8. NSF-supported research.}\ }\textbf {\bibinfo {volume} {374}},\ \bibinfo {pages} {L5} (\bibinfo {year} {1991})}\BibitemShut {NoStop}%
\bibitem [{\citenamefont {Gould}(1992)}]{gould1992femtolensing}%
  \BibitemOpen
  \bibfield  {author} {\bibinfo {author} {\bibfnamefont {A.}~\bibnamefont {Gould}},\ }\href@noop {} {\bibfield  {journal} {\bibinfo  {journal} {Astrophysical Journal, Part 2-Letters (ISSN 0004-637X), vol. 386, Feb. 10, 1992, p. L5-L7.}\ }\textbf {\bibinfo {volume} {386}},\ \bibinfo {pages} {L5} (\bibinfo {year} {1992})}\BibitemShut {NoStop}%
\bibitem [{\citenamefont {Ulmer}\ and\ \citenamefont {Goodman}(1995)}]{ulmer1995femtolensing}%
  \BibitemOpen
  \bibfield  {author} {\bibinfo {author} {\bibfnamefont {A.}~\bibnamefont {Ulmer}}\ and\ \bibinfo {author} {\bibfnamefont {J.}~\bibnamefont {Goodman}},\ }\href@noop {} {\bibfield  {journal} {\bibinfo  {journal} {The Astrophysical Journal, Part 1 (ISSN 0004-637X), vol. 442, no. 1, p. 67-75}\ }\textbf {\bibinfo {volume} {442}},\ \bibinfo {pages} {67} (\bibinfo {year} {1995})}\BibitemShut {NoStop}%
\bibitem [{\citenamefont {Katz}\ \emph {et~al.}(2018)\citenamefont {Katz}, \citenamefont {Kopp}, \citenamefont {Sibiryakov},\ and\ \citenamefont {Xue}}]{katz2018femtolensing}%
  \BibitemOpen
  \bibfield  {author} {\bibinfo {author} {\bibfnamefont {A.}~\bibnamefont {Katz}}, \bibinfo {author} {\bibfnamefont {J.}~\bibnamefont {Kopp}}, \bibinfo {author} {\bibfnamefont {S.}~\bibnamefont {Sibiryakov}},\ and\ \bibinfo {author} {\bibfnamefont {W.}~\bibnamefont {Xue}},\ }\href@noop {} {\bibfield  {journal} {\bibinfo  {journal} {Journal of Cosmology and Astroparticle Physics}\ }\textbf {\bibinfo {volume} {2018}}\bibinfo  {number} { (12)},\ \bibinfo {pages} {005}}\BibitemShut {NoStop}%
\bibitem [{\citenamefont {Eichler}(2017)}]{eichler2017nanolensed}%
  \BibitemOpen
\bibfield  {number} {  }\bibfield  {author} {\bibinfo {author} {\bibfnamefont {D.}~\bibnamefont {Eichler}},\ }\href@noop {} {\bibfield  {journal} {\bibinfo  {journal} {The Astrophysical Journal}\ }\textbf {\bibinfo {volume} {850}},\ \bibinfo {pages} {159} (\bibinfo {year} {2017})}\BibitemShut {NoStop}%
\bibitem [{\citenamefont {Jow}\ \emph {et~al.}(2020)\citenamefont {Jow}, \citenamefont {Foreman}, \citenamefont {Pen},\ and\ \citenamefont {Zhu}}]{jow2020wave}%
  \BibitemOpen
  \bibfield  {author} {\bibinfo {author} {\bibfnamefont {D.~L.}\ \bibnamefont {Jow}}, \bibinfo {author} {\bibfnamefont {S.}~\bibnamefont {Foreman}}, \bibinfo {author} {\bibfnamefont {U.-L.}\ \bibnamefont {Pen}},\ and\ \bibinfo {author} {\bibfnamefont {W.}~\bibnamefont {Zhu}},\ }\href@noop {} {\bibfield  {journal} {\bibinfo  {journal} {Monthly Notices of the Royal Astronomical Society}\ }\textbf {\bibinfo {volume} {497}},\ \bibinfo {pages} {4956} (\bibinfo {year} {2020})}\BibitemShut {NoStop}%
\bibitem [{\citenamefont {Wucknitz}\ \emph {et~al.}(2021)\citenamefont {Wucknitz}, \citenamefont {Spitler},\ and\ \citenamefont {Pen}}]{wucknitz2021cosmology}%
  \BibitemOpen
  \bibfield  {author} {\bibinfo {author} {\bibfnamefont {O.}~\bibnamefont {Wucknitz}}, \bibinfo {author} {\bibfnamefont {L.}~\bibnamefont {Spitler}},\ and\ \bibinfo {author} {\bibfnamefont {U.-L.}\ \bibnamefont {Pen}},\ }\href@noop {} {\bibfield  {journal} {\bibinfo  {journal} {Astronomy \& Astrophysics}\ }\textbf {\bibinfo {volume} {645}},\ \bibinfo {pages} {A44} (\bibinfo {year} {2021})}\BibitemShut {NoStop}%
\bibitem [{\citenamefont {Kader}\ \emph {et~al.}(2022)\citenamefont {Kader}, \citenamefont {Leung}, \citenamefont {Dobbs}, \citenamefont {Masui}, \citenamefont {Michilli}, \citenamefont {Mena-Parra}, \citenamefont {Mckinven}, \citenamefont {Ng}, \citenamefont {Bandura}, \citenamefont {Bhardwaj} \emph {et~al.}}]{kader2022high}%
  \BibitemOpen
  \bibfield  {author} {\bibinfo {author} {\bibfnamefont {Z.}~\bibnamefont {Kader}}, \bibinfo {author} {\bibfnamefont {C.}~\bibnamefont {Leung}}, \bibinfo {author} {\bibfnamefont {M.}~\bibnamefont {Dobbs}}, \bibinfo {author} {\bibfnamefont {K.~W.}\ \bibnamefont {Masui}}, \bibinfo {author} {\bibfnamefont {D.}~\bibnamefont {Michilli}}, \bibinfo {author} {\bibfnamefont {J.}~\bibnamefont {Mena-Parra}}, \bibinfo {author} {\bibfnamefont {R.}~\bibnamefont {Mckinven}}, \bibinfo {author} {\bibfnamefont {C.}~\bibnamefont {Ng}}, \bibinfo {author} {\bibfnamefont {K.}~\bibnamefont {Bandura}}, \bibinfo {author} {\bibfnamefont {M.}~\bibnamefont {Bhardwaj}}, \emph {et~al.},\ }\href@noop {} {\bibfield  {journal} {\bibinfo  {journal} {Physical Review D}\ }\textbf {\bibinfo {volume} {106}},\ \bibinfo {pages} {043016} (\bibinfo {year} {2022})}\BibitemShut {NoStop}%
\bibitem [{\citenamefont {Leung}\ \emph {et~al.}(2025)\citenamefont {Leung}, \citenamefont {Jow}, \citenamefont {Saha}, \citenamefont {Dai}, \citenamefont {Oguri},\ and\ \citenamefont {Koopmans}}]{leung2025wave}%
  \BibitemOpen
  \bibfield  {author} {\bibinfo {author} {\bibfnamefont {C.}~\bibnamefont {Leung}}, \bibinfo {author} {\bibfnamefont {D.}~\bibnamefont {Jow}}, \bibinfo {author} {\bibfnamefont {P.}~\bibnamefont {Saha}}, \bibinfo {author} {\bibfnamefont {L.}~\bibnamefont {Dai}}, \bibinfo {author} {\bibfnamefont {M.}~\bibnamefont {Oguri}},\ and\ \bibinfo {author} {\bibfnamefont {L.~V.}\ \bibnamefont {Koopmans}},\ }\href@noop {} {\bibfield  {journal} {\bibinfo  {journal} {Space Science Reviews}\ }\textbf {\bibinfo {volume} {221}},\ \bibinfo {pages} {1} (\bibinfo {year} {2025})}\BibitemShut {NoStop}%
\bibitem [{\citenamefont {Sugiyama}\ \emph {et~al.}(2020)\citenamefont {Sugiyama}, \citenamefont {Kurita},\ and\ \citenamefont {Takada}}]{sugiyama2020wave}%
  \BibitemOpen
  \bibfield  {author} {\bibinfo {author} {\bibfnamefont {S.}~\bibnamefont {Sugiyama}}, \bibinfo {author} {\bibfnamefont {T.}~\bibnamefont {Kurita}},\ and\ \bibinfo {author} {\bibfnamefont {M.}~\bibnamefont {Takada}},\ }\href@noop {} {\bibfield  {journal} {\bibinfo  {journal} {Monthly Notices of the Royal Astronomical Society}\ }\textbf {\bibinfo {volume} {493}},\ \bibinfo {pages} {3632} (\bibinfo {year} {2020})}\BibitemShut {NoStop}%
\bibitem [{\citenamefont {Ettinger}\ and\ \citenamefont {H{\o}yer}(2000)}]{EH00}%
  \BibitemOpen
  \bibfield  {author} {\bibinfo {author} {\bibfnamefont {M.}~\bibnamefont {Ettinger}}\ and\ \bibinfo {author} {\bibfnamefont {P.}~\bibnamefont {H{\o}yer}},\ }\href@noop {} {\bibfield  {journal} {\bibinfo  {journal} {Advances in Applied Mathematics}\ }\textbf {\bibinfo {volume} {25}},\ \bibinfo {pages} {239} (\bibinfo {year} {2000})},\ \bibinfo {note} {preliminary version in STACS 1999},\ \Eprint {https://arxiv.org/abs/quant-ph/9807029} {quant-ph/9807029} \BibitemShut {NoStop}%
\bibitem [{\citenamefont {Coddington}\ \emph {et~al.}(2016)\citenamefont {Coddington}, \citenamefont {Newbury},\ and\ \citenamefont {Swann}}]{coddington2016dual}%
  \BibitemOpen
  \bibfield  {author} {\bibinfo {author} {\bibfnamefont {I.}~\bibnamefont {Coddington}}, \bibinfo {author} {\bibfnamefont {N.}~\bibnamefont {Newbury}},\ and\ \bibinfo {author} {\bibfnamefont {W.}~\bibnamefont {Swann}},\ }\href@noop {} {\bibfield  {journal} {\bibinfo  {journal} {Optica}\ }\textbf {\bibinfo {volume} {3}},\ \bibinfo {pages} {414} (\bibinfo {year} {2016})}\BibitemShut {NoStop}%
\bibitem [{\citenamefont {Picqu{\'e}}\ and\ \citenamefont {H{\"a}nsch}(2020)}]{picque2020photon}%
  \BibitemOpen
  \bibfield  {author} {\bibinfo {author} {\bibfnamefont {N.}~\bibnamefont {Picqu{\'e}}}\ and\ \bibinfo {author} {\bibfnamefont {T.~W.}\ \bibnamefont {H{\"a}nsch}},\ }\href@noop {} {\bibfield  {journal} {\bibinfo  {journal} {Proceedings of the National Academy of Sciences}\ }\textbf {\bibinfo {volume} {117}},\ \bibinfo {pages} {26688} (\bibinfo {year} {2020})}\BibitemShut {NoStop}%
\bibitem [{\citenamefont {Xu}\ \emph {et~al.}(2024)\citenamefont {Xu}, \citenamefont {Chen}, \citenamefont {H{\"a}nsch},\ and\ \citenamefont {Picqu{\'e}}}]{xu2024near}%
  \BibitemOpen
  \bibfield  {author} {\bibinfo {author} {\bibfnamefont {B.}~\bibnamefont {Xu}}, \bibinfo {author} {\bibfnamefont {Z.}~\bibnamefont {Chen}}, \bibinfo {author} {\bibfnamefont {T.~W.}\ \bibnamefont {H{\"a}nsch}},\ and\ \bibinfo {author} {\bibfnamefont {N.}~\bibnamefont {Picqu{\'e}}},\ }\href@noop {} {\bibfield  {journal} {\bibinfo  {journal} {Nature}\ }\textbf {\bibinfo {volume} {627}},\ \bibinfo {pages} {289} (\bibinfo {year} {2024})}\BibitemShut {NoStop}%
\bibitem [{\citenamefont {Peng}\ \emph {et~al.}(2025)\citenamefont {Peng}, \citenamefont {Mei}, \citenamefont {Gong}, \citenamefont {Zuo}, \citenamefont {Liu}, \citenamefont {Di}, \citenamefont {Liu}, \citenamefont {Gu}, \citenamefont {Li},\ and\ \citenamefont {Yan}}]{peng2025single}%
  \BibitemOpen
  \bibfield  {author} {\bibinfo {author} {\bibfnamefont {D.}~\bibnamefont {Peng}}, \bibinfo {author} {\bibfnamefont {L.}~\bibnamefont {Mei}}, \bibinfo {author} {\bibfnamefont {Z.}~\bibnamefont {Gong}}, \bibinfo {author} {\bibfnamefont {Z.}~\bibnamefont {Zuo}}, \bibinfo {author} {\bibfnamefont {Z.}~\bibnamefont {Liu}}, \bibinfo {author} {\bibfnamefont {Y.}~\bibnamefont {Di}}, \bibinfo {author} {\bibfnamefont {X.}~\bibnamefont {Liu}}, \bibinfo {author} {\bibfnamefont {C.}~\bibnamefont {Gu}}, \bibinfo {author} {\bibfnamefont {W.}~\bibnamefont {Li}},\ and\ \bibinfo {author} {\bibfnamefont {G.}~\bibnamefont {Yan}},\ }\href@noop {} {\bibfield  {journal} {\bibinfo  {journal} {Nature Communications}\ }\textbf {\bibinfo {volume} {16}},\ \bibinfo {pages} {8505} (\bibinfo {year} {2025})}\BibitemShut {NoStop}%
\bibitem [{\citenamefont {Mazelanik}\ \emph {et~al.}(2020)\citenamefont {Mazelanik}, \citenamefont {Leszczy{\'n}ski}, \citenamefont {Lipka}, \citenamefont {Parniak},\ and\ \citenamefont {Wasilewski}}]{mazelanik2020temporal}%
  \BibitemOpen
  \bibfield  {author} {\bibinfo {author} {\bibfnamefont {M.}~\bibnamefont {Mazelanik}}, \bibinfo {author} {\bibfnamefont {A.}~\bibnamefont {Leszczy{\'n}ski}}, \bibinfo {author} {\bibfnamefont {M.}~\bibnamefont {Lipka}}, \bibinfo {author} {\bibfnamefont {M.}~\bibnamefont {Parniak}},\ and\ \bibinfo {author} {\bibfnamefont {W.}~\bibnamefont {Wasilewski}},\ }\href@noop {} {\bibfield  {journal} {\bibinfo  {journal} {Optica}\ }\textbf {\bibinfo {volume} {7}},\ \bibinfo {pages} {203} (\bibinfo {year} {2020})}\BibitemShut {NoStop}%
\bibitem [{\citenamefont {Joshi}\ \emph {et~al.}(2022)\citenamefont {Joshi}, \citenamefont {Sparkes}, \citenamefont {Farsi}, \citenamefont {Gerrits}, \citenamefont {Verma}, \citenamefont {Ramelow}, \citenamefont {Nam},\ and\ \citenamefont {Gaeta}}]{joshi2022picosecond}%
  \BibitemOpen
  \bibfield  {author} {\bibinfo {author} {\bibfnamefont {C.}~\bibnamefont {Joshi}}, \bibinfo {author} {\bibfnamefont {B.~M.}\ \bibnamefont {Sparkes}}, \bibinfo {author} {\bibfnamefont {A.}~\bibnamefont {Farsi}}, \bibinfo {author} {\bibfnamefont {T.}~\bibnamefont {Gerrits}}, \bibinfo {author} {\bibfnamefont {V.}~\bibnamefont {Verma}}, \bibinfo {author} {\bibfnamefont {S.}~\bibnamefont {Ramelow}}, \bibinfo {author} {\bibfnamefont {S.~W.}\ \bibnamefont {Nam}},\ and\ \bibinfo {author} {\bibfnamefont {A.~L.}\ \bibnamefont {Gaeta}},\ }\href@noop {} {\bibfield  {journal} {\bibinfo  {journal} {Optica}\ }\textbf {\bibinfo {volume} {9}},\ \bibinfo {pages} {364} (\bibinfo {year} {2022})}\BibitemShut {NoStop}%
\bibitem [{\citenamefont {Paczynski}(1986)}]{Paczynski1986}%
  \BibitemOpen
  \bibfield  {author} {\bibinfo {author} {\bibfnamefont {B.}~\bibnamefont {Paczynski}},\ }\href {https://doi.org/10.1086/163919} {\bibfield  {journal} {\bibinfo  {journal} {The Astrophysical Journal}\ }\textbf {\bibinfo {volume} {301}},\ \bibinfo {pages} {503} (\bibinfo {year} {1986})}\BibitemShut {NoStop}%
\bibitem [{\citenamefont {{Udalski}}\ \emph {et~al.}(2015)\citenamefont {{Udalski}}, \citenamefont {{Szyma{\'n}ski}},\ and\ \citenamefont {{Szyma{\'n}ski}}}]{Udalski2015}%
  \BibitemOpen
  \bibfield  {author} {\bibinfo {author} {\bibfnamefont {A.}~\bibnamefont {{Udalski}}}, \bibinfo {author} {\bibfnamefont {M.~K.}\ \bibnamefont {{Szyma{\'n}ski}}},\ and\ \bibinfo {author} {\bibfnamefont {G.}~\bibnamefont {{Szyma{\'n}ski}}},\ }\href {https://doi.org/10.48550/arXiv.1504.05966} {\bibfield  {journal} {\bibinfo  {journal} {Acta Astronomica}\ }\textbf {\bibinfo {volume} {65}},\ \bibinfo {pages} {1} (\bibinfo {year} {2015})},\ \Eprint {https://arxiv.org/abs/1504.05966} {arXiv:1504.05966 [astro-ph.SR]} \BibitemShut {NoStop}%
\bibitem [{\citenamefont {{Nunota}}\ \emph {et~al.}(2025)\citenamefont {{Nunota}}, \citenamefont {{Sumi}}, \citenamefont {{Koshimoto}}, \citenamefont {{Rattenbury}}, \citenamefont {{Abe}}, \citenamefont {{Barry}}, \citenamefont {{Bennett}}, \citenamefont {{Bhattacharya}}, \citenamefont {{Fukui}}, \citenamefont {{Hamada}}, \citenamefont {{Hamada}}, \citenamefont {{Hamasaki}}, \citenamefont {{Hirao}}, \citenamefont {{Ishitani Silva}}, \citenamefont {{Itow}}, \citenamefont {{Matsubara}}, \citenamefont {{Miyazaki}}, \citenamefont {{Muraki}}, \citenamefont {{Nagai}}, \citenamefont {{Olmschenk}}, \citenamefont {{Ranc}}, \citenamefont {{Satoh}}, \citenamefont {{Suzuki}}, \citenamefont {{Tristram}}, \citenamefont {{Vandorou}}, \citenamefont {{Yama}},\ and\ \citenamefont {{MOA Collaboration}}}]{Nunota2025}%
  \BibitemOpen
  \bibfield  {author} {\bibinfo {author} {\bibfnamefont {K.}~\bibnamefont {{Nunota}}}, \bibinfo {author} {\bibfnamefont {T.}~\bibnamefont {{Sumi}}}, \bibinfo {author} {\bibfnamefont {N.}~\bibnamefont {{Koshimoto}}}, \bibinfo {author} {\bibfnamefont {N.~J.}\ \bibnamefont {{Rattenbury}}}, \bibinfo {author} {\bibfnamefont {F.}~\bibnamefont {{Abe}}}, \bibinfo {author} {\bibfnamefont {R.}~\bibnamefont {{Barry}}}, \bibinfo {author} {\bibfnamefont {D.~P.}\ \bibnamefont {{Bennett}}}, \bibinfo {author} {\bibfnamefont {A.}~\bibnamefont {{Bhattacharya}}}, \bibinfo {author} {\bibfnamefont {A.}~\bibnamefont {{Fukui}}}, \bibinfo {author} {\bibfnamefont {R.}~\bibnamefont {{Hamada}}}, \bibinfo {author} {\bibfnamefont {S.}~\bibnamefont {{Hamada}}}, \bibinfo {author} {\bibfnamefont {N.}~\bibnamefont {{Hamasaki}}}, \bibinfo {author} {\bibfnamefont {Y.}~\bibnamefont {{Hirao}}}, \bibinfo {author} {\bibfnamefont {S.}~\bibnamefont {{Ishitani Silva}}}, \bibinfo {author} {\bibfnamefont {Y.}~\bibnamefont {{Itow}}}, \bibinfo {author}
  {\bibfnamefont {Y.}~\bibnamefont {{Matsubara}}}, \bibinfo {author} {\bibfnamefont {S.}~\bibnamefont {{Miyazaki}}}, \bibinfo {author} {\bibfnamefont {Y.}~\bibnamefont {{Muraki}}}, \bibinfo {author} {\bibfnamefont {T.}~\bibnamefont {{Nagai}}}, \bibinfo {author} {\bibfnamefont {G.}~\bibnamefont {{Olmschenk}}}, \bibinfo {author} {\bibfnamefont {C.}~\bibnamefont {{Ranc}}}, \bibinfo {author} {\bibfnamefont {Y.~K.}\ \bibnamefont {{Satoh}}}, \bibinfo {author} {\bibfnamefont {D.}~\bibnamefont {{Suzuki}}}, \bibinfo {author} {\bibfnamefont {P.~J.}\ \bibnamefont {{Tristram}}}, \bibinfo {author} {\bibfnamefont {A.}~\bibnamefont {{Vandorou}}}, \bibinfo {author} {\bibfnamefont {H.}~\bibnamefont {{Yama}}},\ and\ \bibinfo {author} {\bibnamefont {{MOA Collaboration}}},\ }\href {https://doi.org/10.3847/1538-4357/ada352} {\bibfield  {journal} {\bibinfo  {journal} {The Astrophysical Journal}\ }\textbf {\bibinfo {volume} {979}},\ \bibinfo {eid} {123} (\bibinfo {year} {2025})},\ \Eprint {https://arxiv.org/abs/2410.23553}
  {arXiv:2410.23553 [astro-ph.GA]} \BibitemShut {NoStop}%
\bibitem [{\citenamefont {{Park}}\ \emph {et~al.}(2018)\citenamefont {{Park}}, \citenamefont {{Gould}}, \citenamefont {{Lee}},\ and\ \citenamefont {{Kim}}}]{Park2018}%
  \BibitemOpen
  \bibfield  {author} {\bibinfo {author} {\bibfnamefont {B.-G.}\ \bibnamefont {{Park}}}, \bibinfo {author} {\bibfnamefont {A.~P.}\ \bibnamefont {{Gould}}}, \bibinfo {author} {\bibfnamefont {C.-U.}\ \bibnamefont {{Lee}}},\ and\ \bibinfo {author} {\bibfnamefont {S.-L.}\ \bibnamefont {{Kim}}},\ }in\ \href {https://doi.org/10.1007/978-3-319-55333-7_124} {\emph {\bibinfo {booktitle} {Handbook of Exoplanets}}},\ \bibinfo {editor} {edited by\ \bibinfo {editor} {\bibfnamefont {H.~J.}\ \bibnamefont {{Deeg}}}\ and\ \bibinfo {editor} {\bibfnamefont {J.~A.}\ \bibnamefont {{Belmonte}}}}\ (\bibinfo {year} {2018})\ p.\ \bibinfo {pages} {124}\BibitemShut {NoStop}%
\bibitem [{\citenamefont {{Lam}}\ \emph {et~al.}(2022)\citenamefont {{Lam}}, \citenamefont {{Lu}}, \citenamefont {{Udalski}}, \citenamefont {{Bond}}, \citenamefont {{Bennett}}, \citenamefont {{Skowron}}, \citenamefont {{Mr{\'o}z}}, \citenamefont {{Poleski}}, \citenamefont {{Sumi}}, \citenamefont {{Szyma{\'n}ski}}, \citenamefont {{Koz{\l}owski}}, \citenamefont {{Pietrukowicz}}, \citenamefont {{Soszy{\'n}ski}}, \citenamefont {{Ulaczyk}}, \citenamefont {{Wyrzykowski}}, \citenamefont {{Miyazaki}}, \citenamefont {{Suzuki}}, \citenamefont {{Koshimoto}}, \citenamefont {{Rattenbury}}, \citenamefont {{Hosek}}, \citenamefont {{Abe}}, \citenamefont {{Barry}}, \citenamefont {{Bhattacharya}}, \citenamefont {{Fukui}}, \citenamefont {{Fujii}}, \citenamefont {{Hirao}}, \citenamefont {{Itow}}, \citenamefont {{Kirikawa}}, \citenamefont {{Kondo}}, \citenamefont {{Matsubara}}, \citenamefont {{Matsumoto}}, \citenamefont {{Muraki}}, \citenamefont {{Olmschenk}}, \citenamefont {{Ranc}}, \citenamefont {{Okamura}}, \citenamefont
  {{Satoh}}, \citenamefont {{Silva}}, \citenamefont {{Toda}}, \citenamefont {{Tristram}}, \citenamefont {{Vandorou}}, \citenamefont {{Yama}}, \citenamefont {{Abrams}}, \citenamefont {{Agarwal}}, \citenamefont {{Rose}},\ and\ \citenamefont {{Terry}}}]{Lam2022}%
  \BibitemOpen
  \bibfield  {author} {\bibinfo {author} {\bibfnamefont {C.~Y.}\ \bibnamefont {{Lam}}}, \bibinfo {author} {\bibfnamefont {J.~R.}\ \bibnamefont {{Lu}}}, \bibinfo {author} {\bibfnamefont {A.}~\bibnamefont {{Udalski}}}, \bibinfo {author} {\bibfnamefont {I.}~\bibnamefont {{Bond}}}, \bibinfo {author} {\bibfnamefont {D.~P.}\ \bibnamefont {{Bennett}}}, \bibinfo {author} {\bibfnamefont {J.}~\bibnamefont {{Skowron}}}, \bibinfo {author} {\bibfnamefont {P.}~\bibnamefont {{Mr{\'o}z}}}, \bibinfo {author} {\bibfnamefont {R.}~\bibnamefont {{Poleski}}}, \bibinfo {author} {\bibfnamefont {T.}~\bibnamefont {{Sumi}}}, \bibinfo {author} {\bibfnamefont {M.~K.}\ \bibnamefont {{Szyma{\'n}ski}}}, \bibinfo {author} {\bibfnamefont {S.}~\bibnamefont {{Koz{\l}owski}}}, \bibinfo {author} {\bibfnamefont {P.}~\bibnamefont {{Pietrukowicz}}}, \bibinfo {author} {\bibfnamefont {I.}~\bibnamefont {{Soszy{\'n}ski}}}, \bibinfo {author} {\bibfnamefont {K.}~\bibnamefont {{Ulaczyk}}}, \bibinfo {author} {\bibfnamefont {{\L}.}~\bibnamefont
  {{Wyrzykowski}}}, \bibinfo {author} {\bibfnamefont {S.}~\bibnamefont {{Miyazaki}}}, \bibinfo {author} {\bibfnamefont {D.}~\bibnamefont {{Suzuki}}}, \bibinfo {author} {\bibfnamefont {N.}~\bibnamefont {{Koshimoto}}}, \bibinfo {author} {\bibfnamefont {N.~J.}\ \bibnamefont {{Rattenbury}}}, \bibinfo {author} {\bibfnamefont {M.~W.}\ \bibnamefont {{Hosek}}}, \bibinfo {author} {\bibfnamefont {F.}~\bibnamefont {{Abe}}}, \bibinfo {author} {\bibfnamefont {R.}~\bibnamefont {{Barry}}}, \bibinfo {author} {\bibfnamefont {A.}~\bibnamefont {{Bhattacharya}}}, \bibinfo {author} {\bibfnamefont {A.}~\bibnamefont {{Fukui}}}, \bibinfo {author} {\bibfnamefont {H.}~\bibnamefont {{Fujii}}}, \bibinfo {author} {\bibfnamefont {Y.}~\bibnamefont {{Hirao}}}, \bibinfo {author} {\bibfnamefont {Y.}~\bibnamefont {{Itow}}}, \bibinfo {author} {\bibfnamefont {R.}~\bibnamefont {{Kirikawa}}}, \bibinfo {author} {\bibfnamefont {I.}~\bibnamefont {{Kondo}}}, \bibinfo {author} {\bibfnamefont {Y.}~\bibnamefont {{Matsubara}}}, \bibinfo {author}
  {\bibfnamefont {S.}~\bibnamefont {{Matsumoto}}}, \bibinfo {author} {\bibfnamefont {Y.}~\bibnamefont {{Muraki}}}, \bibinfo {author} {\bibfnamefont {G.}~\bibnamefont {{Olmschenk}}}, \bibinfo {author} {\bibfnamefont {C.}~\bibnamefont {{Ranc}}}, \bibinfo {author} {\bibfnamefont {A.}~\bibnamefont {{Okamura}}}, \bibinfo {author} {\bibfnamefont {Y.}~\bibnamefont {{Satoh}}}, \bibinfo {author} {\bibfnamefont {S.~I.}\ \bibnamefont {{Silva}}}, \bibinfo {author} {\bibfnamefont {T.}~\bibnamefont {{Toda}}}, \bibinfo {author} {\bibfnamefont {P.~J.}\ \bibnamefont {{Tristram}}}, \bibinfo {author} {\bibfnamefont {A.}~\bibnamefont {{Vandorou}}}, \bibinfo {author} {\bibfnamefont {H.}~\bibnamefont {{Yama}}}, \bibinfo {author} {\bibfnamefont {N.~S.}\ \bibnamefont {{Abrams}}}, \bibinfo {author} {\bibfnamefont {S.}~\bibnamefont {{Agarwal}}}, \bibinfo {author} {\bibfnamefont {S.}~\bibnamefont {{Rose}}},\ and\ \bibinfo {author} {\bibfnamefont {S.~K.}\ \bibnamefont {{Terry}}},\ }\href {https://doi.org/10.3847/2041-8213/ac7442}
  {\bibfield  {journal} {\bibinfo  {journal} {The Astrophysical Journal Letters}\ }\textbf {\bibinfo {volume} {933}},\ \bibinfo {eid} {L23} (\bibinfo {year} {2022})},\ \Eprint {https://arxiv.org/abs/2202.01903} {arXiv:2202.01903 [astro-ph.GA]} \BibitemShut {NoStop}%
\bibitem [{\citenamefont {{Sumi}}\ \emph {et~al.}(2011)\citenamefont {{Sumi}}, \citenamefont {{Kamiya}}, \citenamefont {{Bennett}}, \citenamefont {{Bond}}, \citenamefont {{Abe}}, \citenamefont {{Botzler}}, \citenamefont {{Fukui}}, \citenamefont {{Furusawa}}, \citenamefont {{Hearnshaw}}, \citenamefont {{Itow}}, \citenamefont {{Kilmartin}}, \citenamefont {{Korpela}}, \citenamefont {{Lin}}, \citenamefont {{Ling}}, \citenamefont {{Masuda}}, \citenamefont {{Matsubara}}, \citenamefont {{Miyake}}, \citenamefont {{Motomura}}, \citenamefont {{Muraki}}, \citenamefont {{Nagaya}}, \citenamefont {{Nakamura}}, \citenamefont {{Ohnishi}}, \citenamefont {{Okumura}}, \citenamefont {{Perrott}}, \citenamefont {{Rattenbury}}, \citenamefont {{Saito}}, \citenamefont {{Sako}}, \citenamefont {{Sullivan}}, \citenamefont {{Sweatman}}, \citenamefont {{Tristram}}, \citenamefont {{Udalski}}, \citenamefont {{Szyma{\'n}ski}}, \citenamefont {{Kubiak}}, \citenamefont {{Pietrzy{\'n}ski}}, \citenamefont {{Poleski}}, \citenamefont
  {{Soszy{\'n}ski}}, \citenamefont {{Wyrzykowski}}, \citenamefont {{Ulaczyk}},\ and\ \citenamefont {{Microlensing Observations in Astrophysics (MOA) Collaboration}}}]{Sumi2011}%
  \BibitemOpen
  \bibfield  {author} {\bibinfo {author} {\bibfnamefont {T.}~\bibnamefont {{Sumi}}}, \bibinfo {author} {\bibfnamefont {K.}~\bibnamefont {{Kamiya}}}, \bibinfo {author} {\bibfnamefont {D.~P.}\ \bibnamefont {{Bennett}}}, \bibinfo {author} {\bibfnamefont {I.~A.}\ \bibnamefont {{Bond}}}, \bibinfo {author} {\bibfnamefont {F.}~\bibnamefont {{Abe}}}, \bibinfo {author} {\bibfnamefont {C.~S.}\ \bibnamefont {{Botzler}}}, \bibinfo {author} {\bibfnamefont {A.}~\bibnamefont {{Fukui}}}, \bibinfo {author} {\bibfnamefont {K.}~\bibnamefont {{Furusawa}}}, \bibinfo {author} {\bibfnamefont {J.~B.}\ \bibnamefont {{Hearnshaw}}}, \bibinfo {author} {\bibfnamefont {Y.}~\bibnamefont {{Itow}}}, \bibinfo {author} {\bibfnamefont {P.~M.}\ \bibnamefont {{Kilmartin}}}, \bibinfo {author} {\bibfnamefont {A.}~\bibnamefont {{Korpela}}}, \bibinfo {author} {\bibfnamefont {W.}~\bibnamefont {{Lin}}}, \bibinfo {author} {\bibfnamefont {C.~H.}\ \bibnamefont {{Ling}}}, \bibinfo {author} {\bibfnamefont {K.}~\bibnamefont {{Masuda}}}, \bibinfo {author}
  {\bibfnamefont {Y.}~\bibnamefont {{Matsubara}}}, \bibinfo {author} {\bibfnamefont {N.}~\bibnamefont {{Miyake}}}, \bibinfo {author} {\bibfnamefont {M.}~\bibnamefont {{Motomura}}}, \bibinfo {author} {\bibfnamefont {Y.}~\bibnamefont {{Muraki}}}, \bibinfo {author} {\bibfnamefont {M.}~\bibnamefont {{Nagaya}}}, \bibinfo {author} {\bibfnamefont {S.}~\bibnamefont {{Nakamura}}}, \bibinfo {author} {\bibfnamefont {K.}~\bibnamefont {{Ohnishi}}}, \bibinfo {author} {\bibfnamefont {T.}~\bibnamefont {{Okumura}}}, \bibinfo {author} {\bibfnamefont {Y.~C.}\ \bibnamefont {{Perrott}}}, \bibinfo {author} {\bibfnamefont {N.}~\bibnamefont {{Rattenbury}}}, \bibinfo {author} {\bibfnamefont {T.}~\bibnamefont {{Saito}}}, \bibinfo {author} {\bibfnamefont {T.}~\bibnamefont {{Sako}}}, \bibinfo {author} {\bibfnamefont {D.~J.}\ \bibnamefont {{Sullivan}}}, \bibinfo {author} {\bibfnamefont {W.~L.}\ \bibnamefont {{Sweatman}}}, \bibinfo {author} {\bibfnamefont {P.~J.}\ \bibnamefont {{Tristram}}}, \bibinfo {author} {\bibfnamefont
  {A.}~\bibnamefont {{Udalski}}}, \bibinfo {author} {\bibfnamefont {M.~K.}\ \bibnamefont {{Szyma{\'n}ski}}}, \bibinfo {author} {\bibfnamefont {M.}~\bibnamefont {{Kubiak}}}, \bibinfo {author} {\bibfnamefont {G.}~\bibnamefont {{Pietrzy{\'n}ski}}}, \bibinfo {author} {\bibfnamefont {R.}~\bibnamefont {{Poleski}}}, \bibinfo {author} {\bibfnamefont {I.}~\bibnamefont {{Soszy{\'n}ski}}}, \bibinfo {author} {\bibfnamefont {{\L}.}~\bibnamefont {{Wyrzykowski}}}, \bibinfo {author} {\bibfnamefont {K.}~\bibnamefont {{Ulaczyk}}},\ and\ \bibinfo {author} {\bibnamefont {{Microlensing Observations in Astrophysics (MOA) Collaboration}}},\ }\href {https://doi.org/10.1038/nature10092} {\bibfield  {journal} {\bibinfo  {journal} {\nat}\ }\textbf {\bibinfo {volume} {473}},\ \bibinfo {pages} {349} (\bibinfo {year} {2011})},\ \Eprint {https://arxiv.org/abs/1105.3544} {arXiv:1105.3544 [astro-ph.EP]} \BibitemShut {NoStop}%
\bibitem [{\citenamefont {Mróz}\ \emph {et~al.}(2018)\citenamefont {Mróz}, \citenamefont {Ryu}, \citenamefont {Skowron}, \citenamefont {Udalski}, \citenamefont {Gould}, \citenamefont {Szymański}, \citenamefont {Soszyński}, \citenamefont {Poleski}, \citenamefont {Pietrukowicz}, \citenamefont {Kozłowski}, \citenamefont {Pawlak}, \citenamefont {Ulaczyk}, \citenamefont {Collaboration)}, \citenamefont {Albrow}, \citenamefont {Chung}, \citenamefont {Jung}, \citenamefont {Han}, \citenamefont {Hwang}, \citenamefont {Shin}, \citenamefont {Yee}, \citenamefont {Zhu}, \citenamefont {Cha}, \citenamefont {Kim}, \citenamefont {Kim}, \citenamefont {Kim}, \citenamefont {Lee}, \citenamefont {Lee}, \citenamefont {Lee}, \citenamefont {Park}, \citenamefont {Pogge},\ and\ \citenamefont {Collaboration)}}]{Mroz2018}%
  \BibitemOpen
  \bibfield  {author} {\bibinfo {author} {\bibfnamefont {P.}~\bibnamefont {Mróz}}, \bibinfo {author} {\bibfnamefont {Y.-H.}\ \bibnamefont {Ryu}}, \bibinfo {author} {\bibfnamefont {J.}~\bibnamefont {Skowron}}, \bibinfo {author} {\bibfnamefont {A.}~\bibnamefont {Udalski}}, \bibinfo {author} {\bibfnamefont {A.}~\bibnamefont {Gould}}, \bibinfo {author} {\bibfnamefont {M.~K.}\ \bibnamefont {Szymański}}, \bibinfo {author} {\bibfnamefont {I.}~\bibnamefont {Soszyński}}, \bibinfo {author} {\bibfnamefont {R.}~\bibnamefont {Poleski}}, \bibinfo {author} {\bibfnamefont {P.}~\bibnamefont {Pietrukowicz}}, \bibinfo {author} {\bibfnamefont {S.}~\bibnamefont {Kozłowski}}, \bibinfo {author} {\bibfnamefont {M.}~\bibnamefont {Pawlak}}, \bibinfo {author} {\bibfnamefont {K.}~\bibnamefont {Ulaczyk}}, \bibinfo {author} {\bibfnamefont {T.~O.}\ \bibnamefont {Collaboration)}}, \bibinfo {author} {\bibfnamefont {M.~D.}\ \bibnamefont {Albrow}}, \bibinfo {author} {\bibfnamefont {S.-J.}\ \bibnamefont {Chung}}, \bibinfo {author}
  {\bibfnamefont {Y.~K.}\ \bibnamefont {Jung}}, \bibinfo {author} {\bibfnamefont {C.}~\bibnamefont {Han}}, \bibinfo {author} {\bibfnamefont {K.-H.}\ \bibnamefont {Hwang}}, \bibinfo {author} {\bibfnamefont {I.-G.}\ \bibnamefont {Shin}}, \bibinfo {author} {\bibfnamefont {J.~C.}\ \bibnamefont {Yee}}, \bibinfo {author} {\bibfnamefont {W.}~\bibnamefont {Zhu}}, \bibinfo {author} {\bibfnamefont {S.-M.}\ \bibnamefont {Cha}}, \bibinfo {author} {\bibfnamefont {D.-J.}\ \bibnamefont {Kim}}, \bibinfo {author} {\bibfnamefont {H.-W.}\ \bibnamefont {Kim}}, \bibinfo {author} {\bibfnamefont {S.-L.}\ \bibnamefont {Kim}}, \bibinfo {author} {\bibfnamefont {C.-U.}\ \bibnamefont {Lee}}, \bibinfo {author} {\bibfnamefont {D.-J.}\ \bibnamefont {Lee}}, \bibinfo {author} {\bibfnamefont {Y.}~\bibnamefont {Lee}}, \bibinfo {author} {\bibfnamefont {B.-G.}\ \bibnamefont {Park}}, \bibinfo {author} {\bibfnamefont {R.~W.}\ \bibnamefont {Pogge}},\ and\ \bibinfo {author} {\bibfnamefont {T.~K.}\ \bibnamefont {Collaboration)}},\ }\href
  {https://doi.org/10.3847/1538-3881/aaaae9} {\bibfield  {journal} {\bibinfo  {journal} {The Astronomical Journal}\ }\textbf {\bibinfo {volume} {155}},\ \bibinfo {pages} {121} (\bibinfo {year} {2018})}\BibitemShut {NoStop}%
\bibitem [{\citenamefont {{Mr{\'o}z}}\ \emph {et~al.}(2019)\citenamefont {{Mr{\'o}z}}, \citenamefont {{Udalski}}, \citenamefont {{Bennett}}, \citenamefont {{Ryu}}, \citenamefont {{Sumi}}, \citenamefont {{Shvartzvald}}, \citenamefont {{Skowron}}, \citenamefont {{Poleski}}, \citenamefont {{Pietrukowicz}}, \citenamefont {{Koz{\l}owski}}, \citenamefont {{Szyma{\'n}ski}}, \citenamefont {{Wyrzykowski}}, \citenamefont {{Soszy{\'n}ski}}, \citenamefont {{Ulaczyk}}, \citenamefont {{Rybicki}}, \citenamefont {{Iwanek}}, \citenamefont {{Albrow}}, \citenamefont {{Chung}}, \citenamefont {{Gould}}, \citenamefont {{Han}}, \citenamefont {{Hwang}}, \citenamefont {{Jung}}, \citenamefont {{Shin}}, \citenamefont {{Yee}}, \citenamefont {{Zang}}, \citenamefont {{Cha}}, \citenamefont {{Kim}}, \citenamefont {{Kim}}, \citenamefont {{Kim}}, \citenamefont {{Lee}}, \citenamefont {{Lee}}, \citenamefont {{Lee}}, \citenamefont {{Park}}, \citenamefont {{Pogge}}, \citenamefont {{Abe}}, \citenamefont {{Barry}}, \citenamefont {{Bhattacharya}},
  \citenamefont {{Bond}}, \citenamefont {{Donachie}}, \citenamefont {{Fukui}}, \citenamefont {{Hirao}}, \citenamefont {{Itow}}, \citenamefont {{Kawasaki}}, \citenamefont {{Kondo}}, \citenamefont {{Koshimoto}}, \citenamefont {{Li}}, \citenamefont {{Matsubara}}, \citenamefont {{Muraki}}, \citenamefont {{Miyazaki}}, \citenamefont {{Nagakane}}, \citenamefont {{Ranc}}, \citenamefont {{Rattenbury}}, \citenamefont {{Suematsu}}, \citenamefont {{Sullivan}}, \citenamefont {{Suzuki}}, \citenamefont {{Tristram}}, \citenamefont {{Yonehara}}, \citenamefont {{Maoz}}, \citenamefont {{Kaspi}},\ and\ \citenamefont {{Friedmann}}}]{Mroz2019}%
  \BibitemOpen
  \bibfield  {author} {\bibinfo {author} {\bibfnamefont {P.}~\bibnamefont {{Mr{\'o}z}}}, \bibinfo {author} {\bibfnamefont {A.}~\bibnamefont {{Udalski}}}, \bibinfo {author} {\bibfnamefont {D.~P.}\ \bibnamefont {{Bennett}}}, \bibinfo {author} {\bibfnamefont {Y.-H.}\ \bibnamefont {{Ryu}}}, \bibinfo {author} {\bibfnamefont {T.}~\bibnamefont {{Sumi}}}, \bibinfo {author} {\bibfnamefont {Y.}~\bibnamefont {{Shvartzvald}}}, \bibinfo {author} {\bibfnamefont {J.}~\bibnamefont {{Skowron}}}, \bibinfo {author} {\bibfnamefont {R.}~\bibnamefont {{Poleski}}}, \bibinfo {author} {\bibfnamefont {P.}~\bibnamefont {{Pietrukowicz}}}, \bibinfo {author} {\bibfnamefont {S.}~\bibnamefont {{Koz{\l}owski}}}, \bibinfo {author} {\bibfnamefont {M.~K.}\ \bibnamefont {{Szyma{\'n}ski}}}, \bibinfo {author} {\bibfnamefont {{\L}.}~\bibnamefont {{Wyrzykowski}}}, \bibinfo {author} {\bibfnamefont {I.}~\bibnamefont {{Soszy{\'n}ski}}}, \bibinfo {author} {\bibfnamefont {K.}~\bibnamefont {{Ulaczyk}}}, \bibinfo {author} {\bibfnamefont {K.}~\bibnamefont
  {{Rybicki}}}, \bibinfo {author} {\bibfnamefont {P.}~\bibnamefont {{Iwanek}}}, \bibinfo {author} {\bibfnamefont {M.~D.}\ \bibnamefont {{Albrow}}}, \bibinfo {author} {\bibfnamefont {S.-J.}\ \bibnamefont {{Chung}}}, \bibinfo {author} {\bibfnamefont {A.}~\bibnamefont {{Gould}}}, \bibinfo {author} {\bibfnamefont {C.}~\bibnamefont {{Han}}}, \bibinfo {author} {\bibfnamefont {K.-H.}\ \bibnamefont {{Hwang}}}, \bibinfo {author} {\bibfnamefont {Y.~K.}\ \bibnamefont {{Jung}}}, \bibinfo {author} {\bibfnamefont {I.-G.}\ \bibnamefont {{Shin}}}, \bibinfo {author} {\bibfnamefont {J.~C.}\ \bibnamefont {{Yee}}}, \bibinfo {author} {\bibfnamefont {W.}~\bibnamefont {{Zang}}}, \bibinfo {author} {\bibfnamefont {S.-M.}\ \bibnamefont {{Cha}}}, \bibinfo {author} {\bibfnamefont {D.-J.}\ \bibnamefont {{Kim}}}, \bibinfo {author} {\bibfnamefont {H.-W.}\ \bibnamefont {{Kim}}}, \bibinfo {author} {\bibfnamefont {S.-L.}\ \bibnamefont {{Kim}}}, \bibinfo {author} {\bibfnamefont {C.-U.}\ \bibnamefont {{Lee}}}, \bibinfo {author} {\bibfnamefont
  {D.-J.}\ \bibnamefont {{Lee}}}, \bibinfo {author} {\bibfnamefont {Y.}~\bibnamefont {{Lee}}}, \bibinfo {author} {\bibfnamefont {B.-G.}\ \bibnamefont {{Park}}}, \bibinfo {author} {\bibfnamefont {R.~W.}\ \bibnamefont {{Pogge}}}, \bibinfo {author} {\bibfnamefont {F.}~\bibnamefont {{Abe}}}, \bibinfo {author} {\bibfnamefont {R.}~\bibnamefont {{Barry}}}, \bibinfo {author} {\bibfnamefont {A.}~\bibnamefont {{Bhattacharya}}}, \bibinfo {author} {\bibfnamefont {I.~A.}\ \bibnamefont {{Bond}}}, \bibinfo {author} {\bibfnamefont {M.}~\bibnamefont {{Donachie}}}, \bibinfo {author} {\bibfnamefont {A.}~\bibnamefont {{Fukui}}}, \bibinfo {author} {\bibfnamefont {Y.}~\bibnamefont {{Hirao}}}, \bibinfo {author} {\bibfnamefont {Y.}~\bibnamefont {{Itow}}}, \bibinfo {author} {\bibfnamefont {K.}~\bibnamefont {{Kawasaki}}}, \bibinfo {author} {\bibfnamefont {I.}~\bibnamefont {{Kondo}}}, \bibinfo {author} {\bibfnamefont {N.}~\bibnamefont {{Koshimoto}}}, \bibinfo {author} {\bibfnamefont {M.~C.~A.}\ \bibnamefont {{Li}}}, \bibinfo {author}
  {\bibfnamefont {Y.}~\bibnamefont {{Matsubara}}}, \bibinfo {author} {\bibfnamefont {Y.}~\bibnamefont {{Muraki}}}, \bibinfo {author} {\bibfnamefont {S.}~\bibnamefont {{Miyazaki}}}, \bibinfo {author} {\bibfnamefont {M.}~\bibnamefont {{Nagakane}}}, \bibinfo {author} {\bibfnamefont {C.}~\bibnamefont {{Ranc}}}, \bibinfo {author} {\bibfnamefont {N.~J.}\ \bibnamefont {{Rattenbury}}}, \bibinfo {author} {\bibfnamefont {H.}~\bibnamefont {{Suematsu}}}, \bibinfo {author} {\bibfnamefont {D.~J.}\ \bibnamefont {{Sullivan}}}, \bibinfo {author} {\bibfnamefont {D.}~\bibnamefont {{Suzuki}}}, \bibinfo {author} {\bibfnamefont {P.~J.}\ \bibnamefont {{Tristram}}}, \bibinfo {author} {\bibfnamefont {A.}~\bibnamefont {{Yonehara}}}, \bibinfo {author} {\bibfnamefont {D.}~\bibnamefont {{Maoz}}}, \bibinfo {author} {\bibfnamefont {S.}~\bibnamefont {{Kaspi}}},\ and\ \bibinfo {author} {\bibfnamefont {M.}~\bibnamefont {{Friedmann}}},\ }\href {https://doi.org/10.1051/0004-6361/201834557} {\bibfield  {journal} {\bibinfo  {journal} {Astronomy
  \& Astrophysics}\ }\textbf {\bibinfo {volume} {622}},\ \bibinfo {eid} {A201} (\bibinfo {year} {2019})},\ \Eprint {https://arxiv.org/abs/1811.00441} {arXiv:1811.00441 [astro-ph.EP]} \BibitemShut {NoStop}%
\bibitem [{\citenamefont {Mr{\'{o}}z}\ \emph {et~al.}(2020)\citenamefont {Mr{\'{o}}z}, \citenamefont {Poleski}, \citenamefont {Han}, \citenamefont {Udalski}, \citenamefont {Gould}, \citenamefont {Szyma{\'{n}}ski}, \citenamefont {Soszy{\'{n}}ski}, \citenamefont {Pietrukowicz}, \citenamefont {Koz{\l}owski}, \citenamefont {Skowron}, \citenamefont {Ulaczyk}, \citenamefont {Gromadzki}, \citenamefont {Rybicki}, \citenamefont {Iwanek}, \citenamefont {Wrona}, \citenamefont {Albrow}, \citenamefont {Chung}, \citenamefont {Hwang}, \citenamefont {Ryu}, \citenamefont {Jung}, \citenamefont {Shin}, \citenamefont {Shvartzvald}, \citenamefont {Yee}, \citenamefont {Zang}, \citenamefont {Cha}, \citenamefont {Kim}, \citenamefont {Kim}, \citenamefont {Kim}, \citenamefont {Lee}, \citenamefont {Lee}, \citenamefont {Lee}, \citenamefont {Park}, \citenamefont {Pogge},\ and\ \citenamefont {and}}]{Mroz2020a}%
  \BibitemOpen
  \bibfield  {author} {\bibinfo {author} {\bibfnamefont {P.}~\bibnamefont {Mr{\'{o}}z}}, \bibinfo {author} {\bibfnamefont {R.}~\bibnamefont {Poleski}}, \bibinfo {author} {\bibfnamefont {C.}~\bibnamefont {Han}}, \bibinfo {author} {\bibfnamefont {A.}~\bibnamefont {Udalski}}, \bibinfo {author} {\bibfnamefont {A.}~\bibnamefont {Gould}}, \bibinfo {author} {\bibfnamefont {M.~K.}\ \bibnamefont {Szyma{\'{n}}ski}}, \bibinfo {author} {\bibfnamefont {I.}~\bibnamefont {Soszy{\'{n}}ski}}, \bibinfo {author} {\bibfnamefont {P.}~\bibnamefont {Pietrukowicz}}, \bibinfo {author} {\bibfnamefont {S.}~\bibnamefont {Koz{\l}owski}}, \bibinfo {author} {\bibfnamefont {J.}~\bibnamefont {Skowron}}, \bibinfo {author} {\bibfnamefont {K.}~\bibnamefont {Ulaczyk}}, \bibinfo {author} {\bibfnamefont {M.}~\bibnamefont {Gromadzki}}, \bibinfo {author} {\bibfnamefont {K.}~\bibnamefont {Rybicki}}, \bibinfo {author} {\bibfnamefont {P.}~\bibnamefont {Iwanek}}, \bibinfo {author} {\bibfnamefont {M.}~\bibnamefont {Wrona}}, \bibinfo {author}
  {\bibfnamefont {M.~D.}\ \bibnamefont {Albrow}}, \bibinfo {author} {\bibfnamefont {S.-J.}\ \bibnamefont {Chung}}, \bibinfo {author} {\bibfnamefont {K.-H.}\ \bibnamefont {Hwang}}, \bibinfo {author} {\bibfnamefont {Y.-H.}\ \bibnamefont {Ryu}}, \bibinfo {author} {\bibfnamefont {Y.~K.}\ \bibnamefont {Jung}}, \bibinfo {author} {\bibfnamefont {I.-G.}\ \bibnamefont {Shin}}, \bibinfo {author} {\bibfnamefont {Y.}~\bibnamefont {Shvartzvald}}, \bibinfo {author} {\bibfnamefont {J.~C.}\ \bibnamefont {Yee}}, \bibinfo {author} {\bibfnamefont {W.}~\bibnamefont {Zang}}, \bibinfo {author} {\bibfnamefont {S.-M.}\ \bibnamefont {Cha}}, \bibinfo {author} {\bibfnamefont {D.-J.}\ \bibnamefont {Kim}}, \bibinfo {author} {\bibfnamefont {H.-W.}\ \bibnamefont {Kim}}, \bibinfo {author} {\bibfnamefont {S.-L.}\ \bibnamefont {Kim}}, \bibinfo {author} {\bibfnamefont {C.-U.}\ \bibnamefont {Lee}}, \bibinfo {author} {\bibfnamefont {D.-J.}\ \bibnamefont {Lee}}, \bibinfo {author} {\bibfnamefont {Y.}~\bibnamefont {Lee}}, \bibinfo {author}
  {\bibfnamefont {B.-G.}\ \bibnamefont {Park}}, \bibinfo {author} {\bibfnamefont {R.~W.}\ \bibnamefont {Pogge}},\ and\ \bibinfo {author} {\bibnamefont {and}},\ }\href {https://doi.org/10.3847/1538-3881/ab8aeb} {\bibfield  {journal} {\bibinfo  {journal} {The Astronomical Journal}\ }\textbf {\bibinfo {volume} {159}},\ \bibinfo {pages} {262} (\bibinfo {year} {2020})}\BibitemShut {NoStop}%
\bibitem [{\citenamefont {Mróz}\ \emph {et~al.}(2020)\citenamefont {Mróz}, \citenamefont {Poleski}, \citenamefont {Gould}, \citenamefont {Udalski}, \citenamefont {Sumi}, , , \citenamefont {Szymański}, \citenamefont {Soszyński}, \citenamefont {Pietrukowicz}, \citenamefont {Kozłowski}, \citenamefont {Skowron}, \citenamefont {Ulaczyk}, \citenamefont {Collaboration)}, \citenamefont {Albrow}, \citenamefont {Chung}, \citenamefont {Han}, \citenamefont {Hwang}, \citenamefont {Jung}, \citenamefont {Kim}, \citenamefont {Ryu}, \citenamefont {Shin}, \citenamefont {Shvartzvald}, \citenamefont {Yee}, \citenamefont {Zang}, \citenamefont {Cha}, \citenamefont {Kim}, \citenamefont {Kim}, \citenamefont {Lee}, \citenamefont {Lee}, \citenamefont {Lee}, \citenamefont {Park}, \citenamefont {Pogge},\ and\ \citenamefont {Collaboration)}}]{Mroz2020b}%
  \BibitemOpen
  \bibfield  {author} {\bibinfo {author} {\bibfnamefont {P.}~\bibnamefont {Mróz}}, \bibinfo {author} {\bibfnamefont {R.}~\bibnamefont {Poleski}}, \bibinfo {author} {\bibfnamefont {A.}~\bibnamefont {Gould}}, \bibinfo {author} {\bibfnamefont {A.}~\bibnamefont {Udalski}}, \bibinfo {author} {\bibfnamefont {T.}~\bibnamefont {Sumi}}, , , \bibinfo {author} {\bibfnamefont {M.~K.}\ \bibnamefont {Szymański}}, \bibinfo {author} {\bibfnamefont {I.}~\bibnamefont {Soszyński}}, \bibinfo {author} {\bibfnamefont {P.}~\bibnamefont {Pietrukowicz}}, \bibinfo {author} {\bibfnamefont {S.}~\bibnamefont {Kozłowski}}, \bibinfo {author} {\bibfnamefont {J.}~\bibnamefont {Skowron}}, \bibinfo {author} {\bibfnamefont {K.}~\bibnamefont {Ulaczyk}}, \bibinfo {author} {\bibfnamefont {O.}~\bibnamefont {Collaboration)}}, \bibinfo {author} {\bibfnamefont {M.~D.}\ \bibnamefont {Albrow}}, \bibinfo {author} {\bibfnamefont {S.-J.}\ \bibnamefont {Chung}}, \bibinfo {author} {\bibfnamefont {C.}~\bibnamefont {Han}}, \bibinfo {author} {\bibfnamefont
  {K.-H.}\ \bibnamefont {Hwang}}, \bibinfo {author} {\bibfnamefont {Y.~K.}\ \bibnamefont {Jung}}, \bibinfo {author} {\bibfnamefont {H.-W.}\ \bibnamefont {Kim}}, \bibinfo {author} {\bibfnamefont {Y.-H.}\ \bibnamefont {Ryu}}, \bibinfo {author} {\bibfnamefont {I.-G.}\ \bibnamefont {Shin}}, \bibinfo {author} {\bibfnamefont {Y.}~\bibnamefont {Shvartzvald}}, \bibinfo {author} {\bibfnamefont {J.~C.}\ \bibnamefont {Yee}}, \bibinfo {author} {\bibfnamefont {W.}~\bibnamefont {Zang}}, \bibinfo {author} {\bibfnamefont {S.-M.}\ \bibnamefont {Cha}}, \bibinfo {author} {\bibfnamefont {D.-J.}\ \bibnamefont {Kim}}, \bibinfo {author} {\bibfnamefont {S.-L.}\ \bibnamefont {Kim}}, \bibinfo {author} {\bibfnamefont {C.-U.}\ \bibnamefont {Lee}}, \bibinfo {author} {\bibfnamefont {D.-J.}\ \bibnamefont {Lee}}, \bibinfo {author} {\bibfnamefont {Y.}~\bibnamefont {Lee}}, \bibinfo {author} {\bibfnamefont {B.-G.}\ \bibnamefont {Park}}, \bibinfo {author} {\bibfnamefont {R.~W.}\ \bibnamefont {Pogge}},\ and\ \bibinfo {author} {\bibfnamefont
  {K.}~\bibnamefont {Collaboration)}},\ }\href {https://doi.org/10.3847/2041-8213/abbfad} {\bibfield  {journal} {\bibinfo  {journal} {The Astrophysical Journal Letters}\ }\textbf {\bibinfo {volume} {903}},\ \bibinfo {pages} {L11} (\bibinfo {year} {2020})}\BibitemShut {NoStop}%
\bibitem [{\citenamefont {Kim}\ \emph {et~al.}(2021)\citenamefont {Kim}, \citenamefont {Hwang}, \citenamefont {Gould}, \citenamefont {Yee}, \citenamefont {Ryu}, \citenamefont {Albrow}, \citenamefont {Chung}, \citenamefont {Han}, \citenamefont {Jung}, \citenamefont {Lee}, \citenamefont {Shin}, \citenamefont {Shvartzvald}, \citenamefont {Zang}, \citenamefont {Cha}, \citenamefont {Kim}, \citenamefont {Kim}, \citenamefont {Lee}, \citenamefont {Lee}, \citenamefont {Park},\ and\ \citenamefont {Pogge}}]{Kim2021}%
  \BibitemOpen
  \bibfield  {author} {\bibinfo {author} {\bibfnamefont {H.-W.}\ \bibnamefont {Kim}}, \bibinfo {author} {\bibfnamefont {K.-H.}\ \bibnamefont {Hwang}}, \bibinfo {author} {\bibfnamefont {A.}~\bibnamefont {Gould}}, \bibinfo {author} {\bibfnamefont {J.~C.}\ \bibnamefont {Yee}}, \bibinfo {author} {\bibfnamefont {Y.-H.}\ \bibnamefont {Ryu}}, \bibinfo {author} {\bibfnamefont {M.~D.}\ \bibnamefont {Albrow}}, \bibinfo {author} {\bibfnamefont {S.-J.}\ \bibnamefont {Chung}}, \bibinfo {author} {\bibfnamefont {C.}~\bibnamefont {Han}}, \bibinfo {author} {\bibfnamefont {Y.~K.}\ \bibnamefont {Jung}}, \bibinfo {author} {\bibfnamefont {C.-U.}\ \bibnamefont {Lee}}, \bibinfo {author} {\bibfnamefont {I.-G.}\ \bibnamefont {Shin}}, \bibinfo {author} {\bibfnamefont {Y.}~\bibnamefont {Shvartzvald}}, \bibinfo {author} {\bibfnamefont {W.}~\bibnamefont {Zang}}, \bibinfo {author} {\bibfnamefont {S.-M.}\ \bibnamefont {Cha}}, \bibinfo {author} {\bibfnamefont {D.-J.}\ \bibnamefont {Kim}}, \bibinfo {author} {\bibfnamefont {S.-L.}\
  \bibnamefont {Kim}}, \bibinfo {author} {\bibfnamefont {D.-J.}\ \bibnamefont {Lee}}, \bibinfo {author} {\bibfnamefont {Y.}~\bibnamefont {Lee}}, \bibinfo {author} {\bibfnamefont {B.-G.}\ \bibnamefont {Park}},\ and\ \bibinfo {author} {\bibfnamefont {R.~W.}\ \bibnamefont {Pogge}},\ }\href {https://doi.org/10.3847/1538-3881/abfc4a} {\bibfield  {journal} {\bibinfo  {journal} {The Astronomical Journal}\ }\textbf {\bibinfo {volume} {162}},\ \bibinfo {pages} {15} (\bibinfo {year} {2021})}\BibitemShut {NoStop}%
\bibitem [{\citenamefont {Mroz}\ and\ \citenamefont {Poleski}(2023)}]{Mroz2023}%
  \BibitemOpen
  \bibfield  {author} {\bibinfo {author} {\bibfnamefont {P.}~\bibnamefont {Mroz}}\ and\ \bibinfo {author} {\bibfnamefont {R.}~\bibnamefont {Poleski}},\ }\href@noop {} {\bibinfo {title} {Exoplanet occurrence rates from microlensing surveys}} (\bibinfo {year} {2023}),\ \Eprint {https://arxiv.org/abs/2310.07502} {arXiv:2310.07502 [astro-ph.EP]} \BibitemShut {NoStop}%
\bibitem [{\citenamefont {{Mr{\'o}z}}\ \emph {et~al.}(2024)\citenamefont {{Mr{\'o}z}}, \citenamefont {{Ban}}, \citenamefont {{Marty}},\ and\ \citenamefont {{Poleski}}}]{Mroz2024}%
  \BibitemOpen
  \bibfield  {author} {\bibinfo {author} {\bibfnamefont {P.}~\bibnamefont {{Mr{\'o}z}}}, \bibinfo {author} {\bibfnamefont {M.}~\bibnamefont {{Ban}}}, \bibinfo {author} {\bibfnamefont {P.}~\bibnamefont {{Marty}}},\ and\ \bibinfo {author} {\bibfnamefont {R.}~\bibnamefont {{Poleski}}},\ }\href {https://doi.org/10.3847/1538-3881/ad1106} {\bibfield  {journal} {\bibinfo  {journal} {The Astronomical Journal}\ }\textbf {\bibinfo {volume} {167}},\ \bibinfo {eid} {40} (\bibinfo {year} {2024})},\ \Eprint {https://arxiv.org/abs/2303.04610} {arXiv:2303.04610 [astro-ph.EP]} \BibitemShut {NoStop}%
\bibitem [{\citenamefont {Gould}\ \emph {et~al.}(2022)\citenamefont {Gould}, \citenamefont {Jung}, \citenamefont {Hwang}, \citenamefont {Dong}, \citenamefont {Albrow}, \citenamefont {Chung}, \citenamefont {Han}, \citenamefont {Ryu}, \citenamefont {Shin}, \citenamefont {Shvartzvald}, \citenamefont {Yang}, \citenamefont {Yee}, \citenamefont {Zang}, \citenamefont {Cha}, \citenamefont {Kim}, \citenamefont {Kim}, \citenamefont {Lee}, \citenamefont {Lee}, \citenamefont {Lee}, \citenamefont {Park},\ and\ \citenamefont {Pogge}}]{Gould2022}%
  \BibitemOpen
  \bibfield  {author} {\bibinfo {author} {\bibfnamefont {A.}~\bibnamefont {Gould}}, \bibinfo {author} {\bibfnamefont {Y.~K.}\ \bibnamefont {Jung}}, \bibinfo {author} {\bibfnamefont {K.-H.}\ \bibnamefont {Hwang}}, \bibinfo {author} {\bibfnamefont {S.}~\bibnamefont {Dong}}, \bibinfo {author} {\bibfnamefont {M.~D.}\ \bibnamefont {Albrow}}, \bibinfo {author} {\bibfnamefont {S.-J.}\ \bibnamefont {Chung}}, \bibinfo {author} {\bibfnamefont {C.}~\bibnamefont {Han}}, \bibinfo {author} {\bibfnamefont {Y.-H.}\ \bibnamefont {Ryu}}, \bibinfo {author} {\bibfnamefont {I.-G.}\ \bibnamefont {Shin}}, \bibinfo {author} {\bibfnamefont {Y.}~\bibnamefont {Shvartzvald}}, \bibinfo {author} {\bibfnamefont {H.}~\bibnamefont {Yang}}, \bibinfo {author} {\bibfnamefont {J.~C.}\ \bibnamefont {Yee}}, \bibinfo {author} {\bibfnamefont {W.}~\bibnamefont {Zang}}, \bibinfo {author} {\bibfnamefont {S.-M.}\ \bibnamefont {Cha}}, \bibinfo {author} {\bibfnamefont {D.-J.}\ \bibnamefont {Kim}}, \bibinfo {author} {\bibfnamefont {S.-L.}\ \bibnamefont
  {Kim}}, \bibinfo {author} {\bibfnamefont {C.-U.}\ \bibnamefont {Lee}}, \bibinfo {author} {\bibfnamefont {D.-J.}\ \bibnamefont {Lee}}, \bibinfo {author} {\bibfnamefont {Y.}~\bibnamefont {Lee}}, \bibinfo {author} {\bibfnamefont {B.-G.}\ \bibnamefont {Park}},\ and\ \bibinfo {author} {\bibfnamefont {R.~W.}\ \bibnamefont {Pogge}},\ }\href {https://doi.org/10.5303/JKAS.2022.55.5.173} {\bibfield  {journal} {\bibinfo  {journal} {Journal of The Korean Astronomical Society}\ }\textbf {\bibinfo {volume} {55}},\ \bibinfo {pages} {173} (\bibinfo {year} {2022})},\ \bibinfo {note} {arXiv:2204.03269 [astro-ph]}\BibitemShut {NoStop}%
\bibitem [{\citenamefont {{Gould}}\ \emph {et~al.}(2023)\citenamefont {{Gould}}, \citenamefont {{Ryu}}, \citenamefont {{Yee}}, \citenamefont {{Albrow}}, \citenamefont {{Chung}}, \citenamefont {{Han}}, \citenamefont {{Hwang}}, \citenamefont {{Jung}}, \citenamefont {{Shin}}, \citenamefont {{Shvartzvald}}, \citenamefont {{Yang}}, \citenamefont {{Zang}}, \citenamefont {{Cha}}, \citenamefont {{Kim}}, \citenamefont {{Kim}}, \citenamefont {{Lee}}, \citenamefont {{Lee}}, \citenamefont {{Lee}}, \citenamefont {{Park}}, \citenamefont {{Pogge}},\ and\ \citenamefont {{KMTNet Collaboration}}}]{Gould2023}%
  \BibitemOpen
  \bibfield  {author} {\bibinfo {author} {\bibfnamefont {A.}~\bibnamefont {{Gould}}}, \bibinfo {author} {\bibfnamefont {Y.-H.}\ \bibnamefont {{Ryu}}}, \bibinfo {author} {\bibfnamefont {J.~C.}\ \bibnamefont {{Yee}}}, \bibinfo {author} {\bibfnamefont {M.~D.}\ \bibnamefont {{Albrow}}}, \bibinfo {author} {\bibfnamefont {S.-J.}\ \bibnamefont {{Chung}}}, \bibinfo {author} {\bibfnamefont {C.}~\bibnamefont {{Han}}}, \bibinfo {author} {\bibfnamefont {K.-H.}\ \bibnamefont {{Hwang}}}, \bibinfo {author} {\bibfnamefont {Y.~K.}\ \bibnamefont {{Jung}}}, \bibinfo {author} {\bibfnamefont {I.-G.}\ \bibnamefont {{Shin}}}, \bibinfo {author} {\bibfnamefont {Y.}~\bibnamefont {{Shvartzvald}}}, \bibinfo {author} {\bibfnamefont {H.}~\bibnamefont {{Yang}}}, \bibinfo {author} {\bibfnamefont {W.}~\bibnamefont {{Zang}}}, \bibinfo {author} {\bibfnamefont {S.-M.}\ \bibnamefont {{Cha}}}, \bibinfo {author} {\bibfnamefont {D.-J.}\ \bibnamefont {{Kim}}}, \bibinfo {author} {\bibfnamefont {S.-L.}\ \bibnamefont {{Kim}}}, \bibinfo {author}
  {\bibfnamefont {C.-U.}\ \bibnamefont {{Lee}}}, \bibinfo {author} {\bibfnamefont {D.-J.}\ \bibnamefont {{Lee}}}, \bibinfo {author} {\bibfnamefont {Y.}~\bibnamefont {{Lee}}}, \bibinfo {author} {\bibfnamefont {B.-G.}\ \bibnamefont {{Park}}}, \bibinfo {author} {\bibfnamefont {R.~W.}\ \bibnamefont {{Pogge}}},\ and\ \bibinfo {author} {\bibnamefont {{KMTNet Collaboration}}},\ }\href {https://doi.org/10.3847/1538-3881/ace169} {\bibfield  {journal} {\bibinfo  {journal} {The Astronomical Journal}\ }\textbf {\bibinfo {volume} {166}},\ \bibinfo {eid} {100} (\bibinfo {year} {2023})},\ \Eprint {https://arxiv.org/abs/2306.04870} {arXiv:2306.04870 [astro-ph.EP]} \BibitemShut {NoStop}%
\bibitem [{\citenamefont {{Akeson}}\ \emph {et~al.}(2019)\citenamefont {{Akeson}}, \citenamefont {{Armus}}, \citenamefont {{Bachelet}}, \citenamefont {{Bailey}}, \citenamefont {{Bartusek}}, \citenamefont {{Bellini}}, \citenamefont {{Benford}}, \citenamefont {{Bennett}}, \citenamefont {{Bhattacharya}}, \citenamefont {{Bohlin}}, \citenamefont {{Boyer}}, \citenamefont {{Bozza}}, \citenamefont {{Bryden}}, \citenamefont {{Calchi Novati}}, \citenamefont {{Carpenter}}, \citenamefont {{Casertano}}, \citenamefont {{Choi}}, \citenamefont {{Content}}, \citenamefont {{Dayal}}, \citenamefont {{Dressler}}, \citenamefont {{Dor{\'e}}}, \citenamefont {{Fall}}, \citenamefont {{Fan}}, \citenamefont {{Fang}}, \citenamefont {{Filippenko}}, \citenamefont {{Finkelstein}}, \citenamefont {{Foley}}, \citenamefont {{Furlanetto}}, \citenamefont {{Kalirai}}, \citenamefont {{Gaudi}}, \citenamefont {{Gilbert}}, \citenamefont {{Girard}}, \citenamefont {{Grady}}, \citenamefont {{Greene}}, \citenamefont {{Guhathakurta}}, \citenamefont
  {{Heinrich}}, \citenamefont {{Hemmati}}, \citenamefont {{Hendel}}, \citenamefont {{Henderson}}, \citenamefont {{Henning}}, \citenamefont {{Hirata}}, \citenamefont {{Ho}}, \citenamefont {{Huff}}, \citenamefont {{Hutter}}, \citenamefont {{Jansen}}, \citenamefont {{Jha}}, \citenamefont {{Johnson}}, \citenamefont {{Jones}}, \citenamefont {{Kasdin}}, \citenamefont {{Kelly}}, \citenamefont {{Kirshner}}, \citenamefont {{Koekemoer}}, \citenamefont {{Kruk}}, \citenamefont {{Lewis}}, \citenamefont {{Macintosh}}, \citenamefont {{Madau}}, \citenamefont {{Malhotra}}, \citenamefont {{Mandel}}, \citenamefont {{Massara}}, \citenamefont {{Masters}}, \citenamefont {{McEnery}}, \citenamefont {{McQuinn}}, \citenamefont {{Melchior}}, \citenamefont {{Melton}}, \citenamefont {{Mennesson}}, \citenamefont {{Peeples}}, \citenamefont {{Penny}}, \citenamefont {{Perlmutter}}, \citenamefont {{Pisani}}, \citenamefont {{Plazas}}, \citenamefont {{Poleski}}, \citenamefont {{Postman}}, \citenamefont {{Ranc}}, \citenamefont {{Rauscher}},
  \citenamefont {{Rest}}, \citenamefont {{Roberge}}, \citenamefont {{Robertson}}, \citenamefont {{Rodney}}, \citenamefont {{Rhoads}}, \citenamefont {{Rhodes}}, \citenamefont {{Ryan}}, \citenamefont {{Sahu}}, \citenamefont {{Sand}}, \citenamefont {{Scolnic}}, \citenamefont {{Seth}}, \citenamefont {{Shvartzvald}}, \citenamefont {{Siellez}}, \citenamefont {{Smith}}, \citenamefont {{Spergel}}, \citenamefont {{Stassun}}, \citenamefont {{Street}}, \citenamefont {{Strolger}}, \citenamefont {{Szalay}}, \citenamefont {{Trauger}}, \citenamefont {{Troxel}}, \citenamefont {{Turnbull}}, \citenamefont {{van der Marel}}, \citenamefont {{von der Linden}}, \citenamefont {{Wang}}, \citenamefont {{Weinberg}}, \citenamefont {{Williams}}, \citenamefont {{Windhorst}}, \citenamefont {{Wollack}}, \citenamefont {{Wu}}, \citenamefont {{Yee}},\ and\ \citenamefont {{Zimmerman}}}]{Akeson2019}%
  \BibitemOpen
  \bibfield  {author} {\bibinfo {author} {\bibfnamefont {R.}~\bibnamefont {{Akeson}}}, \bibinfo {author} {\bibfnamefont {L.}~\bibnamefont {{Armus}}}, \bibinfo {author} {\bibfnamefont {E.}~\bibnamefont {{Bachelet}}}, \bibinfo {author} {\bibfnamefont {V.}~\bibnamefont {{Bailey}}}, \bibinfo {author} {\bibfnamefont {L.}~\bibnamefont {{Bartusek}}}, \bibinfo {author} {\bibfnamefont {A.}~\bibnamefont {{Bellini}}}, \bibinfo {author} {\bibfnamefont {D.}~\bibnamefont {{Benford}}}, \bibinfo {author} {\bibfnamefont {D.}~\bibnamefont {{Bennett}}}, \bibinfo {author} {\bibfnamefont {A.}~\bibnamefont {{Bhattacharya}}}, \bibinfo {author} {\bibfnamefont {R.}~\bibnamefont {{Bohlin}}}, \bibinfo {author} {\bibfnamefont {M.}~\bibnamefont {{Boyer}}}, \bibinfo {author} {\bibfnamefont {V.}~\bibnamefont {{Bozza}}}, \bibinfo {author} {\bibfnamefont {G.}~\bibnamefont {{Bryden}}}, \bibinfo {author} {\bibfnamefont {S.}~\bibnamefont {{Calchi Novati}}}, \bibinfo {author} {\bibfnamefont {K.}~\bibnamefont {{Carpenter}}}, \bibinfo {author}
  {\bibfnamefont {S.}~\bibnamefont {{Casertano}}}, \bibinfo {author} {\bibfnamefont {A.}~\bibnamefont {{Choi}}}, \bibinfo {author} {\bibfnamefont {D.}~\bibnamefont {{Content}}}, \bibinfo {author} {\bibfnamefont {P.}~\bibnamefont {{Dayal}}}, \bibinfo {author} {\bibfnamefont {A.}~\bibnamefont {{Dressler}}}, \bibinfo {author} {\bibfnamefont {O.}~\bibnamefont {{Dor{\'e}}}}, \bibinfo {author} {\bibfnamefont {S.~M.}\ \bibnamefont {{Fall}}}, \bibinfo {author} {\bibfnamefont {X.}~\bibnamefont {{Fan}}}, \bibinfo {author} {\bibfnamefont {X.}~\bibnamefont {{Fang}}}, \bibinfo {author} {\bibfnamefont {A.}~\bibnamefont {{Filippenko}}}, \bibinfo {author} {\bibfnamefont {S.}~\bibnamefont {{Finkelstein}}}, \bibinfo {author} {\bibfnamefont {R.}~\bibnamefont {{Foley}}}, \bibinfo {author} {\bibfnamefont {S.}~\bibnamefont {{Furlanetto}}}, \bibinfo {author} {\bibfnamefont {J.}~\bibnamefont {{Kalirai}}}, \bibinfo {author} {\bibfnamefont {B.~S.}\ \bibnamefont {{Gaudi}}}, \bibinfo {author} {\bibfnamefont {K.}~\bibnamefont
  {{Gilbert}}}, \bibinfo {author} {\bibfnamefont {J.}~\bibnamefont {{Girard}}}, \bibinfo {author} {\bibfnamefont {K.}~\bibnamefont {{Grady}}}, \bibinfo {author} {\bibfnamefont {J.}~\bibnamefont {{Greene}}}, \bibinfo {author} {\bibfnamefont {P.}~\bibnamefont {{Guhathakurta}}}, \bibinfo {author} {\bibfnamefont {C.}~\bibnamefont {{Heinrich}}}, \bibinfo {author} {\bibfnamefont {S.}~\bibnamefont {{Hemmati}}}, \bibinfo {author} {\bibfnamefont {D.}~\bibnamefont {{Hendel}}}, \bibinfo {author} {\bibfnamefont {C.}~\bibnamefont {{Henderson}}}, \bibinfo {author} {\bibfnamefont {T.}~\bibnamefont {{Henning}}}, \bibinfo {author} {\bibfnamefont {C.}~\bibnamefont {{Hirata}}}, \bibinfo {author} {\bibfnamefont {S.}~\bibnamefont {{Ho}}}, \bibinfo {author} {\bibfnamefont {E.}~\bibnamefont {{Huff}}}, \bibinfo {author} {\bibfnamefont {A.}~\bibnamefont {{Hutter}}}, \bibinfo {author} {\bibfnamefont {R.}~\bibnamefont {{Jansen}}}, \bibinfo {author} {\bibfnamefont {S.}~\bibnamefont {{Jha}}}, \bibinfo {author} {\bibfnamefont
  {S.}~\bibnamefont {{Johnson}}}, \bibinfo {author} {\bibfnamefont {D.}~\bibnamefont {{Jones}}}, \bibinfo {author} {\bibfnamefont {J.}~\bibnamefont {{Kasdin}}}, \bibinfo {author} {\bibfnamefont {P.}~\bibnamefont {{Kelly}}}, \bibinfo {author} {\bibfnamefont {R.}~\bibnamefont {{Kirshner}}}, \bibinfo {author} {\bibfnamefont {A.}~\bibnamefont {{Koekemoer}}}, \bibinfo {author} {\bibfnamefont {J.}~\bibnamefont {{Kruk}}}, \bibinfo {author} {\bibfnamefont {N.}~\bibnamefont {{Lewis}}}, \bibinfo {author} {\bibfnamefont {B.}~\bibnamefont {{Macintosh}}}, \bibinfo {author} {\bibfnamefont {P.}~\bibnamefont {{Madau}}}, \bibinfo {author} {\bibfnamefont {S.}~\bibnamefont {{Malhotra}}}, \bibinfo {author} {\bibfnamefont {K.}~\bibnamefont {{Mandel}}}, \bibinfo {author} {\bibfnamefont {E.}~\bibnamefont {{Massara}}}, \bibinfo {author} {\bibfnamefont {D.}~\bibnamefont {{Masters}}}, \bibinfo {author} {\bibfnamefont {J.}~\bibnamefont {{McEnery}}}, \bibinfo {author} {\bibfnamefont {K.}~\bibnamefont {{McQuinn}}}, \bibinfo {author}
  {\bibfnamefont {P.}~\bibnamefont {{Melchior}}}, \bibinfo {author} {\bibfnamefont {M.}~\bibnamefont {{Melton}}}, \bibinfo {author} {\bibfnamefont {B.}~\bibnamefont {{Mennesson}}}, \bibinfo {author} {\bibfnamefont {M.}~\bibnamefont {{Peeples}}}, \bibinfo {author} {\bibfnamefont {M.}~\bibnamefont {{Penny}}}, \bibinfo {author} {\bibfnamefont {S.}~\bibnamefont {{Perlmutter}}}, \bibinfo {author} {\bibfnamefont {A.}~\bibnamefont {{Pisani}}}, \bibinfo {author} {\bibfnamefont {A.}~\bibnamefont {{Plazas}}}, \bibinfo {author} {\bibfnamefont {R.}~\bibnamefont {{Poleski}}}, \bibinfo {author} {\bibfnamefont {M.}~\bibnamefont {{Postman}}}, \bibinfo {author} {\bibfnamefont {C.}~\bibnamefont {{Ranc}}}, \bibinfo {author} {\bibfnamefont {B.}~\bibnamefont {{Rauscher}}}, \bibinfo {author} {\bibfnamefont {A.}~\bibnamefont {{Rest}}}, \bibinfo {author} {\bibfnamefont {A.}~\bibnamefont {{Roberge}}}, \bibinfo {author} {\bibfnamefont {B.}~\bibnamefont {{Robertson}}}, \bibinfo {author} {\bibfnamefont {S.}~\bibnamefont {{Rodney}}},
  \bibinfo {author} {\bibfnamefont {J.}~\bibnamefont {{Rhoads}}}, \bibinfo {author} {\bibfnamefont {J.}~\bibnamefont {{Rhodes}}}, \bibinfo {author} {\bibfnamefont {R.}~\bibnamefont {{Ryan}}, \bibfnamefont {Jr.}}, \bibinfo {author} {\bibfnamefont {K.}~\bibnamefont {{Sahu}}}, \bibinfo {author} {\bibfnamefont {D.}~\bibnamefont {{Sand}}}, \bibinfo {author} {\bibfnamefont {D.}~\bibnamefont {{Scolnic}}}, \bibinfo {author} {\bibfnamefont {A.}~\bibnamefont {{Seth}}}, \bibinfo {author} {\bibfnamefont {Y.}~\bibnamefont {{Shvartzvald}}}, \bibinfo {author} {\bibfnamefont {K.}~\bibnamefont {{Siellez}}}, \bibinfo {author} {\bibfnamefont {A.}~\bibnamefont {{Smith}}}, \bibinfo {author} {\bibfnamefont {D.}~\bibnamefont {{Spergel}}}, \bibinfo {author} {\bibfnamefont {K.}~\bibnamefont {{Stassun}}}, \bibinfo {author} {\bibfnamefont {R.}~\bibnamefont {{Street}}}, \bibinfo {author} {\bibfnamefont {L.-G.}\ \bibnamefont {{Strolger}}}, \bibinfo {author} {\bibfnamefont {A.}~\bibnamefont {{Szalay}}}, \bibinfo {author} {\bibfnamefont
  {J.}~\bibnamefont {{Trauger}}}, \bibinfo {author} {\bibfnamefont {M.~A.}\ \bibnamefont {{Troxel}}}, \bibinfo {author} {\bibfnamefont {M.}~\bibnamefont {{Turnbull}}}, \bibinfo {author} {\bibfnamefont {R.}~\bibnamefont {{van der Marel}}}, \bibinfo {author} {\bibfnamefont {A.}~\bibnamefont {{von der Linden}}}, \bibinfo {author} {\bibfnamefont {Y.}~\bibnamefont {{Wang}}}, \bibinfo {author} {\bibfnamefont {D.}~\bibnamefont {{Weinberg}}}, \bibinfo {author} {\bibfnamefont {B.}~\bibnamefont {{Williams}}}, \bibinfo {author} {\bibfnamefont {R.}~\bibnamefont {{Windhorst}}}, \bibinfo {author} {\bibfnamefont {E.}~\bibnamefont {{Wollack}}}, \bibinfo {author} {\bibfnamefont {H.-Y.}\ \bibnamefont {{Wu}}}, \bibinfo {author} {\bibfnamefont {J.}~\bibnamefont {{Yee}}},\ and\ \bibinfo {author} {\bibfnamefont {N.}~\bibnamefont {{Zimmerman}}},\ }\href {https://doi.org/10.48550/arXiv.1902.05569} {\bibfield  {journal} {\bibinfo  {journal} {arXiv e-prints}\ ,\ \bibinfo {eid} {arXiv:1902.05569}} (\bibinfo {year} {2019})},\ \Eprint
  {https://arxiv.org/abs/1902.05569} {arXiv:1902.05569 [astro-ph.IM]} \BibitemShut {NoStop}%
\bibitem [{\citenamefont {Penny}\ \emph {et~al.}(2019)\citenamefont {Penny}, \citenamefont {Scott~Gaudi}, \citenamefont {Kerins}, \citenamefont {Rattenbury}, \citenamefont {Mao}, \citenamefont {Robin},\ and\ \citenamefont {Calchi~Novati}}]{Penny2019}%
  \BibitemOpen
  \bibfield  {author} {\bibinfo {author} {\bibfnamefont {M.~T.}\ \bibnamefont {Penny}}, \bibinfo {author} {\bibfnamefont {B.}~\bibnamefont {Scott~Gaudi}}, \bibinfo {author} {\bibfnamefont {E.}~\bibnamefont {Kerins}}, \bibinfo {author} {\bibfnamefont {N.~J.}\ \bibnamefont {Rattenbury}}, \bibinfo {author} {\bibfnamefont {S.}~\bibnamefont {Mao}}, \bibinfo {author} {\bibfnamefont {A.~C.}\ \bibnamefont {Robin}},\ and\ \bibinfo {author} {\bibfnamefont {S.}~\bibnamefont {Calchi~Novati}},\ }\href {https://doi.org/10.3847/1538-4365/aafb69} {\bibfield  {journal} {\bibinfo  {journal} {The Astrophysical Journal Supplement Series}\ }\textbf {\bibinfo {volume} {241}},\ \bibinfo {pages} {3} (\bibinfo {year} {2019})}\BibitemShut {NoStop}%
\bibitem [{\citenamefont {{Ge}}\ \emph {et~al.}(2022)\citenamefont {{Ge}}, \citenamefont {{Zhang}}, \citenamefont {{Zang}}, \citenamefont {{Deng}}, \citenamefont {{Mao}}, \citenamefont {{Xie}}, \citenamefont {{Liu}}, \citenamefont {{Zhou}}, \citenamefont {{Willis}}, \citenamefont {{Huang}}, \citenamefont {{Howell}}, \citenamefont {{Feng}}, \citenamefont {{Zhu}}, \citenamefont {{Yao}}, \citenamefont {{Liu}}, \citenamefont {{Aizawa}}, \citenamefont {{Zhu}}, \citenamefont {{Li}}, \citenamefont {{Ma}}, \citenamefont {{Ye}}, \citenamefont {{Yu}}, \citenamefont {{Xiang}}, \citenamefont {{Yu}}, \citenamefont {{Liu}}, \citenamefont {{Yang}}, \citenamefont {{Wang}}, \citenamefont {{Shi}}, \citenamefont {{Fang}}, \citenamefont {{Zong}}, \citenamefont {{Liu}}, \citenamefont {{Zhang}}, \citenamefont {{Zhang}}, \citenamefont {{El-Badry}}, \citenamefont {{Shen}}, \citenamefont {{Tam}}, \citenamefont {{Hu}}, \citenamefont {{Yang}}, \citenamefont {{Zou}}, \citenamefont {{Wu}}, \citenamefont {{Lei}}, \citenamefont {{Wei}},
  \citenamefont {{Wu}}, \citenamefont {{Sun}}, \citenamefont {{Wang}}, \citenamefont {{Zhang}}, \citenamefont {{Xu}}, \citenamefont {{Yang}}, \citenamefont {{Li}}, \citenamefont {{Xiang}}, \citenamefont {{Wang}}, \citenamefont {{Wang}}, \citenamefont {{Zhang}}, \citenamefont {{Jia}}, \citenamefont {{Yuan}}, \citenamefont {{Zhang}}, \citenamefont {{Xuesong Wang}}, \citenamefont {{Gan}}, \citenamefont {{Wang}}, \citenamefont {{Zhao}}, \citenamefont {{Liu}}, \citenamefont {{Wei}}, \citenamefont {{Kang}}, \citenamefont {{Yang}}, \citenamefont {{Qi}}, \citenamefont {{Liu}}, \citenamefont {{Zhang}}, \citenamefont {{Zhu}}, \citenamefont {{Zhou}}, \citenamefont {{Zhang}}, \citenamefont {{Yu}}, \citenamefont {{Zhang}}, \citenamefont {{Li}}, \citenamefont {{Tang}}, \citenamefont {{Wang}}, \citenamefont {{Wang}}, \citenamefont {{Li}}, \citenamefont {{Cheng}}, \citenamefont {{Shen}}, \citenamefont {{Li}}, \citenamefont {{Pan}}, \citenamefont {{Yang}}, \citenamefont {{Gao}}, \citenamefont {{Song}}, \citenamefont {{Wang}},
  \citenamefont {{Zhang}}, \citenamefont {{Chen}}, \citenamefont {{Wang}}, \citenamefont {{Zhang}}, \citenamefont {{Wang}}, \citenamefont {{Zeng}}, \citenamefont {{Zheng}}, \citenamefont {{Zhu}}, \citenamefont {{Guo}}, \citenamefont {{Zhang}}, \citenamefont {{Li}}, \citenamefont {{Wen}}, \citenamefont {{Feng}}, \citenamefont {{Chen}}, \citenamefont {{Chen}}, \citenamefont {{Han}}, \citenamefont {{Yang}}, \citenamefont {{Wang}}, \citenamefont {{Duan}}, \citenamefont {{Huang}}, \citenamefont {{Liang}}, \citenamefont {{Bi}}, \citenamefont {{Gai}}, \citenamefont {{Ge}}, \citenamefont {{Guo}}, \citenamefont {{Huang}}, \citenamefont {{Li}}, \citenamefont {{Li}}, \citenamefont {{Li}}, \citenamefont {{Yuxi}}, \citenamefont {{Lu}}, \citenamefont {{Rix}}, \citenamefont {{Shi}}, \citenamefont {{Song}}, \citenamefont {{Tang}}, \citenamefont {{Ting}}, \citenamefont {{Wu}}, \citenamefont {{Wu}}, \citenamefont {{Yang}}, \citenamefont {{Yin}}, \citenamefont {{Gould}}, \citenamefont {{Lee}}, \citenamefont {{Dong}},
  \citenamefont {{Yee}}, \citenamefont {{Shvartzvald}}, \citenamefont {{Yang}}, \citenamefont {{Kuang}}, \citenamefont {{Zhang}}, \citenamefont {{Liao}}, \citenamefont {{Qi}}, \citenamefont {{Yang}}, \citenamefont {{Zhang}}, \citenamefont {{Jiang}}, \citenamefont {{Ou}}, \citenamefont {{Li}}, \citenamefont {{Beck}}, \citenamefont {{Bedding}}, \citenamefont {{Campante}}, \citenamefont {{Chaplin}}, \citenamefont {{Christensen-Dalsgaard}}, \citenamefont {{Garc{\'\i}a}}, \citenamefont {{Gaulme}}, \citenamefont {{Gizon}}, \citenamefont {{Hekker}}, \citenamefont {{Huber}}, \citenamefont {{Khanna}}, \citenamefont {{Li}}, \citenamefont {{Mathur}}, \citenamefont {{Miglio}}, \citenamefont {{Mosser}}, \citenamefont {{Ong}}, \citenamefont {{Santos}}, \citenamefont {{Stello}}, \citenamefont {{Bowman}}, \citenamefont {{Lares-Martiz}}, \citenamefont {{Murphy}}, \citenamefont {{Niu}}, \citenamefont {{Ma}}, \citenamefont {{Moln{\'a}r}}, \citenamefont {{Fu}}, \citenamefont {{De Cat}}, \citenamefont {{Su}},\ and\ \citenamefont
  {{consortium}}}]{Ge2022}%
  \BibitemOpen
  \bibfield  {author} {\bibinfo {author} {\bibfnamefont {J.}~\bibnamefont {{Ge}}}, \bibinfo {author} {\bibfnamefont {H.}~\bibnamefont {{Zhang}}}, \bibinfo {author} {\bibfnamefont {W.}~\bibnamefont {{Zang}}}, \bibinfo {author} {\bibfnamefont {H.}~\bibnamefont {{Deng}}}, \bibinfo {author} {\bibfnamefont {S.}~\bibnamefont {{Mao}}}, \bibinfo {author} {\bibfnamefont {J.-W.}\ \bibnamefont {{Xie}}}, \bibinfo {author} {\bibfnamefont {H.-G.}\ \bibnamefont {{Liu}}}, \bibinfo {author} {\bibfnamefont {J.-L.}\ \bibnamefont {{Zhou}}}, \bibinfo {author} {\bibfnamefont {K.}~\bibnamefont {{Willis}}}, \bibinfo {author} {\bibfnamefont {C.}~\bibnamefont {{Huang}}}, \bibinfo {author} {\bibfnamefont {S.~B.}\ \bibnamefont {{Howell}}}, \bibinfo {author} {\bibfnamefont {F.}~\bibnamefont {{Feng}}}, \bibinfo {author} {\bibfnamefont {J.}~\bibnamefont {{Zhu}}}, \bibinfo {author} {\bibfnamefont {X.}~\bibnamefont {{Yao}}}, \bibinfo {author} {\bibfnamefont {B.}~\bibnamefont {{Liu}}}, \bibinfo {author} {\bibfnamefont {M.}~\bibnamefont
  {{Aizawa}}}, \bibinfo {author} {\bibfnamefont {W.}~\bibnamefont {{Zhu}}}, \bibinfo {author} {\bibfnamefont {Y.-P.}\ \bibnamefont {{Li}}}, \bibinfo {author} {\bibfnamefont {B.}~\bibnamefont {{Ma}}}, \bibinfo {author} {\bibfnamefont {Q.}~\bibnamefont {{Ye}}}, \bibinfo {author} {\bibfnamefont {J.}~\bibnamefont {{Yu}}}, \bibinfo {author} {\bibfnamefont {M.}~\bibnamefont {{Xiang}}}, \bibinfo {author} {\bibfnamefont {C.}~\bibnamefont {{Yu}}}, \bibinfo {author} {\bibfnamefont {S.}~\bibnamefont {{Liu}}}, \bibinfo {author} {\bibfnamefont {M.}~\bibnamefont {{Yang}}}, \bibinfo {author} {\bibfnamefont {M.-T.}\ \bibnamefont {{Wang}}}, \bibinfo {author} {\bibfnamefont {X.}~\bibnamefont {{Shi}}}, \bibinfo {author} {\bibfnamefont {T.}~\bibnamefont {{Fang}}}, \bibinfo {author} {\bibfnamefont {W.}~\bibnamefont {{Zong}}}, \bibinfo {author} {\bibfnamefont {J.}~\bibnamefont {{Liu}}}, \bibinfo {author} {\bibfnamefont {Y.}~\bibnamefont {{Zhang}}}, \bibinfo {author} {\bibfnamefont {L.}~\bibnamefont {{Zhang}}}, \bibinfo {author}
  {\bibfnamefont {K.}~\bibnamefont {{El-Badry}}}, \bibinfo {author} {\bibfnamefont {R.}~\bibnamefont {{Shen}}}, \bibinfo {author} {\bibfnamefont {P.-H.~T.}\ \bibnamefont {{Tam}}}, \bibinfo {author} {\bibfnamefont {Z.}~\bibnamefont {{Hu}}}, \bibinfo {author} {\bibfnamefont {Y.}~\bibnamefont {{Yang}}}, \bibinfo {author} {\bibfnamefont {Y.-C.}\ \bibnamefont {{Zou}}}, \bibinfo {author} {\bibfnamefont {J.-L.}\ \bibnamefont {{Wu}}}, \bibinfo {author} {\bibfnamefont {W.-H.}\ \bibnamefont {{Lei}}}, \bibinfo {author} {\bibfnamefont {J.-J.}\ \bibnamefont {{Wei}}}, \bibinfo {author} {\bibfnamefont {X.-F.}\ \bibnamefont {{Wu}}}, \bibinfo {author} {\bibfnamefont {T.-R.}\ \bibnamefont {{Sun}}}, \bibinfo {author} {\bibfnamefont {F.-Y.}\ \bibnamefont {{Wang}}}, \bibinfo {author} {\bibfnamefont {B.-B.}\ \bibnamefont {{Zhang}}}, \bibinfo {author} {\bibfnamefont {D.}~\bibnamefont {{Xu}}}, \bibinfo {author} {\bibfnamefont {Y.-P.}\ \bibnamefont {{Yang}}}, \bibinfo {author} {\bibfnamefont {W.-X.}\ \bibnamefont {{Li}}}, \bibinfo
  {author} {\bibfnamefont {D.-F.}\ \bibnamefont {{Xiang}}}, \bibinfo {author} {\bibfnamefont {X.}~\bibnamefont {{Wang}}}, \bibinfo {author} {\bibfnamefont {T.}~\bibnamefont {{Wang}}}, \bibinfo {author} {\bibfnamefont {B.}~\bibnamefont {{Zhang}}}, \bibinfo {author} {\bibfnamefont {P.}~\bibnamefont {{Jia}}}, \bibinfo {author} {\bibfnamefont {H.}~\bibnamefont {{Yuan}}}, \bibinfo {author} {\bibfnamefont {J.}~\bibnamefont {{Zhang}}}, \bibinfo {author} {\bibfnamefont {S.}~\bibnamefont {{Xuesong Wang}}}, \bibinfo {author} {\bibfnamefont {T.}~\bibnamefont {{Gan}}}, \bibinfo {author} {\bibfnamefont {W.}~\bibnamefont {{Wang}}}, \bibinfo {author} {\bibfnamefont {Y.}~\bibnamefont {{Zhao}}}, \bibinfo {author} {\bibfnamefont {Y.}~\bibnamefont {{Liu}}}, \bibinfo {author} {\bibfnamefont {C.}~\bibnamefont {{Wei}}}, \bibinfo {author} {\bibfnamefont {Y.}~\bibnamefont {{Kang}}}, \bibinfo {author} {\bibfnamefont {B.}~\bibnamefont {{Yang}}}, \bibinfo {author} {\bibfnamefont {C.}~\bibnamefont {{Qi}}}, \bibinfo {author}
  {\bibfnamefont {X.}~\bibnamefont {{Liu}}}, \bibinfo {author} {\bibfnamefont {Q.}~\bibnamefont {{Zhang}}}, \bibinfo {author} {\bibfnamefont {Y.}~\bibnamefont {{Zhu}}}, \bibinfo {author} {\bibfnamefont {D.}~\bibnamefont {{Zhou}}}, \bibinfo {author} {\bibfnamefont {C.}~\bibnamefont {{Zhang}}}, \bibinfo {author} {\bibfnamefont {Y.}~\bibnamefont {{Yu}}}, \bibinfo {author} {\bibfnamefont {Y.}~\bibnamefont {{Zhang}}}, \bibinfo {author} {\bibfnamefont {Y.}~\bibnamefont {{Li}}}, \bibinfo {author} {\bibfnamefont {Z.}~\bibnamefont {{Tang}}}, \bibinfo {author} {\bibfnamefont {C.}~\bibnamefont {{Wang}}}, \bibinfo {author} {\bibfnamefont {F.}~\bibnamefont {{Wang}}}, \bibinfo {author} {\bibfnamefont {W.}~\bibnamefont {{Li}}}, \bibinfo {author} {\bibfnamefont {P.}~\bibnamefont {{Cheng}}}, \bibinfo {author} {\bibfnamefont {C.}~\bibnamefont {{Shen}}}, \bibinfo {author} {\bibfnamefont {B.}~\bibnamefont {{Li}}}, \bibinfo {author} {\bibfnamefont {Y.}~\bibnamefont {{Pan}}}, \bibinfo {author} {\bibfnamefont {S.}~\bibnamefont
  {{Yang}}}, \bibinfo {author} {\bibfnamefont {W.}~\bibnamefont {{Gao}}}, \bibinfo {author} {\bibfnamefont {Z.}~\bibnamefont {{Song}}}, \bibinfo {author} {\bibfnamefont {J.}~\bibnamefont {{Wang}}}, \bibinfo {author} {\bibfnamefont {H.}~\bibnamefont {{Zhang}}}, \bibinfo {author} {\bibfnamefont {C.}~\bibnamefont {{Chen}}}, \bibinfo {author} {\bibfnamefont {H.}~\bibnamefont {{Wang}}}, \bibinfo {author} {\bibfnamefont {J.}~\bibnamefont {{Zhang}}}, \bibinfo {author} {\bibfnamefont {Z.}~\bibnamefont {{Wang}}}, \bibinfo {author} {\bibfnamefont {F.}~\bibnamefont {{Zeng}}}, \bibinfo {author} {\bibfnamefont {Z.}~\bibnamefont {{Zheng}}}, \bibinfo {author} {\bibfnamefont {J.}~\bibnamefont {{Zhu}}}, \bibinfo {author} {\bibfnamefont {Y.}~\bibnamefont {{Guo}}}, \bibinfo {author} {\bibfnamefont {Y.}~\bibnamefont {{Zhang}}}, \bibinfo {author} {\bibfnamefont {Y.}~\bibnamefont {{Li}}}, \bibinfo {author} {\bibfnamefont {L.}~\bibnamefont {{Wen}}}, \bibinfo {author} {\bibfnamefont {J.}~\bibnamefont {{Feng}}}, \bibinfo {author}
  {\bibfnamefont {W.}~\bibnamefont {{Chen}}}, \bibinfo {author} {\bibfnamefont {K.}~\bibnamefont {{Chen}}}, \bibinfo {author} {\bibfnamefont {X.}~\bibnamefont {{Han}}}, \bibinfo {author} {\bibfnamefont {Y.}~\bibnamefont {{Yang}}}, \bibinfo {author} {\bibfnamefont {H.}~\bibnamefont {{Wang}}}, \bibinfo {author} {\bibfnamefont {X.}~\bibnamefont {{Duan}}}, \bibinfo {author} {\bibfnamefont {J.}~\bibnamefont {{Huang}}}, \bibinfo {author} {\bibfnamefont {H.}~\bibnamefont {{Liang}}}, \bibinfo {author} {\bibfnamefont {S.}~\bibnamefont {{Bi}}}, \bibinfo {author} {\bibfnamefont {N.}~\bibnamefont {{Gai}}}, \bibinfo {author} {\bibfnamefont {Z.}~\bibnamefont {{Ge}}}, \bibinfo {author} {\bibfnamefont {Z.}~\bibnamefont {{Guo}}}, \bibinfo {author} {\bibfnamefont {Y.}~\bibnamefont {{Huang}}}, \bibinfo {author} {\bibfnamefont {G.}~\bibnamefont {{Li}}}, \bibinfo {author} {\bibfnamefont {H.}~\bibnamefont {{Li}}}, \bibinfo {author} {\bibfnamefont {T.}~\bibnamefont {{Li}}}, \bibinfo {author} {\bibnamefont {{Yuxi}}}, \bibinfo
  {author} {\bibnamefont {{Lu}}}, \bibinfo {author} {\bibfnamefont {H.-W.}\ \bibnamefont {{Rix}}}, \bibinfo {author} {\bibfnamefont {J.}~\bibnamefont {{Shi}}}, \bibinfo {author} {\bibfnamefont {F.}~\bibnamefont {{Song}}}, \bibinfo {author} {\bibfnamefont {Y.}~\bibnamefont {{Tang}}}, \bibinfo {author} {\bibfnamefont {Y.-S.}\ \bibnamefont {{Ting}}}, \bibinfo {author} {\bibfnamefont {T.}~\bibnamefont {{Wu}}}, \bibinfo {author} {\bibfnamefont {Y.}~\bibnamefont {{Wu}}}, \bibinfo {author} {\bibfnamefont {T.}~\bibnamefont {{Yang}}}, \bibinfo {author} {\bibfnamefont {Q.-Z.}\ \bibnamefont {{Yin}}}, \bibinfo {author} {\bibfnamefont {A.}~\bibnamefont {{Gould}}}, \bibinfo {author} {\bibfnamefont {C.-U.}\ \bibnamefont {{Lee}}}, \bibinfo {author} {\bibfnamefont {S.}~\bibnamefont {{Dong}}}, \bibinfo {author} {\bibfnamefont {J.~C.}\ \bibnamefont {{Yee}}}, \bibinfo {author} {\bibfnamefont {Y.}~\bibnamefont {{Shvartzvald}}}, \bibinfo {author} {\bibfnamefont {H.}~\bibnamefont {{Yang}}}, \bibinfo {author} {\bibfnamefont
  {R.}~\bibnamefont {{Kuang}}}, \bibinfo {author} {\bibfnamefont {J.}~\bibnamefont {{Zhang}}}, \bibinfo {author} {\bibfnamefont {S.}~\bibnamefont {{Liao}}}, \bibinfo {author} {\bibfnamefont {Z.}~\bibnamefont {{Qi}}}, \bibinfo {author} {\bibfnamefont {J.}~\bibnamefont {{Yang}}}, \bibinfo {author} {\bibfnamefont {R.}~\bibnamefont {{Zhang}}}, \bibinfo {author} {\bibfnamefont {C.}~\bibnamefont {{Jiang}}}, \bibinfo {author} {\bibfnamefont {J.-W.}\ \bibnamefont {{Ou}}}, \bibinfo {author} {\bibfnamefont {Y.}~\bibnamefont {{Li}}}, \bibinfo {author} {\bibfnamefont {P.}~\bibnamefont {{Beck}}}, \bibinfo {author} {\bibfnamefont {T.~R.}\ \bibnamefont {{Bedding}}}, \bibinfo {author} {\bibfnamefont {T.~L.}\ \bibnamefont {{Campante}}}, \bibinfo {author} {\bibfnamefont {W.~J.}\ \bibnamefont {{Chaplin}}}, \bibinfo {author} {\bibfnamefont {J.}~\bibnamefont {{Christensen-Dalsgaard}}}, \bibinfo {author} {\bibfnamefont {R.~A.}\ \bibnamefont {{Garc{\'\i}a}}}, \bibinfo {author} {\bibfnamefont {P.}~\bibnamefont {{Gaulme}}}, \bibinfo
  {author} {\bibfnamefont {L.}~\bibnamefont {{Gizon}}}, \bibinfo {author} {\bibfnamefont {S.}~\bibnamefont {{Hekker}}}, \bibinfo {author} {\bibfnamefont {D.}~\bibnamefont {{Huber}}}, \bibinfo {author} {\bibfnamefont {S.}~\bibnamefont {{Khanna}}}, \bibinfo {author} {\bibfnamefont {Y.}~\bibnamefont {{Li}}}, \bibinfo {author} {\bibfnamefont {S.}~\bibnamefont {{Mathur}}}, \bibinfo {author} {\bibfnamefont {A.}~\bibnamefont {{Miglio}}}, \bibinfo {author} {\bibfnamefont {B.}~\bibnamefont {{Mosser}}}, \bibinfo {author} {\bibfnamefont {J.~M.~J.}\ \bibnamefont {{Ong}}}, \bibinfo {author} {\bibfnamefont {{\^A}.~R.~G.}\ \bibnamefont {{Santos}}}, \bibinfo {author} {\bibfnamefont {D.}~\bibnamefont {{Stello}}}, \bibinfo {author} {\bibfnamefont {D.~M.}\ \bibnamefont {{Bowman}}}, \bibinfo {author} {\bibfnamefont {M.}~\bibnamefont {{Lares-Martiz}}}, \bibinfo {author} {\bibfnamefont {S.}~\bibnamefont {{Murphy}}}, \bibinfo {author} {\bibfnamefont {J.-S.}\ \bibnamefont {{Niu}}}, \bibinfo {author} {\bibfnamefont {X.-Y.}\
  \bibnamefont {{Ma}}}, \bibinfo {author} {\bibfnamefont {L.}~\bibnamefont {{Moln{\'a}r}}}, \bibinfo {author} {\bibfnamefont {J.-N.}\ \bibnamefont {{Fu}}}, \bibinfo {author} {\bibfnamefont {P.}~\bibnamefont {{De Cat}}}, \bibinfo {author} {\bibfnamefont {J.}~\bibnamefont {{Su}}},\ and\ \bibinfo {author} {\bibfnamefont {t.~E.}\ \bibnamefont {{consortium}}},\ }\href {https://doi.org/10.48550/arXiv.2206.06693} {\bibfield  {journal} {\bibinfo  {journal} {arXiv e-prints}\ ,\ \bibinfo {eid} {arXiv:2206.06693}} (\bibinfo {year} {2022})},\ \Eprint {https://arxiv.org/abs/2206.06693} {arXiv:2206.06693 [astro-ph.IM]} \BibitemShut {NoStop}%
\bibitem [{\citenamefont {Johnson}\ \emph {et~al.}(2020)\citenamefont {Johnson}, \citenamefont {Penny}, \citenamefont {Gaudi}, \citenamefont {Kerins}, \citenamefont {Rattenbury}, \citenamefont {Robin}, \citenamefont {Calchi~Novati},\ and\ \citenamefont {Henderson}}]{Johnson2020}%
  \BibitemOpen
  \bibfield  {author} {\bibinfo {author} {\bibfnamefont {S.~A.}\ \bibnamefont {Johnson}}, \bibinfo {author} {\bibfnamefont {M.}~\bibnamefont {Penny}}, \bibinfo {author} {\bibfnamefont {B.~S.}\ \bibnamefont {Gaudi}}, \bibinfo {author} {\bibfnamefont {E.}~\bibnamefont {Kerins}}, \bibinfo {author} {\bibfnamefont {N.~J.}\ \bibnamefont {Rattenbury}}, \bibinfo {author} {\bibfnamefont {A.~C.}\ \bibnamefont {Robin}}, \bibinfo {author} {\bibfnamefont {S.}~\bibnamefont {Calchi~Novati}},\ and\ \bibinfo {author} {\bibfnamefont {C.~B.}\ \bibnamefont {Henderson}},\ }\href {https://doi.org/10.3847/1538-3881/aba75b} {\bibfield  {journal} {\bibinfo  {journal} {The Astronomical Journal}\ }\textbf {\bibinfo {volume} {160}},\ \bibinfo {pages} {123} (\bibinfo {year} {2020})}\BibitemShut {NoStop}%
\bibitem [{\citenamefont {{Lee}}(2017)}]{Lee2017}%
  \BibitemOpen
  \bibfield  {author} {\bibinfo {author} {\bibfnamefont {C.-H.}\ \bibnamefont {{Lee}}},\ }\href {https://doi.org/10.3390/universe3030053} {\bibfield  {journal} {\bibinfo  {journal} {Universe}\ }\textbf {\bibinfo {volume} {3}},\ \bibinfo {eid} {53} (\bibinfo {year} {2017})},\ \Eprint {https://arxiv.org/abs/1711.05298} {arXiv:1711.05298 [astro-ph.IM]} \BibitemShut {NoStop}%
\bibitem [{\citenamefont {{Nitz}}\ \emph {et~al.}(2021)\citenamefont {{Nitz}}, \citenamefont {{Capano}}, \citenamefont {{Kumar}}, \citenamefont {{Wang}}, \citenamefont {{Kastha}}, \citenamefont {{Sch{\"a}fer}}, \citenamefont {{Dhurkunde}},\ and\ \citenamefont {{Cabero}}}]{Nitz2021BHMerger}%
  \BibitemOpen
  \bibfield  {author} {\bibinfo {author} {\bibfnamefont {A.~H.}\ \bibnamefont {{Nitz}}}, \bibinfo {author} {\bibfnamefont {C.~D.}\ \bibnamefont {{Capano}}}, \bibinfo {author} {\bibfnamefont {S.}~\bibnamefont {{Kumar}}}, \bibinfo {author} {\bibfnamefont {Y.-F.}\ \bibnamefont {{Wang}}}, \bibinfo {author} {\bibfnamefont {S.}~\bibnamefont {{Kastha}}}, \bibinfo {author} {\bibfnamefont {M.}~\bibnamefont {{Sch{\"a}fer}}}, \bibinfo {author} {\bibfnamefont {R.}~\bibnamefont {{Dhurkunde}}},\ and\ \bibinfo {author} {\bibfnamefont {M.}~\bibnamefont {{Cabero}}},\ }\href {https://doi.org/10.3847/1538-4357/ac1c03} {\bibfield  {journal} {\bibinfo  {journal} {The Astrophysical Journal}\ }\textbf {\bibinfo {volume} {922}},\ \bibinfo {eid} {76} (\bibinfo {year} {2021})},\ \Eprint {https://arxiv.org/abs/2105.09151} {arXiv:2105.09151 [astro-ph.HE]} \BibitemShut {NoStop}%
\bibitem [{\citenamefont {{Abbott}}\ \emph {et~al.}(2023)\citenamefont {{Abbott}}, \citenamefont {{Abbott}}, \citenamefont {{Acernese}}, \citenamefont {{Ackley}}, \citenamefont {{Adams}}, \citenamefont {{Adhikari}}, \citenamefont {{Adhikari}}, \citenamefont {{Adya}}, \citenamefont {{Affeldt}}, \citenamefont {{Agarwal}}, \citenamefont {{Agathos}}, \citenamefont {{Agatsuma}}, \citenamefont {{Aggarwal}}, \citenamefont {{Aguiar}}, \citenamefont {{Aiello}}, \citenamefont {{Ain}}, \citenamefont {{Ajith}}, \citenamefont {{Akutsu}}, \citenamefont {{de Alarc{\'o}n}}, \citenamefont {{Akcay}}, \citenamefont {{Albanesi}}, \citenamefont {{Allocca}}, \citenamefont {{Altin}}, \citenamefont {{Amato}}, \citenamefont {{Anand}}, \citenamefont {{Anand}}, \citenamefont {{Ananyeva}}, \citenamefont {{Anderson}}, \citenamefont {{Anderson}}, \citenamefont {{Ando}}, \citenamefont {{Andrade}}, \citenamefont {{Andres}}, \citenamefont {{Andri{\'c}}}, \citenamefont {{Angelova}}, \citenamefont {{Ansoldi}}, \citenamefont {{Antelis}},
  \citenamefont {{Antier}}, \citenamefont {{Antonini}}, \citenamefont {{Appert}}, \citenamefont {{Arai}}, \citenamefont {{Arai}}, \citenamefont {{Arai}}, \citenamefont {{Araki}}, \citenamefont {{Araya}}, \citenamefont {{Araya}}, \citenamefont {{Areeda}}, \citenamefont {{Ar{\`e}ne}}, \citenamefont {{Aritomi}}, \citenamefont {{Arnaud}}, \citenamefont {{Arogeti}}, \citenamefont {{Aronson}}, \citenamefont {{Arun}}, \citenamefont {{Asada}}, \citenamefont {{Asali}}, \citenamefont {{Ashton}}, \citenamefont {{Aso}}, \citenamefont {{Assiduo}}, \citenamefont {{Aston}}, \citenamefont {{Astone}}, \citenamefont {{Aubin}}, \citenamefont {{Austin}}, \citenamefont {{Babak}}, \citenamefont {{Badaracco}}, \citenamefont {{Bader}}, \citenamefont {{Badger}}, \citenamefont {{Bae}}, \citenamefont {{Bae}}, \citenamefont {{Baer}}, \citenamefont {{Bagnasco}}, \citenamefont {{Bai}}, \citenamefont {{Baiotti}}, \citenamefont {{Baird}}, \citenamefont {{Bajpai}}, \citenamefont {{Ball}}, \citenamefont {{Ballardin}}, \citenamefont
  {{Ballmer}}, \citenamefont {{Balsamo}}, \citenamefont {{Baltus}}, \citenamefont {{Banagiri}}, \citenamefont {{Bankar}}, \citenamefont {{Barayoga}}, \citenamefont {{Barbieri}}, \citenamefont {{Barish}}, \citenamefont {{Barker}}, \citenamefont {{Barneo}}, \citenamefont {{Barone}}, \citenamefont {{Barr}}, \citenamefont {{Barsotti}}, \citenamefont {{Barsuglia}}, \citenamefont {{Barta}}, \citenamefont {{Bartlett}}, \citenamefont {{Barton}}, \citenamefont {{Bartos}}, \citenamefont {{Bassiri}}, \citenamefont {{Basti}}, \citenamefont {{Bawaj}}, \citenamefont {{Bayley}}, \citenamefont {{Baylor}}, \citenamefont {{Bazzan}}, \citenamefont {{B{\'e}csy}}, \citenamefont {{Bedakihale}}, \citenamefont {{Bejger}}, \citenamefont {{Belahcene}}, \citenamefont {{Benedetto}}, \citenamefont {{Beniwal}}, \citenamefont {{Bennett}}, \citenamefont {{Bentley}}, \citenamefont {{Benyaala}}, \citenamefont {{Bergamin}}, \citenamefont {{Berger}}, \citenamefont {{Bernuzzi}}, \citenamefont {{Berry}}, \citenamefont {{Bersanetti}},
  \citenamefont {{Bertolini}}, \citenamefont {{Betzwieser}}, \citenamefont {{Beveridge}}, \citenamefont {{Bhandare}}, \citenamefont {{Bhardwaj}}, \citenamefont {{Bhattacharjee}}, \citenamefont {{Bhaumik}}, \citenamefont {{Bilenko}}, \citenamefont {{Billingsley}}, \citenamefont {{Bini}}, \citenamefont {{Birney}}, \citenamefont {{Birnholtz}}, \citenamefont {{Biscans}}, \citenamefont {{Bischi}}, \citenamefont {{Biscoveanu}}, \citenamefont {{Bisht}}, \citenamefont {{Biswas}}, \citenamefont {{Bitossi}}, \citenamefont {{Bizouard}}, \citenamefont {{Blackburn}}, \citenamefont {{Blair}}, \citenamefont {{Blair}}, \citenamefont {{Blair}}, \citenamefont {{Bobba}}, \citenamefont {{Bode}}, \citenamefont {{Boer}}, \citenamefont {{Bogaert}}, \citenamefont {{Boldrini}}, \citenamefont {{Bonavena}}, \citenamefont {{Bondu}}, \citenamefont {{Bonilla}}, \citenamefont {{Bonnand}}, \citenamefont {{Booker}}, \citenamefont {{Boom}}, \citenamefont {{Bork}}, \citenamefont {{Boschi}}, \citenamefont {{Bose}}, \citenamefont {{Bose}},
  \citenamefont {{Bossilkov}}, \citenamefont {{Boudart}}, \citenamefont {{Bouffanais}}, \citenamefont {{Bozzi}}, \citenamefont {{Bradaschia}}, \citenamefont {{Brady}}, \citenamefont {{Bramley}}, \citenamefont {{Branch}}, \citenamefont {{Branchesi}}, \citenamefont {{Brandt}}, \citenamefont {{Brau}}, \citenamefont {{Breschi}}, \citenamefont {{Briant}}, \citenamefont {{Briggs}}, \citenamefont {{Brillet}}, \citenamefont {{Brinkmann}}, \citenamefont {{Brockill}}, \citenamefont {{Brooks}}, \citenamefont {{Brooks}}, \citenamefont {{Brown}}, \citenamefont {{Brunett}}, \citenamefont {{Bruno}}, \citenamefont {{Bruntz}}, \citenamefont {{Bryant}}, \citenamefont {{Bulik}}, \citenamefont {{Bulten}}, \citenamefont {{Buonanno}}, \citenamefont {{Buscicchio}}, \citenamefont {{Buskulic}}, \citenamefont {{Buy}}, \citenamefont {{Byer}}, \citenamefont {{Cadonati}}, \citenamefont {{Cagnoli}}, \citenamefont {{Cahillane}}, \citenamefont {{Bustillo}}, \citenamefont {{Callaghan}}, \citenamefont {{Callister}}, \citenamefont {{Calloni}},
  \citenamefont {{Cameron}}, \citenamefont {{Camp}}, \citenamefont {{Canepa}}, \citenamefont {{Canevarolo}}, \citenamefont {{Cannavacciuolo}}, \citenamefont {{Cannon}}, \citenamefont {{Cao}}, \citenamefont {{Cao}}, \citenamefont {{Capocasa}}, \citenamefont {{Capote}},\ and\ \citenamefont {{Carapella}}}]{LIGO2023BHCatalogue}%
  \BibitemOpen
  \bibfield  {author} {\bibinfo {author} {\bibfnamefont {R.}~\bibnamefont {{Abbott}}}, \bibinfo {author} {\bibfnamefont {T.~D.}\ \bibnamefont {{Abbott}}}, \bibinfo {author} {\bibfnamefont {F.}~\bibnamefont {{Acernese}}}, \bibinfo {author} {\bibfnamefont {K.}~\bibnamefont {{Ackley}}}, \bibinfo {author} {\bibfnamefont {C.}~\bibnamefont {{Adams}}}, \bibinfo {author} {\bibfnamefont {N.}~\bibnamefont {{Adhikari}}}, \bibinfo {author} {\bibfnamefont {R.~X.}\ \bibnamefont {{Adhikari}}}, \bibinfo {author} {\bibfnamefont {V.~B.}\ \bibnamefont {{Adya}}}, \bibinfo {author} {\bibfnamefont {C.}~\bibnamefont {{Affeldt}}}, \bibinfo {author} {\bibfnamefont {D.}~\bibnamefont {{Agarwal}}}, \bibinfo {author} {\bibfnamefont {M.}~\bibnamefont {{Agathos}}}, \bibinfo {author} {\bibfnamefont {K.}~\bibnamefont {{Agatsuma}}}, \bibinfo {author} {\bibfnamefont {N.}~\bibnamefont {{Aggarwal}}}, \bibinfo {author} {\bibfnamefont {O.~D.}\ \bibnamefont {{Aguiar}}}, \bibinfo {author} {\bibfnamefont {L.}~\bibnamefont {{Aiello}}}, \bibinfo
  {author} {\bibfnamefont {A.}~\bibnamefont {{Ain}}}, \bibinfo {author} {\bibfnamefont {P.}~\bibnamefont {{Ajith}}}, \bibinfo {author} {\bibfnamefont {T.}~\bibnamefont {{Akutsu}}}, \bibinfo {author} {\bibfnamefont {P.~F.}\ \bibnamefont {{de Alarc{\'o}n}}}, \bibinfo {author} {\bibfnamefont {S.}~\bibnamefont {{Akcay}}}, \bibinfo {author} {\bibfnamefont {S.}~\bibnamefont {{Albanesi}}}, \bibinfo {author} {\bibfnamefont {A.}~\bibnamefont {{Allocca}}}, \bibinfo {author} {\bibfnamefont {P.~A.}\ \bibnamefont {{Altin}}}, \bibinfo {author} {\bibfnamefont {A.}~\bibnamefont {{Amato}}}, \bibinfo {author} {\bibfnamefont {C.}~\bibnamefont {{Anand}}}, \bibinfo {author} {\bibfnamefont {S.}~\bibnamefont {{Anand}}}, \bibinfo {author} {\bibfnamefont {A.}~\bibnamefont {{Ananyeva}}}, \bibinfo {author} {\bibfnamefont {S.~B.}\ \bibnamefont {{Anderson}}}, \bibinfo {author} {\bibfnamefont {W.~G.}\ \bibnamefont {{Anderson}}}, \bibinfo {author} {\bibfnamefont {M.}~\bibnamefont {{Ando}}}, \bibinfo {author} {\bibfnamefont
  {T.}~\bibnamefont {{Andrade}}}, \bibinfo {author} {\bibfnamefont {N.}~\bibnamefont {{Andres}}}, \bibinfo {author} {\bibfnamefont {T.}~\bibnamefont {{Andri{\'c}}}}, \bibinfo {author} {\bibfnamefont {S.~V.}\ \bibnamefont {{Angelova}}}, \bibinfo {author} {\bibfnamefont {S.}~\bibnamefont {{Ansoldi}}}, \bibinfo {author} {\bibfnamefont {J.~M.}\ \bibnamefont {{Antelis}}}, \bibinfo {author} {\bibfnamefont {S.}~\bibnamefont {{Antier}}}, \bibinfo {author} {\bibfnamefont {F.}~\bibnamefont {{Antonini}}}, \bibinfo {author} {\bibfnamefont {S.}~\bibnamefont {{Appert}}}, \bibinfo {author} {\bibfnamefont {K.}~\bibnamefont {{Arai}}}, \bibinfo {author} {\bibfnamefont {K.}~\bibnamefont {{Arai}}}, \bibinfo {author} {\bibfnamefont {Y.}~\bibnamefont {{Arai}}}, \bibinfo {author} {\bibfnamefont {S.}~\bibnamefont {{Araki}}}, \bibinfo {author} {\bibfnamefont {A.}~\bibnamefont {{Araya}}}, \bibinfo {author} {\bibfnamefont {M.~C.}\ \bibnamefont {{Araya}}}, \bibinfo {author} {\bibfnamefont {J.~S.}\ \bibnamefont {{Areeda}}}, \bibinfo
  {author} {\bibfnamefont {M.}~\bibnamefont {{Ar{\`e}ne}}}, \bibinfo {author} {\bibfnamefont {N.}~\bibnamefont {{Aritomi}}}, \bibinfo {author} {\bibfnamefont {N.}~\bibnamefont {{Arnaud}}}, \bibinfo {author} {\bibfnamefont {M.}~\bibnamefont {{Arogeti}}}, \bibinfo {author} {\bibfnamefont {S.~M.}\ \bibnamefont {{Aronson}}}, \bibinfo {author} {\bibfnamefont {K.~G.}\ \bibnamefont {{Arun}}}, \bibinfo {author} {\bibfnamefont {H.}~\bibnamefont {{Asada}}}, \bibinfo {author} {\bibfnamefont {Y.}~\bibnamefont {{Asali}}}, \bibinfo {author} {\bibfnamefont {G.}~\bibnamefont {{Ashton}}}, \bibinfo {author} {\bibfnamefont {Y.}~\bibnamefont {{Aso}}}, \bibinfo {author} {\bibfnamefont {M.}~\bibnamefont {{Assiduo}}}, \bibinfo {author} {\bibfnamefont {S.~M.}\ \bibnamefont {{Aston}}}, \bibinfo {author} {\bibfnamefont {P.}~\bibnamefont {{Astone}}}, \bibinfo {author} {\bibfnamefont {F.}~\bibnamefont {{Aubin}}}, \bibinfo {author} {\bibfnamefont {C.}~\bibnamefont {{Austin}}}, \bibinfo {author} {\bibfnamefont {S.}~\bibnamefont
  {{Babak}}}, \bibinfo {author} {\bibfnamefont {F.}~\bibnamefont {{Badaracco}}}, \bibinfo {author} {\bibfnamefont {M.~K.~M.}\ \bibnamefont {{Bader}}}, \bibinfo {author} {\bibfnamefont {C.}~\bibnamefont {{Badger}}}, \bibinfo {author} {\bibfnamefont {S.}~\bibnamefont {{Bae}}}, \bibinfo {author} {\bibfnamefont {Y.}~\bibnamefont {{Bae}}}, \bibinfo {author} {\bibfnamefont {A.~M.}\ \bibnamefont {{Baer}}}, \bibinfo {author} {\bibfnamefont {S.}~\bibnamefont {{Bagnasco}}}, \bibinfo {author} {\bibfnamefont {Y.}~\bibnamefont {{Bai}}}, \bibinfo {author} {\bibfnamefont {L.}~\bibnamefont {{Baiotti}}}, \bibinfo {author} {\bibfnamefont {J.}~\bibnamefont {{Baird}}}, \bibinfo {author} {\bibfnamefont {R.}~\bibnamefont {{Bajpai}}}, \bibinfo {author} {\bibfnamefont {M.}~\bibnamefont {{Ball}}}, \bibinfo {author} {\bibfnamefont {G.}~\bibnamefont {{Ballardin}}}, \bibinfo {author} {\bibfnamefont {S.~W.}\ \bibnamefont {{Ballmer}}}, \bibinfo {author} {\bibfnamefont {A.}~\bibnamefont {{Balsamo}}}, \bibinfo {author} {\bibfnamefont
  {G.}~\bibnamefont {{Baltus}}}, \bibinfo {author} {\bibfnamefont {S.}~\bibnamefont {{Banagiri}}}, \bibinfo {author} {\bibfnamefont {D.}~\bibnamefont {{Bankar}}}, \bibinfo {author} {\bibfnamefont {J.~C.}\ \bibnamefont {{Barayoga}}}, \bibinfo {author} {\bibfnamefont {C.}~\bibnamefont {{Barbieri}}}, \bibinfo {author} {\bibfnamefont {B.~C.}\ \bibnamefont {{Barish}}}, \bibinfo {author} {\bibfnamefont {D.}~\bibnamefont {{Barker}}}, \bibinfo {author} {\bibfnamefont {P.}~\bibnamefont {{Barneo}}}, \bibinfo {author} {\bibfnamefont {F.}~\bibnamefont {{Barone}}}, \bibinfo {author} {\bibfnamefont {B.}~\bibnamefont {{Barr}}}, \bibinfo {author} {\bibfnamefont {L.}~\bibnamefont {{Barsotti}}}, \bibinfo {author} {\bibfnamefont {M.}~\bibnamefont {{Barsuglia}}}, \bibinfo {author} {\bibfnamefont {D.}~\bibnamefont {{Barta}}}, \bibinfo {author} {\bibfnamefont {J.}~\bibnamefont {{Bartlett}}}, \bibinfo {author} {\bibfnamefont {M.~A.}\ \bibnamefont {{Barton}}}, \bibinfo {author} {\bibfnamefont {I.}~\bibnamefont {{Bartos}}}, \bibinfo
  {author} {\bibfnamefont {R.}~\bibnamefont {{Bassiri}}}, \bibinfo {author} {\bibfnamefont {A.}~\bibnamefont {{Basti}}}, \bibinfo {author} {\bibfnamefont {M.}~\bibnamefont {{Bawaj}}}, \bibinfo {author} {\bibfnamefont {J.~C.}\ \bibnamefont {{Bayley}}}, \bibinfo {author} {\bibfnamefont {A.~C.}\ \bibnamefont {{Baylor}}}, \bibinfo {author} {\bibfnamefont {M.}~\bibnamefont {{Bazzan}}}, \bibinfo {author} {\bibfnamefont {B.}~\bibnamefont {{B{\'e}csy}}}, \bibinfo {author} {\bibfnamefont {V.~M.}\ \bibnamefont {{Bedakihale}}}, \bibinfo {author} {\bibfnamefont {M.}~\bibnamefont {{Bejger}}}, \bibinfo {author} {\bibfnamefont {I.}~\bibnamefont {{Belahcene}}}, \bibinfo {author} {\bibfnamefont {V.}~\bibnamefont {{Benedetto}}}, \bibinfo {author} {\bibfnamefont {D.}~\bibnamefont {{Beniwal}}}, \bibinfo {author} {\bibfnamefont {T.~F.}\ \bibnamefont {{Bennett}}}, \bibinfo {author} {\bibfnamefont {J.~D.}\ \bibnamefont {{Bentley}}}, \bibinfo {author} {\bibfnamefont {M.}~\bibnamefont {{Benyaala}}}, \bibinfo {author} {\bibfnamefont
  {F.}~\bibnamefont {{Bergamin}}}, \bibinfo {author} {\bibfnamefont {B.~K.}\ \bibnamefont {{Berger}}}, \bibinfo {author} {\bibfnamefont {S.}~\bibnamefont {{Bernuzzi}}}, \bibinfo {author} {\bibfnamefont {C.~P.~L.}\ \bibnamefont {{Berry}}}, \bibinfo {author} {\bibfnamefont {D.}~\bibnamefont {{Bersanetti}}}, \bibinfo {author} {\bibfnamefont {A.}~\bibnamefont {{Bertolini}}}, \bibinfo {author} {\bibfnamefont {J.}~\bibnamefont {{Betzwieser}}}, \bibinfo {author} {\bibfnamefont {D.}~\bibnamefont {{Beveridge}}}, \bibinfo {author} {\bibfnamefont {R.}~\bibnamefont {{Bhandare}}}, \bibinfo {author} {\bibfnamefont {U.}~\bibnamefont {{Bhardwaj}}}, \bibinfo {author} {\bibfnamefont {D.}~\bibnamefont {{Bhattacharjee}}}, \bibinfo {author} {\bibfnamefont {S.}~\bibnamefont {{Bhaumik}}}, \bibinfo {author} {\bibfnamefont {I.~A.}\ \bibnamefont {{Bilenko}}}, \bibinfo {author} {\bibfnamefont {G.}~\bibnamefont {{Billingsley}}}, \bibinfo {author} {\bibfnamefont {S.}~\bibnamefont {{Bini}}}, \bibinfo {author} {\bibfnamefont
  {R.}~\bibnamefont {{Birney}}}, \bibinfo {author} {\bibfnamefont {O.}~\bibnamefont {{Birnholtz}}}, \bibinfo {author} {\bibfnamefont {S.}~\bibnamefont {{Biscans}}}, \bibinfo {author} {\bibfnamefont {M.}~\bibnamefont {{Bischi}}}, \bibinfo {author} {\bibfnamefont {S.}~\bibnamefont {{Biscoveanu}}}, \bibinfo {author} {\bibfnamefont {A.}~\bibnamefont {{Bisht}}}, \bibinfo {author} {\bibfnamefont {B.}~\bibnamefont {{Biswas}}}, \bibinfo {author} {\bibfnamefont {M.}~\bibnamefont {{Bitossi}}}, \bibinfo {author} {\bibfnamefont {M.~A.}\ \bibnamefont {{Bizouard}}}, \bibinfo {author} {\bibfnamefont {J.~K.}\ \bibnamefont {{Blackburn}}}, \bibinfo {author} {\bibfnamefont {C.~D.}\ \bibnamefont {{Blair}}}, \bibinfo {author} {\bibfnamefont {D.~G.}\ \bibnamefont {{Blair}}}, \bibinfo {author} {\bibfnamefont {R.~M.}\ \bibnamefont {{Blair}}}, \bibinfo {author} {\bibfnamefont {F.}~\bibnamefont {{Bobba}}}, \bibinfo {author} {\bibfnamefont {N.}~\bibnamefont {{Bode}}}, \bibinfo {author} {\bibfnamefont {M.}~\bibnamefont {{Boer}}},
  \bibinfo {author} {\bibfnamefont {G.}~\bibnamefont {{Bogaert}}}, \bibinfo {author} {\bibfnamefont {M.}~\bibnamefont {{Boldrini}}}, \bibinfo {author} {\bibfnamefont {L.~D.}\ \bibnamefont {{Bonavena}}}, \bibinfo {author} {\bibfnamefont {F.}~\bibnamefont {{Bondu}}}, \bibinfo {author} {\bibfnamefont {E.}~\bibnamefont {{Bonilla}}}, \bibinfo {author} {\bibfnamefont {R.}~\bibnamefont {{Bonnand}}}, \bibinfo {author} {\bibfnamefont {P.}~\bibnamefont {{Booker}}}, \bibinfo {author} {\bibfnamefont {B.~A.}\ \bibnamefont {{Boom}}}, \bibinfo {author} {\bibfnamefont {R.}~\bibnamefont {{Bork}}}, \bibinfo {author} {\bibfnamefont {V.}~\bibnamefont {{Boschi}}}, \bibinfo {author} {\bibfnamefont {N.}~\bibnamefont {{Bose}}}, \bibinfo {author} {\bibfnamefont {S.}~\bibnamefont {{Bose}}}, \bibinfo {author} {\bibfnamefont {V.}~\bibnamefont {{Bossilkov}}}, \bibinfo {author} {\bibfnamefont {V.}~\bibnamefont {{Boudart}}}, \bibinfo {author} {\bibfnamefont {Y.}~\bibnamefont {{Bouffanais}}}, \bibinfo {author} {\bibfnamefont
  {A.}~\bibnamefont {{Bozzi}}}, \bibinfo {author} {\bibfnamefont {C.}~\bibnamefont {{Bradaschia}}}, \bibinfo {author} {\bibfnamefont {P.~R.}\ \bibnamefont {{Brady}}}, \bibinfo {author} {\bibfnamefont {A.}~\bibnamefont {{Bramley}}}, \bibinfo {author} {\bibfnamefont {A.}~\bibnamefont {{Branch}}}, \bibinfo {author} {\bibfnamefont {M.}~\bibnamefont {{Branchesi}}}, \bibinfo {author} {\bibfnamefont {J.}~\bibnamefont {{Brandt}}}, \bibinfo {author} {\bibfnamefont {J.~E.}\ \bibnamefont {{Brau}}}, \bibinfo {author} {\bibfnamefont {M.}~\bibnamefont {{Breschi}}}, \bibinfo {author} {\bibfnamefont {T.}~\bibnamefont {{Briant}}}, \bibinfo {author} {\bibfnamefont {J.~H.}\ \bibnamefont {{Briggs}}}, \bibinfo {author} {\bibfnamefont {A.}~\bibnamefont {{Brillet}}}, \bibinfo {author} {\bibfnamefont {M.}~\bibnamefont {{Brinkmann}}}, \bibinfo {author} {\bibfnamefont {P.}~\bibnamefont {{Brockill}}}, \bibinfo {author} {\bibfnamefont {A.~F.}\ \bibnamefont {{Brooks}}}, \bibinfo {author} {\bibfnamefont {J.}~\bibnamefont {{Brooks}}},
  \bibinfo {author} {\bibfnamefont {D.~D.}\ \bibnamefont {{Brown}}}, \bibinfo {author} {\bibfnamefont {S.}~\bibnamefont {{Brunett}}}, \bibinfo {author} {\bibfnamefont {G.}~\bibnamefont {{Bruno}}}, \bibinfo {author} {\bibfnamefont {R.}~\bibnamefont {{Bruntz}}}, \bibinfo {author} {\bibfnamefont {J.}~\bibnamefont {{Bryant}}}, \bibinfo {author} {\bibfnamefont {T.}~\bibnamefont {{Bulik}}}, \bibinfo {author} {\bibfnamefont {H.~J.}\ \bibnamefont {{Bulten}}}, \bibinfo {author} {\bibfnamefont {A.}~\bibnamefont {{Buonanno}}}, \bibinfo {author} {\bibfnamefont {R.}~\bibnamefont {{Buscicchio}}}, \bibinfo {author} {\bibfnamefont {D.}~\bibnamefont {{Buskulic}}}, \bibinfo {author} {\bibfnamefont {C.}~\bibnamefont {{Buy}}}, \bibinfo {author} {\bibfnamefont {R.~L.}\ \bibnamefont {{Byer}}}, \bibinfo {author} {\bibfnamefont {L.}~\bibnamefont {{Cadonati}}}, \bibinfo {author} {\bibfnamefont {G.}~\bibnamefont {{Cagnoli}}}, \bibinfo {author} {\bibfnamefont {C.}~\bibnamefont {{Cahillane}}}, \bibinfo {author} {\bibfnamefont {J.~C.}\
  \bibnamefont {{Bustillo}}}, \bibinfo {author} {\bibfnamefont {J.~D.}\ \bibnamefont {{Callaghan}}}, \bibinfo {author} {\bibfnamefont {T.~A.}\ \bibnamefont {{Callister}}}, \bibinfo {author} {\bibfnamefont {E.}~\bibnamefont {{Calloni}}}, \bibinfo {author} {\bibfnamefont {J.}~\bibnamefont {{Cameron}}}, \bibinfo {author} {\bibfnamefont {J.~B.}\ \bibnamefont {{Camp}}}, \bibinfo {author} {\bibfnamefont {M.}~\bibnamefont {{Canepa}}}, \bibinfo {author} {\bibfnamefont {S.}~\bibnamefont {{Canevarolo}}}, \bibinfo {author} {\bibfnamefont {M.}~\bibnamefont {{Cannavacciuolo}}}, \bibinfo {author} {\bibfnamefont {K.~C.}\ \bibnamefont {{Cannon}}}, \bibinfo {author} {\bibfnamefont {H.}~\bibnamefont {{Cao}}}, \bibinfo {author} {\bibfnamefont {Z.}~\bibnamefont {{Cao}}}, \bibinfo {author} {\bibfnamefont {E.}~\bibnamefont {{Capocasa}}}, \bibinfo {author} {\bibfnamefont {E.}~\bibnamefont {{Capote}}},\ and\ \bibinfo {author} {\bibfnamefont {G.}~\bibnamefont {{Carapella}}},\ }\href {https://doi.org/10.1103/PhysRevX.13.011048}
  {\bibfield  {journal} {\bibinfo  {journal} {Physical Review X}\ }\textbf {\bibinfo {volume} {13}},\ \bibinfo {eid} {011048} (\bibinfo {year} {2023})},\ \Eprint {https://arxiv.org/abs/2111.03634} {arXiv:2111.03634 [astro-ph.HE]} \BibitemShut {NoStop}%
\bibitem [{\citenamefont {Vittorini}\ \emph {et~al.}(2014)\citenamefont {Vittorini}, \citenamefont {Hucul}, \citenamefont {Inlek}, \citenamefont {Crocker},\ and\ \citenamefont {Monroe}}]{vittorini2014entanglement}%
  \BibitemOpen
  \bibfield  {author} {\bibinfo {author} {\bibfnamefont {G.}~\bibnamefont {Vittorini}}, \bibinfo {author} {\bibfnamefont {D.}~\bibnamefont {Hucul}}, \bibinfo {author} {\bibfnamefont {I.}~\bibnamefont {Inlek}}, \bibinfo {author} {\bibfnamefont {C.}~\bibnamefont {Crocker}},\ and\ \bibinfo {author} {\bibfnamefont {C.}~\bibnamefont {Monroe}},\ }\href@noop {} {\bibfield  {journal} {\bibinfo  {journal} {Physical Review A}\ }\textbf {\bibinfo {volume} {90}},\ \bibinfo {pages} {040302} (\bibinfo {year} {2014})}\BibitemShut {NoStop}%
\bibitem [{\citenamefont {Childs}\ and\ \citenamefont {Van~Dam}(2010)}]{childs2010quantum}%
  \BibitemOpen
  \bibfield  {author} {\bibinfo {author} {\bibfnamefont {A.~M.}\ \bibnamefont {Childs}}\ and\ \bibinfo {author} {\bibfnamefont {W.}~\bibnamefont {Van~Dam}},\ }\href@noop {} {\bibfield  {journal} {\bibinfo  {journal} {Reviews of Modern Physics}\ }\textbf {\bibinfo {volume} {82}},\ \bibinfo {pages} {1} (\bibinfo {year} {2010})}\BibitemShut {NoStop}%
\bibitem [{\citenamefont {Bacon}\ \emph {et~al.}(2006)\citenamefont {Bacon}, \citenamefont {Childs},\ and\ \citenamefont {van Dam}}]{bacon2005optimal}%
  \BibitemOpen
  \bibfield  {author} {\bibinfo {author} {\bibfnamefont {D.}~\bibnamefont {Bacon}}, \bibinfo {author} {\bibfnamefont {A.~M.}\ \bibnamefont {Childs}},\ and\ \bibinfo {author} {\bibfnamefont {W.}~\bibnamefont {van Dam}},\ }\bibfield  {journal} {\bibinfo  {journal} {Chicago Journal of Theoretical Computer Science}\ }\href {https://doi.org/10.4086/cjtcs.2006.002} {10.4086/cjtcs.2006.002} (\bibinfo {year} {2006}),\ \Eprint {https://arxiv.org/abs/quant-ph/0501044} {quant-ph/0501044} \BibitemShut {NoStop}%
\bibitem [{\citenamefont {Regev}(2004)}]{regev2004quantum}%
  \BibitemOpen
  \bibfield  {author} {\bibinfo {author} {\bibfnamefont {O.}~\bibnamefont {Regev}},\ }\href@noop {} {\bibfield  {journal} {\bibinfo  {journal} {SIAM Journal on Computing}\ }\textbf {\bibinfo {volume} {33}},\ \bibinfo {pages} {738} (\bibinfo {year} {2004})}\BibitemShut {NoStop}%
\bibitem [{\citenamefont {Davis}\ \emph {et~al.}(2017)\citenamefont {Davis}, \citenamefont {Saulnier}, \citenamefont {Karpi{\'n}ski},\ and\ \citenamefont {Smith}}]{davis2017pulsed}%
  \BibitemOpen
  \bibfield  {author} {\bibinfo {author} {\bibfnamefont {A.~O.}\ \bibnamefont {Davis}}, \bibinfo {author} {\bibfnamefont {P.~M.}\ \bibnamefont {Saulnier}}, \bibinfo {author} {\bibfnamefont {M.}~\bibnamefont {Karpi{\'n}ski}},\ and\ \bibinfo {author} {\bibfnamefont {B.~J.}\ \bibnamefont {Smith}},\ }\href@noop {} {\bibfield  {journal} {\bibinfo  {journal} {Optics Express}\ }\textbf {\bibinfo {volume} {25}},\ \bibinfo {pages} {12804} (\bibinfo {year} {2017})}\BibitemShut {NoStop}%
\bibitem [{\citenamefont {Zhu}\ \emph {et~al.}(2022)\citenamefont {Zhu}, \citenamefont {Chen}, \citenamefont {Yu}, \citenamefont {Shao}, \citenamefont {Hu}, \citenamefont {Xin}, \citenamefont {Yeh}, \citenamefont {Ghosh}, \citenamefont {He}, \citenamefont {Reimer} \emph {et~al.}}]{zhu2022spectral}%
  \BibitemOpen
  \bibfield  {author} {\bibinfo {author} {\bibfnamefont {D.}~\bibnamefont {Zhu}}, \bibinfo {author} {\bibfnamefont {C.}~\bibnamefont {Chen}}, \bibinfo {author} {\bibfnamefont {M.}~\bibnamefont {Yu}}, \bibinfo {author} {\bibfnamefont {L.}~\bibnamefont {Shao}}, \bibinfo {author} {\bibfnamefont {Y.}~\bibnamefont {Hu}}, \bibinfo {author} {\bibfnamefont {C.}~\bibnamefont {Xin}}, \bibinfo {author} {\bibfnamefont {M.}~\bibnamefont {Yeh}}, \bibinfo {author} {\bibfnamefont {S.}~\bibnamefont {Ghosh}}, \bibinfo {author} {\bibfnamefont {L.}~\bibnamefont {He}}, \bibinfo {author} {\bibfnamefont {C.}~\bibnamefont {Reimer}}, \emph {et~al.},\ }\href@noop {} {\bibfield  {journal} {\bibinfo  {journal} {Light: Science \& Applications}\ }\textbf {\bibinfo {volume} {11}},\ \bibinfo {pages} {327} (\bibinfo {year} {2022})}\BibitemShut {NoStop}%
\bibitem [{\citenamefont {Cheng}\ \emph {et~al.}(2019)\citenamefont {Cheng}, \citenamefont {Zou}, \citenamefont {Guo}, \citenamefont {Wang}, \citenamefont {Han},\ and\ \citenamefont {Tang}}]{cheng2019broadband}%
  \BibitemOpen
  \bibfield  {author} {\bibinfo {author} {\bibfnamefont {R.}~\bibnamefont {Cheng}}, \bibinfo {author} {\bibfnamefont {C.-L.}\ \bibnamefont {Zou}}, \bibinfo {author} {\bibfnamefont {X.}~\bibnamefont {Guo}}, \bibinfo {author} {\bibfnamefont {S.}~\bibnamefont {Wang}}, \bibinfo {author} {\bibfnamefont {X.}~\bibnamefont {Han}},\ and\ \bibinfo {author} {\bibfnamefont {H.~X.}\ \bibnamefont {Tang}},\ }\href@noop {} {\bibfield  {journal} {\bibinfo  {journal} {Nature Communications}\ }\textbf {\bibinfo {volume} {10}},\ \bibinfo {pages} {4104} (\bibinfo {year} {2019})}\BibitemShut {NoStop}%
\bibitem [{\citenamefont {Xiao}\ \emph {et~al.}(2022)\citenamefont {Xiao}, \citenamefont {Wei}, \citenamefont {Xu}, \citenamefont {Ma}, \citenamefont {Liu}, \citenamefont {Zhang}, \citenamefont {Tao}, \citenamefont {Li}, \citenamefont {Wang}, \citenamefont {You} \emph {et~al.}}]{xiao2022superconducting}%
  \BibitemOpen
  \bibfield  {author} {\bibinfo {author} {\bibfnamefont {Y.}~\bibnamefont {Xiao}}, \bibinfo {author} {\bibfnamefont {S.}~\bibnamefont {Wei}}, \bibinfo {author} {\bibfnamefont {J.}~\bibnamefont {Xu}}, \bibinfo {author} {\bibfnamefont {R.}~\bibnamefont {Ma}}, \bibinfo {author} {\bibfnamefont {X.}~\bibnamefont {Liu}}, \bibinfo {author} {\bibfnamefont {X.}~\bibnamefont {Zhang}}, \bibinfo {author} {\bibfnamefont {T.~H.}\ \bibnamefont {Tao}}, \bibinfo {author} {\bibfnamefont {H.}~\bibnamefont {Li}}, \bibinfo {author} {\bibfnamefont {Z.}~\bibnamefont {Wang}}, \bibinfo {author} {\bibfnamefont {L.}~\bibnamefont {You}}, \emph {et~al.},\ }\href@noop {} {\bibfield  {journal} {\bibinfo  {journal} {ACS Photonics}\ }\textbf {\bibinfo {volume} {9}},\ \bibinfo {pages} {3450} (\bibinfo {year} {2022})}\BibitemShut {NoStop}%
\bibitem [{\citenamefont {Kahl}\ \emph {et~al.}(2017)\citenamefont {Kahl}, \citenamefont {Ferrari}, \citenamefont {Kovalyuk}, \citenamefont {Vetter}, \citenamefont {Lewes-Malandrakis}, \citenamefont {Nebel}, \citenamefont {Korneev}, \citenamefont {Goltsman},\ and\ \citenamefont {Pernice}}]{kahl2017spectrally}%
  \BibitemOpen
  \bibfield  {author} {\bibinfo {author} {\bibfnamefont {O.}~\bibnamefont {Kahl}}, \bibinfo {author} {\bibfnamefont {S.}~\bibnamefont {Ferrari}}, \bibinfo {author} {\bibfnamefont {V.}~\bibnamefont {Kovalyuk}}, \bibinfo {author} {\bibfnamefont {A.}~\bibnamefont {Vetter}}, \bibinfo {author} {\bibfnamefont {G.}~\bibnamefont {Lewes-Malandrakis}}, \bibinfo {author} {\bibfnamefont {C.}~\bibnamefont {Nebel}}, \bibinfo {author} {\bibfnamefont {A.}~\bibnamefont {Korneev}}, \bibinfo {author} {\bibfnamefont {G.}~\bibnamefont {Goltsman}},\ and\ \bibinfo {author} {\bibfnamefont {W.}~\bibnamefont {Pernice}},\ }\href@noop {} {\bibfield  {journal} {\bibinfo  {journal} {Optica}\ }\textbf {\bibinfo {volume} {4}},\ \bibinfo {pages} {557} (\bibinfo {year} {2017})}\BibitemShut {NoStop}%
\bibitem [{\citenamefont {Zhong}\ \emph {et~al.}(2025)\citenamefont {Zhong}, \citenamefont {Liu}, \citenamefont {Yin}, \citenamefont {Zhao}, \citenamefont {Wang}, \citenamefont {Ren}, \citenamefont {Dou},\ and\ \citenamefont {Xue}}]{zhong2025broadband}%
  \BibitemOpen
  \bibfield  {author} {\bibinfo {author} {\bibfnamefont {W.}~\bibnamefont {Zhong}}, \bibinfo {author} {\bibfnamefont {Y.}~\bibnamefont {Liu}}, \bibinfo {author} {\bibfnamefont {Q.}~\bibnamefont {Yin}}, \bibinfo {author} {\bibfnamefont {R.}~\bibnamefont {Zhao}}, \bibinfo {author} {\bibfnamefont {C.}~\bibnamefont {Wang}}, \bibinfo {author} {\bibfnamefont {W.}~\bibnamefont {Ren}}, \bibinfo {author} {\bibfnamefont {X.}~\bibnamefont {Dou}},\ and\ \bibinfo {author} {\bibfnamefont {X.}~\bibnamefont {Xue}},\ }\href@noop {} {\bibfield  {journal} {\bibinfo  {journal} {Light: Science \& Applications}\ }\textbf {\bibinfo {volume} {14}},\ \bibinfo {pages} {293} (\bibinfo {year} {2025})}\BibitemShut {NoStop}%
\bibitem [{\citenamefont {Nagoro}\ \emph {et~al.}(2025)\citenamefont {Nagoro}, \citenamefont {Sato}, \citenamefont {Tezuka},\ and\ \citenamefont {Horikiri}}]{nagoro2025single}%
  \BibitemOpen
  \bibfield  {author} {\bibinfo {author} {\bibfnamefont {Y.}~\bibnamefont {Nagoro}}, \bibinfo {author} {\bibfnamefont {H.}~\bibnamefont {Sato}}, \bibinfo {author} {\bibfnamefont {H.}~\bibnamefont {Tezuka}},\ and\ \bibinfo {author} {\bibfnamefont {T.}~\bibnamefont {Horikiri}},\ }\href@noop {} {\bibfield  {journal} {\bibinfo  {journal} {Optics Express}\ }\textbf {\bibinfo {volume} {33}},\ \bibinfo {pages} {40997} (\bibinfo {year} {2025})}\BibitemShut {NoStop}%
\bibitem [{\citenamefont {Heshami}\ \emph {et~al.}(2016)\citenamefont {Heshami}, \citenamefont {England}, \citenamefont {Humphreys}, \citenamefont {Bustard}, \citenamefont {Acosta}, \citenamefont {Nunn},\ and\ \citenamefont {Sussman}}]{heshami16}%
  \BibitemOpen
  \bibfield  {author} {\bibinfo {author} {\bibfnamefont {K.}~\bibnamefont {Heshami}}, \bibinfo {author} {\bibfnamefont {D.~G.}\ \bibnamefont {England}}, \bibinfo {author} {\bibfnamefont {P.~C.}\ \bibnamefont {Humphreys}}, \bibinfo {author} {\bibfnamefont {P.~J.}\ \bibnamefont {Bustard}}, \bibinfo {author} {\bibfnamefont {V.~M.}\ \bibnamefont {Acosta}}, \bibinfo {author} {\bibfnamefont {J.}~\bibnamefont {Nunn}},\ and\ \bibinfo {author} {\bibfnamefont {B.~J.}\ \bibnamefont {Sussman}},\ }\href@noop {} {\bibfield  {journal} {\bibinfo  {journal} {Journal of Modern Optics}\ }\textbf {\bibinfo {volume} {63}},\ \bibinfo {pages} {2005} (\bibinfo {year} {2016})}\BibitemShut {NoStop}%
\bibitem [{\citenamefont {Shinbrough}\ \emph {et~al.}(2023)\citenamefont {Shinbrough}, \citenamefont {Pearson}, \citenamefont {Fang}, \citenamefont {Goldschmidt},\ and\ \citenamefont {Lorenz}}]{shinbrough23}%
  \BibitemOpen
  \bibfield  {author} {\bibinfo {author} {\bibfnamefont {K.}~\bibnamefont {Shinbrough}}, \bibinfo {author} {\bibfnamefont {D.~R.}\ \bibnamefont {Pearson}}, \bibinfo {author} {\bibfnamefont {B.}~\bibnamefont {Fang}}, \bibinfo {author} {\bibfnamefont {E.~A.}\ \bibnamefont {Goldschmidt}},\ and\ \bibinfo {author} {\bibfnamefont {V.~O.}\ \bibnamefont {Lorenz}},\ }in\ \href {https://doi.org/10.1016/bs.aamop.2023.04.001} {\emph {\bibinfo {booktitle} {Advances {{In Atomic}}, {{Molecular}}, and {{Optical Physics}}}}},\ Vol.~\bibinfo {volume} {72}\ (\bibinfo  {publisher} {Elsevier},\ \bibinfo {year} {2023})\ pp.\ \bibinfo {pages} {297--360}\BibitemShut {NoStop}%
\bibitem [{\citenamefont {Afzelius}\ \emph {et~al.}(2009)\citenamefont {Afzelius}, \citenamefont {Simon}, \citenamefont {De~Riedmatten},\ and\ \citenamefont {Gisin}}]{afzelius2009multimode}%
  \BibitemOpen
  \bibfield  {author} {\bibinfo {author} {\bibfnamefont {M.}~\bibnamefont {Afzelius}}, \bibinfo {author} {\bibfnamefont {C.}~\bibnamefont {Simon}}, \bibinfo {author} {\bibfnamefont {H.}~\bibnamefont {De~Riedmatten}},\ and\ \bibinfo {author} {\bibfnamefont {N.}~\bibnamefont {Gisin}},\ }\href@noop {} {\bibfield  {journal} {\bibinfo  {journal} {Physical Review A—Atomic, Molecular, and Optical Physics}\ }\textbf {\bibinfo {volume} {79}},\ \bibinfo {pages} {052329} (\bibinfo {year} {2009})}\BibitemShut {NoStop}%
\bibitem [{\citenamefont {Ma}\ \emph {et~al.}(2021)\citenamefont {Ma}, \citenamefont {Ma}, \citenamefont {Zhou}, \citenamefont {Li},\ and\ \citenamefont {Guo}}]{ma2021one}%
  \BibitemOpen
  \bibfield  {author} {\bibinfo {author} {\bibfnamefont {Y.}~\bibnamefont {Ma}}, \bibinfo {author} {\bibfnamefont {Y.-Z.}\ \bibnamefont {Ma}}, \bibinfo {author} {\bibfnamefont {Z.-Q.}\ \bibnamefont {Zhou}}, \bibinfo {author} {\bibfnamefont {C.-F.}\ \bibnamefont {Li}},\ and\ \bibinfo {author} {\bibfnamefont {G.-C.}\ \bibnamefont {Guo}},\ }\href@noop {} {\bibfield  {journal} {\bibinfo  {journal} {Nature Communications}\ }\textbf {\bibinfo {volume} {12}},\ \bibinfo {pages} {2381} (\bibinfo {year} {2021})}\BibitemShut {NoStop}%
\bibitem [{\citenamefont {Mejia}\ \emph {et~al.}(2025)\citenamefont {Mejia}, \citenamefont {Nicolas}, \citenamefont {Rodriguez}, \citenamefont {Goldner},\ and\ \citenamefont {Afzelius}}]{Mejia2025-rf}%
  \BibitemOpen
  \bibfield  {author} {\bibinfo {author} {\bibfnamefont {T.~S.}\ \bibnamefont {Mejia}}, \bibinfo {author} {\bibfnamefont {L.}~\bibnamefont {Nicolas}}, \bibinfo {author} {\bibfnamefont {A.~G.}\ \bibnamefont {Rodriguez}}, \bibinfo {author} {\bibfnamefont {P.}~\bibnamefont {Goldner}},\ and\ \bibinfo {author} {\bibfnamefont {M.}~\bibnamefont {Afzelius}},\ }\href@noop {} {\bibfield  {journal} {\bibinfo  {journal} {arXiv preprint arXiv:2507.13973}\ } (\bibinfo {year} {2025})}\BibitemShut {NoStop}%
\bibitem [{\citenamefont {Reck}\ \emph {et~al.}(1994)\citenamefont {Reck}, \citenamefont {Zeilinger}, \citenamefont {Bernstein},\ and\ \citenamefont {Bertani}}]{Reck1994OpticsUnitary}%
  \BibitemOpen
  \bibfield  {author} {\bibinfo {author} {\bibfnamefont {M.}~\bibnamefont {Reck}}, \bibinfo {author} {\bibfnamefont {A.}~\bibnamefont {Zeilinger}}, \bibinfo {author} {\bibfnamefont {H.~J.}\ \bibnamefont {Bernstein}},\ and\ \bibinfo {author} {\bibfnamefont {P.}~\bibnamefont {Bertani}},\ }\href {https://doi.org/10.1103/PhysRevLett.73.58} {\bibfield  {journal} {\bibinfo  {journal} {Physical Review Letters}\ }\textbf {\bibinfo {volume} {73}},\ \bibinfo {pages} {58} (\bibinfo {year} {1994})}\BibitemShut {NoStop}%
\bibitem [{\citenamefont {Barak}\ and\ \citenamefont {Ben-Aryeh}(2007)}]{Barak2007OpticalFT}%
  \BibitemOpen
  \bibfield  {author} {\bibinfo {author} {\bibfnamefont {R.}~\bibnamefont {Barak}}\ and\ \bibinfo {author} {\bibfnamefont {Y.}~\bibnamefont {Ben-Aryeh}},\ }\href {https://doi.org/10.1364/JOSAB.24.000231} {\bibfield  {journal} {\bibinfo  {journal} {Journal of the Optical Society of America B}\ }\textbf {\bibinfo {volume} {24}},\ \bibinfo {pages} {231} (\bibinfo {year} {2007})}\BibitemShut {NoStop}%
\bibitem [{\citenamefont {Hillerkuss}\ \emph {et~al.}(2010)\citenamefont {Hillerkuss}, \citenamefont {Winter}, \citenamefont {Teschke}, \citenamefont {Marculescu}, \citenamefont {Li}, \citenamefont {Sigurdsson}, \citenamefont {Worms}, \citenamefont {Ezra}, \citenamefont {Narkiss}, \citenamefont {Freude},\ and\ \citenamefont {Leuthold}}]{Hillerkuss2010OpticalFTProcessing}%
  \BibitemOpen
  \bibfield  {author} {\bibinfo {author} {\bibfnamefont {D.}~\bibnamefont {Hillerkuss}}, \bibinfo {author} {\bibfnamefont {M.}~\bibnamefont {Winter}}, \bibinfo {author} {\bibfnamefont {M.}~\bibnamefont {Teschke}}, \bibinfo {author} {\bibfnamefont {A.}~\bibnamefont {Marculescu}}, \bibinfo {author} {\bibfnamefont {J.}~\bibnamefont {Li}}, \bibinfo {author} {\bibfnamefont {G.}~\bibnamefont {Sigurdsson}}, \bibinfo {author} {\bibfnamefont {K.}~\bibnamefont {Worms}}, \bibinfo {author} {\bibfnamefont {S.~B.}\ \bibnamefont {Ezra}}, \bibinfo {author} {\bibfnamefont {N.}~\bibnamefont {Narkiss}}, \bibinfo {author} {\bibfnamefont {W.}~\bibnamefont {Freude}},\ and\ \bibinfo {author} {\bibfnamefont {J.}~\bibnamefont {Leuthold}},\ }\href {https://doi.org/10.1364/OE.18.009324} {\bibfield  {journal} {\bibinfo  {journal} {Optics Express}\ }\textbf {\bibinfo {volume} {18}},\ \bibinfo {pages} {9324} (\bibinfo {year} {2010})}\BibitemShut {NoStop}%
\bibitem [{\citenamefont {Guo}\ \emph {et~al.}(2025)\citenamefont {Guo}, \citenamefont {Banerji}, \citenamefont {Chin}, \citenamefont {Chowdhury},\ and\ \citenamefont {Ling}}]{Guo2025HighlyEfficientBroadbandOptical}%
  \BibitemOpen
  \bibfield  {author} {\bibinfo {author} {\bibfnamefont {Y.}~\bibnamefont {Guo}}, \bibinfo {author} {\bibfnamefont {A.}~\bibnamefont {Banerji}}, \bibinfo {author} {\bibfnamefont {J.~B.}\ \bibnamefont {Chin}}, \bibinfo {author} {\bibfnamefont {A.}~\bibnamefont {Chowdhury}},\ and\ \bibinfo {author} {\bibfnamefont {A.}~\bibnamefont {Ling}},\ }\href@noop {} {\bibinfo {title} {Highly efficient and broadband optical delay line towards a quantum memory}} (\bibinfo {year} {2025})\BibitemShut {NoStop}%
\bibitem [{\citenamefont {Duan}\ and\ \citenamefont {Kimble}(2004)}]{Duan2004PhotonicGate}%
  \BibitemOpen
  \bibfield  {author} {\bibinfo {author} {\bibfnamefont {L.-M.}\ \bibnamefont {Duan}}\ and\ \bibinfo {author} {\bibfnamefont {H.~J.}\ \bibnamefont {Kimble}},\ }\href {https://doi.org/10.1103/PhysRevLett.92.127902} {\bibfield  {journal} {\bibinfo  {journal} {Physical Review Letters}\ }\textbf {\bibinfo {volume} {92}},\ \bibinfo {pages} {127902} (\bibinfo {year} {2004})}\BibitemShut {NoStop}%
\bibitem [{\citenamefont {{Jackman}}\ \emph {et~al.}(2023)\citenamefont {{Jackman}}, \citenamefont {{Shkolnik}}, \citenamefont {{Million}}, \citenamefont {{Fleming}}, \citenamefont {{Richey-Yowell}},\ and\ \citenamefont {{Loyd}}}]{Jackman2023}%
  \BibitemOpen
  \bibfield  {author} {\bibinfo {author} {\bibfnamefont {J.~A.~G.}\ \bibnamefont {{Jackman}}}, \bibinfo {author} {\bibfnamefont {E.~L.}\ \bibnamefont {{Shkolnik}}}, \bibinfo {author} {\bibfnamefont {C.}~\bibnamefont {{Million}}}, \bibinfo {author} {\bibfnamefont {S.}~\bibnamefont {{Fleming}}}, \bibinfo {author} {\bibfnamefont {T.}~\bibnamefont {{Richey-Yowell}}},\ and\ \bibinfo {author} {\bibfnamefont {R.~O.~P.}\ \bibnamefont {{Loyd}}},\ }\href {https://doi.org/10.1093/mnras/stac3135} {\bibfield  {journal} {\bibinfo  {journal} {Monthly Notices of the Royal Astronomical Society}\ }\textbf {\bibinfo {volume} {519}},\ \bibinfo {pages} {3564} (\bibinfo {year} {2023})},\ \Eprint {https://arxiv.org/abs/2210.15688} {arXiv:2210.15688 [astro-ph.SR]} \BibitemShut {NoStop}%
\bibitem [{\citenamefont {{Howard}}\ \emph {et~al.}(2020)\citenamefont {{Howard}}, \citenamefont {{Corbett}}, \citenamefont {{Law}}, \citenamefont {{Ratzloff}}, \citenamefont {{Galliher}}, \citenamefont {{Glazier}}, \citenamefont {{Gonzalez}}, \citenamefont {{Vasquez Soto}}, \citenamefont {{Fors}}, \citenamefont {{del Ser}},\ and\ \citenamefont {{Haislip}}}]{Howard2020}%
  \BibitemOpen
  \bibfield  {author} {\bibinfo {author} {\bibfnamefont {W.~S.}\ \bibnamefont {{Howard}}}, \bibinfo {author} {\bibfnamefont {H.}~\bibnamefont {{Corbett}}}, \bibinfo {author} {\bibfnamefont {N.~M.}\ \bibnamefont {{Law}}}, \bibinfo {author} {\bibfnamefont {J.~K.}\ \bibnamefont {{Ratzloff}}}, \bibinfo {author} {\bibfnamefont {N.}~\bibnamefont {{Galliher}}}, \bibinfo {author} {\bibfnamefont {A.~L.}\ \bibnamefont {{Glazier}}}, \bibinfo {author} {\bibfnamefont {R.}~\bibnamefont {{Gonzalez}}}, \bibinfo {author} {\bibfnamefont {A.}~\bibnamefont {{Vasquez Soto}}}, \bibinfo {author} {\bibfnamefont {O.}~\bibnamefont {{Fors}}}, \bibinfo {author} {\bibfnamefont {D.}~\bibnamefont {{del Ser}}},\ and\ \bibinfo {author} {\bibfnamefont {J.}~\bibnamefont {{Haislip}}},\ }\href {https://doi.org/10.3847/1538-4357/abb5b4} {\bibfield  {journal} {\bibinfo  {journal} {The Astrophysical Journal}\ }\textbf {\bibinfo {volume} {902}},\ \bibinfo {eid} {115} (\bibinfo {year} {2020})},\ \Eprint {https://arxiv.org/abs/2010.00604}
  {arXiv:2010.00604 [astro-ph.SR]} \BibitemShut {NoStop}%
\bibitem [{Note2()}]{Note2}%
  \BibitemOpen
  \bibinfo {note} {The ‘5’ in M5V indicates a spectral subclass, specifying the star’s temperature within the M-type sequence (a lower number indicates a hotter star). The ‘V’ signifies luminosity class five, denoting a main-sequence star that generates energy through hydrogen fusion in its core.}\BibitemShut {Stop}%
\bibitem [{\citenamefont {{Paudel}}\ \emph {et~al.}(2024)\citenamefont {{Paudel}}, \citenamefont {{Barclay}}, \citenamefont {{Youngblood}}, \citenamefont {{Quintana}}, \citenamefont {{Schlieder}}, \citenamefont {{Vega}}, \citenamefont {{Gilbert}}, \citenamefont {{Osten}}, \citenamefont {{Peacock}}, \citenamefont {{Tristan}}, \citenamefont {{Feliz}}, \citenamefont {{Boyd}}, \citenamefont {{Davenport}}, \citenamefont {{Huber}}, \citenamefont {{Kowalski}}, \citenamefont {{Monsue}},\ and\ \citenamefont {{Silverstein}}}]{Paudel2024}%
  \BibitemOpen
  \bibfield  {author} {\bibinfo {author} {\bibfnamefont {R.~R.}\ \bibnamefont {{Paudel}}}, \bibinfo {author} {\bibfnamefont {T.}~\bibnamefont {{Barclay}}}, \bibinfo {author} {\bibfnamefont {A.}~\bibnamefont {{Youngblood}}}, \bibinfo {author} {\bibfnamefont {E.~V.}\ \bibnamefont {{Quintana}}}, \bibinfo {author} {\bibfnamefont {J.~E.}\ \bibnamefont {{Schlieder}}}, \bibinfo {author} {\bibfnamefont {L.~D.}\ \bibnamefont {{Vega}}}, \bibinfo {author} {\bibfnamefont {E.~A.}\ \bibnamefont {{Gilbert}}}, \bibinfo {author} {\bibfnamefont {R.~A.}\ \bibnamefont {{Osten}}}, \bibinfo {author} {\bibfnamefont {S.}~\bibnamefont {{Peacock}}}, \bibinfo {author} {\bibfnamefont {I.~I.}\ \bibnamefont {{Tristan}}}, \bibinfo {author} {\bibfnamefont {D.~L.}\ \bibnamefont {{Feliz}}}, \bibinfo {author} {\bibfnamefont {P.~T.}\ \bibnamefont {{Boyd}}}, \bibinfo {author} {\bibfnamefont {J.~R.~A.}\ \bibnamefont {{Davenport}}}, \bibinfo {author} {\bibfnamefont {D.}~\bibnamefont {{Huber}}}, \bibinfo {author} {\bibfnamefont {A.~F.}\
  \bibnamefont {{Kowalski}}}, \bibinfo {author} {\bibfnamefont {T.}~\bibnamefont {{Monsue}}},\ and\ \bibinfo {author} {\bibfnamefont {M.~L.}\ \bibnamefont {{Silverstein}}},\ }\href {https://doi.org/10.3847/1538-4357/ad487d} {\bibfield  {journal} {\bibinfo  {journal} {The Astrophysical Journal}\ }\textbf {\bibinfo {volume} {971}},\ \bibinfo {eid} {24} (\bibinfo {year} {2024})},\ \Eprint {https://arxiv.org/abs/2404.12310} {arXiv:2404.12310 [astro-ph.SR]} \BibitemShut {NoStop}%
\bibitem [{\citenamefont {{Saha}}\ \emph {et~al.}(2019)\citenamefont {{Saha}}, \citenamefont {{Vivas}}, \citenamefont {{Olszewski}}, \citenamefont {{Smith}}, \citenamefont {{Olsen}}, \citenamefont {{Blum}}, \citenamefont {{Valdes}}, \citenamefont {{Claver}}, \citenamefont {{Calamida}}, \citenamefont {{Walker}}, \citenamefont {{Matheson}}, \citenamefont {{Narayan}}, \citenamefont {{Soraisam}}, \citenamefont {{Cunha}}, \citenamefont {{Axelrod}}, \citenamefont {{Bloom}}, \citenamefont {{Cenko}}, \citenamefont {{Frye}}, \citenamefont {{Juric}}, \citenamefont {{Kaleida}}, \citenamefont {{Kunder}}, \citenamefont {{Miller}}, \citenamefont {{Nidever}},\ and\ \citenamefont {{Ridgway}}}]{Saha2019}%
  \BibitemOpen
  \bibfield  {author} {\bibinfo {author} {\bibfnamefont {A.}~\bibnamefont {{Saha}}}, \bibinfo {author} {\bibfnamefont {A.~K.}\ \bibnamefont {{Vivas}}}, \bibinfo {author} {\bibfnamefont {E.~W.}\ \bibnamefont {{Olszewski}}}, \bibinfo {author} {\bibfnamefont {V.}~\bibnamefont {{Smith}}}, \bibinfo {author} {\bibfnamefont {K.}~\bibnamefont {{Olsen}}}, \bibinfo {author} {\bibfnamefont {R.}~\bibnamefont {{Blum}}}, \bibinfo {author} {\bibfnamefont {F.}~\bibnamefont {{Valdes}}}, \bibinfo {author} {\bibfnamefont {J.}~\bibnamefont {{Claver}}}, \bibinfo {author} {\bibfnamefont {A.}~\bibnamefont {{Calamida}}}, \bibinfo {author} {\bibfnamefont {A.~R.}\ \bibnamefont {{Walker}}}, \bibinfo {author} {\bibfnamefont {T.}~\bibnamefont {{Matheson}}}, \bibinfo {author} {\bibfnamefont {G.}~\bibnamefont {{Narayan}}}, \bibinfo {author} {\bibfnamefont {M.}~\bibnamefont {{Soraisam}}}, \bibinfo {author} {\bibfnamefont {K.}~\bibnamefont {{Cunha}}}, \bibinfo {author} {\bibfnamefont {T.}~\bibnamefont {{Axelrod}}}, \bibinfo {author}
  {\bibfnamefont {J.~S.}\ \bibnamefont {{Bloom}}}, \bibinfo {author} {\bibfnamefont {S.~B.}\ \bibnamefont {{Cenko}}}, \bibinfo {author} {\bibfnamefont {B.}~\bibnamefont {{Frye}}}, \bibinfo {author} {\bibfnamefont {M.}~\bibnamefont {{Juric}}}, \bibinfo {author} {\bibfnamefont {C.}~\bibnamefont {{Kaleida}}}, \bibinfo {author} {\bibfnamefont {A.}~\bibnamefont {{Kunder}}}, \bibinfo {author} {\bibfnamefont {A.}~\bibnamefont {{Miller}}}, \bibinfo {author} {\bibfnamefont {D.}~\bibnamefont {{Nidever}}},\ and\ \bibinfo {author} {\bibfnamefont {S.}~\bibnamefont {{Ridgway}}},\ }\href {https://doi.org/10.3847/1538-4357/ab07ba} {\bibfield  {journal} {\bibinfo  {journal} {The Astrophysical Journal}\ }\textbf {\bibinfo {volume} {874}},\ \bibinfo {eid} {30} (\bibinfo {year} {2019})},\ \Eprint {https://arxiv.org/abs/1902.05637} {arXiv:1902.05637 [astro-ph.GA]} \BibitemShut {NoStop}%
\bibitem [{\citenamefont {{Stanek}}(1996)}]{Stanek1996}%
  \BibitemOpen
  \bibfield  {author} {\bibinfo {author} {\bibfnamefont {K.~Z.}\ \bibnamefont {{Stanek}}},\ }\href {https://doi.org/10.1086/309976} {\bibfield  {journal} {\bibinfo  {journal} {\apjl}\ }\textbf {\bibinfo {volume} {460}},\ \bibinfo {pages} {L37} (\bibinfo {year} {1996})},\ \Eprint {https://arxiv.org/abs/astro-ph/9512137} {arXiv:astro-ph/9512137 [astro-ph]} \BibitemShut {NoStop}%
\bibitem [{Gem(2024)}]{GeminiObservatory2024}%
  \BibitemOpen
  \href {https://www.gemini.edu/observing/telescopes-and-sites/sites#MK%20optical%20extinction%20curve} {} (\bibinfo {year} {2024})\BibitemShut {NoStop}%
\bibitem [{\citenamefont {{Padovani}}\ and\ \citenamefont {{Cirasuolo}}(2023)}]{Padovani2023}%
  \BibitemOpen
  \bibfield  {author} {\bibinfo {author} {\bibfnamefont {P.}~\bibnamefont {{Padovani}}}\ and\ \bibinfo {author} {\bibfnamefont {M.}~\bibnamefont {{Cirasuolo}}},\ }\href {https://doi.org/10.1080/00107514.2023.2266921} {\bibfield  {journal} {\bibinfo  {journal} {Contemporary Physics}\ }\textbf {\bibinfo {volume} {64}},\ \bibinfo {pages} {47} (\bibinfo {year} {2023})},\ \Eprint {https://arxiv.org/abs/2312.04299} {arXiv:2312.04299 [astro-ph.IM]} \BibitemShut {NoStop}%
\bibitem [{\citenamefont {{Carr}}\ and\ \citenamefont {{Hawking}}(1974)}]{Carr1974}%
  \BibitemOpen
  \bibfield  {author} {\bibinfo {author} {\bibfnamefont {B.~J.}\ \bibnamefont {{Carr}}}\ and\ \bibinfo {author} {\bibfnamefont {S.~W.}\ \bibnamefont {{Hawking}}},\ }\href {https://doi.org/10.1093/mnras/168.2.399} {\bibfield  {journal} {\bibinfo  {journal} {\mnras}\ }\textbf {\bibinfo {volume} {168}},\ \bibinfo {pages} {399} (\bibinfo {year} {1974})}\BibitemShut {NoStop}%
\bibitem [{\citenamefont {Carr}\ and\ \citenamefont {Kühnel}(2020)}]{Carr2020}%
  \BibitemOpen
  \bibfield  {author} {\bibinfo {author} {\bibfnamefont {B.}~\bibnamefont {Carr}}\ and\ \bibinfo {author} {\bibfnamefont {F.}~\bibnamefont {Kühnel}},\ }\href {https://doi.org/10.1146/annurev-nucl-050520-125911} {\bibfield  {journal} {\bibinfo  {journal} {Annual Review of Nuclear and Particle Science}\ }\textbf {\bibinfo {volume} {70}},\ \bibinfo {pages} {355–394} (\bibinfo {year} {2020})}\BibitemShut {NoStop}%
\bibitem [{\citenamefont {Mathis}\ \emph {et~al.}(1977)\citenamefont {Mathis}, \citenamefont {Rumpl},\ and\ \citenamefont {Nordsieck}}]{mathis1977size}%
  \BibitemOpen
  \bibfield  {author} {\bibinfo {author} {\bibfnamefont {J.~S.}\ \bibnamefont {Mathis}}, \bibinfo {author} {\bibfnamefont {W.}~\bibnamefont {Rumpl}},\ and\ \bibinfo {author} {\bibfnamefont {K.~H.}\ \bibnamefont {Nordsieck}},\ }\href@noop {} {\bibfield  {journal} {\bibinfo  {journal} {Astrophysical Journal, Part 1, vol. 217, Oct. 15, 1977, p. 425-433. NSF-supported research.}\ }\textbf {\bibinfo {volume} {217}},\ \bibinfo {pages} {425} (\bibinfo {year} {1977})}\BibitemShut {NoStop}%
\bibitem [{\citenamefont {{Armstrong}}\ \emph {et~al.}(1995)\citenamefont {{Armstrong}}, \citenamefont {{Rickett}},\ and\ \citenamefont {{Spangler}}}]{Armstrong1995DMVariation}%
  \BibitemOpen
  \bibfield  {author} {\bibinfo {author} {\bibfnamefont {J.~W.}\ \bibnamefont {{Armstrong}}}, \bibinfo {author} {\bibfnamefont {B.~J.}\ \bibnamefont {{Rickett}}},\ and\ \bibinfo {author} {\bibfnamefont {S.~R.}\ \bibnamefont {{Spangler}}},\ }\href {https://doi.org/10.1086/175515} {\bibfield  {journal} {\bibinfo  {journal} {The Astrophysical Journal}\ }\textbf {\bibinfo {volume} {443}},\ \bibinfo {pages} {209} (\bibinfo {year} {1995})}\BibitemShut {NoStop}%
\bibitem [{\citenamefont {{You}}\ \emph {et~al.}(2007)\citenamefont {{You}}, \citenamefont {{Hobbs}}, \citenamefont {{Coles}}, \citenamefont {{Manchester}}, \citenamefont {{Edwards}}, \citenamefont {{Bailes}}, \citenamefont {{Sarkissian}}, \citenamefont {{Verbiest}}, \citenamefont {{van Straten}}, \citenamefont {{Hotan}}, \citenamefont {{Ord}}, \citenamefont {{Jenet}}, \citenamefont {{Bhat}},\ and\ \citenamefont {{Teoh}}}]{You2007DMVariation}%
  \BibitemOpen
  \bibfield  {author} {\bibinfo {author} {\bibfnamefont {X.~P.}\ \bibnamefont {{You}}}, \bibinfo {author} {\bibfnamefont {G.}~\bibnamefont {{Hobbs}}}, \bibinfo {author} {\bibfnamefont {W.~A.}\ \bibnamefont {{Coles}}}, \bibinfo {author} {\bibfnamefont {R.~N.}\ \bibnamefont {{Manchester}}}, \bibinfo {author} {\bibfnamefont {R.}~\bibnamefont {{Edwards}}}, \bibinfo {author} {\bibfnamefont {M.}~\bibnamefont {{Bailes}}}, \bibinfo {author} {\bibfnamefont {J.}~\bibnamefont {{Sarkissian}}}, \bibinfo {author} {\bibfnamefont {J.~P.~W.}\ \bibnamefont {{Verbiest}}}, \bibinfo {author} {\bibfnamefont {W.}~\bibnamefont {{van Straten}}}, \bibinfo {author} {\bibfnamefont {A.}~\bibnamefont {{Hotan}}}, \bibinfo {author} {\bibfnamefont {S.}~\bibnamefont {{Ord}}}, \bibinfo {author} {\bibfnamefont {F.}~\bibnamefont {{Jenet}}}, \bibinfo {author} {\bibfnamefont {N.~D.~R.}\ \bibnamefont {{Bhat}}},\ and\ \bibinfo {author} {\bibfnamefont {A.}~\bibnamefont {{Teoh}}},\ }\href {https://doi.org/10.1111/j.1365-2966.2007.11617.x} {\bibfield
  {journal} {\bibinfo  {journal} {\mnras}\ }\textbf {\bibinfo {volume} {378}},\ \bibinfo {pages} {493} (\bibinfo {year} {2007})},\ \Eprint {https://arxiv.org/abs/astro-ph/0702366} {arXiv:astro-ph/0702366 [astro-ph]} \BibitemShut {NoStop}%
\bibitem [{\citenamefont {Donner}\ \emph {et~al.}(2020)\citenamefont {Donner}, \citenamefont {Verbiest}, \citenamefont {Tiburzi}, \citenamefont {Osłowski}, \citenamefont {Künsemöller}, \citenamefont {Bak~Nielsen}, \citenamefont {Grießmeier}, \citenamefont {Serylak}, \citenamefont {Kramer}, \citenamefont {Anderson}, \citenamefont {Wucknitz}, \citenamefont {Keane}, \citenamefont {Kondratiev}, \citenamefont {Sobey}, \citenamefont {McKee}, \citenamefont {Bilous}, \citenamefont {Breton}, \citenamefont {Brüggen}, \citenamefont {Ciardi}, \citenamefont {Hoeft}, \citenamefont {van Leeuwen},\ and\ \citenamefont {Vocks}}]{Donner_2020}%
  \BibitemOpen
  \bibfield  {author} {\bibinfo {author} {\bibfnamefont {J.~Y.}\ \bibnamefont {Donner}}, \bibinfo {author} {\bibfnamefont {J.~P.~W.}\ \bibnamefont {Verbiest}}, \bibinfo {author} {\bibfnamefont {C.}~\bibnamefont {Tiburzi}}, \bibinfo {author} {\bibfnamefont {S.}~\bibnamefont {Osłowski}}, \bibinfo {author} {\bibfnamefont {J.}~\bibnamefont {Künsemöller}}, \bibinfo {author} {\bibfnamefont {A.-S.}\ \bibnamefont {Bak~Nielsen}}, \bibinfo {author} {\bibfnamefont {J.-M.}\ \bibnamefont {Grießmeier}}, \bibinfo {author} {\bibfnamefont {M.}~\bibnamefont {Serylak}}, \bibinfo {author} {\bibfnamefont {M.}~\bibnamefont {Kramer}}, \bibinfo {author} {\bibfnamefont {J.~M.}\ \bibnamefont {Anderson}}, \bibinfo {author} {\bibfnamefont {O.}~\bibnamefont {Wucknitz}}, \bibinfo {author} {\bibfnamefont {E.}~\bibnamefont {Keane}}, \bibinfo {author} {\bibfnamefont {V.}~\bibnamefont {Kondratiev}}, \bibinfo {author} {\bibfnamefont {C.}~\bibnamefont {Sobey}}, \bibinfo {author} {\bibfnamefont {J.~W.}\ \bibnamefont {McKee}}, \bibinfo
  {author} {\bibfnamefont {A.~V.}\ \bibnamefont {Bilous}}, \bibinfo {author} {\bibfnamefont {R.~P.}\ \bibnamefont {Breton}}, \bibinfo {author} {\bibfnamefont {M.}~\bibnamefont {Brüggen}}, \bibinfo {author} {\bibfnamefont {B.}~\bibnamefont {Ciardi}}, \bibinfo {author} {\bibfnamefont {M.}~\bibnamefont {Hoeft}}, \bibinfo {author} {\bibfnamefont {J.}~\bibnamefont {van Leeuwen}},\ and\ \bibinfo {author} {\bibfnamefont {C.}~\bibnamefont {Vocks}},\ }\href {https://doi.org/10.1051/0004-6361/202039517} {\bibfield  {journal} {\bibinfo  {journal} {Astronomy \& Astrophysics}\ }\textbf {\bibinfo {volume} {644}},\ \bibinfo {pages} {A153} (\bibinfo {year} {2020})}\BibitemShut {NoStop}%
\bibitem [{\citenamefont {Roddier}(1981)}]{roddier1981v}%
  \BibitemOpen
  \bibfield  {author} {\bibinfo {author} {\bibfnamefont {F.}~\bibnamefont {Roddier}},\ }in\ \href@noop {} {\emph {\bibinfo {booktitle} {Progress in optics}}},\ Vol.~\bibinfo {volume} {19}\ (\bibinfo  {publisher} {Elsevier},\ \bibinfo {year} {1981})\ pp.\ \bibinfo {pages} {281--376}\BibitemShut {NoStop}%
\bibitem [{\citenamefont {Kellerer}\ and\ \citenamefont {Tokovinin}(2007)}]{kellerer2007atmospheric}%
  \BibitemOpen
  \bibfield  {author} {\bibinfo {author} {\bibfnamefont {A.}~\bibnamefont {Kellerer}}\ and\ \bibinfo {author} {\bibfnamefont {A.}~\bibnamefont {Tokovinin}},\ }\href@noop {} {\bibfield  {journal} {\bibinfo  {journal} {Astronomy \& Astrophysics}\ }\textbf {\bibinfo {volume} {461}},\ \bibinfo {pages} {775} (\bibinfo {year} {2007})}\BibitemShut {NoStop}%
\bibitem [{\citenamefont {Osborn}\ \emph {et~al.}(2018)\citenamefont {Osborn}, \citenamefont {Wilson}, \citenamefont {Sarazin}, \citenamefont {Butterley}, \citenamefont {Chac{\'o}n}, \citenamefont {Derie}, \citenamefont {Farley}, \citenamefont {Haubois}, \citenamefont {Laidlaw}, \citenamefont {LeLouarn} \emph {et~al.}}]{osborn2018optical}%
  \BibitemOpen
  \bibfield  {author} {\bibinfo {author} {\bibfnamefont {J.}~\bibnamefont {Osborn}}, \bibinfo {author} {\bibfnamefont {R.}~\bibnamefont {Wilson}}, \bibinfo {author} {\bibfnamefont {M.}~\bibnamefont {Sarazin}}, \bibinfo {author} {\bibfnamefont {T.}~\bibnamefont {Butterley}}, \bibinfo {author} {\bibfnamefont {A.}~\bibnamefont {Chac{\'o}n}}, \bibinfo {author} {\bibfnamefont {F.}~\bibnamefont {Derie}}, \bibinfo {author} {\bibfnamefont {O.}~\bibnamefont {Farley}}, \bibinfo {author} {\bibfnamefont {X.}~\bibnamefont {Haubois}}, \bibinfo {author} {\bibfnamefont {D.}~\bibnamefont {Laidlaw}}, \bibinfo {author} {\bibfnamefont {M.}~\bibnamefont {LeLouarn}}, \emph {et~al.},\ }\href@noop {} {\bibfield  {journal} {\bibinfo  {journal} {Monthly Notices of the Royal Astronomical Society}\ }\textbf {\bibinfo {volume} {478}},\ \bibinfo {pages} {825} (\bibinfo {year} {2018})}\BibitemShut {NoStop}%
\bibitem [{\citenamefont {Zernike}(1938)}]{Zernike1938}%
  \BibitemOpen
  \bibfield  {author} {\bibinfo {author} {\bibfnamefont {F.}~\bibnamefont {Zernike}},\ }\href {https://doi.org/10.1016/S0031-8914(38)80203-2} {\bibfield  {journal} {\bibinfo  {journal} {Physica}\ }\textbf {\bibinfo {volume} {5}},\ \bibinfo {pages} {785 } (\bibinfo {year} {1938})}\BibitemShut {NoStop}%
\bibitem [{\citenamefont {Zhou}\ \emph {et~al.}(2000)\citenamefont {Zhou}, \citenamefont {Leung},\ and\ \citenamefont {Chuang}}]{Zhou2000GateConstruction}%
  \BibitemOpen
  \bibfield  {author} {\bibinfo {author} {\bibfnamefont {X.}~\bibnamefont {Zhou}}, \bibinfo {author} {\bibfnamefont {D.~W.}\ \bibnamefont {Leung}},\ and\ \bibinfo {author} {\bibfnamefont {I.~L.}\ \bibnamefont {Chuang}},\ }\href {https://doi.org/10.1103/PhysRevA.62.052316} {\bibfield  {journal} {\bibinfo  {journal} {Phys. Rev. A}\ }\textbf {\bibinfo {volume} {62}},\ \bibinfo {pages} {052316} (\bibinfo {year} {2000})}\BibitemShut {NoStop}%
\bibitem [{\citenamefont {Mie}(1908)}]{mie1908beitrage}%
  \BibitemOpen
  \bibfield  {author} {\bibinfo {author} {\bibfnamefont {G.}~\bibnamefont {Mie}},\ }\href@noop {} {\bibfield  {journal} {\bibinfo  {journal} {Annalen der physik}\ }\textbf {\bibinfo {volume} {330}},\ \bibinfo {pages} {377} (\bibinfo {year} {1908})}\BibitemShut {NoStop}%
\bibitem [{\citenamefont {Li}(2008)}]{li2008optical}%
  \BibitemOpen
  \bibfield  {author} {\bibinfo {author} {\bibfnamefont {A.}~\bibnamefont {Li}},\ }in\ \href@noop {} {\emph {\bibinfo {booktitle} {Small bodies in planetary systems}}}\ (\bibinfo  {publisher} {Springer},\ \bibinfo {year} {2008})\ pp.\ \bibinfo {pages} {1--22}\BibitemShut {NoStop}%
\bibitem [{\citenamefont {Chafa{\"\i}}\ \emph {et~al.}(2012)\citenamefont {Chafa{\"\i}}, \citenamefont {Gu{\'e}don}, \citenamefont {Lecu{\'e}},\ and\ \citenamefont {Pajor}}]{chafai2012interactions}%
  \BibitemOpen
  \bibfield  {author} {\bibinfo {author} {\bibfnamefont {D.}~\bibnamefont {Chafa{\"\i}}}, \bibinfo {author} {\bibfnamefont {O.}~\bibnamefont {Gu{\'e}don}}, \bibinfo {author} {\bibfnamefont {G.}~\bibnamefont {Lecu{\'e}}},\ and\ \bibinfo {author} {\bibfnamefont {A.}~\bibnamefont {Pajor}},\ }\href@noop {} {\emph {\bibinfo {title} {Interactions between compressed sensing random matrices and high dimensional geometry}}},\ Vol.~\bibinfo {volume} {37}\ (\bibinfo  {publisher} {Soci{\'e}t{\'e} Math{\'e}matique de France Paris},\ \bibinfo {year} {2012})\BibitemShut {NoStop}%
\bibitem [{\citenamefont {Mr{\'o}z}\ \emph {et~al.}(2025)\citenamefont {Mr{\'o}z}, \citenamefont {Dong}, \citenamefont {M{\'e}rand}, \citenamefont {Shangguan}, \citenamefont {Woillez}, \citenamefont {Gould}, \citenamefont {Udalski}, \citenamefont {Eisenhauer}, \citenamefont {Ryu}, \citenamefont {Wu} \emph {et~al.}}]{mroz2025observations}%
  \BibitemOpen
  \bibfield  {author} {\bibinfo {author} {\bibfnamefont {P.}~\bibnamefont {Mr{\'o}z}}, \bibinfo {author} {\bibfnamefont {S.}~\bibnamefont {Dong}}, \bibinfo {author} {\bibfnamefont {A.}~\bibnamefont {M{\'e}rand}}, \bibinfo {author} {\bibfnamefont {J.}~\bibnamefont {Shangguan}}, \bibinfo {author} {\bibfnamefont {J.}~\bibnamefont {Woillez}}, \bibinfo {author} {\bibfnamefont {A.}~\bibnamefont {Gould}}, \bibinfo {author} {\bibfnamefont {A.}~\bibnamefont {Udalski}}, \bibinfo {author} {\bibfnamefont {F.}~\bibnamefont {Eisenhauer}}, \bibinfo {author} {\bibfnamefont {Y.-H.}\ \bibnamefont {Ryu}}, \bibinfo {author} {\bibfnamefont {Z.}~\bibnamefont {Wu}}, \emph {et~al.},\ }\href@noop {} {\bibfield  {journal} {\bibinfo  {journal} {The Astrophysical Journal}\ }\textbf {\bibinfo {volume} {980}},\ \bibinfo {pages} {47} (\bibinfo {year} {2025})}\BibitemShut {NoStop}%
\end{thebibliography}%
\end{document}